\documentclass[preprint,11pt,authoryear]{elsarticle}
\usepackage{amsmath}
\usepackage{mathrsfs}
\usepackage{amsfonts}
\usepackage{amssymb}
\usepackage{amsthm}
\usepackage{mathrsfs}
\usepackage{bbm}
\usepackage{latexsym}
\usepackage{graphicx,psfrag,epsf}
\usepackage{enumerate}
\usepackage{color}
\usepackage{colortbl}
\usepackage{natbib}
\usepackage{multirow}
\usepackage{hyperref}
\usepackage{hypernat}

\newcommand{\toDistSt}{\overset{st}{\longrightarrow}}
\newcommand{\eqDist}{\stackrel{\mathcal{D}}{=}}

 \newtheorem{theorem}{Theorem} 
 \newtheorem{lemma}[theorem]{Lemma} 
 \newdefinition{remark}{Remark} 
 \newdefinition{proposition}{Proposition}
\newdefinition{assumption}{Assumption} 

\def\cF{\mathcal{F}}

\def\cN{\mathcal{N}}

\def\bE{\mathbb{E}}

\def\bP{\mathbb{P}}

\def\bR{\mathbb{R}}

\def\sF{\mathscr{F}}

\newcommand{\sgn}{{\rm sgn}}
\newcommand{\wh}{\widehat}
\newcommand{\wt}{\widetilde}

\newsavebox\CBox 
\newcommand{\best}[1]{\sbox\CBox{#1}\resizebox{1.05\wd\CBox}{\ht\CBox}{\textbf{#1}}}

\journal{}

\begin{document}
\allowdisplaybreaks
\begin{frontmatter}

 \title{Efficient Integrated Volatility Estimation  in the  Presence of  Infinite Variation Jumps via Debiased Truncated Realized Variations\tnoteref{t1,t2}} \tnotetext[t1]{The second author's research is partially supported by the NSF Grants: DMS-2015323, DMS-1613016.} \tnotetext[t2]{An earlier version of this manuscript was circulated under the title ``Efficient Volatility Estimation for L\'evy Processes with Jumps of Unbounded Variation". The results therein were constrained to L\'evy processes, whereas here we consider a much larger class of It\^o semimartingales.}

\author[1]{B. Cooper Boniece}
\ead{cooper.boniece@drexel.edu}

\author[2]{Jos\'e E. Figueroa-L\'opez\corref{cor1}%
}
\ead{figueroa-lopez@wustl.edu}

\author[2]{Yuchen Han}
\ead{y.han@wustl.edu}

 \cortext[cor1]{Corresponding author}

 \affiliation[1]{organization={Drexel University}, 
                 addressline={Department of Mathematics},
                 postcode={19104, PA}, 
                 city={Philadelphia}, 
                 country={USA}}

 \affiliation[2]{organization={Washington University in St. Louis},
                 addressline={Department of Statistics and Data Science}, 
                 city={St. Louis},
                 postcode={63130, MO}, 
                 country={USA}}


\begin{abstract}
Statistical inference for stochastic processes based on high frequency observations has been an active research area for more than two decades. One of the most well-known and widely studied problems has been the estimation of the quadratic variation of the continuous component of an It\^o semimartingale with jumps. Several rate- and variance-efficient estimators have been proposed in the literature when the jump component is of bounded variation. However, to date, very few methods can deal with jumps of unbounded variation. By developing new high-order expansions of the truncated moments of a locally stable L\'evy process, we propose a new rate- and variance-efficient volatility estimator for a class of It\^o semimartingales whose jumps behave locally like those of a stable L\'evy process with  Blumenthal-Getoor index $Y\in (1,8/5)$ (hence,  of unbounded variation). The proposed method is based on a two-step debiasing procedure for the truncated realized quadratic variation of the process and can also cover the case $Y<1$. Our Monte Carlo experiments indicate that the method outperforms other efficient alternatives in the literature in the setting covered by our theoretical framework.\end{abstract}


\end{frontmatter}

\section{Introduction}\label{s:1}
Statistical inference for stochastic processes based on high-frequency observations has attracted considerable attention in the literature for more than two decades.  Among the many problems studied to date, arguably none has received more attention than that of the estimation of the continuous (or predictable) quadratic variation of an It\^o semimartingale $X=\{X_{t}\}_{t\geq{}0}$. Specifically, if 
\begin{align}\label{eq:MnMdlX0}
	X_t:=X_0+ X^{c}_t+X^{j}_t:= X_0+\int_0^t b_s ds+\int_0^t \sigma_s dW_s+X^{j}_t,\quad t\in[0,T],
\end{align}
where $X_0\in\bR$, $W=\{W_t\}_{t\geq{}0}$ is a Wiener process and $X^{j}=\{X^{j}_{t}\}_{t\geq{}0}$ is a pure-jump It\^o semimartingale, our estimation target is 
\[
	IV_T=\int_0^T\sigma^2_s ds.
\]
This quantity, also known as the \emph{integrated volatility} or \emph{integrated variance} of $X$, has many applications, especially in finance, where $X$ typically models the log-return process of a risky asset and $IV_T$ measures the overall uncertainty or variability inherent in $X$ during the time period $[0,T]$. When $X$ is observed at evenly spaced times $0=t_0<t_1<\ldots<t_n=T$, in the absence of jumps, an efficient estimator of $IV_T$ is given by the realized quadratic variation $\widehat{IV}_T=\sum_{i=1}^{n}(X_{t_i}-X_{t_{i-1}})^2$ in the so-called high-frequency (or infill) asymptotic regime; i.e., when $\max_{i}(t_i-t_{i-1})\to{}0$ and $T\equiv t_n$ is fixed.  In the presence of jumps, $\widehat{IV}_T$ is no longer even consistent for $IV_T$, instead converging to $IV_T+\sum_{s\leq T}(\Delta X_s)^2$, where $\Delta X_s:= X_s-X_{s^-}$ denotes the jump at time $s$. To account for jumps, several estimators have been proposed, among which the most well-known are the truncated realized quadratic variation  and the multipower variations. We focus on the first class, which, unlike the second, is both rate- and variance-efficient, {in the Cramer-Rao lower bound sense,} when jumps are of bounded variation under certain additional conditions.%

The truncated realized quadratic variation (TRQV),  also called truncated realized volatility, was first introduced by \cite{mancini2001} and \cite{mancini2004} and is defined as 
\begin{align}\label{eq:TRV0}
\wh{C}_{n}(\varepsilon)=\sum_{i=1}^{n}(\Delta_i^n X )^2{\bf 1}_{\{|\Delta_i^n X|\leq\varepsilon\}},
\end{align}
where $\varepsilon=\varepsilon_n>0$ is a tuning parameter converging to $0$ at a suitable rate. Above, $\Delta_{i}^{n}X:=X_{t_{i}}-X_{t_{i-1}}$ is the $i$-th increment of $(X_{t})_{t\geq 0}$ based on evenly spaced observations $X_{t_{0}},\ldots,X_{t_{n}}$ over a fixed time interval $[0,T]$ (i.e., $t_{i}=ih_n$ with $h_{n}=T/n$). 
It is shown in \cite{mancini2009non} that TRQV is consistent when either the jumps have finite activity or stem from an infinite-activity L\'evy process. 
In a semimartingale model with L\'evy jumps of  bounded variation, \cite{cont2011nonparametric} showed that the TRQV admits a feasible central limit theorem (CLT), provided that $\varepsilon_n = ch_n^{\omega}$ with some 
  ${ \omega} \in [\frac{1}{4-Y}, \frac{1}{2})$, where $Y\in [0,1)$ denotes the Blumenthal-Getoor index.  
In \cite{jacod2008asymptotic}, consistency was established for a general It\^o semimartingale $X$, and a corresponding CLT is given when the jumps of $X$ are of  bounded variation. In that case, the TRQV attains the optimal rate and asymptotic variance of $\sqrt{ h_n}$ and $2\int_{0}^{T}\sigma_s^4ds$, respectively.

However, in the presence of jumps of unbounded variation,  arguably the most relevant for financial applications (see, e.g.,  \cite{SahaliaJacod2009}, \cite{Belomestny:2010}, \cite{FigueroaLopez:2012},  and the results in Table \ref{EmpMies2005} below), the situation is notably different,  and the available literature on TRQV offers an incomplete picture. In \cite{cont2011nonparametric}, it is shown that when jumps stem from a L\'evy process with stable-like small-jumps of infinite variation, the TRQV estimator $\wh{C}_n(\varepsilon)$ converges to $IV_T$ at a rate slower than $\sqrt{h_n}$. Further, in \cite{mancini2011speed} it is shown that when the jump component $J$ is a  symmetric $Y$-stable L\'evy process and $\varepsilon_n=h_n^{\omega}$ with ${\omega}\in(0,1/2)$, the decomposition $\widehat{C}_{n}( \varepsilon_n)-IV_T=\sqrt{h_n}Z_n+\mathcal{R}_n$ holds, where $Z_{n}$ converges stably in law to $\mathcal{N}(0,2\int_0^T\sigma_s^4 ds)$, while $\mathcal{R}_n$ is precisely of order $\varepsilon_n^{2-Y}$ in the sense that $\mathcal{R}_n=O_P(\varepsilon_n^{2-Y})$ and $\varepsilon_n^{2-Y}=O_P(\mathcal{R}_n)$ (which decays too slowly to allow for efficiency when $Y>1$). In \cite{AmorinoGloter20},  a smoothed version of the TRQV estimator  of the form $\widehat{C}^{Sm}_{n}(\varepsilon)=\sum_{i=1}^{n}(X_{t_i}-X_{t_{i-1}})^2\varphi((X_{t_i}-X_{t_{i-1}})/\varepsilon)$ is considered\footnote{The authors in \cite{AmorinoGloter20}, in fact, consider the more general estimation of $\int_0^T f(X_s) \sigma_s^2 ds$ for functions $f$ of polynomial growth, for which $IV_T$ is a special case.}, where $\varphi\in C^{\infty}$ vanishes in $\mathbb{R}\backslash{}(-2,2)$ and $\varphi(x)=1$ for $x\in(-1,1)$. In that case, using the truncation level $\varepsilon_n:=h_n^{ \omega}$, it is shown that $\widehat{C}^{Sm}_{n}( \varepsilon_n)-IV_T=\sqrt{h_n}Z_n+\mathcal{R}_n$ with $\mathcal{R}_n$ such that $\varepsilon_n^{-(2-Y)}\mathcal{R}_n\to c_{Y}\int \varphi(u)|u|^{1-Y}du$, for a constant $c_Y\neq{}0$, and still $Z_{n}\stackrel{st}{\rightarrow}\mathcal{N}(0,2\int_0^T\sigma_s^4 ds)$. By taking $\varphi$ such that $\int \varphi(u)|u|^{1-Y}du=0$, a ``bias-corrected" estimator was considered under the additional condition that $Y<4/3$. Specifically, the resulting estimator is such that, for any { $\tilde{\epsilon}>0$}, $\widehat{C}^{Sm}_{n}( \varepsilon_n)-IV_T=o_{P}(h_n^{1/2-{ \tilde\epsilon}})$, ``nearly" attaining the optimal statistical error $O_{P}(h_n^{1/2})$. Unfortunately, the construction of such an estimator requires knowledge or accurate estimation of the jump intensity index $Y$, and no feasible CLT was proved when jumps are of unbounded variation even assuming $Y$ is known. 

Apart from TQRV-based approaches, efficient estimation of $IV_T$ when the jumps have unbounded variation is intrinsically limited in the general case. In  \cite{jacod:reiss:2014},  it was shown that when the jump intensity index $Y>1$, the best possible convergence rate, in a minimax sense, over certain ``bounded" classes of semimartingales, is of order $(n\log n)^{-(2-Y)/2}$.  Nevertheless, in principle, a faster convergence rate may be attainable if one constrains the process $X$ to belong to a certain semiparametric class such as when the jumps exhibit a ``locally stable''-like behavior. Obviously, the fastest possible rate one can hope to achieve is $n^{-1/2}$, which coincides with the one attained by the realized quadratic variation in the continuous case and is known to be optimal in a minimax sense.
 
The first (and to-date, only) rate- and variance-efficient estimator of the integrated volatility known in the literature for semimartingales when $Y>1$ was proposed by \cite{JacodTodorov:2014},  under a locally-stable assumption on jumps, but with some notable additional restrictions: these results require either that the jump intensity index $Y<3/2,$ or that the ``small" jumps of the process $X$ are ``symmetric"\footnote{\cite{JacodTodorov:2014} also constructed an estimator that is rate-efficient even in the presence of  asymmetric jumps, but its asymptotic variance is twice as big as the optimal value $2\int_{0}^{T}\sigma_s^4ds$ and, thus, it is not variance efficient.}. Their estimator is based on locally estimating the volatility from the empirical characteristic function of the process' increments over disjoint time intervals shrinking to $0$, but still containing an increasing number of observations. It requires two debiasing steps, which are simpler to explain for a L\'evy process $X$ with symmetric $Y$-stable jump component $ X^j$. The first debiasing step is meant to reduce the bias introduced when attempting to estimate $\log\mathbb{E}(\cos({u} X_{h_n}/\sqrt{h_n}))$ with $\log\big\{\frac{1}{n}\sum_{i=1}^{n}\cos\big({ u} \Delta_i^n X/\sqrt{h_n}\big)\big\}$. The second debiasing step is aimed at eliminating the second term 
in the expansion $-2\log\mathbb{E}(\cos({ u} X_{h_n}/\sqrt{h_n}))
/{ u}^2=\sigma^2+2|\gamma|^{Y} { u}^{Y-2}h_n^{1-Y/2}+O(h_n)$, which otherwise diverges when multiplied by the optimal scaling $h_{n}^{-1/2}=  n^{1/2}$. Using an extension of this approach, Jacod and Todorov were able to apply these techniques to a more general class of It\^o semimartingales in \cite{jacod2016efficient}, even allowing any $Y<2$, though only rate-efficient, but not variance-efficient, estimators were ultimately constructed.

On the other hand, in the special case of L\'evy processes, efficiency across the full range $0<Y<2$ without symmetry requirements has been attained by  \cite{Mies:2020}, again under a locally stable assumption, 
via a generalized method of moments. Specifically, for some suitable  smooth functions $f_{1},f_{2},\dots,f_{m}$ and a scaling factor $u_{n}\rightarrow\infty$, \cite{Mies:2020}  {proposed to search} for the parameter values $\wh{\boldsymbol{\theta}}=(\wh{\theta}_{1},\ldots,\wh{\theta}_{m})$ such that
\begin{align}\label{eq:MF00}
\frac{1}{n}\sum_{i=1}^{n}f_{j}\big(u_{n}\Delta_{i}^{n}X\big)-\bE_{\wh{\boldsymbol{\theta}}}\Big(f_{j}\big(u_{n}\Delta^{n}_{i}\wt{ X}\big)\Big)=0,\quad j=1,\ldots,m,
\end{align}
where $\wt{X}$ is the superposition of a Brownian motion and independent stable L\'{e}vy processes closely approximating $X$ in a certain sense. The distribution measure $\bP_{\boldsymbol{\theta}}$ of  $\wt{X}$ depends on some parameters $\boldsymbol{\theta}=(\theta_1,\ldots,\theta_m)$, one of which is the volatility $\sigma$ of $X$, and $\bE_{\boldsymbol{\theta}}(\cdot)$ denotes the expectation with respect to $\bP_{\boldsymbol{\theta}}$. %
Aside from the fact that this method can only be applied to L\'evy processes, it also suffers from other drawbacks. First, its finite-sample performance critically depends on the chosen moment functions $f_{1},\ldots,f_{m}$. Secondly, its implementation is  
computationally expensive  and may lead to numerical issues since it involves solving a system of nonlinear equations (including possible non-existence of solutions to \eqref{eq:MF00} over finite samples). Moreover,  in {addition} to the required use of a numerical solver to determine the {values} of $\widehat{\pmb{\theta}}$ in \eqref{eq:MF00}, the expectations appearing therein need to be numerically   approximated since explicit expressions for the moments $\bE_{\boldsymbol{\theta}}(f_{j}(u_{n}\Delta^{n}_{i}\wt{ X}))$ are typically not available. This fact introduces additional numerical errors that complicates its performance.%

In this paper, we consider a new method to estimate the integrated volatility $IV_T=\int_0^T\sigma_s^2ds$ of an It\^o semimartingale whose jump component is given by a stochastic integral with respect to a tempered-stable-like  L\'evy process $J$ of unbounded variation.
To the best of our knowledge, our method, together with \cite{JacodTodorov:2014}, are the only efficient methods to deal with jumps of unbounded variation for semimartingales. The idea is natural. We simply apply debiasing steps similar to those of \cite{JacodTodorov:2014} to the TRQV of  \cite{mancini2009non}. To give the heuristics as to why this strategy works, consider a small-time expansion of the truncated moments $\mathbb{E}(X_{h_n}^{2k}{\bf 1}_{\{|X_{h_n}|\leq{}\varepsilon_n\}})$ of a L\'evy process $X_t=bt+\sigma W_{t}+X^j_t$ in the asymptotic regime $h_n,  \varepsilon_n\to{}0$ with $\varepsilon_n/\sqrt{h_n}\to\infty$. Using a variety of techniques, including a change of probability measure, Fourier-based methods, and small-large jump decompositions, we show the following two expansions:
\begin{align*}
	\mathbb{E}\left[X^2_{h_n}{\bf 1}_{\{|X_{h_n}|\leq{}\varepsilon_n\}}\right]&=\sigma^2h_n+c_1 h_n\varepsilon_n^{2-Y} +c_2 h_n^2 \varepsilon^{-Y}_n+ {\rm h.o.t.},\\
    \bE\left[X_{h_{n}}^{2k}\,{\bf 1}_{\{|X_{h_{n}}|\leq\varepsilon\}}\right]
    &= (2k-1)!!\,\sigma^{2k} h_n^k + c_3\,h_n\varepsilon_n^{2k-Y}+ {\rm h.o.t.},
\end{align*}
for certain constants $c_1,c_2,c_3\neq{}0$ that are explicitly computed.   Hereafter, {\rm h.o.t.} stands for `higher order terms'. The expansions above are the most precise of their type in the literature and are of great interest in their own right. Based on the first expansion above, it is easy to see that the rescaled bias ${\mathbb{E}[h_n^{-1/2}(\wh{C}_{n}(\varepsilon)-\sigma^2T)]}$ satisfies
\begin{align}
	\mathbb{E}\left[h_n^{-\frac{1}{2}}\left(\wh{C}_{n}(\varepsilon)-\sigma^2T\right)\right]&=
	T c_1 h_n^{-\frac{1}{2}}\varepsilon_n^{2-Y} +Tc_2 h_n^{\frac{1}{2}} \varepsilon^{-Y}_n+{\rm h.o.t.}, \label{e:bias_expansion}
\end{align}
which suggests the necessity of the condition $h_n^{-1/2}\varepsilon_n^{2-Y} =o(1)$ for a feasible CLT for $\widehat C_n(\varepsilon)$ at the rate $\sqrt{h_n}$. However, together with the (necessary) condition $\varepsilon_n/\sqrt{h_n}\to\infty$, this can happen only if $Y<1$, and removal of the  first terms in \eqref{e:bias_expansion} is  consequently necessary for efficient estimation when jumps are of unbounded variation.  
To that end,  note that for any $\zeta>1$, 
\begin{align*}
	\mathbb{E}\left(\wh{C}_{n}(\zeta\varepsilon)-\wh{C}_{n}(\varepsilon)\right)&= c_1(\zeta^{2-Y}-1)\varepsilon_n^{2-Y}+{\rm h.o.t.},\\
	\mathbb{E}\left(\wh{C}_{n}(\zeta^2\varepsilon)-2\wh{C}_{n}(\zeta\varepsilon)+\wh{C}_{n}(\varepsilon)\right)&=
	 c_1(\zeta^{2-Y}-1)^2\varepsilon_n^{2-Y}+{\rm h.o.t.}.
\end{align*}
The above formulas motivate the ``bias-corrected" estimator
\begin{align}\label{Bias1Corr}
	\wh{C}'_{n}(\varepsilon;\zeta):=\wh{C}_{n}(\varepsilon)-\frac{\left(\wh{C}_{n}(\zeta\varepsilon)-\wh{C}_{n}(\varepsilon)\right)^2}{\wh{C}_{n}(\zeta^2\varepsilon)-2\wh{C}_{n}(\zeta\varepsilon)+\wh{C}_{n}(\varepsilon)},
\end{align}
which is the essence of the debiasing procedure of \cite{JacodTodorov:2014}. 
As we shall see, the story is more complicated than what the simple heuristics above suggest. 
Our main result shows that, for  the class of It\^o semimartingales described above, the estimator \eqref{Bias1Corr} is indeed rate- and variance-efficient provided that $Y\in(1,4/3)$. Furthermore, if $Y\in (1,8/5)$, a second bias-correcting step will achieve both rate- and variance-efficiency (for the case $ 8/5\leq Y<2$, see remark at the end of Section \ref{Sec:DebiasMthd}).  Even though our main motivation lies in incorporating jumps of unbounded variation, we show that the debiasing steps will still achieve efficiency in the case that $Y<1$ (though, of course, no debiasing is needed in that case because the jumps are of bounded variation).

Though our approach is natural, mathematically establishing its efficiency is highly nontrivial,  
starting from the new high-order expansions of the truncated moments of L\'evy processes -- which, beyond heuristics, 
ultimately play a key role in analyzing asymptotics for 
the debiasing technique -- to the application of Jacod's stable central limit theorem for semimartingales (in particular, the verification of the asymptotic `orthogonality' condition (2.2.40) in \cite{JacodProtter}). 
Our Monte Carlo experiments indicate improved performance compared to \cite{JacodTodorov:2014} (and also \cite{Mies:2020}) for the important class of CGMY L\'evy processes (cf. \cite{CarrGemanMadanYor:2002}) and for a Heston stochastic volatility model with CGMY jumps in the range of values of $Y$ covered by our theoretical framework.

 If we limit ourselves to a L\'evy model, our approach is more computationally efficient and numerically stable than that in \cite{Mies:2020}\footnote{Though the method of \cite{Mies:2020} allows for simultaneous estimation of several parameters of the model (such as both $\sigma$ and $Y$).}. For more general semimartingales, 
  our procedure is simpler than that in \cite{JacodTodorov:2014} since it does not require an extra debiasing step to correct the nonlinear nature of the logarithmic transformation employed therein nor does it require a symmetrization step to deal with asymmetric jump components. Furthermore, our method does not rely on a `localization'  technique 
in the sense that it does not need to break the data into disjoint blocks where the integrated volatility is locally estimated. The latter step introduces an additional tuning parameter absent from our method.  

Let us emphasize that our method is the first variance- and rate-efficient nonparametric method for integrated volatility free of complete symmetry assumptions on small jumps that is capable of exceeding the limit $Y<3/2$ imposed in \cite{JacodTodorov:2014,jacod2016efficient}. Symmetry is potentially a strong assumption for financial returns as there is a general belief that significant losses are more likely than significant gains. For instance, recently \cite{Asmussen:2022} examined several empirical studies from the literature and observed that the majority display negative skewness. Of course,  skewness may arise from either large or small jumps, though large jump asymmetry has received most attention in empirical work. Indeed, there are few studies to date that estimate the intensity of small positive and negative jumps separately.  An exception is \cite{Bianchi}, who, using  MLE applied to daily S\&P500 index data from 1996-2006, obtained the estimates $\hat{C}_+=0.7119$, $\hat{C}_-=0.5412$, $\hat{G}=59.94$, $\hat{M}=59.94$, $\hat{Y}_+=1.0457$, and $\hat{Y}^-=1.1521$ under a pure-jump L\'evy model with L\'evy density $C_{+}e^{-x/G}|x|^{-Y_+-1}{\bf 1}_{x>0}+C_-e^{-|x|/M}|x|^{-Y_{-}-1}{\bf 1}_{x<0}$, which points to asymmetry in small jumps.  For further illustration,
in Table \ref{EmpMies2005} below, we fit a L\'evy model $\sigma W_t+X_t^j$, with semi-parametric L\'evy density $(C_+{\bf 1}_{x>0}+C_-{\bf 1}_{x<0})q(x)|x|^{-Y-1}$ (here, $q(x)\stackrel{x\to{}0}{\longrightarrow}1$). We use Mies' method of moments (cf. \cite{Mies:2020}) for different stocks and frequencies\footnote{As pointed out in \cite{Mies:2020} (see also the paragraph after \eqref{eq:MF00} above), numerical issues can arise related to feasibility of the estimating equations associated with the method. We are only presenting the results when the algorithm successfully finishes and yields reasonable values.}
 over a 1-year period in 2005. It is clear that $\hat{C}_+$ is different from $\hat{C}_-$, sometimes by a relatively large value, indicating further evidence for asymmetry in small jump behavior. Furthermore, all values of $\hat{Y}$ are larger than $1$, indicating the presence of a jump component of unbounded variation. 

Finally, let us also remark that our result opens the doors to attain rate- and variance-efficient estimators free of symmetry requirements beyond the mark $8/5$ or in more general semiparametric models with successive Blumethal-Getoor indices by considering further debiasing steps. These directions will be investigated in further work.

\begin{table}[ht]
\centering
\begin{tabular}{cccccc}
Stock &Freq. & $\hat{\sigma}$ & $\hat{C}_{+}$ & $\hat{C}_{-}$ & $\hat{Y}$\\\hline
INTC & 1 min & 0.216 & 0.0096 & 0.0075 & 1.43\\
INTC & 5 sec & 0.241 & 0.0311 & 0.0292 & 1.60\\
PFE & 5 min &0.180 & 0.0232 & 0.0199 & 1.14\\
PFE & 1 min &0.196 & 0.0105 & 0.0066 & 1.37\\
AMAT & 5 sec & 0.344 & 0.0014 & 0.0012 & 1.85 \\
SPY & 5 sec & 0.103 &  0.0003 &  0.00005 & 1.82\\
AMGN & 1 min & 0.211 & 0.0032 & 0.0038 & 1.53\\
MOT & 1 min  & 0.244 & 0.0183 & 0.0066 & 1.33
\\
\hfill
\end{tabular}
\caption{\small Parameter estimates for a L\'evy model $\sigma W_t+X_t^j$ with semi-parametric L\'evy density $(C_+{\bf 1}_{x>0}+C_-{\bf 1}_{x<0})q(x)|x|^{-Y-1}$ with $q(x)\stackrel{x\to{}0}{\longrightarrow}1$ applied to stock data over a 1-year horizon at different sampling frequencies. Parameter estimates were obtained via the method \cite{Mies:2020} with $M=1$. Intraday data was obtained from the NYSE TAQ database of 2005 trades via Wharton's WRDS system.
}
\label{EmpMies2005}
\end{table}

The rest of this paper is organized as follows. Section \ref{Sec:Model} introduces the framework and assumptions as well as some known preliminary results from the literature.  Section \ref{Sec:DebiasMthd} introduces the debiasing method and main results of the paper. Section \ref{constCGMY} illustrates the performance of our method via Monte Carlo simulations and compares it to the method in \cite{JacodTodorov:2014}.   The proofs of the key results are deferred to appendix sections. 

\section{Setting and background}\label{Sec:Model}
In this section, we introduce  the model, main assumptions, and some notation.
We consider a 1-dimensional It\^o semimartingale $X=(X_{t})_{t\in\bR_{+}}$ of the form \eqref{eq:MnMdlX0}, defined on a complete filtered probability space $(\Omega,\sF,(\sF_{t})_{t\in\bR_{+}},\bP)$.  Since it has no impact on the value of the increments of $X$, for simplicity throughout we assume $X_0=0$.  We assume the jump component $X^{j}$ can be decomposed into a sum of an infinite-variation process $X^{j,\infty}$ and a finite variation process $X^{j,0}$ given,  for $t\in\bR_+$, as
\begin{align}\nonumber
X^j_{t}&:= X_t^{j,\infty} + X_t^{j,0}\\
&:=\int_0^t \chi_{s^-} dJ_s +  \int_0^t \int \Big\{ \delta_0(s,{ z})\mathfrak p_0(ds,dz) +\delta_1(s,{ z})\mathfrak p_1(ds,dz) \Big\},  
\label{eq:MnMdlX}
\end{align}
where 
$\chi=\{\chi_t\}_{t\geq{}0}$ is an adapted process satisfying appropriate integrability conditions,
$J:=(J_{t})_{t\in\bR_{+}}$ is an independent pure-jump L\'{e}vy  process with L\'{e}vy triplet $(b,0,\nu)$, $\mathfrak p_0,\mathfrak p_1,$ are Poisson random measures on $\bR_+\times \bR$ with intensities $\mathfrak q_i(ds,dz)=ds\otimes \lambda_i(dz)$, where  the $\lambda_i$'s  are $\sigma$-finite measures on $\bR$, and $\mathfrak p_1$ is assumed independent of $J$. The specific conditions on $\nu$, $\lambda_i$, and on the coefficient processes $\delta_i$, $\chi$, and $\sigma$ are given below.

The L\'{e}vy measure $\nu$ is assumed to admit a density $s:\bR_{0}\rightarrow\bR_{+}$  of the form
\begin{align}\label{eq:levyden}
s(x):=\frac{d\nu}{dx}:=\big(C_{+}{\bf 1}_{(0,\infty)}(x)+C_{-}{\bf 1}_{(-\infty,0)}(x)\big)q(x)\,|x|^{-1-Y}.
\end{align}
Above, $\bR_{0}:=\mathbb{R}\backslash\{0\}$, $C_{\pm}>0$, $Y\in(1,2)$, and $q:\bR_{0}\rightarrow\bR_{+}$ is a bounded Borel-measurable function satisfying the following conditions:

\begin{assumption}\label{assump:Funtq}
\hfill
\begin{itemize}
\item [(i)] $q(x)\rightarrow 1$, as $x\rightarrow 0$;
\item [(ii)] there exist $\alpha_{\pm}\neq 0$ such that
    \begin{align*}
    \int_{0}^{1}\big|q(x)-1-\alpha_{+}x\big|x^{-Y-1}dx+\int_{-1}^{0}\big|q(x)-1-\alpha_{-}x\big| |x|^{-Y-1}dx<\infty .
    \end{align*}
\end{itemize}
\end{assumption}

These processes are sometimes called ``stable-like L\'evy processes" and were studied in \cite{FigueroaLopezGongHoudre:2016,FigueroaLopezOlafsson:2016} and many other works. In simple terms, condition (i) above says that the small jumps of the L\'evy process $X$ behave like those of a $Y$-stable L\'evy process with L\'evy measure
\begin{align}\label{Dfntildenu0}
\wt{\nu}(dx):=\big(C_{+}{\bf 1}_{(0,\infty)}(x)+C_{-}{\bf 1}_{(-\infty,0)}(x)\big)|x|^{-Y-1}dx.
\end{align}
The condition $Y\in(1,2)$ implies that $J$ has unbounded variation in that sense that $\sum_{i=1}^{n}|J_{t_i}-J_{t_{i-1}}|\to{}\infty$, a.s., as the partition $0=t_0<t_1<\dots<t_n=T$ is such that $\max\{t_i-t_{i-1}\}\to{}0$.  

As discussed in Section \ref{s:1}, in view of \cite{jacod:reiss:2014}, the locally $Y$-stable aspect of Assumption \ref{assump:Funtq} is crucially important for our results, and similar assumptions have been made by other authors (e.g, \cite{Mies:2020,JacodTodorov:2014}).  Though not completely general, the class is still relevant in applications, as many of the models proposed in the literature (especially, in finance) fall within this class. Nevertheless, from a theoretical point of view, it remains to be seen as to what the broadest assumptions may be under which one can still attain estimation efficiency.
 
As in \cite{FigueroaLopezGongHoudre:2016}, it will be important for our analysis to apply a density transformation technique \cite[Section 6.33]{Sato:1999} to ``transform" the process $J$ into a stable L\'evy process. Concretely, we can change the probability measure from $\bP$ to another locally absolutely continuous measure $\wt{\bP}$, under which $W$ is still a standard Brownian motion independent of $J$, but, under $\wt{\bP}$, $J$ has L\'evy triplet $(\tilde{b},0,\tilde{\nu})$, where $\wt{\nu}(dx)$ is given as in \eqref{Dfntildenu0} and $\wt{b}:=\bar b+\int_{0<|x|\leq 1} x(\tilde{\nu}-\nu)(dx)$. The key assumption above is (ii), which would allow us to decompose  the log-density process 
\[
	U_t:=\ln
\frac{d\wt{\bP}\big|_{\sF_{t}}}{d\bP\big|_{\sF_{t}}},
\]
 as a sum of a bounded variation process and two spectrally one-sided $Y$-stable L\'evy processes.

Finally, we give the conditions on  $\sigma$ and the coefficient processes $b$, $\delta_0,\delta_1$, and $\chi$ in \eqref{eq:MnMdlX0} and \eqref{eq:MnMdlX}.
\begin{assumption}\label{assump:Coef0} \hfill
\begin{enumerate}
	\item[(i)] $\sigma$ is c\`adl\`ag adapted.
	\item[(ii)] 
The process $\chi$ is given as
$$
	\chi_t = \chi_0+\int_0^t b^\chi_sds+\int_0^t\Sigma^\chi_s dB_s.
$$

\item[(iii)] The processes $W,B$ are Brownian motions independent of $(J,\mathfrak p_0,\mathfrak p_1)$; $(\delta_1,\mathfrak p_1)$ is independent of $(X^c,J,\chi,\mathfrak p_0,\delta_0)$; {$J$ is independent of $\sigma$}.

	\item[(iv)] The processes $\Sigma^\chi$, $b$, and $b^\chi$ are c\`adl\`ag adapted,  and $\delta_0,\delta_1$ are predictable. There is also a sequence  $\{\tau_n\}_{n\geq{}1}$ of stopping times increasing to infinity,  nonnegative $\lambda_i(dz)$-integrable functions $H_i$, and a  positive sequence  $\{M_n\}_{n\geq{}1}$ such that
$$
t\leq \tau_n \implies \begin{cases}
|\sigma_t| + |b_{t}| + |b^\chi_{t}| +|\Sigma^\chi_{t}| \leq M_n ,\\
|\bE( \sigma_{t+s }-\sigma_{t}|\mathcal F_t)|+\bE( |\sigma_{t+s }-\sigma_{t}|^2|\mathcal F_t)\leq M_n s,\\
( |\delta_0(t,z)|\wedge 1)^{r_0}  \leq M_n H_0(z),\\
( |\delta_1(t,z)|\wedge 1)^{r_1}  \leq M_n H_1(z),
\end{cases}
$$
for some $r_0 \in [0,\frac{Y}{2+Y})$ and $r_1\in[0,Y/2)$.

\end{enumerate}
\end{assumption}
Above, the parameters $r_0,r_1$ control the degree of activity in the nuisance finite-variation jump terms.  Two such terms are included to allow for a broader range of finite-variation jump activity in our model setup.  Note that $Y/2$ is always bigger than $Y/(2+Y)$, which shows that when the bounded variation jump component is independent from the other processes, we can incorporate a wider range of jump activity. These restrictions effectively guarantee that the bias introduced by finite variation components are negligible in comparison to leading bias terms arising from the locally-stable jumps in $X$.


%

\section{ Main results}\label{Sec:DebiasMthd}
In this section, we construct an efficient estimator for the integrated volatility $IV_T=\int_0^T \sigma_s^2ds$ based on the well-studied estimator TRQV \eqref{eq:TRV0}. 

Throughout, we assume the process $X=\{X_{t}\}_{t\geq{}0}$ is sampled at $n$ evenly spaced observations, $X_{t_1}, X_{t_2}, \ldots, X_{t_n}$, during a fixed time interval $[0,T]$, where for $i=0,\dots,n$,  $t_{i}=t_{i,n}= ih_{n}$ with $h_n=T/n$, and, for simplicity, assume that $T=1$. As usual, we define the increments of a generic process $V=\{V_t\}_{t\geq{}0}$ as $\Delta_{i}^{n}V:=V_{t_{i}}-V_{t_{i-1}}$, $i=1,\dots,n$. We often use the shorthand notation:
\[
	V_i^n=V_{t_{i}},\quad \mathbb{E}_{i}[\cdot]=\mathbb{E}[\cdot|\mathcal{F}_{t_i}].
\]

As mentioned in the introduction, in the presence of jumps of unbounded variation, the TRQV estimator $\wh{C}_n(\varepsilon)$ is not efficient since it possesses a bias that vanishes at a rate slower than $n^{-1/2}$, the rate at which the ``centered'' TRQV
\begin{equation}\label{limitCLT0}
	\overline{C}_{n}(\varepsilon)=\sum_{i=1}^{n}\left\{\big(\Delta_{i}^{n}X\big)^{2}{\bf 1}_{\{|\Delta_{i}^{n}X|\leq\varepsilon\}}- \bE_{i-1}\left[\big(\Delta_{i}^{n}X\big)^{2}{\bf 1}_{\{|\Delta_{i}^{n}X|\leq\varepsilon\}}\right]\right\},
\end{equation}
admits a CLT.
  To overcome this, our idea is to apply the debiasing procedure of \cite{JacodTodorov:2014} to the TRQV.  
 As mentioned in the introduction, 
our procedure is simpler than \cite{JacodTodorov:2014} since it does not require an extra debiasing step to account for the logarithmic transformation nor does it require a symmetrization step to deal with asymmetric L\'evy measures. Furthermore, our method does not have to be applied in each subinterval of a partition of the time horizon, which introduces another tuning parameter.

Before constructing our estimator, we first establish the asymptotic behavior of TRQV \eqref{eq:TRV0} with a fully specified centering quantity $A(\varepsilon,h)$ rather than the inexplicit centering $\bE_{i-1}\left[\big(\Delta_{i}^{n}X\big)^{2}{\bf 1}_{\{|\Delta_{i}^{n}X|\leq\varepsilon\}}\right]$ 
 of \eqref{limitCLT0}. It  also  characterizes the structure of the bias $A(\varepsilon,h)$ in the threshold parameter $\varepsilon$ that will ultimately be exploited in our debiasing procedure. 
Below and throughout the rest of the paper, we use the usual notation $a_{n}\ll b_n$, whenever $a_n/b_n\to0$ as $n\to\infty$.%
\begin{proposition}\label{prop1}
Suppose that $Y\in (0,1)\cup (1, 8/5)$ and $h_n^{\frac{3}{2(2+Y)}\wedge\frac{1}{2}}\ll \varepsilon_n\ll  h_n^{\frac{1}{4-Y}}$. Let 
\begin{align}
 \wt{Z}_n(\varepsilon) :=  \sqrt{n}\left(\widehat C_n(\varepsilon) - \int_0^1\sigma_s^2ds - A(\varepsilon,h)  \right),\label{e:def_Zn}%
\end{align}
where 
\begin{align}\label{e:def_a(eps,h)0}
    A(\varepsilon,h) 
    :=\frac{\bar{C}}{2-Y}\int_0^1|\chi_s|^{Y}ds\varepsilon^{2-Y} - \bar{C}\frac{(Y+1)(Y+2)}{2Y} \int_0^1|\chi_s|^Y\sigma^2_sds h\varepsilon^{-Y},
\end{align}
and $\bar{C}:=C_++C_-$. Then, as $n\to\infty$,
 \begin{align}\label{limitCLT}
\wt{Z}_n( \varepsilon_n)
\toDistSt  \mathcal{N}\left(0,2\int_0^1\sigma_s^4ds\right).
\end{align}
\end{proposition}

\begin{remark}\label{WhatIfYl1a}
As expected, the statement above shows that, when $Y<1$, we have $\sqrt{n}\left(\widehat C_n(\varepsilon) - \int_0^1\sigma_s^2ds\right)\toDistSt  \mathcal{N}\left(0,2\int_0^1\sigma_s^4ds\right)$, whenever $h_n^{\frac{1}{2}}\ll \varepsilon_n\ll  h_n^{\frac{1}{2(2-Y)}}$ and the indices $r_0$ and $r_1$ in Assumption \ref{assump:Coef0} are less than $Y$ (i.e., the constraints $r_0 \in [0,\frac{Y}{2+Y})$ and $r_1\in[0,Y/2)$ are needed in Proposition \ref{prop1} only in the case $Y>1$). Indeed, in this case, since the leading order bias is $O(\varepsilon^{2-Y}),$ so taking $\varepsilon \ll h_n^{\frac{1}{2(2-Y)}}$ renders it asymptotically negligible.
\end{remark}

The next theorem establishes the stable convergence (in particular, the convergence rate) of the difference $\widetilde{Z}_n(\zeta\varepsilon) -\widetilde{Z}_n(\varepsilon)$, for some $\zeta>1$, which is the second main technical result we use to deduce the efficiency of our debiased estimator.
\begin{proposition}\label{thm:singleclt}
Suppose $Y\in (0,1)\cup (1,  8/5)$, and  $h_n^{\frac{4}{8+Y}}\ll \varepsilon_n\ll  h_n^{\frac{1}{4-Y}}$. With the notation of Proposition \ref{prop1}, for arbitrary $\zeta>1$,
 \begin{align}
\label{e:joint_CLT_expression}
u_n^{-1}\left(\wt Z_n( \zeta\varepsilon_n) - \wt Z_n( \varepsilon_n)\right)
&\toDistSt
\mathcal{N}\left(0,
\frac{\bar{C}}{4-Y}\int_0^1|\chi_s|^Yds(\zeta^{4-Y}-1)\right),
\end{align}
 as $n\to\infty$, where $u_n := h_{n}^{-\frac{1}{2}}\varepsilon_{n}^{\frac{4-Y}{2}}\to 0$.
\end{proposition}
\begin{remark}
Note that \eqref{e:joint_CLT_expression} implies that
\begin{align*}
	&\varepsilon_n^{-Y/2}\left(\frac{\widehat C_n(\zeta\varepsilon)-\widehat C_n(\varepsilon)}{\varepsilon^{2-Y}}-\frac{\bar{C}}{2-Y}(\zeta^{2-Y}-1)\int_0^1|\chi_s|^{Y}ds\right)\\
	&\quad \toDistSt  \mathcal{N}\left(0,\frac{\bar{C}}{4-Y}(\zeta^{4-Y}-1)\int_{0}^{1}|\gamma_s|^Yds\right).
\end{align*}
In particular,
\begin{align}\label{SmplImpFMLb}
	 \frac{\widehat C_n(\zeta\varepsilon)-\widehat C_n(\varepsilon)}{\varepsilon^{2-Y}}\stackrel{\mathbb{P}}{\longrightarrow}{}\frac{\bar{C}}{2-Y}(\zeta^{2-Y}-1)\int_0^1|\chi_s|^{Y}ds.
\end{align}
Expression \eqref{SmplImpFMLb} plays a role in our numerical implementation in Section \ref{constCGMY}. %
\end{remark}

We are now in a position to introduce our proposed estimator.  To this end, we will exploit the structure of the bias term $A(\varepsilon,h)$ in $\varepsilon$. The idea is simple. Suppose that a function $f(x)$ takes the form $a+bx^{\alpha}$ for any $a\in\mathbb{R}$ and $\alpha,b\neq{}0$. Then, it is easy to see that, for any  $\zeta>1$, 
\[
	f(x)-\frac{(f(\zeta x)-f(x))^2}{f(\zeta^2x)-2f(\zeta x)+f(x)}=a+bx^\alpha-\frac{b^2x^{2\alpha}(\zeta^\alpha-1)^2}{bx^{\alpha}(\zeta^{2\alpha}-2\zeta^\alpha+1)}=a,
\]
 hence, recovering $a$ without requiring knowledge of $b$ and $\alpha$. These heuristics suggest the following debiasing procedure to successively remove each  term appearing in $A(\varepsilon,h)$. For any $\zeta_1,\zeta_2>1$, in a first step, we compute
\begin{align}
    \wt C_n '(\varepsilon, \zeta_1) &= \widehat C_n(\varepsilon)-\frac{\left(\widehat C_n(\zeta
    _1\varepsilon)-\widehat C_n(\varepsilon)\right)^2}{\widehat C_n(\zeta_1^2\varepsilon)-2\widehat C_n(\zeta_1 \varepsilon) + \widehat C_n(\varepsilon)},\label{e:debiased1}
    \end{align}
and, in the second step,
    \begin{align}
    \wt C_n ''(\varepsilon, \zeta_2,\zeta_1) &= \wt C_n '(\varepsilon, \zeta_1)-\frac{\left(\wt C_n '(\zeta_2\varepsilon, \zeta_1)-\wt C_n '(\varepsilon, \zeta_1)\right)^2}{\wt C_n '(\zeta_2^2\varepsilon, \zeta_1)-2 \wt C_n '(\zeta_2\varepsilon, \zeta_1) + \wt C_n '(\varepsilon, \zeta_1)} .\label{e:debiased2}
\end{align}
The next theorem is the main result of the paper.  It establishes the rate- and variance-efficiency of the two-step debiased estimator $ \wt C_n ''(\varepsilon, \zeta_2,\zeta_1)$ provided   $Y\in (0,1)\cup (1,8/5)$.
\begin{theorem} \label{thm:debiasclt}
Suppose that  $Y\in (0,1)\cup (1,{8/5})$, and  $h_n^{\frac{4}{8+Y}}\ll \varepsilon_n\ll  h_n^{\frac{1}{4-Y}\vee \frac{2}{4+Y}}$. Then, for any fixed $\zeta_1,\zeta_2>1$, as $n\to\infty$,
\begin{align*}
     \sqrt{n}\left(\wt C_n ''(\varepsilon, \zeta_2,\zeta_1) - \int_0^1\sigma_s^2ds \right) \toDistSt  \mathcal{N}\left(0, 2\int_0^1\sigma_s^4ds\right).
\end{align*}
\end{theorem}
 It is customary to use power thresholds of the form $\varepsilon_n = c_0 h_n^{ \omega}$, where $c_0>0$ and ${ \omega}>0$ are some constants\footnote{ Though, it is shown in  \cite{gong2021} that in the case of a L\'evy process $X$ with jump component $J$ as in Section \ref{Sec:Model}, the  MSE-optimal threshold $\varepsilon_n^{*}$ is such that $\varepsilon_n^{*}\sim\sqrt{(2-Y)\sigma^{2}h_n\ln(1/h_n)}$, as $n\to\infty$.}. In that case, the assumption   $h_n^{\frac{4}{8+Y}}\ll \varepsilon_n\ll  h_n^{\frac{1}{4-Y}\vee { \frac{2}{4+Y}}}$ in Theorem \ref{thm:debiasclt} becomes
 \begin{align}\label{IDWBHT}
   \left( \frac{1}{4-Y} \vee { \frac{2}{4+Y}}\right) < \omega < \frac{4}{8+Y}.
\end{align}
\begin{remark}
Note that, if  $Y>1$ -- the case where debiasing is strictly necessary for estimation efficiency --  the value of $ \omega=5/12$ satisfies the above constraint  \eqref{IDWBHT} for any value {$Y$} of the possible range  $4/5<Y<8/5$ (on relaxations of these constraints, see Remark \ref{DRemH}).
\end{remark}
\begin{remark}
The case $Y=1$ is excluded from the above statements since part of our arguments rely on moment estimates for the truncated increments of $Y$-stable L\'evy processes, whose characteristic function differs slightly when $Y=1$, though this case can be handled similarly with minor adjustments to our arguments.
\end{remark}

\begin{remark}\label{DRemH}
	As a consequence of the proof of Theorem \ref{thm:debiasclt}, it follows that if $1<Y<4/3$, then only one debiasing step is needed to achieve efficiency. That is, for $1<Y<4/3$, we already have 
	\begin{align*}
     \sqrt{n}\left(\wt C_n'(\varepsilon, \zeta_1) - \int_0^1\sigma_s^2ds \right) \toDistSt  \mathcal{N}\left(0, 2\int_0^1\sigma_s^4ds\right),
\end{align*}
whenever  $h_n^{\frac{1}{2Y}}\ll \varepsilon_n\ll  h_n^{\frac{2}{4+Y}}$.
	If $4/3\leq Y<8/5$, the proof of Theorem \ref{thm:debiasclt} shows a second debiasing step is required. These two facts suggest that further debiasing steps similar to \eqref{e:debiased1}-\eqref{e:debiased2} could be used to handle values of $Y$ larger than $8/5$, or more broadly, less restrictive conditions on the jump measure of X at zero. %
	This conjecture requires significant further analysis beyond the scope of the present paper and, hence, we leave it for future research.
	\end{remark}

\section{Monte Carlo performance  with CGMY jumps} \label{constCGMY}

In this section, we study the performance of the two-step debiasing procedure of the previous section in two settings: the case of a L\'evy process with a CGMY jump component $J$ (cf. \cite{CarrGemanMadanYor:2002}) and a Heston stochastic volatility model with again a CGMY jump component. Following \cite{JacodTodorov:2014}, we also consider variants of our debiasing procedure that make use of the sign of the bias terms that lead to further improved finite-sample performance.

\subsection{Constant volatility}
We start by considering simulated data from the model \eqref{eq:MnMdlX0} and \eqref{eq:MnMdlX},
where the coefficient processes $\sigma$, $b$, and $\chi$ are constants, {$\delta_0=\delta_1\equiv 0$ (no bounded variation components)}, and $\{J_t\}_{t\geq{}0}$ is a CGMY process, independent of the Brownian motion $\{W_{t}\}_{t\geq{}0}$, with L\'evy measure having a $q$-function, in the notation of \eqref{eq:levyden}, of the form:
\begin{align*}
q(x)=e^{-Mx}{\bf 1}_{(0,\infty)}(x)+e^{Gx}{\bf 1}_{(-\infty,0)}(x),
\end{align*}
and $C_{+}=C_{-}=C$. Thus, the conditions of Assumption \ref{assump:Funtq} are satisfied with $\alpha_{+}=-M$ and $\alpha_{-}=G$. For simplicity we take $b=0$ and $\chi=1$, and adopt the parameter setting
\begin{align}\label{USPH0}
C=0.028,\quad G=2.318,\quad M=4.025,
\end{align}
which are similar to those used in \cite{FigueroaLopezOlafsson:2016}\footnote{\cite{FigueroaLopezOlafsson:2016} considers the asymmetric case  $\nu(dx)=C_{\sgn (x)}\bar{q}(x)|x|^{-1-Y}\,dx$ with $C_+=0.015$ and $C_-=0.041$. Here, we take $C=(C_++C_-))/2$ in order to simplify the simulation of the model.  The parameter values of $C_+$, $C_-$, $G$, and $M$ used in \cite{FigueroaLopezOlafsson:2016} were taken from \cite{Kawai}, who calibrated the tempered stable model using market option prices.}, and take $\sigma=0.2$, $0.4$, and $Y=1.25$, $1.35$, $1.5$, $1.7$, respectively. We fix $T=1$ year and 
$n=252(6.5)(60)$, which corresponds to a frequency of $1$ minute (assuming $252$ trading days and a $6.5$-hour trading period each day). 

In a fashion similar to \cite{JacodTodorov:2014}, for the threshold $\varepsilon = c_0h^{\omega}$,  we take $c_0 = \sigma_{BV}$, where
\[
\sigma_{BV}^2:= \frac{\pi}{2} \sum_{i=2}^n |\Delta_{i-1}^nX||\Delta_i^nX|,
\] 
which is the standard bipower variation estimator of $\sigma^2$ first introduced by \cite{Barndorff}. {For the value of ${ \omega}$ we take ${ \omega} = \frac{5}{12}$, which, as mentioned above, satisfies the condition \eqref{IDWBHT} for any $Y\in(1,8/5)$.}
We compare the performance of the following estimators:
\begin{enumerate} 
    \item TRQV: $\wh{C}_{n}(\varepsilon)=\sum_{i=1}^{n}\big(\Delta_{i}^{n}X\big)^{2}{\bf 1}_{\{|\Delta_{i}^{n}X|\leq\varepsilon\}}$;
    
    \item 1-step debiasing estimator removing positive bias:
    \begin{equation}
    \begin{gathered} \label{eq:stp1_pb}
    \widetilde C_{n,pb}'(\varepsilon, \zeta_1, p_1) = \widehat C_n(\varepsilon)- \eta_1 \left(\widehat C_n(\zeta_1\varepsilon)-\widehat C_n(\varepsilon)\right),\\%
    \text{where} \quad \eta_1 = \frac{\widehat C_n(p_1\, \zeta
    _1\varepsilon)-\widehat C_n(p_1\, \varepsilon)}{\widehat C_n(p_1\, \zeta_1^2\varepsilon)-2\widehat C_n(p_1\, \zeta_1 \varepsilon) + \widehat C_n(p_1\, \varepsilon)} \vee 0,%
   \end{gathered}
   \end{equation}
with $\zeta_1=1.45$ and $p_1=0.6$, which were selected to achieve favorable performance across all considered values of $Y$ and $\sigma$. If $\wt C_{n,pb} '(\varepsilon, \zeta_1,p_1)$ is negative, we recompute $\eta_1$ with $\varepsilon = 2 \varepsilon/3$.  This method is inspired by \cite{JacodTodorov:2014} and is motivated by the following decomposition of the bias correction term of \eqref{e:debiased1} into a product of two factors: 
    \[\frac{\left(\widehat C_n(\zeta
    _1\varepsilon)-\widehat C_n(\varepsilon)\right)}{\widehat C_n(\zeta_1^2\varepsilon)-2\widehat C_n(\zeta_1 \varepsilon) + \widehat C_n(\varepsilon)} \times \left(\widehat C_n(\zeta
    _1\varepsilon)-\widehat C_n(\varepsilon)\right),\]
    where, due to \eqref{SmplImpFMLb}, the first factor estimates $(\zeta_1^{2-Y} - 1)^{-1}$, which is positive. So, we should expect $\eta_1>0$.
    
    \item With $\widetilde C_{n,pb} '(\varepsilon):=\widetilde C_{n,pb} '(\varepsilon, \zeta_1, p_1)$ defined as in Step 1, the 2-step debiasing estimator removing negative bias is given by:
    \begin{equation} \label{eq:stp2_nb}
    \begin{gathered} 
    \widetilde C_{n,nb} ''(\varepsilon, \zeta_2, \zeta_1,p_2,p_1) = \widetilde C_{n,pb} '(\varepsilon)-  \eta_2 \left(\left(\widetilde C_{n,pb} '(\zeta_2\varepsilon)-\widetilde C_{n,pb} '(\varepsilon)\right) \vee 0 \right),\\
 \text{where} \quad\eta_2 = \frac{\widetilde C_{n,pb} '(p_2\, \zeta_2\varepsilon,)-\widetilde C_{n,pb} '(p_2\, \varepsilon)}{\widetilde C_{n,pb} '(p_2\, \zeta_2^2\varepsilon)-2 \widetilde C_{n,pb} '(p_2\, \zeta_2\varepsilon) + \widetilde C_{n,pb} '(p_2\, \varepsilon)} \wedge 0,
 \end{gathered}
 \end{equation}
with $\zeta_1=1.2$, $\zeta_2=1.2$, $p_1=0.65$, and $p_2=0.75$.
    If it turns out that $\wt C_{n,nb} ''(\varepsilon, \zeta_2, \zeta_1,p_2,p_1)$ is negative, we recompute $\eta_2$ with $\varepsilon= 2 \varepsilon/3$.  The reason for this adjustment is the fact that $\eta_2$ is expected to be negative since it serves as estimate of $(\zeta_2^{-Y} - 1)^{-1}$. The values of the tuning parameters $\zeta_1, \zeta_2, p_1, p_2$ were selected for  `overall favorable' performance for all considered values of $Y$ and $\sigma$.
\end{enumerate}

We further compare the simulated performance of the above estimators to the estimators proposed in \cite{JacodTodorov:2014} {and \cite{Mies:2020}}. Specifically, we use the equation (5.3) in the paper \cite{JacodTodorov:2014}:
$$
\begin{gathered}
    \wh C_{\text{JT}, 53}(u_n, \zeta) = \widehat C_{\text{JT}}(u_n)- \eta \left( \left(\widehat C_{\text{JT}}(\zeta u_n)-\widehat C_{\text{JT}}(u_n)\right)\wedge 0\right),\\
 \text{with}\quad \eta = \frac{\widehat C_{\text{JT}}(p_0\, \zeta u_n)-\widehat C_{\text{JT}}(p_0\, u_n)}{\widehat C_{\text{JT}}(p_0\, \zeta^2u_n)-2\widehat C_{\text{JT}}(p_0\, \zeta u_n) + \widehat C_{\text{JT}}(p_0\, u_n)} \wedge 0,
\end{gathered}
$$
where $\widehat C_{\text{JT}}$ denotes their ``nonsymmetrized'' two-step debiased estimator (see Eq.~(3.1) in therein). For the tuning parameters, we take
$$
	\zeta=1.5, \quad u_{n}=(\ln(1/h_{n}))^{-1/30}/\sigma_{BV}, \quad p_0=0.2,
$$
where the values of $\zeta$ and $u_n$ were those  suggested by \cite{JacodTodorov:2014} and the value of $p_0$ was chosen for favorable estimation performance.
Note that, since a L\'evy model has constant volatility, it is not necessary to localize the estimator and, hence, we treat the 1-year data as one block, which corresponds to taking $k_n=252(6.5)(60)$ in the notation of \cite{JacodTodorov:2014}.  
For the moment estimator proposed in \cite{Mies:2020}, we use the same moment functions and the parameter settings as suggested by \cite{Mies:2020}. We denote this estimator $\wh C_{\text{M},4}$. We also examine the performance of another moment estimator, denoted $\wh C_{\text{M},3}$, that is computed under a similar algorithm to \cite{Mies:2020} but with $3$ different moment functions suggested in \cite{gong2021}.  
We remark that that the moment functions used in the construction of $\wh C_{\text{M},3}$ do not satisfy the strict constraints imposed in \cite{Mies:2020}, and therefore the asymptotic efficiency of the estimator $\wh C_{\text{M},3}$ has not been established. We refer to \cite{gong2021} for more details about the computations of $\wh C_{\text{M},3}$ and $\wh C_{\text{M},4}$.

The sample means, standard deviations (SDs), the average and SD of relative errors, the mean squared errors (MSEs), and median of absolute deviations (MADs) for each of the estimators described above are reported in Tables \ref{table:s2y125}-\ref{table:s4y17}. In addition,   we show the results corresponding to
`case-by-case favorably-tuned' versions of $\wt C_{n,pb}''$, and $\wh C_{\text{JT}, 53}$; i.e., their tuning parameters were chosen for achieving the `best' performance for each pair $(Y,\sigma)$ based on a grid search; we distinguish these estimators and their corresponding tuning parameters by the superscript $^*$.  That is, $\wt C_{n,pb}''$ is based on the choices $(\zeta_1,\zeta_2,p_1,p_2)=(1.2,1.2,.65,.75)$ across all considered values of $Y$ and $\sigma$.   With these values of the tuning parameters, $\wt C_{n,pb}''$ exhibits generally good performance overall.  However, for each given fixed pair $(Y,\sigma)$, its counterpart $\wt C_{n,pb}''^*$ is tuned to have superior performance  for those particular values of $Y$ and $\sigma$. For instance, as shown in Table \ref{table:s2y125}, when $Y=1.25$ and $\sigma=0.2$, the estimator $\wt C_{n,pb}''$ attains an MSE of $1.45\times 10^{-7}$, whereas the choice of parameters  $(\zeta_1^*,\zeta_2^*,p_1^*,p_2^*)=(1.35,1.20,0.5,0.85)$ leads to an MSE of 
$4.70\times 10^{-8}$ for $\wt C_{n,pb}''^*$. 

We provide a broad summary of our simulation results. For the  estimators using the case-by-case tuned parameters $(\zeta_1^*,\zeta_2^*,p_1^*,p_2^*)$, based on both MSE and MAD, $\widetilde C_{n,nb} ''^*$ outperforms every other estimator considered in each setting, except when $Y=0.8$ when the performance between $\widetilde C_{n,nb} ''^*$ and $\wh C_{\text{JT}, 53}^*$ is comparable (c.f. Figure \ref{Fig:density}, top row) but $\wh C_{\text{JT}, 53}^*$ has a slight edge in both MAD and MSE when $\sigma=0.4$ and a slight edge in MAD when $\sigma=0.2$. 
 Using the parameters $(\zeta_1,\zeta_2,p_1,p_2)=(1.2,1.2,.65,.75)$, compared to the method of \cite{Mies:2020},  $\widetilde C_{n,nb} ''$ outperforms $\wh C_{\text{M},4}$ in all cases, as measured by MSE and MAD in Tables \ref{table:s2y125}-\ref{table:s4y17}.
When $\sigma=0.2$ and $Y=1.25$, or when $\sigma=0.4$ and $Y=1.25, 1.35, 1.5, 1.7$, $\widetilde C_{n,nb} ''$ outperforms $\wh C_{\text{M},3}$ (though, as noted earlier, computation of  $\widetilde C_{n,nb} ''$ is much faster and  more numerically stable than that of either $\wh C_{\text{M},3}$ or $\wh C_{\text{M},4}$).
Next, compared with \cite{JacodTodorov:2014}, when $\sigma=0.2$ and $Y=1.25, 1.35,1.5$, or when $\sigma=0.4$ and $Y=1.25$ or $1.35$, $\widetilde C_{n,nb} ''$ has superior performance compared to $\widehat C_{\text{JT}, 53}$ as measured by MSE and MAD. Generally, $\widetilde C_{n,nb} ''$ significantly outperforms $\widetilde C_{n,pb}'$. Table \ref{table:s2y15} also shows that, though when $\sigma=0.2$ and $Y=1.5$, $\widetilde C_{n,nb} ''$ has larger MSE and MAD than $\widehat{C}_n$ and $\widetilde C_{n,pb}'$, it still performs better than $\widehat C_{\text{JT}, 53}$: the MSE and MAD of $\widetilde C_{n,nb} ''$ in these cases are approximately {22\% and 85\%} of those of $\widehat C_{\text{JT}, 53}$, respectively.
Tables \ref{table:s2y17} and \ref{table:s4y17} show that, when $Y=1.7$, which is not covered by our theoretical framework (see Remark \ref{DRemH}), $\widetilde C_{n,nb} ''$ has slightly larger MSE and MAD than $\widehat C_{\text{JT}, 53}$. %
Overall, we conclude that our debiasing procedure outperforms $\widehat C_{\text{JT}, 53}$ when $Y = 1.25, 1.35, 1.5$, and outperforms $\wh C_{\text{M},4}$ in all parameter settings considered, and the estimation performance  of $\widehat C_{\text{JT}, 53}$ and $\widetilde C_{n,nb} ''$ is comparable when $Y=0.8$.
 
We also study the asymptotic approximation for the sampling distribution of $\widetilde C_{n,nb} ''$ and $\wh C_{\text{JT}, 53}$ based on  the case-by-case optimally-tuned estimators $\widetilde C_{n,nb} ''^*$ and $\wh C_{\text{JT}, 53}^*$.
Figure \ref{Fig:density} shows their normalized simulated sampling distributions, $\sqrt{n}(\hat\sigma - \sigma^2)/\sqrt{2\sigma^4}$, and the   theoretical asymptotic normal distribution, $\cN(0,1)$.
Compared to $\wh C_{\text{JT}, 53}^*$, the simulated distribution of $\widetilde C_{n,nb} ''^*$ yields a better match with the asymptotic normal distribution, especially for $Y\leq 1.5$. Note that, in general, the values of $\widetilde C_{n,nb} ''^*$ are much less spread out than those of $\wh C_{\text{JT}, 53}^*$. Though the simulated distributions of $\wh C_{\text{M},3}$ and $\wh C_{\text{M},4}$ are not shown in Figure \ref{Fig:density}, we can see that $\widetilde C_{n,nb} ''^*$ performs much better than $\wh C_{\text{M},3}$ and $\wh C_{\text{M},4}$ as seen in Tables \ref{table:s2y125}-\ref{table:s4y17}.
\begin{figure}
\begin{center}
\includegraphics[scale=0.70]{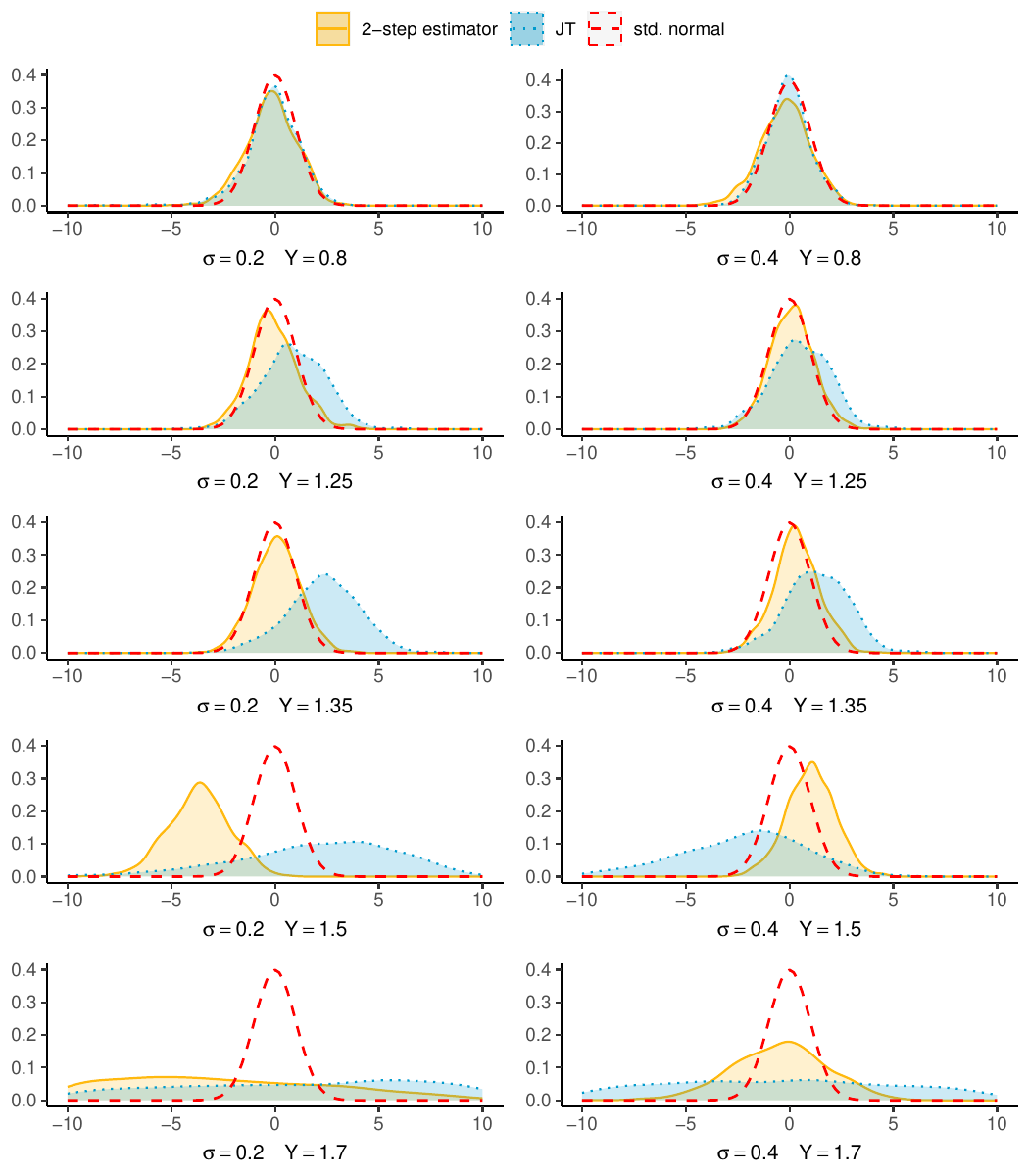}
\end{center}
\caption{\small Simulated and theoretical asymptotic distributions of the estimators based on simulated $1$-minute observations of $1000$ paths over a $1$ year time horizon. The bold dashed red curve is the theoretical asymptotic distribution, $\cN(0,1)$. The solid yellow curve is the simulated distribution of $({2\sigma^4})^{-1/2}\sqrt n(\widetilde C_{n,nb} ''^*-\sigma^2)$. The dotted blue curve is the normalized simulated distribution of $({2\sigma^4})^{-1/2}\sqrt n(\wh C_{\text{JT}}^*-\sigma^2)$.}
\label{Fig:density}
\end{figure}

{

\begin{table}[ht]
\centering
\scalebox{0.92}{
\begin{tabular}{ccccccc}
\hline
\multicolumn{7}{c}{$\sigma=0.2$,\quad $Y=0.8$}\\
\hline
 \hline
& \begin{tabular}{c} Sample\\ Mean\end{tabular} & \begin{tabular}{c} Sample\\ SD \end{tabular} & \begin{tabular}{c} Mean of\\ RE \end{tabular} & \begin{tabular}{c} SD of\\ RE \end{tabular} & \begin{tabular}{c} MSE \end{tabular} & \begin{tabular}{c} MAD \end{tabular}\\
\hline

  $\wh{C}_{n}$ & 0.036881 & 0.0002 & -0.0780 & 0.0050 & 9.77E-06 & 3.12E-03 \\ 
  $\widetilde C_{n,nb} ''^*$ &  0.039902 & 0.0002 & -0.0024 & 0.0047 & \best{4.50E-08} & 1.44E-04 \\ 
  $\wh C_{\text{JT}, 53}^*$  &  0.039997 & 0.0003 & -0.0001 & 0.0070 & 7.80E-08 & \best{1.25E-04} \\ \hline
  $\wh C_{\text{M},3}$ & 0.038777 & 0.0002 & -0.0306 & 0.0055 & 1.55E-06 & 1.23E-03 \\ 
  $\wh C_{\text{M},4}$ & 0.040530 & 0.0003 & 0.0132 & 0.0077 & 3.75E-07 & 5.08E-04 \\ 
  $\wt C_{n,pb}'$ & 0.036881 & 0.0002 & -0.0780 & 0.0050 & 9.77E-06 & 3.12E-03 \\ 
  $\wt C_{n,nb}''$ & 0.038108 & 0.0002 & -0.0473 & 0.0048 & 3.62E-06 & 1.90E-03 \\ 
  $\wh C_{\text{JT}, 53}$ & 0.039916 & 0.0010 & -0.0021 & 0.0241 & 9.40E-07 & 1.39E-04 \\ %
  \hline\\
  \vspace{-.4 cm}
\end{tabular}
}
\caption{\small Estimation based on simulated $1$-minute observations of $1000$ paths over a $1$ year time horizon. The parameters are $Y=1.25$ and $\sigma=0.2$. {The best tuning parameters are $\zeta_1^*=1.4$, $\zeta_2^*=1.35$, $p_1^*=0.5$, $p_2^*=0.85$;  $\zeta^*=1.2$, $p_0^*=0.2$.} 
}
\label{table:s2y08}
\end{table}

\begin{table}[ht]
\centering
\scalebox{0.92}{
\begin{tabular}{ccccccc}
\hline
\multicolumn{7}{c}{$\sigma=0.2$,\quad $Y=1.25$}\\
\hline
 \hline
& \begin{tabular}{c} Sample\\ Mean\end{tabular} & \begin{tabular}{c} Sample\\ SD \end{tabular} & \begin{tabular}{c} Mean of\\ RE \end{tabular} & \begin{tabular}{c} SD of\\ RE \end{tabular} & \begin{tabular}{c} MSE \end{tabular} & \begin{tabular}{c} MAD \end{tabular}\\
\hline
  $\wh{C}_{n}$ & 0.037544 & 0.000205 & -0.0614 & 0.0051 & 6.07E-06 & 2.45E-03 \\ 
  $\widetilde C_{n,nb} ''^*$ & 0.039987 & 0.000216 & -0.0003 & 0.0054 & \best{4.70E-08} & \best{1.40E-04} \\ 
  $\wh C_{\text{JT}, 53}^*$ & 0.040156 & 0.000300 & 0.0039 & 0.0075 & 1.14E-07 & 2.25E-04 \\ \hline
  $\wh C_{\text{M},3}$ & 0.039661 & 0.000268 & -0.0085 & 0.0067 & 1.87E-07 & 3.28E-04 \\
  $\wh C_{\text{M},4}$ & 0.046369 & 0.000939 & 0.1592 & 0.0235 & 4.14E-05 & 6.37E-03 \\ 
  $\wt C_{n,pb}'$ & 0.037544 & 0.000205 & -0.0614 & 0.0051 & 6.07E-06 & 2.45E-03 \\ 
  $\wt C_{n,nb}''$ & 0.040323 & 0.000203 & 0.0081 & 0.0051 & 1.45E-07 & 3.15E-04 \\ 
  $\wh C_{\text{JT}, 53}$ & 0.040742 & 0.000257 & 0.0185 & 0.0064 & 6.16E-07 & 7.60E-04 \\ 
  \hline\\
  \vspace{-.4 cm}
\end{tabular}
}
\caption{\small Estimation based on simulated $1$-minute observations of $1000$ paths over a $1$ year time horizon. The parameters are $Y=1.25$ and $\sigma=0.2$. {The best tuning parameters are $\zeta_1^*=1.35$, $\zeta_2^*=1.2$, $p_1^*=0.5$, $p_2^*=0.85$; $\zeta^*=1.5$, $p_0^*=0.1$.} 
}
\label{table:s2y125}
\end{table}

\begin{table}[ht]
\centering
\scalebox{0.92}{
\begin{tabular}{ccccccc}
\hline
\multicolumn{7}{c}{$\sigma=0.2$,\quad $Y=1.35$}\\
\hline\hline
& \begin{tabular}{c} Sample\\Mean\end{tabular} & \begin{tabular}{c} Sample\\ SD \end{tabular} & \begin{tabular}{c} Mean of\\ RE \end{tabular} & \begin{tabular}{c} SD of\\ RE \end{tabular} & \begin{tabular}{c} MSE \end{tabular} & \begin{tabular}{c} MAD \end{tabular}\\
\hline
  $\wh{C}_{n}$ & 0.038264 & 0.000203 & -0.0434 & 0.0051 & 3.05E-06 & 1.73E-03 \\ 
  $\widetilde C_{n,nb} ''^*$ & 0.040010 & 0.000206 & 0.0003 & 0.0052 & \best{4.27E-08} & \best{1.33E-04} \\ 
  $\wh C_{\text{JT}, 53}^*$ & 0.040418 & 0.000316 & 0.0104 & 0.0079 & 2.74E-07 & 4.30E-04 \\
   \hline
  $\wh C_{\text{M},3}$ & 0.040461 & 0.000277 & 0.0115 & 0.0069 & 2.89E-07 & 4.70E-04 \\ 
  $\wh C_{\text{M},4}$ & 0.050576 & 0.001121 & 0.2644 & 0.0280 & 1.13E-04 & 1.06E-02 \\ 
  $\wt C_{n,pb}'$ & 0.038264 & 0.000203 & -0.0434 & 0.0051 & 3.05E-06 & 1.73E-03 \\ 
  $\wt C_{n,nb}''$ & 0.041044 & 0.000211 & 0.0261 & 0.0053 & 1.13E-06 & 1.05E-03 \\ 
  $\wh C_{\text{JT}, 53}$ & 0.041547 & 0.001111 & 0.0387 & 0.0278 & 3.63E-06 & 1.59E-03 \\ 
  \hline\\
  \vspace{-.4 cm}
\end{tabular}
}
\caption{\small Estimation based on simulated $1$-minute observations of $1000$ paths over a $1$ year time horizon. The parameters are $Y=1.35$ and $\sigma=0.2$. {The best tuning parameters are $\zeta_1^*=1.35$, $\zeta_2^*=1.1$, $p_1^*=0.6$, $p_2^*=0.75$; $\zeta^*=1.5$, $p_0^*=0.1$.} 
}
\label{table:s2y135}
\end{table}

\begin{table}[ht]
\centering
\scalebox{0.92}{
\begin{tabular}{ccccccc}
\hline
\multicolumn{7}{c}{$\sigma=0.2$,\quad $Y=1.5$}\\
\hline\hline
& \begin{tabular}{c} Sample\\ Mean\end{tabular} & \begin{tabular}{c} Sample\\ SD \end{tabular} & \begin{tabular}{c} Mean of\\ RE \end{tabular} & \begin{tabular}{c} SD of\\ RE \end{tabular} & \begin{tabular}{c} MSE \end{tabular} & \begin{tabular}{c} MAD \end{tabular}\\
\hline
  $\wh{C}_{n}$ & 0.041326 & 0.000223 & 0.0332 & 0.0056 & 1.81E-06 & 1.33E-03 \\ 
  $\widetilde C_{n,nb} ''^*$ & 0.039333 & 0.000259 & -0.0167 & 0.0065 & \best{5.12E-07} & 6.58E-04 \\ 
  $\wh C_{\text{JT}, 53}^*$ & 0.040371 & 0.000740 & 0.0093 & 0.0185 & 6.85E-07 & \best{6.20E-04} \\ 
  \hline
  $\wh C_{\text{M},3}$ & 0.044022 & 0.000311 & 0.1006 & 0.0078 & 1.63E-05 & 4.02E-03 \\
  $\wh C_{\text{M},4}$ & 0.063113 & 0.001598 & 0.5778 & 0.0400 & 5.37E-04 & 2.30E-02 \\ 
  $\wt C_{n,pb}'$ & 0.041326 & 0.000223 & 0.0332 & 0.0056 & 1.81E-06 & 1.33E-03 \\ 
  $\wt C_{n,nb}''$ & 0.044331 & 0.000232 & 0.1083 & 0.0058 & 1.88E-05 & 4.33E-03 \\ 
  $\wh C_{\text{JT}, 53}$ & 0.040092 & 0.009267 & 0.0023 & 0.2317 & 8.59E-05 & 5.12E-03 \\ 
  \hline\\
  \vspace{-.4 cm}
\end{tabular}
}
\caption{\small Estimation based on simulated $1$-minute observations of $1000$ paths over a $1$ year time horizon. The parameters are $Y=1.5$ and $\sigma=0.2$. {The best tuning parameters are $\zeta_1^*=1.45$, $\zeta_2^*=1.3$, $p_1^*=0.1$, $p_2^*=0.2$; $\zeta^*=1.5$, $p_0^*=0.1$.}
}
\label{table:s2y15}
\end{table}

\begin{table}[ht]
\centering
\scalebox{0.92}{
\begin{tabular}{ccccccc}
\hline
\multicolumn{7}{c}{$\sigma=0.2$,\quad $Y=1.7$}\\
\hline\hline
& \begin{tabular}{c} Sample\\ Mean\end{tabular} & \begin{tabular}{c} Sample\\ SD \end{tabular} & \begin{tabular}{c} Mean of\\ RE \end{tabular} & \begin{tabular}{c} SD of\\ RE \end{tabular} & \begin{tabular}{c} MSE \end{tabular} & \begin{tabular}{c} MAD \end{tabular}\\
\hline
  $\wh{C}_{n}$ & 0.063772 & 0.000359 & 0.5943 & 0.0090 & 5.65E-04 & 2.38E-02 \\ 
  $\widetilde C_{n,nb} ''^*$ & 0.038238 & 0.001842 & -0.0441 & 0.0461 & \best{6.50E-06} &\best{1.91E-03} \\ 
  $\wh C_{\text{JT}, 53}^*$ & 0.039833 & 0.006706 & -0.0042 & 0.1677 & 4.50E-05 & 4.38E-03 \\ \hline
  $\wh C_{\text{M},3}$ & 0.069247 & 0.000469 & 0.7312 & 0.0117 & 8.56E-04 & 2.93E-02 \\ 
  $\wh C_{\text{M},4}$ & 0.110533 & 0.002272 & 1.7633 & 0.0568 & 4.98E-03 & 7.05E-02 \\ 
  $\wt C_{n,pb}'$ & 0.063772 & 0.000359 & 0.5943 & 0.0090 & 5.65E-04 & 2.38E-02 \\ 
  $\wt C_{n,nb}''$ & 0.069421 & 0.001187 & 0.7355 & 0.0297 & 8.67E-04 & 2.93E-02 \\ 
  $\wh C_{\text{JT}, 53}$ & 0.066487 & 0.000488 & 0.6622 & 0.0122 & 7.02E-04 & 2.65E-02 \\ 
  \hline\\
  \vspace{-.4 cm}
\end{tabular}
}
\caption{\small Estimation based on simulated $1$-minute observations of $1000$ paths over a $1$ year time horizon. The parameters are $Y=1.7$ and $\sigma=0.2$. {The best tuning parameters are $\zeta_1^*=1.3$, $\zeta_2^*=1.3$, $p_1^*=0.4$, $p_2^*=0.8$; $\zeta^*=2$, $p_0^*=0.4$.} {The parameters for $\wt C_{n,pb}'$ and $\wt C_{n,nb}''$ are $\zeta_1=1.2$, $\zeta_2=1.2$, $p_1=0.6$, and $p_2=0.85$.}}
\label{table:s2y17}
\end{table}

\begin{table}[ht]
\centering
\scalebox{0.92}{
\begin{tabular}{ccccccc}
\hline
\multicolumn{7}{c}{$\sigma=0.4$,\quad $Y=0.8$}\\
\hline
 \hline
& \begin{tabular}{c} Sample\\ Mean\end{tabular} & \begin{tabular}{c} Sample\\ SD \end{tabular} & \begin{tabular}{c} Mean of\\ RE \end{tabular} & \begin{tabular}{c} SD of\\ RE \end{tabular} & \begin{tabular}{c} MSE \end{tabular} & \begin{tabular}{c} MAD \end{tabular}\\
\hline
%
  $\wh{C}_{n}$ & 0.147467 & 0.0008 & -0.0783 & 0.0049 & 1.58E-04 & 1.26E-02 \\ 
  $\widetilde C_{n,nb} ''^*$ & 0.159698 & 0.0008 & -0.0019 & 0.0047 & {6.66E-07} & 5.17E-04 \\ 
  $\wh C_{\text{JT}, 53}^*$  &0.159963 & 0.0007 & -0.0002 & 0.0047 & \best{5.61E-07} & \best{4.75E-04} \\ 
  \hline
  $\wh C_{\text{M},3}$ & 0.160908 & 0.0008 & 0.0057 & 0.0051 & 1.49E-06 & 9.39E-04 \\ 
  $\wh C_{\text{M},4}$ & 0.160531 & 0.0008 & 0.0033 & 0.0049 & 9.00E-07 & 6.34E-04 \\ 
  $\wt C_{n,pb}'$ & 0.147467 & 0.0008 & -0.0783 & 0.0049 & 1.58E-04 & 1.26E-02 \\ 
  $\wt C_{n,nb}''$ & 0.152373 & 0.0007 & -0.0477 & 0.0047 & 5.87E-05 & 7.61E-03 \\ 
  $\wh C_{\text{JT}, 53}$ & 0.159312 & 0.0038 & -0.0043 & 0.0235 & 1.47E-05 & 5.73E-04 \\ 
  \hline\\
  \vspace{-.4 cm}
\end{tabular}
}
\caption{\small Estimation based on simulated $1$-minute observations of $1000$ paths over a $1$ year time horizon. The parameters are $Y=1.25$ and $\sigma=0.2$. {The best tuning parameters are $\zeta_1^*=1.4$, $\zeta_2^*=1.35$, $p_1^*=0.5$, $p_2^*=0.85$; $\zeta^*=1.5$, $p_0^*=0.1$.} 
}
\label{table:s2y08}
\end{table}

\begin{table}[ht]
\centering
\scalebox{0.92}{
\begin{tabular}{ccccccc}
\hline
\multicolumn{7}{c}{$\sigma=0.4$,\quad $Y=1.25$}\\
\hline\hline
& \begin{tabular}{c} Sample\\ Mean\end{tabular} & \begin{tabular}{c} Sample\\ SD \end{tabular} & \begin{tabular}{c} Mean of\\ RE \end{tabular} & \begin{tabular}{c} SD of\\ RE \end{tabular} & \begin{tabular}{c} MSE \end{tabular} & \begin{tabular}{c} MAD \end{tabular}\\
\hline
  $\wh{C}_{n}$ & 0.148523 & 0.000775 & -0.0717 & 0.0048 & 1.32E-04 & 1.14E-02 \\ 
  $\widetilde C_{n,nb} ''^*$ & 0.160102 & 0.000754 & 0.0006 & 0.0047 & \best{5.79E-07} & \best{5.24E-04} \\ 
  $\wh C_{\text{JT}, 53}^*$ & 0.160376 & 0.001062 & 0.0023 & 0.0066 & 1.27E-06 & 7.70E-04 \\ \hline
  $\wh C_{\text{M},3}$ & 0.162777 & 0.000825 & 0.0174 & 0.0052 & 8.39E-06 & 2.78E-03 \\ 
  $\wh C_{\text{M},4}$ & 0.166381 & 0.001159 & 0.0399 & 0.0072 & 4.21E-05 & 6.32E-03 \\ 
  $\wt C_{n,pb}'$ & 0.148523 & 0.000775 & -0.0717 & 0.0048 & 1.32E-04 & 1.14E-02 \\ 
  $\wt C_{n,nb}''$ & 0.159879 & 0.000785 & -0.0008 & 0.0049 & 6.31E-07 & 5.34E-04 \\ 
  $\wh C_{\text{JT}, 53}$ & 0.160376 & 0.001062 & 0.0023 & 0.0066 & 1.27E-06 & 7.70E-04 \\ 
   \hline\\
  \vspace{-.4 cm}
\end{tabular}
}
\caption{\small Estimation based on simulated $1$-minute observations of $1000$ paths over a $1$ year time horizon. The parameters are $Y=1.25$ and $\sigma=0.4$. {The best tuning parameters are $\zeta_1^*=1.35$, $\zeta_2^*=1.3$, $p_1^*=0.5$, $p_2^*=0.85$; $\zeta^*=1.5$, $p_0^*=0.2$.}
}
\label{table:s4y125}
\end{table}

\begin{table}[ht]
\centering
\scalebox{0.92}{
\begin{tabular}{ccccccc}
\hline
\multicolumn{7}{c}{$\sigma=0.4$,\quad $Y=1.35$}\\
\hline\hline
& \begin{tabular}{c} Sample\\ Mean\end{tabular} & \begin{tabular}{c} Sample\\ SD \end{tabular} & \begin{tabular}{c} Mean of\\ RE \end{tabular} & \begin{tabular}{c} SD of\\ RE \end{tabular} & \begin{tabular}{c} MSE \end{tabular} & \begin{tabular}{c} MAD \end{tabular}\\
\hline
  $\wh{C}_{n}$ & 0.149582 & 0.000788 & -0.0651 & 0.0049 & 1.09E-04 & 1.04E-02 \\ 
  $\widetilde C_{n,nb} ''^*$ & 0.160284 & 0.000757 & 0.0018 & 0.0047 & \best{6.53E-07} & \best{5.65E-04} \\ 
  $\wh C_{\text{JT}, 53}^*$ & 0.160950 & 0.001125 & 0.0059 & 0.0070 & 2.17E-06 & 1.09E-03 \\ 
  \hline 
  $\wh C_{\text{M},3}$ & 0.164366 & 0.000860 & 0.0273 & 0.0054 & 1.98E-05 & 4.38E-03 \\ 
  $\wh C_{\text{M},4}$ & 0.170625 & 0.001292 & 0.0664 & 0.0081 & 1.15E-04 & 1.06E-02 \\ 
  $\wt C_{n,pb}'$ & 0.149582 & 0.000788 & -0.0651 & 0.0049 & 1.09E-04 & 1.04E-02 \\ 
  $\wt C_{n,nb}''$ & 0.160937 & 0.000758 & 0.0059 & 0.0047 & 1.45E-06 & 9.54E-04 \\ 
  $\wh C_{\text{JT}, 53}$ & 0.160950 & 0.001125 & 0.0059 & 0.0070 & 2.17E-06 & 1.09E-03 \\
   \hline\\
  \vspace{-.4 cm}
\end{tabular}
}
\caption{\small Estimation based on simulated $1$-minute observations of $1000$ paths over a $1$ year time horizon. The parameters are $Y=1.35$ and $\sigma=0.4$. {The best tuning parameters are $\zeta_1^*=1.2$, $\zeta_2^*=1.2$, $p_1^*=0.6$, $p_2^*=0.85$; $\zeta^*=1.5$, $p_0^*=0.2$.}
}
\label{table:s4y135}
\end{table}

\begin{table}[ht]
\centering
\scalebox{0.92}{
\begin{tabular}{ccccccc}
\hline
\multicolumn{7}{c}{$\sigma=0.4$,\quad $Y=1.5$}\\
\hline\hline
& \begin{tabular}{c} Sample\\ Mean\end{tabular} & \begin{tabular}{c} Sample\\ SD \end{tabular} & \begin{tabular}{c} Mean of\\ RE \end{tabular} & \begin{tabular}{c} SD of\\ RE \end{tabular} & \begin{tabular}{c} MSE \end{tabular} & \begin{tabular}{c} MAD \end{tabular}\\
\hline
  $\wh{C}_{n}$ & 0.153623 & 0.000811 & -0.0399 & 0.0051 & 4.13E-05 & 6.39E-03 \\ 
  $\widetilde C_{n,nb} ''^*$ & 0.160721 & 0.000826 & 0.0045 & 0.0052 & \best{1.20E-06} & \best{8.13E-04} \\ 
  $\wh C_{\text{JT}, 53}^*$ & 0.158146 & 0.002499 & -0.0116 & 0.0156 & 9.68E-06 & 1.89E-03 \\ \hline
  $\wh C_{\text{M},3}$ & 0.168253 & 0.000981 & 0.0516 & 0.0061 & 6.91E-05 & 8.24E-03 \\ 
  $\wh C_{\text{M},4}$ & 0.183132 & 0.001785 & 0.1446 & 0.0112 & 5.38E-04 & 2.31E-02 \\ 
  $\wt C_{n,pb}'$ & 0.153623 & 0.000811 & -0.0399 & 0.0051 & 4.13E-05 & 6.39E-03 \\ 
  $\wt C_{n,nb}''$ & 0.165112 & 0.000800 & 0.0319 & 0.0050 & 2.68E-05 & 5.13E-03 \\ 
  $\wh C_{\text{JT}, 53}$ & 0.163082 & 0.001305 & 0.0193 & 0.0082 & 1.12E-05 & 3.08E-03 \\ 
   \hline\\
  \vspace{-.4 cm}
\end{tabular}
}
\caption{\small Estimation based on simulated $1$-minute observations of $1000$ paths over a $1$ year time horizon. The parameters are $Y=1.5$ and $\sigma=0.4$. {The best tuning parameters are $\zeta_1^*=1.35$, $\zeta_2^*=1.1$, $p_1^*=0.5$, $p_2^*=0.9$; $\zeta^*=1.4$, $p_0^*=0.2$.}
}
\label{table:s4y15}
\end{table}

\begin{table}[ht]
\centering
\scalebox{0.92}{
\begin{tabular}{ccccccc}
\hline
\multicolumn{7}{c}{$\sigma=0.4$,\quad $Y=1.7$}\\
\hline\hline
& \begin{tabular}{c} Sample\\ Mean\end{tabular} & \begin{tabular}{c} Sample\\ SD \end{tabular} & \begin{tabular}{c} Mean of\\ RE \end{tabular} & \begin{tabular}{c} SD of\\ RE \end{tabular} & \begin{tabular}{c} MSE \end{tabular} & \begin{tabular}{c} MAD \end{tabular}\\
\hline
  $\wh{C}_{n}$ & 0.178820 & 0.000946 & 0.1176 & 0.0059 & 3.55E-04 & 1.88E-02 \\ 
  $\widetilde C_{n,nb} ''^*$ & 0.159715 & 0.001594 & -0.0018 & 0.0100 & \best{2.62E-06} & \best{1.09E-03} \\ 
  $\wh C_{\text{JT}, 53}^*$ & 0.151628 & 0.013172 & -0.0523 & 0.0823 & 2.44E-04 & 6.87E-03 \\ \hline
  $\wh C_{\text{M},3}$ & 0.192848 & 0.001222 & 0.2053 & 0.0076 & 1.08E-03 & 3.29E-02 \\
  $\wh C_{\text{M},4}$ & 0.230606 & 0.002466 & 0.4413 & 0.0154 & 4.99E-03 & 7.05E-02 \\ 
  $\wt C_{n,pb}'$ & 0.178820 & 0.000946 & 0.1176 & 0.0059 & 3.55E-04 & 1.88E-02 \\ 
  $\wt C_{n,nb}''$ & 0.192674 & 0.000972 & 0.2042 & 0.0061 & 1.07E-03 & 3.27E-02 \\ 
  $\wh C_{\text{JT}, 53}$ & 0.151628 & 0.013172 & -0.0523 & 0.0823 & 2.44E-04 & 6.87E-03 \\ 
   \hline\\
  \vspace{-.4 cm}
\end{tabular}
}
\caption{\small Estimation based on simulated $1$-minute observations of $1000$ paths over a $1$ year time horizon. The parameters are $Y=1.7$ and $\sigma=0.4$. {The best tuning parameters are $\zeta_1^*=1.2$, $\zeta_2^*=1.35$, $p_1^*=0.2$, $p_2^*=0.3$; $\zeta^*=1.5$, $p_0^*=0.2$.}
}
\label{table:s4y17}
\end{table}
}

\subsection{Stochastic volatility}
In this section, we apply our two-step debiasing procedure to estimate the daily integrated variance under a stochastic volatility model with a CGMY jump component and compare it with the estimator of Jacod and Todorov \cite{JacodTodorov:2014}. 

Specifically, we consider the following Heston model:
\begin{align*}
X_{t}=1+\int_{0}^{t}\sqrt{V_{s}}\,dW_{s}+J_{t},\quad V_{t}=\theta+\int_{0}^{t}\kappa\big(\theta-V_{s}\big)\,ds+\xi\int_{0}^{t}\sqrt{V_{s}}\,dB_{s},
\end{align*}
where $\{W_{t}\}_{t\geq 0}$ and $\{B_{t}\}_{t\geq 0}$ are two {correlated} standard Brownian motions {with correlation $\rho$} and $\{J_{t}\}_{t\geq 0}$ is a CGMY L\'{e}vy process independent of $\{W_{t}\}_{t\geq 0}$ and $\{B_{t}\}_{t\geq 0}$. The parameters are set as
\begin{align*}
\kappa=5,\quad\xi=0.5,\quad\theta=0.16, \quad {\rho = -0.5}.
\end{align*}
The values of $\kappa$, $\xi$, and {$\rho$} above are borrowed from \cite{ZhangMyklandAitSahalia:2005}. The CGMY parameters are the same as those in the previous section.

We consider $1$-min observations over a one-year ($252$ days) time horizon with $6.5$ trading hours per day. We break each path into $252$ blocks (one for each day) and estimate the integrated volatility $IV = \int_{t}^{t+1/252}V_{s}ds$ for each day ($t=0,1/252,\dots,251/252$). As suggested and used in \cite{JacodTodorov:2014}, 
 to improve the stability of the estimates, the estimated bias terms in \eqref{eq:stp1_pb} and \eqref{eq:stp2_nb} are split into two components each: $\big(\widehat C_n(\zeta_1\varepsilon)-\widehat C_n(\varepsilon)\big)$ and $\big(\widetilde C_{n,pb} '(\zeta_2\varepsilon, \zeta_1, p_1)-\widetilde C_{n,pb} '(\varepsilon, \zeta_1, p_1)\big)$.  These are computed
using the data in each day, and the factors $\eta_1$, $\eta_2$, which only depend on  $Y$,  are computed using the data during the whole time horizon. In practice, one would precompute $\eta_1$ and $\eta_2$ using historical data over 1 year and use those values to compute the daily integrated volatility afterward. The precise formulas for our estimators are described below:
\begin{enumerate} 
    
    \item 1-step debiasing estimator removing positive bias:
    \begin{align*}
    &\overline C_{n,pb}'(\varepsilon, \zeta_1)_{t} = \widehat C_n(\varepsilon)_{t}- \eta_1 \left(\widehat C_n(\zeta_1\varepsilon)_{t}-\widehat C_n(\varepsilon)_{t}\right),%
    \\
    &\eta_1 = \frac{\sum_{i=0}^{251}\left(\widehat C_n(p_1\, \zeta
    _1\varepsilon)_{\frac{i}{252}}-\widehat C_n(p_1\, \varepsilon)_{\frac{i}{252}}\right)}{\sum_{i=0}^{251}\left(\widehat C_n(p_1\, \zeta_1^2\varepsilon)_{\frac{i}{252}}-2\widehat C_n(p_1\, \zeta_1 \varepsilon)_{\frac{i}{252}} + \widehat C_n(p_1\, \varepsilon)_{\frac{i}{252}}\right)} \vee 0,%
    \end{align*}
    \normalsize
    
    \item With the estimator $\overline C_{n,pb}'(\varepsilon)_{t}:=\overline C_{n,pb}'(\varepsilon, \zeta_1)_{t}$ defined in Step 1 above, the 2-step debiasing estimator removing negative bias is given by:\small
    \begin{align*}
    &\overline C_{n,nb} ''(\varepsilon, \zeta_2, \zeta_1)_{t} = \overline C_{n,pb} '(\varepsilon)_{t}-  \eta_2 \left(\left(\overline C_{n,pb} '(\zeta_2\varepsilon)_{t}-\overline C_{n,pb} '(\varepsilon)_{t}\right) \vee 0 \right),\\
    &\eta_2 = \frac{\sum_{i=0}^{251}\left(\overline C_{n,pb} '(p_2\, \zeta_2\varepsilon)_{\frac{i}{252}}-\overline C_{n,pb} '(p_2\, \varepsilon)_{\frac{i}{252}}\right)}{\sum_{i=0}^{251}\left(\overline C_{n,pb} '(p_2\, \zeta_2^2\varepsilon)_{\frac{i}{252}}-2 \overline C_{n,pb} '(p_2\, \zeta_2\varepsilon)_{\frac{i}{252}} + \overline C_{n,pb} '(p_2\, \varepsilon)_{\frac{i}{252}}\right)} \wedge 0,
\end{align*}
\normalsize
with the same parameters used in the previous subsection $\zeta_1=1.2$, $\zeta_2=1.2$, $p_1=0.65$, and  $p_2=0.75$.
\end{enumerate}

For the estimator of {\cite{JacodTodorov:2014}}, we use equation (5.3) therein with tuning parameter $k_{n}=130$ (number of observation in each block), $\xi=1.5$, and $u_{n}=(-\ln h_{n})^{-1/30}/\sqrt{BV}$ (these values were suggested in \cite{JacodTodorov:2014}). Here, $BV$ is the bipower variation of the previous day. The resulting estimator is denoted as $\overline C_{\text{JT}, 53}$. To assess the accuracy of the different methods, we compute the Median Absolute Deviation (MAD) around the true value, $IV_t=\int_{t}^{t+1/252}V_{s}ds$, and the MSE, i.e. the sample mean of $(\widehat{IV}_t - IV_t)^2$, for 5 arbitrarily chosen days over $1000$ simulation paths. 

The results are shown in Table \ref{tab9}. The last column of Table \ref{tab9} shows the sample means of the MSE and MAD over the 252 days for the different estimators. When $Y=0.8,1.25, 1.35, 1.5$, all MSEs and MADs of $\overline C_{n,nb} ''$ are smaller than those of $\overline C_{\text{JT}, 53}$, with the MSE of $\overline C_{n,nb} ''$ typically half of that of $\overline C_{\text{JT}, 53}$ or smaller. For the case $Y=1.7$ (outside the scope of our theoretical framework), the MSE and MAD of $\overline C_{n,nb} ''$ is slightly larger than those of $\overline C_{\text{JT}, 53}$.

This behavior can also be observed in Figure \ref{Fig:path53}, which shows the true daily integrated volatility (dashed red line)  for one fixed simulated path compared with the estimates corresponding to $\overline C_{n,nb} ''$ (solid black line) and $\overline C_{\text{JT}, 53}$ in \cite{JacodTodorov:2014} (dotted blue line). From the figure, we conclude that for this specific stochastic volatility model, our debiasing method achieves significant improvement when $Y\leq{}1.5$. For $Y=1.7$, both estimators $\overline C_{n,nb} ''$ and  $\overline C_{\text{JT}, 53}$ are very close and significantly overestimate the true daily integrated volatility for this simulated path. This changes from path to path, though $\overline C_{n,nb} ''$ and  $\overline C_{\text{JT}, 53}$ are typically close when $Y=1.7$.

\begin{figure}[!hptb]
\begin{center}
\vspace{-0.3cm}

\includegraphics[width=0.95\linewidth]{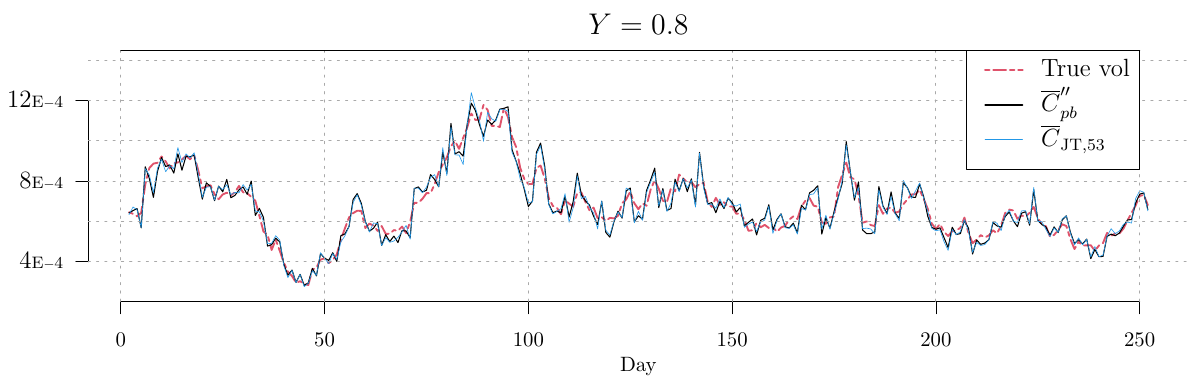}\\ \vspace{1ex}
\includegraphics[width=0.95\linewidth]{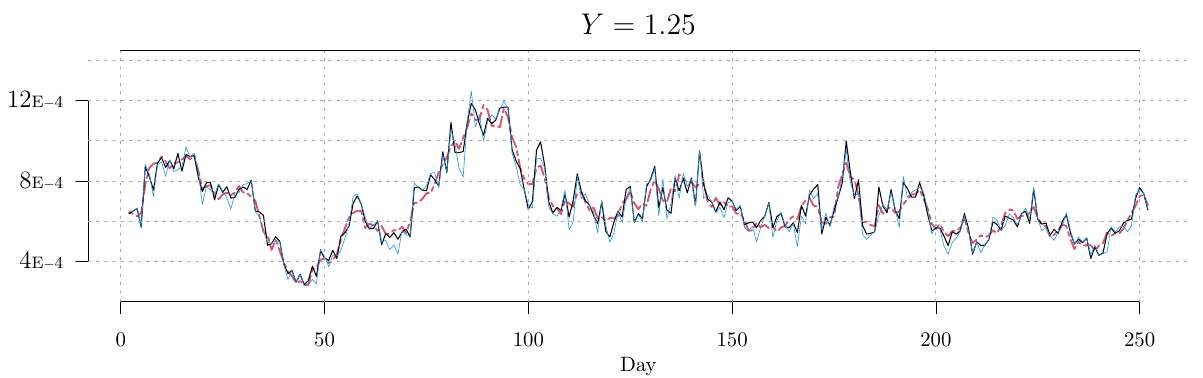}\\ \vspace{1ex}
\includegraphics[width=0.95\linewidth]{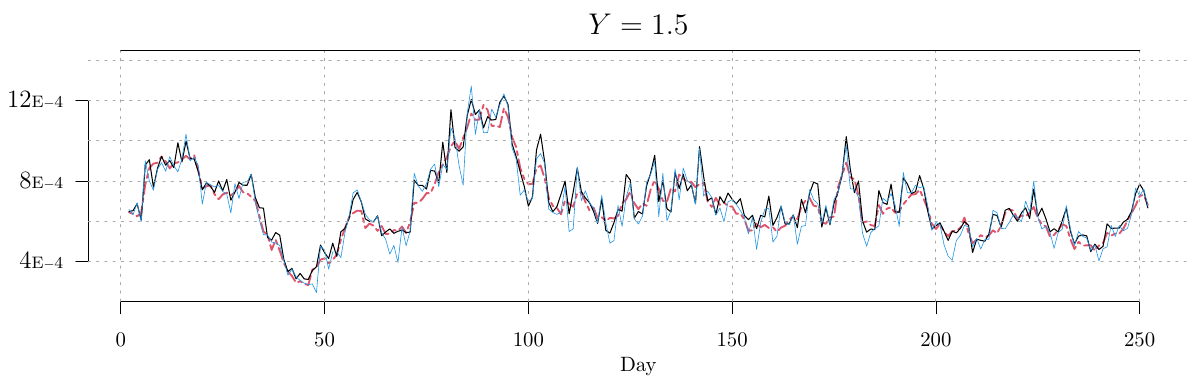}\\ \vspace{1ex}
\includegraphics[width=0.95\linewidth]{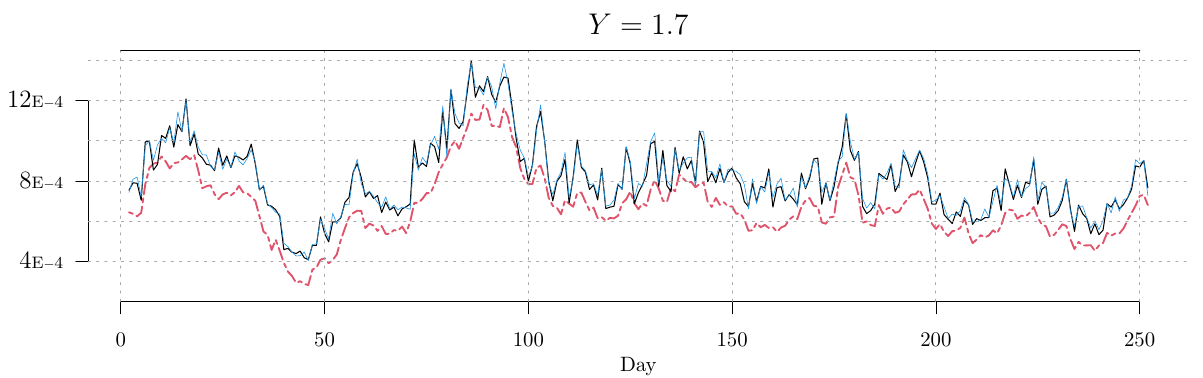}
\end{center}
\caption{\small Plots of daily integrated volatility estimates $\overline C_{n,nb} ''(\varepsilon, \zeta_2, \zeta_1,p_2,p_1)_{(\frac{t}{252}, \frac{t+1}{252}]} $. Except for a change in the jump index $Y$ of the CGMY component, the same path of $X_t$ is for each plot.  Above are the plots for $Y=0.8, 1.25, 1.5, 1.7$ ($Y=1.35$ is similar to $Y=1.25$ and was removed to save space).  The dashed red line corresponds to the true daily integrated volatility, while the solid black (respectively, blue) line corresponds to the daily estimates using our debiasing estimator $\overline C_{n,nb} ''$ (respectively, the estimator given in expression (5.3) in \cite{JacodTodorov:2014}).} \label{Fig:path53}
\end{figure}


\begin{table}[ht!]
\vspace{0.5cm}
\centering
\scalebox{0.81}{
\begin{tabular}{c|cccccccc}
\hline
\multicolumn{9}{c}{Estimation performance in a Heston model}\\ 
\hline\hline
$Y$ & \multicolumn{2}{c}{Method} & Day 2 & Day 63 & Day 126 & Day 189 & Day 252 & {Average}\\
\hline
\multirow{4}{*}{$0.8$} 
& \multirow{2}{*}{$\overline C_{n,nb} ''$} 
&MSE & 2.38E-09 & 2.77E-09 & 2.64E-09 & 2.64E-09 & 2.39E-09 & \best{2.57E-09} \\ 
 &&MAD &  2.86E-05 & 2.92E-05 & 2.90E-05 & 3.02E-05 & 2.82E-05 & \best{2.93E-05} \\ 
  &\multirow{2}{*}{$\overline C_{\text{JT}, 53}$}
  &MSE  & 7.89E-09 & 8.30E-09 & 8.09E-09 & 8.98E-09 & 8.17E-09 & 8.68E-09 \\ 
  &&MAD & 3.34E-05 & 3.51E-05 & 3.23E-05 & 3.38E-05 & 3.34E-05 & 3.40E-05 \\  \hline
\multirow{4}{*}{$1.25$} 
  & \multirow{2}{*}{$\overline C_{n,nb} ''$}& MSE &  2.50E-09 & 2.37E-09 & 2.48E-09 & 2.57E-09 & 2.45E-09 & \best{2.45E-09}\\ 
  & & MAD & 2.90E-05 & 2.93E-05 & 2.98E-05 & 2.91E-05 & 2.89E-05 & \best{2.88E-05} \\ 
  &\multirow{2}{*}{$\overline C_{\text{JT}, 53}$}& MSE &  5.80E-09 & 5.49E-09 & 5.26E-09 & 5.21E-09 & 5.79E-09 & 5.52E-09 \\ 
  & & MAD & 3.49E-05 & 3.43E-05 & 3.45E-05 & 3.26E-05 & 3.44E-05 & 3.40E-05 \\ 
\hline
\multirow{4}{*}{$1.35$} 
  & \multirow{2}{*}{$\overline C_{n,nb} ''$}& MSE & 2.52E-09 & 2.40E-09 & 2.55E-09 & 2.55E-09 & 2.45E-09 & \best{2.47E-09} \\ 
   & & MAD & 2.83E-05 & 2.93E-05 & 2.97E-05 & 2.91E-05 & 2.81E-05 & \best{2.90E-05} \\ 
  &\multirow{2}{*}{$\overline C_{\text{JT}, 53}$}& MSE &  7.63E-09 & 7.31E-09 & 6.85E-09 & 6.50E-09 & 7.55E-09 & 7.25E-09 \\ 
  & & MAD & 3.63E-05 & 3.55E-05 & 3.51E-05 & 3.29E-05 & 3.53E-05 & 3.46E-05 \\ 
\hline
\multirow{4}{*}{$1.5$} 
  & \multirow{2}{*}{$\overline C_{n,nb} ''$}& MSE & 3.05E-09 & 3.12E-09 & 3.10E-09 & 3.19E-09 & 3.01E-09 & \best{3.07E-09}\\ 
   & & MAD & 3.25E-05 & 3.48E-05 & 3.37E-05 & 3.30E-05 & 3.22E-05 & \best{3.31E-05} \\ 
  &\multirow{2}{*}{$\overline C_{\text{JT}, 53}$}& MSE &  6.25E-09 & 6.27E-09 & 5.76E-09 & 5.62E-09 & 6.21E-09 & 5.98E-09 \\ 
  & & MAD & 3.65E-05 & 3.76E-05 & 3.85E-05 & 3.53E-05 & 3.60E-05 & 3.66E-05 \\ 
\hline
\multirow{4}{*}{$1.7$} 
  & \multirow{2}{*}{$\overline C_{n,nb} ''$}& MSE & 2.08E-08 & 2.14E-08 & 2.04E-08 & 2.13E-08 & 2.09E-08 & 2.07E-08 \\ 
   & & MAD & 1.23E-04 & 1.28E-04 & 1.26E-04 & 1.30E-04 & 1.24E-04 & 1.26E-04 \\ 
  &\multirow{2}{*}{$\overline C_{\text{JT}, 53}$}& MSE &  1.89E-08 & 1.86E-08 & 1.82E-08 & 1.75E-08 & 1.88E-08 & \best{1.82E-08} \\ 
  & & MAD & 1.06E-04 & 1.07E-04 & 1.05E-04 & 1.04E-04 & 1.06E-04 & \best{1.05E-04} 
  \vspace{-.1 cm}
\end{tabular}
}
\caption{\small The MSE and MADs for $\overline C_{n,nb} ''$ and $\overline C_{\text{JT}, 53}$ . The results are based on simulated $1$-minute observations of $1000$ paths over a one-year time horizon with $Y=0.8,1.25, 1.35, 1.5, 1.7$. The parameters for debiasing method are {$\zeta_1=1.2$, $\zeta_2=1.2$, $p_1=0.65$, and $p_2=0.75$} in all cases.  The smallest average MSE and smallest average MAD is displayed in bold in each row.}\label{tab9}
\end{table}
\newpage

\appendix
\section{Proofs of the main results}\label{AppMnRslt}

Throughout all appendices, we routinely make use of the following fact: under the assumption  $h_n^{4/(8+Y)}\ll \varepsilon_n$ made throughout Section \ref{Sec:DebiasMthd}, it holds  $\exp\big(-\frac{\varepsilon^2_n}{2\sigma^2h_n}\big)\ll h_n^s$ for any $s>0$ . For notational simplicity, we also often omit the subscript $n$ in $h_n$ and $\varepsilon_n$. We denote by $C$ or $K$  generic constants independent of $n$ that may be different from line to line.  In all proofs, based on  Assumption \ref{assump:Coef0}, by a standard localization argument, we may assume without restriction that $|\sigma|, |b|, |b^\chi|, |\chi|, |\Sigma^\chi|$, are almost surely bounded by a nonrandom constant, and $ (|\delta_i(t,z)|\wedge1)^{r_i} \leq K H_i(z) $, $i=0,1$.   Further, we have the following estimates valid for all $s,t>0$:
$$
 \mathbb{E}\left(\left.|\chi_{t+s}-\chi_{t}|^p\right|\mathcal{F}_t\right)\leq{}K s^{p/2}, \quad p\geq 1; \quad \mathbb{E}\left(\left.|\sigma_{t+s}-\sigma_{t}|^2\right|\mathcal{F}_t\right)\leq{}K s.
$$
 In the proofs below, we will show that we can neglect the finite variation jump component $X^{j,0}$ and prove the results for the It\^o semimartingale 
\[
	X':=X-X^{j,0}=X^{c}_t+X_t^{j,\infty};
\]
 i.e., the process consisting of the continuous component $X^{c}_t=\int_0^t b_s ds+\int_0^t \sigma_s dW_s$ and the infinite variation jump component $X_t^{j,\infty}=\int_0^t \chi_{s^-} dJ_s$. 

\begin{proof}[Proof of Proposition \ref{prop1}]
Note that $\wt{Z}_n(\varepsilon)$ in \eqref{e:def_Zn}
 can be decomposed as  follows:
\begin{align*}
    \wt{Z}_{n}(\varepsilon)
    &= \sqrt{n}\,\sum_{i=1}^{n}\left(\big(\Delta_{i}^{n}X\big)^{2}\,{\bf 1}_{\{|\Delta_{i}^{n}X|\leq\varepsilon\}}-\big(\Delta_{i}^{n}X'\big)^{2}\,{\bf 1}_{\{|\Delta_{i}^{n}X'|\leq\varepsilon\}} \right)\\
    &\quad + \sqrt{n}\,\sum_{i=1}^{n}\left(\big(\Delta_{i}^{n}X'\big)^{2}\,{\bf 1}_{\{|\Delta_{i}^{n}X'|\leq\varepsilon\}}- \bE_{i-1}\big[\big(\Delta_{i}^{n}X'\big)^{2}\,{\bf 1}_{\{|\Delta_{i}^{n}X'|\leq\varepsilon\}}\big] \right) \\
    &\quad +\sqrt{n}\,\sum_{i=1}^{n}\left(\bE_{i-1}\big[\big(\Delta_{i}^{n}X'\big)^{2}\,{\bf 1}_{\{|\Delta_{i}^{n}X'|\leq\varepsilon\}}\big]-\sigma_{t_{i-1}}^2h-\widehat{A}_i(\varepsilon,h)h \right)\\
    &\quad +\sqrt{n}\left(\sum_{i=1}^{n}\left(\sigma_{t_{i-1}}^2h+\widehat{A}_i(\varepsilon,h)h\right) -\int_0^{1}\sigma^2_sds-A(\varepsilon,h) \right)\\
    &=: T_0+T_1 + T_2+T_3,
\end{align*}
where $\widehat{A}_i(\varepsilon,h)$ is defined as in \eqref{e:def_hat_a(eps,h)}.   Lemma \ref{l:trucated_moments_extended_to_X} implies $T_0=o_P(1)$, and Lemma \ref{ErroDiscr} directly implies that $T_2=o_P(1)$.
 We also have $T_3=o_P(1)$, which follows from the fact that
	$$\sum_{i=1}^{n}\sigma^2_{t_{i-1}}h-\int_0^1\sigma_s^2ds=o_{P}(n^{-1/2}),$$
	$$\varepsilon^{2-Y}\sum_{i=1}^{n}|\chi_{t_{i-1}}|^{Y}h-\varepsilon^{2-Y}\int_0^1|\chi_s|^Yds=o_{P}(n^{-1/2}),$$
	$$\text{and} \quad h\varepsilon^{-Y}\sum_{i=1}^{n}|\chi_{t_{i-1}}|^{Y}\sigma_{t_{i-1}}^2h-h\varepsilon^{-Y}\int_0^1|\chi_s|^Y\sigma_s^2 ds=o_{P}(n^{-1/2}).$$
The above is established in Lemma \ref{l:Riemann} of \ref{OtherSuperTechProofs}. 
 

We now show $T_1 \toDistSt  \mathcal{N}\big(0, 2\int_0^1\sigma_s^4ds\big)$ by applying the martingale difference CLT (Theorem 2.2.15  of \cite{JacodProtter}), which will complete the proof. Define 
\[
	\xi_i^n := \sqrt{n}\left(\big(\Delta_{i}^{n}X'\big)^{2}{\bf 1}_{\{|\Delta_{i}^{n}X'|\leq\varepsilon\}}-\bE_{i-1}\big(\big(\Delta_{i}^{n}X'\big)^{2}\,{\bf 1}_{\{|\Delta_{i}^{n}X'|\leq\varepsilon\}}\big)\right).
\] 
We first need to show that 
 $V_n:=\sum_{i=1}^n \mathbb{E}_{i-1}\left[(\xi_i^n)^2\right]\stackrel{\mathbb{P}}{\longrightarrow}2\int_0^1\sigma_s^4ds$.
The left-hand side can be written as 
\begin{align*}
	V_n
		=n\sum_{i=1}^n \left(\bE_{i-1}\big[\big(\Delta_{i}^{n}X'\big)^{4}\,{\bf 1}_{\{|\Delta_{i}^{n}X'|\leq\varepsilon\}}\big]
		-\big[\bE_{i-1}\big(\Delta_{i}^{n}X'\big)^{2}\,{\bf 1}_{\{|\Delta_{i}^{n}X'|\leq\varepsilon\}}\big]^2\right).	
\end{align*}
Lemmas \ref{ErroDiscr} and \ref{ErroDiscrHigherOrder} imply that 
	\begin{align*}
		\mathbb{E}_{i-1}\left[(\Delta_i^nX')^2{\bf 1}_{\{|\Delta_i^nX'|\leq{}\varepsilon\}}\right]&= \sigma_{t_{i-1}}^2h+O_P(h\varepsilon^{2-Y}),\\
    \bE_{i-1}\left[\left(\Delta_i^n {X'}\right)^{4}\,{\bf 1}_{\{|\Delta_i^n  {X'}|\leq\varepsilon\}}\right]
    &= 3\sigma^{4}_{t_{i-1}} h^2 + O_P(h\varepsilon^{4-Y}).
\end{align*}
Since, due to our conditions on $\varepsilon$, $h\varepsilon^{4-Y}\ll h^2$, we clearly have
\begin{align*}
	V_n=n\sum_{i=1}^n \left(3\sigma^{4}_{t_{i-1}} h^2 +  o_P(h^2)-\big(\sigma_{t_{i-1}}^2h+ O_P(h\varepsilon^{2-Y})\big)^2\right)\stackrel{\mathbb{P}}{\longrightarrow}2\int_0^1\sigma_s^4ds.
\end{align*}
Next, we show that 
$\sum_{i=1}^n \mathbb{E}_{i-1}\left[(\xi_i^n)^4\right]\stackrel{\mathbb{P}}{\longrightarrow}0$.
 We have
\begin{align*}
	 \mathbb{E}_{i-1}&\left[(\xi_i^n)^4\right]\\
	 & \leq n^2K \Big(\mathbb{E}_{i-1}[(\Delta_{i}^{n}X'\big)^{8}{\bf 1}_{\{|\Delta_{i}^{n}X'|\leq\varepsilon\}}] + 
\big[\bE_{i-1}\big(\Delta_{i}^{n}X'\big)^{2}\,{\bf 1}_{\{|\Delta_{i}^{n}X'|\leq\varepsilon\}}\big]^4\Big).
\end{align*}
Lemma \ref{ErroDiscrHigherOrder} implies that, 
$$\bE_{i-1}\left[\left(\Delta_i^n {X'}\right)^{2k}\,{\bf 1}_{\{|\Delta_i^n  {X'}|\leq\varepsilon\}}\right]=O_P(h^{k}) + O_P(h\varepsilon^{2k-Y}),$$  for $k\geq{}1$. Then, $\sum_{i=1}^{n}\mathbb{E}_{i-1}\left[(\xi_i^n)^4\right] =n^3\big( O_P(h^4)+ O_P(h\varepsilon^{8-Y})\big)=o_P(1)$, since our assumption $\varepsilon\ll h^{\frac{1}{4-Y}}$ implies that $\varepsilon \ll h^{\frac{1}{4-Y/2}}$, which is equivalent to $n^3 h\varepsilon^{8-Y} \ll1$. It remains to check the condition (2.2.40) in \cite{JacodProtter}:
\begin{align}\label{NdTchCnd}
	\sum_{i=1}^{n}\mathbb{E}_{i-1}\left[\xi_{i}^n({\Delta_i^nM})\right]\,\stackrel{\mathbb{P}}{\longrightarrow}\,0,
\end{align}
when $M=W$ or $M$ is a square-integrable martingale orthogonal to $W$. This is proved in Lemma \ref{ErroDiscrHigherOrder22} of \ref{OtherSuperTechProofs}.
\end{proof}
%
\begin{proof}[Proof of Proposition \ref{thm:singleclt}]
We consider the decomposition: 
\begin{align*}
&u_n^{-1}\left(\wt Z_n( \zeta\varepsilon_n) - \wt Z_n( \varepsilon_n)\right)\\
&=\varepsilon_n^{\frac{Y-4}{2}}\sum_{i=1}^{n}\left(\big(\Delta_{i}^{n}X\big)^{2}\,{\bf 1}_{\{\varepsilon<|\Delta_{i}^{n}X|\leq\zeta\varepsilon\}}- \big(\Delta_{i}^{n}X'\big)^{2}\,{\bf 1}_{\{\varepsilon<|\Delta_{i}^{n}X'|\leq\zeta\varepsilon\}} \right) \\
&\quad +\varepsilon_n^{\frac{Y-4}{2}}\sum_{i=1}^{n}\left(\big(\Delta_{i}^{n}X'\big)^{2}\,{\bf 1}_{\{\varepsilon<|\Delta_{i}^{n}X'|\leq\zeta\varepsilon\}}- \bE_{i-1}\big[\big(\Delta_{i}^{n}X'\big)^{2}\,{\bf 1}_{\{\varepsilon<|\Delta_{i}^{n}X'|\leq\zeta\varepsilon\}}\big] \right) \\
    &\quad+\varepsilon_n^{\frac{Y-4}{2}}\,\sum_{i=1}^{n}\left(\bE_{i-1}\big[\big(\Delta_{i}^{n}X'\big)^{2}\,{\bf 1}_{\{\varepsilon<|\Delta_{i}^{n}X'|\leq\zeta\varepsilon\}}\big]-\widehat{A}_i(\zeta\varepsilon,h)h +\widehat{A}_i(\varepsilon,h)h \right)\\
    &\quad+\varepsilon_n^{\frac{Y-4}{2}}\left(\sum_{i=1}^{n}\left(\widehat{A}_i(\zeta\varepsilon,h)-\widehat{A}_i(\varepsilon,h)\right)h -A(\zeta\varepsilon,h)+A(\varepsilon,h) \right)\\
    &=: T_0+T_1 + T_2+T_3,
\end{align*}
 where $\widehat{A}_i(\varepsilon,h)$ is defined as in \eqref{e:def_hat_a(eps,h)}.
From Lemma \ref{l:trucated_moments_extended_to_X} we  have $$T_0= n\varepsilon^{(Y-4)/2}o_P(h\varepsilon^{2-Y/2})=o_P(1).$$
Lemma \ref{ErroDiscr2} shows that $T_2=o_P(1)$. 
To show that $T_3=o_P(1)$, Lemma \ref{l:Riemann} implies
\begin{align*}
\sum_{i=1}^{n}&\left(\widehat{A}_i(\zeta\varepsilon,h)-\widehat{A}_i(\varepsilon,h)\right)h -A(\zeta\varepsilon,h)+A(\varepsilon,h)=o_P(h^{\frac{1}{2}}\varepsilon^{2-Y})+o_P(h^{\frac{3}{2}}\varepsilon^{-Y}),
\end{align*}
and each of the terms above is $o_P(\varepsilon^{\frac{4-Y}{2}})$ since $h^{\frac{1}{2}}\varepsilon^{2-Y}\gg  h^{\frac{3}{2}}\varepsilon^{-Y}$ and {$h^{\frac{1}{2}}\varepsilon^{2-Y}\ll \varepsilon^{\frac{4-Y}{2}}$, which is implied from our assumption $h^{\frac{4}{8+Y}}\ll \varepsilon$ since $\frac{1}{Y}>\frac{4}{8+Y}$.}
It remains to show $T_1 \toDistSt  \mathcal{N}\big(0, 2\int_0^1\sigma_s^4ds\big)$, for which we shall again use Theorem 2.2.15 in \cite{JacodProtter}. Define 
\[
	\tilde\xi_i^n := \varepsilon^{\frac{Y-4}{2}}\left(\big(\Delta_{i}^{n}X'\big)^{2}{\bf 1}_{\{\varepsilon<|\Delta_{i}^{n}X'|\leq\zeta\varepsilon\}}-\bE_{i-1}\big(\big(\Delta_{i}^{n}X'\big)^{2}\,{\bf 1}_{\{\varepsilon<|\Delta_{i}^{n}X'|\leq\zeta\varepsilon\}}\big)\right).
\] 
We first need to compute $V_n:=\sum_{i=1}^n \mathbb{E}_{i-1}\left[(\tilde\xi_i^n)^2\right]$.  We clearly have
\begin{align*}
	V_n
		&=\varepsilon^{Y-4}\sum_{i=1}^n \left(\bE_{i-1}\big[\big(\Delta_{i}^{n}X'\big)^{4}\,{\bf 1}_{\{\varepsilon<|\Delta_{i}^{n}X'|\leq\zeta\varepsilon\}}\big]\right.\\
		&\qquad\qquad\quad \left.
		-\bE_{i-1}\big[\big(\Delta_{i}^{n}X'\big)^{2}\,{\bf 1}_{\{\varepsilon<|\Delta_{i}^{n}X'|\leq\zeta\varepsilon\}}\big]^2\right).	
\end{align*}
Lemmas \ref{ErroDiscr} and \ref{ErroDiscrHigherOrder} imply that 
	\begin{align*}
		\mathbb{E}_{i-1}(\Delta_i^nX')^2{\bf 1}_{\{\varepsilon<|\Delta_i^nX'|\leq{}\zeta\varepsilon\}}&= 
		O_P(h\varepsilon^{2-Y}),\\
    \bE_{i-1}\left(\Delta_i^n {X'}\right)^{4}\,{\bf 1}_{\{\varepsilon<|\Delta_i^n  {X'}|\leq\zeta\varepsilon\}}
    &= \frac{\bar{C}|\chi_{t_{i-1}}|^{Y}}{4-Y}\,h\varepsilon^{4-Y}(\zeta^{4-Y}-1)+ o_P(h\varepsilon^{4-Y}).
\end{align*}
Therefore,\small 
\begin{align*}
	V_n&=\varepsilon^{Y-4}\sum_{i=1}^n \left( \frac{\bar{C}}{4-Y}|\chi_{t_{i-1}}|^{Y}\,h\varepsilon^{4-Y}(\zeta^{4-Y}-1)+o_P(h\varepsilon^{4-Y})+O_P(h\varepsilon^{2-Y})^2\right)\\
	&\quad\stackrel{\mathbb{P}}{\longrightarrow} \frac{\bar{C}}{4-Y}(\zeta^{4-Y}-1)\int_0^1|\chi_s|^Yds,
\end{align*}
\normalsize
since $\varepsilon^{Y-4}n(h\varepsilon^{2-Y})^2 
=
	h\varepsilon^{-Y}\ll 1$.
Next, we show 
$\sum_{i=1}^n \mathbb{E}_{i-1}(\tilde\xi_i^n)^4\stackrel{\mathbb{P}}{\longrightarrow}0$. 
We have 
\small
\begin{align*}
	 \mathbb{E}_{i-1}(\tilde\xi_i^n)^4 \leq K \varepsilon^{2Y-8}\Big(\mathbb{E}_{i-1}[&(\Delta_{i}^{n}X'\big)^{8}{\bf 1}_{\{\varepsilon<|\Delta_{i}^{n}X'|\leq\zeta\varepsilon\}}]\\
	 &
	+\big[\bE_{i-1}\big(\Delta_{i}^{n}X'\big)^{2}\,{\bf 1}_{\{\varepsilon<|\Delta_{i}^{n}X'|\leq\zeta\varepsilon\}}\big]^4\Big).
\end{align*}
\normalsize
Lemma \ref{ErroDiscrHigherOrder} implies that for $k\geq{}1$, 
$$\bE_{i-1}\left(\Delta_i^n {X'}\right)^{2k}\,{\bf 1}_{\{\varepsilon<|\Delta_{i}^{n}X'|\leq\zeta\varepsilon\}}=O_P(h\varepsilon^{2k-Y}),$$ and thus,
	 $\sum_{i=1}^{n}\mathbb{E}_{i-1}(\tilde\xi_i^n)^4=\varepsilon^{2Y-8} O_P(\varepsilon^{8-Y})=o_P(1)$.
It remains to check the condition (2.2.40) in \cite{JacodProtter}:
\begin{align}\label{NdTchCndbb}
	\sum_{i=1}^{n}\mathbb{E}_{i-1}\left[\tilde\xi_{i}^n({\Delta_i^nM})\right]\,\stackrel{\mathbb{P}}{\longrightarrow}\,0,
\end{align}
when $M=W$ or $M$ is a square-integrable martingale orthogonal to $W$. 
 The proof is much more involved and technical than that of \eqref{NdTchCnd} and is given in Lemma \ref{ErroDiscrHigherOrderb}  of \ref{OtherSuperTechProofs}.
\end{proof}
\smallskip

\begin{proof}[Proof of Theorem \ref{thm:debiasclt}] 
Recall the notation of Propositions \ref{prop1} and \ref{thm:singleclt}.
In addition, we set up the following notation:
\begin{align*}
a_1(\varepsilon) &:=  \frac{\bar{C}}{2-Y}\int_0^1|\chi_s|^{Y}ds\varepsilon^{2-Y}=: {\kappa_1}\varepsilon^{2-Y},\\
a_2(\varepsilon)&:= a_2(\varepsilon,h):= - \bar{C}\frac{(Y+1)(Y+2)}{2Y}  \int_0^1|\chi_s|^Y\sigma^2_sds h \varepsilon^{-Y}=:{\kappa_2}h \varepsilon^{-Y},\\
    \Phi_n&:= u_n^{-1}\left(\wt{Z}_n(\zeta_1\varepsilon)-\wt{Z}_n(\varepsilon) \right) = O_P(1),\\
        \Psi_n&:={u_n^{-1}\left(\wt{Z}_n(\zeta^2_1\varepsilon)-2\wt{Z}_n(\zeta_1\varepsilon)+\wt{Z}_n(\varepsilon)\right)} = O_P(1),
\end{align*}
where the stochastic boundedness of $\Phi_n$ and $\Psi_n$ is a consequence of Proposition \ref{thm:singleclt}.
The proof is obtained in two steps.
\bigskip
\noindent
\textbf{Step 1.} 
We first analyze the behavior of $\wt C_n '(\varepsilon, \zeta_1)={\wh C_n (\varepsilon)}-\widehat a_1(\varepsilon)$, where
\begin{align}\label{DfnHata1a}
\widehat a_1(\varepsilon) := \frac{\left(\widehat C_n(\zeta
    _1 \varepsilon)-\widehat C_n( \varepsilon)\right)^2}{\widehat C_n(\zeta_1^2 \varepsilon)-2\widehat C_n(\zeta_1  \varepsilon) + \widehat C_n( \varepsilon)}.
\end{align}
  If we let {$\eta_1(\zeta)=\zeta^{2-Y}-1$} and $\eta_2(\zeta)=\zeta^{-Y}-1$, then, for $i=1,2$,  we have
\begin{align*}
\begin{split}
a_i(\zeta_1 \varepsilon)-a_i(\varepsilon)=\eta_i(\zeta_1) a_i(\varepsilon),\quad a_i(\zeta_1^2\varepsilon)-2a_i(\zeta_1\varepsilon))+a_i(\varepsilon)=
  \eta_i^2(\zeta_1) a_i(\varepsilon).
\end{split}
  \end{align*}
For simplicity, we often omit the variable $\zeta_1$ on $\eta_i(\zeta_1)$.  Also,  note that, by definition, $\widehat C_n(\varepsilon) =  \sqrt{h} \wt Z_n(\varepsilon) + \int_0^1 \sigma_s^2ds + A(h,\varepsilon)$ and $A(\varepsilon,h)=a_1(\varepsilon)+a_2(\varepsilon,h)$. Therefore, we may write
    \begin{align} \label{DfnHata1b}
	{\widehat a_1(\varepsilon)}&=\frac{(\eta_1 a_1(\varepsilon) + \eta_2 a_2(\varepsilon) +  \sqrt h u_n \Phi_n)^2}{\eta_1^2 a_1(\varepsilon)+{\eta_2^2 a_2(\varepsilon)}+ \sqrt h u_n \Psi_n }.
\end{align}
 By expanding the squares in the numerator and using the notation
   $$\widetilde a_1(\varepsilon) :={a_1(\varepsilon) + \wt\eta_2 a_2(\varepsilon):= a_1(\varepsilon)+ \frac{2\eta_1\eta_2-\eta_2^2}{\eta_1^2}a_2(\varepsilon)},
$$
we may express $\widehat a_1(\varepsilon)$ as\small
\begin{align} 
	\widehat a_1(\varepsilon)&= \widetilde a_1(\varepsilon)+\frac{ {\eta_2^2(1-\wt \eta_2)}a^2_2(\varepsilon) + \sqrt h \left[2 u_n \Phi_n(\eta_1 a_1(\varepsilon) + \eta_2 a_2(\varepsilon)) - \widetilde a_1(\varepsilon) u_n \Psi_n \right] + h u_n^2\Phi_n^2}{\eta_1^2 a_1(\varepsilon)+ {\eta_2^2 a_2(\varepsilon)}+ \sqrt h u_n \Psi_n }\notag\\
&=\widetilde a_1(\varepsilon)+\sqrt h \times \frac{O( h^{3/2}\varepsilon^{-2Y}) +  O_P(u_n \varepsilon^{2-Y})+ O_P(\sqrt h u_n^{2})   }{\eta_1^2 a_1(\varepsilon)+ {o(a_1(\varepsilon)) + O_P(\sqrt h u_n)} }\notag\\
&= \widetilde a_1(\varepsilon)+\sqrt h \times \frac{O(h^{3/2}\varepsilon^{-2-Y}) +  O_P(u_n )+ O_P(h^{-1/2}\varepsilon^2)   }{\eta_1^2\kappa_1 + {o(1)} +  O_P(\varepsilon^{Y/2}) }\notag\\
&=\widetilde a_1(\varepsilon)+ {\sqrt h \times O_P(u_n),}\label{e:firststep_remainder}
\end{align}
\normalsize
 where in the last equality we use our assumption $h_n^{\frac{4}{8+Y}}\ll \varepsilon_n$ to conclude that {$h^{3/2}\varepsilon^{-2-Y}\ll u_n$}.
Then, we see that $ \wt C'_n(\varepsilon,\zeta_1)=  \widehat C_n(\varepsilon) - \widehat a_1(\varepsilon)$ is given by 
\begin{align}\nonumber
\wt  C'_n(\varepsilon,\zeta_1)
& =  \sqrt{h} \wt Z_n(\varepsilon) + \int_0^1\sigma^2_sds + A(h,\varepsilon)-\wt a_1(\varepsilon) + O_P(h^{1/2}u_n)\\
\nonumber
& =  \sqrt{h} \wt Z_n(\varepsilon)\\
&\quad + \int_0^1\sigma^2_sds + a_1(\varepsilon)+a_2(\varepsilon,h)-[a_1(\varepsilon) + \wt\eta_2 a_2(\varepsilon)] + O_P(h^{1/2}u_n)\nonumber\\
& =  \sqrt{h} \wt Z_n(\varepsilon) + \int_0^1 \sigma^2_s ds + a'_2(\varepsilon) + O_P(h^{1/2}u_n),\label{e:C'_breakdown}
\end{align}
where {$a_2'(\varepsilon)=(1-\wt\eta_2)a_2(\varepsilon)$}.
So, 
\begin{align}
\wt{Z}_n'( \varepsilon) &:= \sqrt{n}\left(\wt C_n '(\varepsilon, \zeta_1) - \int_0^1 \sigma^2_sds - a'_2(\varepsilon)   \right)=\wt{Z}_n( \varepsilon)  +O_P(u_n),\label{e:step1_essential_asymptotics}
\end{align}
where the $O_P(u_n)$ term is a consequence of expression \eqref{e:C'_breakdown}. Then, by Proposition \ref{thm:singleclt},
\begin{equation}\label{e:Z'_clt}
\wt{Z}'_n( \varepsilon) = \wt{Z}_n(\varepsilon) + O_P(u_n) \toDistSt\mathcal{N}\Big(0, 2\int_0^1\sigma^4_sds\Big),
\end{equation}
since $u_n\to 0$ by our Assumption $\varepsilon_n\ll  h_n^{\frac{1}{4-Y}}$. Note that if $\varepsilon_n\gg h_n^{\frac{1}{2Y}}$, then $\sqrt{n}a_{2}'\ll 1$ and in place of \eqref{e:firststep_remainder} we have $\widetilde a_1(\varepsilon)+ {\sqrt h \times o_P(1)}$, from which we conclude that $$\sqrt{n}\left(\wt C_n '(\varepsilon, \zeta_1) - \int_0^1\sigma_s^2ds\right)\toDistSt \mathcal{N}\Big(0, 2\int_0^1\sigma^4_sds\Big).$$

\noindent
\textbf{Step 2.} Now we analyze ${\wt C_n ''(\varepsilon, \zeta_2,\zeta_1)}= \wt C_n '(\varepsilon, \zeta_1) - \widehat a'_2(\varepsilon,\zeta_1,\zeta_2)$, where
$$
\widehat a'_2(\varepsilon,\zeta_1,\zeta_2):= \frac{\left(\wt C_n '(\zeta_2\varepsilon, \zeta_1)-\wt C_n '(\varepsilon, \zeta_1)\right)^2}{\wt C_n '(\zeta_2^2\varepsilon, \zeta_1)-2 \wt C_n '(\zeta_2\varepsilon, \zeta_1) + \wt C_n '(\varepsilon, \zeta_1)}.
$$
For simplicity, we omit the dependence on $\zeta_1$ and $\zeta_2$ in $\wt C_n '(\varepsilon, \zeta_1), C_n ''(\varepsilon, \zeta_1,\zeta_2)$, etc.
First, analogous to  $\Phi_n,\Psi_n$ defined in Step 1, we define
\begin{align*}
    \Phi_n' :=& u_n^{-1}\left(\wt{Z}_n'(\zeta_2\varepsilon)-\wt{Z}'(\varepsilon)\right)=u_n^{-1}\left(\wt{Z}_n(\zeta\varepsilon)-  
    \wt{Z}_n(\varepsilon)+O_P(u_n)\right)= O_P(1),\\
        \Psi'_n :=& u_n^{-1}\left(\wt{Z}'_n(\zeta^2_2\varepsilon)-2\wt{Z}'_n(\zeta_2\varepsilon)+\wt{Z}'_n(\varepsilon)\right) = O_P(1),
\end{align*}
where the stochastic boundedness of $\Phi'_n,\Psi'_n$ follows from \eqref{e:step1_essential_asymptotics} and \eqref{e:joint_CLT_expression}.
Now,  by definition \eqref{e:step1_essential_asymptotics}, $ \widetilde C'_n(\varepsilon) =  \sqrt{h} \wt Z'_n(\varepsilon) + \int_0^1 \sigma^2_sds + a'_2(\varepsilon)$.  Also, with the notation {$\eta_2'(\zeta)=\zeta^{-Y}-1$}, the term ${a_2'(\varepsilon)=(1-\wt\eta_2)a_2(\varepsilon)}=:\kappa_2' h \varepsilon^{-Y}$ satisfies
   $$\begin{gathered}
	a_2'(\zeta_2 \varepsilon)-a'_2(\varepsilon)=\eta_2'(\zeta_2) a_2'(\varepsilon),\\
  a'_2(\zeta_2^2\varepsilon)-2a'_2(\zeta_2\varepsilon))+a'_2(\varepsilon)=(\eta_2')^2(\zeta_2) a'_2(h,\varepsilon).
  \end{gathered}
  $$
  Therefore, we may express
\begin{align*} \widehat a_2'(\varepsilon)&=\frac{(\eta_2' a_2'(\varepsilon) +  \sqrt h u_n \Phi'_n)^2}{(\eta_2')^2 a_2'(\varepsilon)+ \sqrt h u_n \Psi_n' }\\
&= a_2'(\varepsilon) + \frac{ \sqrt h a_2'(\varepsilon) u_n(2\Phi'_n   -  \Psi'_n )  + h u_n^2(\Phi_n')^2}{(\eta_2')^2 a_2'(\varepsilon)+ \sqrt h u_n \Psi_n' }\\
&= a_2'(\varepsilon)+ \sqrt h \times \frac{  O_P(u_n )+ O_P(h^{-3/2}\varepsilon^4)   }{(\eta_2')^2\kappa_2' + o_P(1) }\\
&= a_2'(\varepsilon)+ \sqrt h \times O_P(u_n),%
\end{align*}
where, in the last equality we used that $\varepsilon_n\ll h_n^{2/(4+Y)}$ to conclude that 
\begin{align*}
	(a_2')^{-1}\sqrt{h}u_n&=O_P(h^{-1}\varepsilon^{Y}\varepsilon^{\frac{4-Y}{2}})=O_P(h^{-1}\varepsilon^{\frac{4+Y}{2}})=o_P(1),\\
	(a_2')^{-1}h^{1/2}u_n^2&=h^{-3/2}\varepsilon^{4}\ll u_n=h_{n}^{-\frac{1}{2}}\varepsilon^{\frac{4-Y}{2}}.
\end{align*}
Finally, 
\begin{align*}
	\sqrt{n}&\left(\wt C_n ''({\varepsilon, \zeta_2,\zeta_1})-\int_0^1\sigma_s^2ds\right)\\
	&= h^{-1/2}\left(\wt C_n '(\varepsilon, \zeta_1)-\int_0^1\sigma^2_s ds - \widehat a'_2(\varepsilon,\zeta_1,\zeta_2)\right)\\
	&= h^{-1/2}\left(\wt C_n '(\varepsilon, \zeta_1)-\int_0^1\sigma_s^2ds - a'_2(\varepsilon)\right)+O(u_n)\\
	&=\wt Z'_n(\varepsilon) + O_P(u_n)\toDistSt \cN\left(0, 2\int_0^1\sigma^4_sds\right),
\end{align*}
where the third and fourth limit follow from \eqref{e:step1_essential_asymptotics} and \eqref{e:Z'_clt}, respectively.  
\end{proof}

\section{Asymptotic expansions for truncated moments }\label{s:proofs_of_lemmasItoSmrt}

In this section, we provide high-order asymptotic expansions for the truncated moments of the It\^o semimartingale $X$. 
As in \ref{AppMnRslt}, we denote $C$ or $K$ generic constants that may be different from line to line.
 
 To simplify some proofs, we now lay out some additional notation related to the process $J$.
Let $N$ be the Poisson jump measure of $J$ and let $\bar{N}$ be its compensated measure. Observe that due to condition (ii) in Assumption \ref{assump:Funtq}, there exists $\delta_0\in(0,1)$ such that $q(x)>0$ for all $|x|\leq{}\delta_0$. Next, let $\breve J$ be a pure-jump L\'evy process independent  of $J$ with triplet $(0,0,\breve \nu)$, where $\breve{\nu}(dx)=e^{-|x|^p}(C_{+}{\bf 1}_{(0,\infty)}(x)+C_{-}{\bf 1}_{(-\infty,0)}(x)){\bf 1}_{|x|>\delta_0}\,|x|^{-1-Y}dx$, for a fixed $p<1\wedge Y$, and define the L\'evy process
\begin{equation}\label{e:def_J^infty}
	J^\infty_t=\Big(\bar b+\int_{\delta_0<|x|\leq 1}x\nu(dx)\Big)t+\int_{0}^t \int_{|x|\leq{}\delta_0}x\bar{N}(ds,dx)+\breve{J}_t.
\end{equation}
In other words, $J^\infty$ has L\'evy measure $\nu(dx){\bf 1}_{\{|x|\leq \delta_0\}} + \breve{\nu}(dx){\bf 1}_{\{|x|>\delta_0\}}$, and, in particular,  $J^\infty_t$ satisfies all the conditions of {assumptions of \ref{VeryTechProofs}, 
and, thus, we can apply the asymptotics of the truncated moments established therein.}
 Next, we write
\begin{equation}\label{e:def_J^0}
	J^0_t:=J_t-J^\infty_t=\int_{0}^t \int_{|x|>{}\delta_0}x{N}(ds,dx)-\breve{J}_t,
\end{equation}
and observe $J^0$ has finite jump activity.

As an intermediate step, we first establish estimates for the truncated moments of the process 
\begin{align}\label{def:X-prime}
 X_t' := \int_0^t b_s ds+\int_0^t \sigma_s dW_s+\int_0^t \chi_ s dJ^\infty_s.
\end{align}
In a subsequent step, we show the same estimates also hold for $X$ up to asymptotically negligible terms (Lemma \ref{l:trucated_moments_extended_to_X}, below).
Note that
$$
X_t-X'_t = X_t^{j,0} + \int_0^t \chi_sdJ_s^0.
$$
In other words, the process $X'$ includes the continuous component of $X$ and the infinite variation component $\int_0^t \chi_ s dJ^\infty_s$, and  $X-X'$ contains only finite variation terms.
\begin{lemma}\label{ErroDiscr}
Let 
\begin{align}\label{e:def_hat_a(eps,h)}
    \widehat{A}_i(\varepsilon,h) 
    = \frac{\bar{C}|\chi_{t_{i-1}}|^Y}{2-Y}\varepsilon^{2-Y} - \bar{C}\frac{(Y+1)(Y+2)}{2Y} \sigma_{t_{i-1}}^2|\chi_{t_{i-1}}|^Y\ h\varepsilon^{-Y}.
\end{align}	
Suppose that  $Y\in(0,1)\cup(1,\frac{8}{5})$  and $h_n^{\frac{3}{2(2+Y)}\wedge \frac12}\ll \varepsilon_n\ll  h_n^{\frac{1}{4-s}}$, for some  $s\in(0,4)$.
Then, for any $i=1,\dots,n$, 
		\begin{align}\label{CTNITP}
		\mathbb{E}_{i-1}\left[(\Delta_i^nX')^2{\bf 1}_{\{|\Delta_i^nX'|\leq{}\varepsilon\}}\right]&= \sigma_{t_{i-1}}^2h+\widehat{A}_i(\varepsilon,h)h +o_P(h^{3/2}).
	\end{align}
\end{lemma}
\begin{proof}
We use the following notation 
\begin{align}\nonumber
	&x_i:=b_{t_{i-1}}h+\sigma_{t_{i-1}}\Delta_i^n W
	+\chi_{t_{i-1}} \Delta_i^nJ^\infty=:x_{i,1}+x_{i,2}+x_{i,3},\\
	\label{NtnEXs}
	\qquad &\mathcal{E}_{i}:=\Delta_i^n X'-x_i,\quad
	\mathcal{E}_{i,1}=\int_{t_{i-1}}^{t_i}b_s ds-{b}_{t_{i-1}}h,\\
	\nonumber
	&\mathcal{E}_{i,2}=\int_{t_{i-1}}^{t_i}\sigma_s dW_s-\sigma_{t_{i-1}}\Delta_i W,
	\qquad \mathcal{E}_{i,3}=\int_{t_{i-1}}^{t_i}\chi_s dJ_s^\infty-\chi_{t_{i-1}} \Delta_i^nJ^\infty.
\end{align}
The following estimates, established in Lemma \ref{l:ldiscretiz_error_moment_bounds} of \ref{OtherSuperTechProofs},
are often used:
\begin{align}\label{MDEE0}
\mathbb{E}_{i-1}\left[\left|\mathcal{E}_{i,\ell}\right|^p\right]\leq C \begin{cases}  h^{p},& \ell=1, ~p>0\\
h^{\frac{(2 \wedge p)+p}{2},}&\ell=2, ~p>0 \\
 h^{1+\frac{p}{2}}   & \ell=3,~ p \in [1,\infty)\cap (Y,\infty).
  \end{cases}
\end{align}
In particular  $\mathbb{E}_{i-1}\left[\left|\mathcal{E}_{i}\right|^p\right]\leq C h^{1+\frac p2}$ for all $p\geq 2$.
By our expansion for the truncated second-order moment of L\'evy process given in {Proposition \ref{lemma:2E} of \ref{VeryTechProofs}},
we can easily see that 
$$
	\mathbb{E}_{i-1}\left[x_i^2{\bf 1}_{\{|x_i|\leq{}\varepsilon\}}\right]=
	 \sigma_{t_{i-1}}^2h+\widehat{A}_i(\varepsilon,h)h +o_P(h^{\frac{3}{2}}),
$$
because the higher-order terms $h^3\varepsilon^{-Y-2}$, $h^{2}\varepsilon^{2-2Y}$, and $h\varepsilon^{2-\bar\delta}$ ($\bar\delta>0$ arbitrary) are smaller than $h^{\frac{3}{2}}$ due to the restrictions $Y<8/5$ and $h_n^{\frac{3}{2(2+Y)}\wedge \frac12}\ll \varepsilon_n\ll  h_n^{\frac{1}{4-s}}$.
Therefore, for \eqref{CTNITP} to hold, it suffices that 
\begin{align}\label{DBCD2_0}
	\quad{\mathcal{R}}_i:=\mathbb{E}_{i-1}\left[(\Delta_i^nX')^2{\bf 1}_{\{|\Delta_i^nX'|\leq{}\varepsilon\}}\right]-\mathbb{E}_{i-1}\left[x_i^2{\bf 1}_{\{|x_i|\leq{}\varepsilon\}}\right]
	=o_P(h^{3/2}).
\end{align}
Clearly,
\begin{align*}
	|\mathcal{R}_{i}|&\leq 
	2\mathbb{E}_{i-1}\left[x_i^2{\bf 1}_{\{|x_i+\mathcal{E}_i|\leq{}\varepsilon<|x_i|\}} \right]+2\mathbb{E}_{i-1}\left[\mathcal{E}_i^2 {\bf 1}_{\{|x_i+\mathcal{E}_i|\leq{}\varepsilon<|x_i|\}}\right]\\
	&\quad +2|\mathbb{E}_{i-1}\left[x_i\mathcal{E}_i{\bf 1}_{\{|x_i+\mathcal{E}_i|\leq{}\varepsilon,|x_i|\leq{}\varepsilon\}}\right]|\\
	&\quad
	+\mathbb{E}_{i-1}\left[x_i^2{\bf 1}_{\{|x_i|\leq{}\varepsilon<|x_i+\mathcal{E}_i|\}}\right]
	+\mathbb{E}_{i-1}\left[\mathcal{E}_i^2{\bf 1}_{\{|x_i+\mathcal{E}_i|\leq{}\varepsilon,|x_i|\leq{}\varepsilon\}}\right]=\sum_{\ell=1}^5\mathcal{R}_{i,\ell}.
\end{align*}
The terms $\mathcal{R}_{i,2}$ and $\mathcal{R}_{i,5}$ are straightforward since $\mathcal{R}_{i,\ell}\leq{}
	\mathbb{E}_{i-1}\left[\mathcal{E}_i^2\right]\leq C h^{2}=o_P(h^{3/2})$. For the term $\mathcal{R}_{i,3}$, expanding the product $x_i\mathcal{E}_i$, we get 9 terms of the form  $A_{\ell,\ell'}:=|\mathbb{E}_{i-1}\left[\mathcal{E}_{i,\ell}x_{i,\ell'}{\bf 1}_{\{|x_i+\mathcal{E}_i|\leq{}\varepsilon,|x_i|\leq{}\varepsilon\}}\right]|$  (one for each pair $\ell,\ell'\in\{1,2,3\})$.  For terms with $x_{i,1}$, since $|x_{i,1}|\leq Kh$ clearly
	\begin{align*}
	A_{\ell,1}&
	\leq{}C h\mathbb{E}_{i-1}\left[|\mathcal{E}_{i,\ell}|\right]
	\leq{}Ch\mathbb{E}_{i-1}\left[|\mathcal{E}_{i,\ell}|^2\right]^{\frac{1}{2}}\leq{}Ch^2.
\end{align*}
For terms involving $x_{i,3}$, Lemma \ref{lemma:J2k} of \ref{MoreMoreAuxLm}
implies
\begin{align*}
	A_{\ell,3}&
	\leq C
	\mathbb{E}_{i-1}\left[|\Delta_i^nJ^{ \infty} | |\mathcal{E}_{i,\ell}|{\bf 1}_{\{|x_i|\leq{}\varepsilon\}}\right]\\
	&\leq C
	\mathbb{E}_{i-1}\left[|\Delta_i^nJ^{ \infty}|^2|{\bf 1}_{\{|x_i|\leq{}\varepsilon\}}\right]^{\frac{1}{2}}\mathbb{E}_{i-1}\left[\mathcal{E}_{i,\ell}^2|\right]^{\frac{1}{2}}\leq C h^{\frac{1}{2}}\varepsilon^{\frac{2-Y}{2}}h=o_P(h^{3/2}).
\end{align*}
 The terms involving $x_{i,2}$ are more delicate. We start with 
 \begin{equation}
	A_{\ell,2}
	\leq C|\mathbb{E}_{i-1}\left[\Delta_i^nW\mathcal{E}_{i,\ell}\right]|
	+C\mathbb{E}_{i-1}\left[|\Delta_i^nW||\mathcal{E}_{i,\ell}|{\bf 1}_{\{|x_i+\mathcal{E}_i|>{}\varepsilon\,\text{ or }\,|x_i|>{}\varepsilon\}}\right].\label{e:x_i2_term_decomp}
\end{equation}
Clearly, $\mathbb{E}_{i-1}\left[\Delta_i^nW\mathcal{E}_{i,\ell}\right]=0$ for $\ell=1,3$. For $\ell=2$, $\mathbb{E}_{i-1}\left[\Delta_i^nW\mathcal{E}_{i,2}\right]=\mathbb{E}_{i-1}\left[\int_{t_{i-1}}^{t_{i}}(\sigma_s-\sigma_{t_{i-1}})ds\right]=O_P(h^2)$. 
 For the second term in \eqref{e:x_i2_term_decomp} on the event $\{|x_i+\mathcal{E}_i|>{}\varepsilon\,\text{ or }\,|x_i|>{}\varepsilon\}$, we have that $|x_{i,\ell}|>\varepsilon/4$ for at least one $\ell$ or $|\mathcal{E}_{i}|>\varepsilon/4$. The case $|x_{i,1}|>\varepsilon/4$ is eventually impossible for $n$ large enough (since $b$ is bounded), while both cases $|x_{i,2}|>\varepsilon/4$ and $|\mathcal{E}_{i}|>\varepsilon/4$ are straightforward to handle using the Markov's and H\"older's inequalities. For instance, for any $m\geq1$,
\begin{align}\nonumber
	\mathbb{E}_{i-1}&\left[|\Delta_i^nW||\mathcal{E}_{i,\ell}|{\bf 1}_{\{|\mathcal{E}_i|>{}\varepsilon/4\}}\right]\\
	&\leq\nonumber
	\frac{C}{\varepsilon^m}\mathbb{E}_{i-1}\left[|\Delta_i^nW||\mathcal{E}_{i,\ell}||\mathcal{E}_i|^{m}\right]\\
	\nonumber
	&\leq{}\nonumber
	\frac{C}{\varepsilon^m}\mathbb{E}_{i-1}\left[|\Delta_i^nW|^{3}\right]^{\frac{1}{3}}
	\mathbb{E}_{i-1}\left[|\mathcal{E}_{i,\ell}|^{3}\right]^{\frac{1}{3}}\mathbb{E}_{i-1}\left[|\mathcal{E}_i|^{3m}\right]^{\frac{1}{3}}\\
	&\leq \frac{C}{\varepsilon^m} h^{\frac{1}{2}}h^{\frac{5}{6}}h^{\frac{1}{3}+\frac{m}{2}}=Ch^{\frac{1}{2}}h^{\frac{5}{6}}h^{\frac{1}{3}}\left(\frac{h}{\varepsilon^2}\right)^{\frac{m}{2}}.\label{CMInTr}
\end{align}
So, by picking $m=1$, we can make this term $o_P(h^{3/2})$. The remaining term is when $|x_{i,3}|>\varepsilon/4$. In that case, for any $p,q> 2$ such that $\frac{1}{p}+\frac{1}{q}=\frac{1}{2}$,  applying {Lemma \ref{lemma:|J|>eps}},
\begin{align*}
	&\mathbb{E}_{i-1}\left[|\Delta_i^nW||\mathcal{E}_{i,\ell}|{\bf 1}_{\{|x_{i,3}|>{}\varepsilon/4\}}\right]\\
	&\leq{}
	C\mathbb{E}_{i-1}\left[|\Delta_i^nW|^{p}\right]^{\frac{1}{p}}
	\mathbb{P}_{i-1}\left[|\Delta_i^n  J^{\infty}|>\frac{\varepsilon}{4}\right]^{\frac{1}{q}}\mathbb{E}_{i-1}\left[|\mathcal{E}_{i,\ell}|^{2}\right]^{\frac{1}{2}}\\
	&\leq C h^{\frac{3}{2}}(h\varepsilon^{-Y})^{\frac{1}{q}}\ll h^{\frac{3}{2}},
\end{align*}
since, by our assumption on $\varepsilon$, we have $h^{\frac{1}{Y}}\ll \varepsilon$. We then conclude that 
$\mathcal{R}_{i,3}=o_P(h^{\frac{3}{2}})$.

It remains to analyze $\mathcal{R}_{i,1}$ and $\mathcal{R}_{i,4}$. The proof is similar in both cases and we only give the details for the second case to save space. For some $\delta=\delta_n\to{}0$ ($0<\delta<\varepsilon$), whose precise asymptotic behavior will be determined below, we consider the decomposition:
\[
	\mathcal{R}_{i,4}\leq\mathbb{E}_{i-1}\left[x_i^2{\bf 1}_{\{\varepsilon-\delta<|x_i|\leq{}\varepsilon\}}\right]+\mathbb{E}_{i-1}\left[x_i^2{\bf 1}_{\{|x_i|\leq{}\varepsilon-\delta, \varepsilon\leq|x_i+\mathcal{E}_i|\}}\right]=:\mathcal{D}_{i,1}+\mathcal{D}_{i,2}.
\]
By the  expansion
in Proposition \ref{prop:EX2},
we have 
\begin{align}\nonumber
	\mathcal{D}_{i,1}&=Ch[\varepsilon^{2-Y}-(\varepsilon-\delta)^{2-Y}]+C' h^2[\varepsilon^{-Y}-(\varepsilon-\delta)^{-Y}]+o_P(h^{\frac{3}{2}})\\
	&= Ch\varepsilon^{1-Y}\delta+o_P(h\varepsilon^{1-Y}\delta)+o_P(h^{\frac{3}{2}}). \label{e:D_i1_derivative}
\end{align}
Therefore, to obtain $\mathcal{D}_{i,1}=o_P(h^{\frac{3}{2}})$, we require
\begin{align}\label{CndDlt1}
	\delta\ll h^{\frac{1}{2}}\varepsilon^{Y-1}.
\end{align}
For $\mathcal{D}_{i,2}$, note that $|x_i|\leq{}\varepsilon-\delta$ and $\varepsilon\leq|x_i+\mathcal{E}_i|$ imply that  $|\mathcal{E}_i|>\delta$. Also, $\varepsilon\leq|x_i+\mathcal{E}_i|$ implies that  $|x_{i,\ell}|>\varepsilon/4$ for at least one $\ell$ or $|\mathcal{E}_{i}|>\varepsilon/4$. The case $|x_{i,1}|>\varepsilon/4$ is eventually impossible for $n$ large enough, while both cases $|x_{i,2}|>\varepsilon/4$ and $|\mathcal{E}_{i}|>\varepsilon/4$ are again straightforward to handle using Markov's and H\"older's inequalities as in \eqref{CMInTr}. Therefore, we need only to consider the case when $|x_{i,3}|>\varepsilon/4$ and $\max\{|x_{i,1}|,|x_{i,2}|,|\mathcal{E}_{i}|\}\leq{}\varepsilon/4$. In particular, since $|x_{i}|=|x_{i,1}+x_{i,2}+x_{i,3}+\mathcal{E}_{i}|<\varepsilon$, we have $\varepsilon/4<|x_{i,3}|\leq{}C\varepsilon$ for some $C$. Then, we are left to analyze the following term:
\begin{align*}
	\mathbb{E}_{i-1}\left[x_i^2{\bf 1}_{\{\frac{\varepsilon}{4}<|x_{i,3}|\leq{}C\varepsilon, |\mathcal{E}_i|>\delta\}}\right]&\leq{}
	C\sum_{\ell=1}^{3}\mathbb{E}_{i-1}\left[x_{i,\ell}^2{\bf 1}_{\{\frac{\varepsilon}{4}<|x_{i,3}|\leq{}C\varepsilon, |\mathcal{E}_i|>\delta\}}\right]\\
	&=:
	\sum_{\ell=1}^{3}\mathcal{V}_{i,\ell}.
\end{align*}
Clearly, $\mathcal{V}_{i,1}=O_P(h^2)=o_P(h^{3/2})$. For $\mathcal{V}_{i,2}$,  by H\"older's inequality, for any $p,q> 1$, $r\geq 2$, such that $\frac{1}{p}+\frac{1}{q}+\frac{1}{r}=1$, recalling \eqref{MDEE0}  and Lemma \ref{lemma:|J|>eps} below,
\begin{align*}
	\mathcal{V}_{i,2}&\leq C\frac{1}{\delta}
	\mathbb{E}_{i-1}\left[(\Delta_i^nW)^2 {\bf 1}_{\{\frac{\varepsilon}{4}<|x_{i,3}|\}}|\mathcal{E}_i|\right]\\
	&\leq{}C\frac{1}{\delta}
	\mathbb{E}_{i-1}\left[(\Delta_i^nW)^{2p}\right]^{\frac{1}{p}}
	\mathbb{P}_{i-1}\left[\frac{\varepsilon}{4}<|x_{i,3}|\right]^{\frac{1}{q}}
	\mathbb{E}_{i-1}\left[|\mathcal{E}_i|^{r}\right]^{\frac{1}{r}}\\
	&\leq C\frac{1}{\delta} h (h\varepsilon^{-Y})^{\frac{1}{q}}
	(h^{1+\frac{r}{2}})^{\frac{1}{r}}=\frac{C}{\delta}h^{\frac{5}{2}-\frac{1}{p}}\varepsilon^{-\frac{Y}{q}}.
\end{align*}
If $Y<1$, we take $q\to\infty$ and $p=r=2$, to conclude that we only need $\delta\gg h^{1/2}$ for $\mathcal{V}_{i,2}=o_P(h^{3/2})$ to hold. This is consistent with \eqref{CndDlt1} (meaning they can be met simultaneously for at least one choice of the sequence $\delta$), since $Y<1$.

 Now we consider the case $Y\in (1,8/5)$. Clearly 
 $$\mathcal{V}_{i,2}\leq C\frac{1}{\delta}
	\mathbb{E}_{i-1}\left[(\Delta_i^nW)^2 {\bf 1}_{\{\frac{\varepsilon}{4}<|x_{i,3}|\}}(|\mathcal{E}_{i,1}|+|\mathcal{E}_{i,2}|+|\mathcal{E}_{i,3}|)\right].$$
	Note ${\bf 1}_{\{\frac{\varepsilon}{4}<|x_{i,3}|\}}\leq {\bf 1}_{\{ K\varepsilon <|\Delta_i^nJ^{ \infty}|\}}$. Thus,
\begin{align*}
	\mathbb{E}_{i-1}&\left[(\Delta_i^nW)^2 {\bf 1}_{\{\frac{\varepsilon}{4}<|x_{i,3}|\}}|\mathcal{E}_{i,1}|\right] \\
	& \leq  h\mathbb{E}_{i-1}(\Delta_i^nW)^2 \bP\left( K\varepsilon <|\Delta_i^nJ^{ \infty}|\right)\\	
	& \leq   C h^{3}\varepsilon^{-Y}.
\end{align*}
 Similarly,
\begin{align*}
	\mathbb{E}_{i-1}&\left[(\Delta_i^nW)^2 {\bf 1}_{\{\frac{\varepsilon}{4}<|x_{i,3}|\}}|\mathcal{E}_{i,2}|\right] \\
	& \leq  \left(\mathbb{E}_{i-1}(\Delta_i^nW)^4 
	\mathbb{E}_{i-1}|\mathcal{E}_{i,2}|^2\right)^{1/2}\bP\left( K\varepsilon <|\Delta_i^nJ^{ \infty}|\right)\\
	& \leq   C h^{3}\varepsilon^{-Y}.
\end{align*}
For $\ell=3$, with   $p,q> 1$, $Y<r<2$, such that $\frac{1}{p}+\frac{1}{q}+\frac{1}{r}=1$,  from \eqref{MDEE0} we have
\begin{align*}
	\mathbb{E}_{i-1}\left[(\Delta_i^nW)^2 {\bf 1}_{\{\frac{\varepsilon}{4}<|x_{i,3}|\}}|\mathcal{E}_{i,3}|\right] & \leq C h (h\varepsilon^{-Y})^{\frac{1}{q}}
	(h^{1+ \frac{r}{2}})^{\frac{1}{r}}=Ch^{\frac{5}{2}-\frac{1}{p}}\varepsilon^{-\frac{Y}{q}}.
\end{align*}
 Then, taking $r$ close to $Y$, $p$ large, and $q$ close to $Y/(Y-1)$, we obtain, for some $s',s''>0$ that can be made arbitrarily small, 
\[
	{\mathcal{V}}_{i,2}
	\leq \frac{C}{\delta}\left(h^{\frac{5}{2}-s'}\varepsilon^{1-Y-s''}+h^{3}\varepsilon^{-Y}\right) = O\left(\delta^{-1}h^{\frac{5}{2}-s'}\varepsilon^{1-Y-s''} \right),
\]
where we used $h^{\frac12}\ll \varepsilon$ to conclude that $h^{3}\varepsilon^{-Y}\ll h^{\frac{5}{2}}\varepsilon^{1-Y}\ll h^{\frac{5}{2}-s'}\varepsilon^{1-Y-s''}.$
Thus, for $\mathcal{V}_{i,2}=o_P(h^{3/2})$ to hold, it suffices that for some appropriately small $s',s''>0$,
\begin{align}\label{CndDlt2}
 h^{1-s'}\varepsilon^{1-Y-s''}\ll \delta.
\end{align}
The conditions \eqref{CndDlt1} and \eqref{CndDlt2} are consistent,
since 
\begin{align}\label{CndOnSOfEps1}
	h^{1-s'}\varepsilon^{1-Y-s''}\ll h^{\frac{1}{2}}\varepsilon^{Y-1}\;\Longleftrightarrow\;  h^{\frac{1-2s'}{4Y-4+2s''}}\ll  \varepsilon,\;
\end{align}
which is implied by our condition $\varepsilon\gg h ^{\frac{3}{2(2+Y)}}$ 
when provided $s',s''$ are both chosen small enough, since $\frac{3}{2(2+Y)}<\frac{1}{4-4Y}$ when $Y<8/5$. 

It remains to analyze $\mathcal{V}_{i,3}$. 
 As a consequence of Lemma \ref{l:V_i3} of \ref{VeryTechProofs}
\footnote{This estimate is shaper than what can be obtained by applying simply H\"older's inequality and, hence, require some special handling}, provided that the condition $\delta \gg h^{1/2}\varepsilon^{Y/2}$ holds,
we have 
$$
	\bE{\mathcal{V}}_{i,3}\leq C\bE[(\Delta_i^nJ^\infty)^2 {\mathbf 1}_{\{|\Delta_i^nJ^\infty|\leq \varepsilon\}}{\bf 1}_{\{|\mathcal E_{i}|>\delta_n\}}] = O(h^2\varepsilon^{2-2Y})=o(h^{3/2}),
$$
where the second equality follows from $\varepsilon\gg h ^{\frac{3}{2(2+Y)}}$. The  condition $\delta \gg h^{1/2}\varepsilon^{Y/2}$  can be met  under \eqref{CndDlt1} since $h^{1/2}\varepsilon^{Y/2} \ll h^{1/2}\varepsilon^{Y-1}$. Thus, \eqref{DBCD2_0} holds and this concludes the proof.
\end{proof}

\begin{lemma}\label{ErroDiscr2}
Let $\widehat{A}_i(\varepsilon,h) $ be as in \eqref{e:def_hat_a(eps,h)}
and suppose that $Y\in(0,1)\cup(1,\frac{8}{5})$ and ${h_n^{\frac{3}{2(2+Y)}\wedge\frac12}}\ll \varepsilon_n\ll  h_n^{\frac{1}{4-Y}}$. Then, for any $i=1,\dots,n$ and $\zeta>1$, 
		\begin{align}\nonumber
		&\mathbb{E}_{i-1}\left[(\Delta_i^nX')^2{\bf 1}_{\{\varepsilon<|\Delta_i^nX'|\leq{}\zeta\varepsilon\}}\right]\\
		&\quad= \widehat{A}_i(\zeta\varepsilon,h)h -\widehat{A}_i(\varepsilon,h)h+
		o_P\left(h\varepsilon_n^{\frac{4-Y}{2}}\right). \label{e:trucated_difference_estimate}
	\end{align}	
\end{lemma}
\begin{proof} We use the same notation as in \eqref{NtnEXs}. From 
the  expansion 
in Proposition  \ref{lemma:2E} of \ref{s:proofs_of_lemmasItoSmrt},
we have that 
\[
	\mathbb{E}_{i-1}\left[x_i^2{\bf 1}_{\{\varepsilon<|x_i|\leq{}\zeta\varepsilon\}}\right]=\widehat{A}_i(\zeta\varepsilon,h)h -\widehat{A}_i(\varepsilon,h)h+o_P\left(h\varepsilon_n^{\frac{4-Y}{2}}\right),
\]
since all higher-order terms $h^{3}\varepsilon^{-Y-2}$, $h^2\varepsilon^{2-2Y}$, $h^{\frac{3}{2}}\varepsilon^{1-\frac{Y}{2}}$, and $h \varepsilon^{2-\bar\delta}$ are all $o(h\varepsilon^{\frac{4-Y}{2}})$, 
for $\bar\delta$ small enough.
Therefore, it suffices to show that 
\begin{align}
\nonumber
	{\overline{\mathcal{R}}}_i&:=\mathbb{E}_{i-1}\left[(\Delta_i^nX')^2{\bf 1}_{\{\varepsilon\leq|\Delta_i^nX'|\leq{}\zeta\varepsilon\}}\right]-\mathbb{E}_{i-1}\left[x_i^2{\bf 1}_{\{\varepsilon\leq|x_i|\leq{}\zeta\varepsilon\}}\right]\\
	\label{DBCD2_1}
	&\quad 
	=o_P\Big(h\varepsilon_n^{\frac{4-Y}{2}}\Big).
\end{align}
We have the decomposition:
\begin{align*}
	|\overline{\mathcal{R}}_{i}|&\leq 
	2\mathbb{E}_{i-1}\left[x_i^2{\bf 1}_{\{\varepsilon<|x_i+\mathcal{E}_i|\leq{}\zeta\varepsilon,|x_i|>\zeta\varepsilon\}} \right]+
	2\mathbb{E}_{i-1}\left[x_i^2{\bf 1}_{\{\varepsilon<|x_i+\mathcal{E}_i|\leq{}\zeta\varepsilon,|x_i|<\varepsilon\}} \right]\\
	&\quad +4\mathbb{E}_{i-1}\left[\mathcal{E}_i^2\right] +\mathbb{E}_{i-1}\left[|x_i\mathcal{E}_i|{\bf 1}_{\{\varepsilon<|x_i+\mathcal{E}_i|\leq{}\zeta\varepsilon,\varepsilon<|x_i|\leq{}\zeta\varepsilon\}}\right]\\
	&\quad
	+\mathbb{E}_{i-1}\left[x_i^2{\bf 1}_{\{\varepsilon<|x_i|\leq{}\zeta\varepsilon,|x_i+\mathcal{E}_i|>\zeta\varepsilon\}} \right]+
	2\mathbb{E}_{i-1}\left[x_i^2{\bf 1}_{\{\varepsilon<|x_i|\leq{}\zeta\varepsilon,|x_i+\mathcal{E}_i|<\varepsilon\}} \right]\\
	&=:\sum_{\ell=1}^6\overline{\mathcal{R}}_{i,\ell}.
\end{align*}
The term $\overline{\mathcal{R}}_{i,3}$ is clearly $O_P(h^2)$ and hence, $o_P(h\varepsilon^{(4-Y)/2})$. 
For $\overline{\mathcal{R}}_{i,4}$, by Cauchy's inequality, the expansion
in Proposition  \ref{lemma:2E} of \ref{s:proofs_of_lemmasItoSmrt}
and \eqref{MDEE0}, 
\begin{align*}
	\overline{\mathcal{R}}_{i,4}\leq{}2\mathbb{E}_{i-1}\left[|x_{i}\mathcal{E}_{i}|{\bf 1}_{\{\varepsilon<|x_i|\leq{}\zeta\varepsilon\}}\right]&\leq{}C 
	\mathbb{E}_{i-1}\left[x_{i}^2{\bf 1}_{\{\varepsilon<|x_i|\leq{}\zeta\varepsilon\}}\right]^{\frac{1}{2}}\mathbb{E}_{i-1}\left[|\mathcal{E}_{i}|^2\right]^{\frac{1}{2}}\\
	&\leq C\left(h\varepsilon^{2-Y}\right)^{\frac{1}{2}}(h^2)^{\frac{1}{2}}=h^{\frac{3}{2}}\varepsilon^{\frac{2-Y}{2}},
\end{align*}
which is $o_P(h\varepsilon^{(4-Y)/2})$ since  $h^{\frac{1}{2}}\ll\varepsilon$. 
The proofs of the remaining terms are similar. We give only one of those for simplicity. Consider $\overline{\mathcal{R}}_{i,1}$. We decompose it as 
\begin{align*}
	(1/2)\overline{\mathcal{R}}_{i,1}&\leq\mathbb{E}_{i-1}\left[x_i^2{\bf 1}_{\{\zeta\varepsilon<|x_i|\leq{}\zeta\varepsilon+\delta\}}\right]+\mathbb{E}_{i-1}\left[x_i^2{\bf 1}_{\{\varepsilon<|x_i+\mathcal{E}_{i}|\leq{}\zeta\varepsilon, |x_i|>\zeta\varepsilon+\delta\}}\right]\\
	&=:\overline{\mathcal{D}}_{i,1}+\overline{\mathcal{D}}_{i,2}.
\end{align*}
By the 
{expansion in 
Proposition  \ref{lemma:2E} of \ref{s:proofs_of_lemmasItoSmrt}, we have} 
\begin{align*}
	\overline{\mathcal{D}}_{i,1}&=Ch[(\zeta\varepsilon)^{2-Y}-(\zeta\varepsilon-\delta)^{2-Y}]+C' h^2[(\zeta\varepsilon)^{-Y}-(\zeta\varepsilon-\delta)^{-Y}]+o_P\left(h\varepsilon^{\frac{4-Y}{2}}\right)\\
	&=
	Ch\varepsilon^{1-Y}\delta+o_P(h\varepsilon^{1-Y}\delta)+o_P\left(h\varepsilon^{\frac{4-Y}{2}}\right).
\end{align*}
Therefore, to obtain $\overline{\mathcal{D}}_{i,1}= o_P(h\varepsilon^{(4-Y)/2})$, we require
\begin{align}\label{CndDlt1b}
	\delta\ll \varepsilon^{1+\frac{Y}{2}}.
\end{align}
For $\overline{\mathcal{D}}_{i,2}$, we follow the same 
analysis as in the proof of Lemma \ref{ErroDiscr}.  Indeed, the arguments following expression \eqref{CndDlt1} show that under the condition
\begin{align}\label{CndDlt1b2}
	\delta\gg h^{1/2}\varepsilon^{Y/2},
\end{align}
we have
\begin{align}\label{CIANBD}
	\overline{\mathcal{D}}_{i,2}=O_P\left(\delta^{-1}h^{\frac{5}{2}}\varepsilon^{1-Y}\right)+O_P(h^2 \varepsilon^{2-2Y}) 
+o_P\left(h\varepsilon_n^{\frac{4-Y}{2}}\right).
\end{align}
Observe the condition \eqref{CndDlt1b2} is consistent with \eqref{CndDlt1b} since $h^{1/2}\varepsilon^{Y/2}\ll \varepsilon^{1+Y/2}\iff h^{1/2} \ll \varepsilon$.  
For the first term on the right-hand side of \eqref{CIANBD} to be $o_P(h\varepsilon^{(4-Y)/2})$, we require $h^{\frac{3}{2}}\varepsilon^{-1-\frac{Y}{2}}\ll \delta$, which is consistent with the condition \eqref{CndDlt1b}  since $h^{\frac{3}{2}}\varepsilon^{-1-\frac{Y}{2}}\ll \varepsilon^{1+\frac{Y}{2}}$  under $\varepsilon \gg h^{\frac{3}{2(2+Y)}}$.
Therefore taking any $\delta\to 0$ such that $(h^{1/2}\varepsilon\ \vee\ h^{\frac{3}{2}}\varepsilon^{-1-\frac{Y}{2}})\ll \delta \ll \varepsilon^{1+Y/2}$, we obtain $\overline{\mathcal{D}}_{i,2}=o_P\Big(h\varepsilon_n^{\frac{4-Y}{2}}\Big)$, which establishes \eqref{DBCD2_1} and completes the proof.
\end{proof}

\begin{lemma}\label{ErroDiscrHigherOrder}
	Suppose that $\sqrt{h}\ll \varepsilon\ll{}h^{\frac{1}{4-Y}}$  and $Y\in(0,1)\cup(1,8/5)$. Then, for any $k\geq{}2$, 
	\begin{align}
    \label{CTNITPbb}
    &\bE_{i-1}\left[\left(\Delta_i^n {X'}\right)^{2k}\,{\bf 1}_{\{|\Delta_i^n  {X'}|\leq\varepsilon\}}\right]\\
    \nonumber
    &\quad= (2k-1)!!\,\sigma^{2k}_{t_{i-1}} h^k + \frac{\bar{C}|\chi_{t_{i-1}}|^{Y}}{2k-Y}\,h\varepsilon^{2k-Y}+  o_P\big(h\varepsilon^{2k-Y}\big)+O_P\big(h^{k+\frac{1}{2}}\big).
\end{align}
\end{lemma}
\begin{proof}
	We use the same notation as in \eqref{NtnEXs}. The proof is similar to that of Lemma \ref{ErroDiscr}. From the  expansion 
	of Proposition \ref{prop:EX2},
under our assumptions, we have that 
\begin{align}\label{BndExp3NH}
	\mathbb{E}_{i-1}\left[x_i^{2k}{\bf 1}_{\{{ |x_i|\leq{}\varepsilon}\}}\right]=d_1\sigma^{2k}_{t_{i-1}} h^k + d_2 h\varepsilon^{2k-Y}+  o_P\left(h\varepsilon^{2k-Y}\right),
\end{align}
where $d_1=(2k-1)!!$ and $d_2=\frac{\bar{C}}{2k-Y}|\chi_{t_{i-1}}|^{Y}$.
Therefore, for \eqref{CTNITPbb} to hold, it suffices to show  that 
\begin{align}\nonumber
	\widetilde{\mathcal{R}}_i&:=\mathbb{E}_{i-1}\left[(\Delta_i^nX)^{2k}{\bf 1}_{\{|\Delta_i^nX|\leq{}\varepsilon\}}\right]-\mathbb{E}_{i-1}\left[x_i^{2k}{\bf 1}_{\{|x_i|\leq{}\varepsilon\}}\right]\\
	\label{DBCD2}
	&=o_P\left(h\varepsilon^{2k-Y}\right).
\end{align}
Consider the decomposition:
\begin{align*}
	|\widetilde{\mathcal{R}}_{i}|&\leq 
	C\mathbb{E}_{i-1}\left[x_i^{2k}{\bf 1}_{\{|x_i+\mathcal{E}_i|\leq{}\varepsilon<|x_i|\}} \right]+C\mathbb{E}_{i-1}\left[\mathcal{E}_i^{2k} {\bf 1}_{\{|x_i+\mathcal{E}_i|\leq{}\varepsilon<|x_i|\}}\right]\\
	&\quad +\sum_{\ell=0}^{2k-1}\binom{2k}{\ell}\mathbb{E}_{i-1}\left[\left|x_i^{\ell}\mathcal{E}_i^{2k-\ell}
	\right|{\bf 1}_{\{|x_i+\mathcal{E}_i|\leq{}\varepsilon,|x_i|\leq{}\varepsilon\}}\right]\\
	&\quad
	+\mathbb{E}_{i-1}\left[x_i^{2k}{\bf 1}_{\{|x_i|\leq{}\varepsilon<|x_i+\mathcal{E}_i|\}}\right]=\sum_{ m=1}^4\widetilde{\mathcal{R}}_{i,m}.
\end{align*}
By \eqref{MDEE0}, the term $\widetilde{\mathcal{R}}_{i,2}=O_{P}(h^{1+k})$ and, thus, is $o_P(h\varepsilon^{2k-Y})$. 
In light of \eqref{BndExp3NH}, the $\ell$--th summand appearing in $\widetilde{\mathcal{R}}_{i,3}$ is bounded by a constant times
\begin{align*}
\mathbb{E}_{i-1}\left[\left|x_i^{\ell}\mathcal{E}_i^{2k-\ell}
	\right|{\bf 1}_{\{|x_i|\leq{}\varepsilon\}}\right]
	&\leq \mathbb{E}_{i-1}\left[\left|x_i\right|^{2\ell}{\bf 1}_{\{|x_i|\leq{}\varepsilon\}}\right]^{\frac{1}{2}}
	\mathbb{E}_{i-1}\left[\left|\mathcal{E}_i
	\right|^{4k-2\ell}\right]^{\frac{1}{2}}\\
	&\leq \left(O_P(h^{\frac{\ell}{2}})+O_P(h^{\frac{1}{2}}\varepsilon^{\frac{2\ell-Y}{2}})\right)(h^{1+2k-\ell})^{\frac{1}{2}}\\
	&=O_P\big(h^{k+\frac{1}{2}}\big)+O_P\big(h^{1+k-\frac{\ell}{2}}\varepsilon^{\ell-\frac{Y}{2}}\big).
\end{align*}
The second term above is $o_P(h\varepsilon^{2k-Y})$ when $\varepsilon\gg\sqrt{h}$.

It remains to analyze $\widetilde{\mathcal{R}}_{i,1}$ and $\widetilde{\mathcal{R}}_{i,4}$. The proof is similar in both cases and we only give the details for the second case to save space. For some $\delta\to{}0$ ($0<\delta<\varepsilon$), whose precise asymptotic behavior will be determined below, we consider the decomposition:
\[
	\widetilde{\mathcal{R}}_{i,4}\leq\mathbb{E}_{i-1}\left[x_i^{2k}{\bf 1}_{\{\varepsilon-\delta<|x_i|\leq{}\varepsilon\}}\right]+\mathbb{E}_{i-1}\left[x_i^{2k}{\bf 1}_{\{|x_i|\leq{}\varepsilon-\delta, \varepsilon\leq|x_i+\mathcal{E}_i|\}}\right]=:\widetilde{\mathcal{D}}_{i,1}+\widetilde{\mathcal{D}}_{i,2}.
\]
By {the expansion of Proposition \ref{prop:EX2},} 
we have 
\begin{align*}
	\widetilde{\mathcal{D}}_{i,1}&=Ch[\varepsilon^{2k-Y}-(\varepsilon-\delta)^{2k-Y}]+o_P\left(h\varepsilon^{2k-Y}\right)\\
	&= O_P(h\varepsilon^{2k-1-Y}\delta)+o_P\left(h\varepsilon^{2k-Y}\right).
\end{align*}
Thus, to obtain $\widetilde{\mathcal{D}}_{i,1}=o_P\left(h\varepsilon^{2k-Y}\right)$, we require
\begin{align}\label{CndDlt1cc}
	\delta\ll \varepsilon.
\end{align}
As in the proof of Lemma \ref{ErroDiscr}, when dealing with $\mathcal{D}_{i,2}$, it suffices to analyze the term:
\begin{align*}
	\mathbb{E}_{i-1}\left[x_i^{2k}{\bf 1}_{\{\frac{\varepsilon}{4}<|x_{i,3}||\leq{}C\varepsilon, |\mathcal{E}_i|>\delta\}}\right]&\leq{}
	C\sum_{\ell=1}^{3}\mathbb{E}_{i-1}\left[x_{i,\ell}^{2k}{\bf 1}_{\{\frac{\varepsilon}{4}<|x_{i,3}||\leq{}C\varepsilon, |\mathcal{E}_i|>\delta\}}\right]\\
	&=: \sum_{\ell=1}^{3}\widetilde{\mathcal{V}}_{i,\ell}.
\end{align*}
Clearly, $\widetilde{\mathcal{V}}_{i,1}=O_P(h^{2k})=o_P(h\varepsilon^{2k-Y})$. For $\widetilde{\mathcal{V}}_{i,2}$, by H\"older's inequality, for any  $p,q>1$ and $r>Y$ such that $\frac{1}{p}+\frac{1}{q}+\frac{1}{r}=1$:
\begin{align*}
	\widetilde{\mathcal{V}}_{i,2}&\leq C\frac{1}{\delta}
	\mathbb{E}_{i-1}\left[(\Delta_i^nW)^{2k} {\bf 1}_{\{\frac{\varepsilon}{4}<|x_{i,3}|\}}|\mathcal{E}_i|\right]\\
	&\leq{}C\frac{1}{\delta}
	\mathbb{E}_{i-1}\left[(\Delta_i^nW)^{2kp}\right]^{\frac{1}{p}}
	\mathbb{P}_{i-1}\left[\frac{\varepsilon}{4}<|x_{i,3}|\right]^{\frac{1}{q}}
	\mathbb{E}_{i-1}\left[|\mathcal{E}_i|^{r}\right]^{\frac{1}{r}}\\
	&\leq C\frac{1}{\delta} h^k (h\varepsilon^{-Y})^{\frac{1}{q}}
	(h^{1+\frac{r}{2}})^{\frac{1}{r}}=\frac{C}{\delta}h^{k+\frac{3}{2}-\frac{1}{p}}\varepsilon^{-\frac{Y}{q}}.
\end{align*}
 For the above to be smaller than $h\varepsilon^{2k-Y}$, we need $\delta\gg h^{k+\frac{1}{2}-\frac{1}{p}}\varepsilon^{Y-2k-\frac{Y}{q}}$. For $\delta$ to be consistent with \eqref{CndDlt1cc}, we need that $h^{k+\frac{1}{2}-\frac{1}{p}}\varepsilon^{Y-2k-\frac{Y}{q}-1}\ll1$. This is always possible if we take $q$ close to $1$, $p\nearrow\infty$, and $r\nearrow\infty$ because $h^{k+\frac{1}{2}}\varepsilon^{-2k-1}\ll 1$ under our condition $\varepsilon\gg \sqrt{h}$.

It remains to analyze $\widetilde{\mathcal{V}}_{i,3}$.  {Under the additional constraint
\begin{align}\label{CndDlt3cc}
	h^{\frac{1}{2}}\varepsilon^{Y/2}\ll \delta,
\end{align}
Lemma \ref{l:V_i3} of \ref{VeryTechProofs}
implies
\begin{align*}
	\widetilde{\mathcal{V}}_{i,3}&\leq 
	\mathbb{E}_{i-1}\left[(\Delta_i^n J^{\infty})^{2k} {\bf 1}_{\{|\Delta_i^n J^{\infty}|\leq{}C\varepsilon\}}{\bf 1}_{\{|\mathcal{E}_i|>\delta_n\}}\right] = O(h^2 \varepsilon^{2k-2Y}),
\end{align*}
which implies $\widetilde{\mathcal{V}}_{i,3}=o_P(h\varepsilon^{2k-Y})$ since $h^2 \varepsilon^{2k-2Y} \ll h \varepsilon^{2k-Y}$  under our condition $\varepsilon\gg h^{1/2}$.  The conditions \eqref{CndDlt1cc} and \eqref{CndDlt3cc} are consistent since $h^{\frac{1}{2}}\varepsilon^{Y/2}\ll \varepsilon$, under our conditions.
Thus,  we conclude that
\eqref{DBCD2} holds}, which completes the proof.
\end{proof}
\begin{lemma} \label{l:trucated_moments_extended_to_X} 
Provided $h_n^{\frac{1}{2}-s}\ll \varepsilon_n\ll  h_n^{\frac{1}{4-Y}}$ for some $ s\in(0,1/2)$, and $r_0\vee r_1\in (0,Y\wedge 1)$, for any $\alpha\geq 1$,
\begin{align}\label{e:X'_to_X_k2_pre}
&\big|\big(\Delta_i^n X \big)^{2k} {\bf 1}_{\{|\Delta_i^n X|\leq \varepsilon\}} - (\Delta_i^n X')^{2k} {\bf 1}_{\{|\Delta_i^n X'|\leq \varepsilon\}} \big|\\
\nonumber
&= O_P(h\varepsilon^{2k-r_1}) + O_P(h\varepsilon^{2k-\alpha r_0}) + O_P( h^2 \varepsilon^{2k-Y-r_1\alpha})+  O_P(h\varepsilon^{2k-Y-1+\alpha}
).
\end{align}
In particular, when $h_n^{\frac{4}{8+Y}}\ll \varepsilon_n\ll  h_n^{\frac{1}{4-Y}},$ with $r_0,r_1$ as in Assumption \ref{assump:Coef0}, 
\begin{equation}\label{e:X'_to_X_k2}
\big|\big(\Delta_i^n X \big)^{2k} {\bf 1}_{\{|\Delta_i^n X|\leq \varepsilon\}} - (\Delta_i^n X')^{2k} {\bf 1}_{\{|\Delta_i^n X'|\leq \varepsilon\}} \big| = o_P(h \varepsilon^{2-Y/2}),
\end{equation}
and the estimates \eqref{CTNITP}, \eqref{e:trucated_difference_estimate} and \eqref{CTNITPbb} hold with $X$ in place of $X'.$ 
\end{lemma}
\begin{proof}
Let $\mathcal{D}_{2k}:=\big|\big(\Delta_i^n X \big)^{2k} {\bf 1}_{\{|\Delta_i^n X|\leq \varepsilon\}} - (\Delta_i^n X')^{2k} {\bf 1}_{\{|\Delta_i^n X'|\leq \varepsilon\}} \big|$ and $r=r_0\vee r_1$. 
 Let us recall the definition of $J^0$ and $X'$ in \eqref{e:def_J^0} and \eqref{def:X-prime}, respectively.
Write $V_t=X_t-X'_t= X_t^{j,0} + \int_0^t \chi_sdJ_s^0$, and set
$$
X_t^{j,0} =\int_0^t \int \delta_0(s,{ z})\mathfrak p_0(ds,dz) + \int_0^t \int  \delta_1(s,{ z})\mathfrak p_1(ds,dz) =: Y_t^0 + Y_t^1.
$$
For any fixed integer $k\geq 1$, 
we have
\begin{align}
\nonumber
\mathcal{D}_{2k}
&= \big|(\Delta_i^n X)^{2k}  - (\Delta_i^n X')^{2k}\big| {\bf 1}_{\{|\Delta_i^n X|\leq \varepsilon, ~|\Delta_i^n X'|\leq \varepsilon\}}\\
\nonumber
 & \qquad + (\Delta_i^n X)^{2k}  {\bf 1}_{\{|\Delta_i^n X|\leq  \varepsilon,~|\Delta_i^n X'|> \varepsilon\}}\\
 \nonumber
  & \qquad + (\Delta_i^n X')^{2k}  {\bf 1}_{\{|\Delta_i^n X|>  \varepsilon, ~|\Delta_i^n X'|\leq \varepsilon\}}\\
  \nonumber
    & =: T_1 + T_2 + T_3 . 
\end{align}
A Taylor expansion of $(x'+v)^{2k}$ at $v=0$ gives $|(x'+v)^{2k} - (x')^{2k}| \leq K |v|\big(|x'|^{2k-1} + |v|^{2k-1}\big)$, so
 \begin{equation}\label{e:T_1_semimg_momentEst}
 T_1 \leq  K\big(
 |\Delta_i^n V|^{2k} + |\Delta_i^n V| |\Delta_i^n X'|^{2k-1}\big) {\bf 1}_{\{|\Delta_i^n X|\leq \varepsilon, ~|\Delta_i^n X'|\leq \varepsilon\}.}
\end{equation}
Note that Corollary 2.1.9 in \cite{JacodProtter} implies that for each $p\geq 1$,
 \begin{equation}\label{InBndVarWK}
 \bE_{i-1} \bigg( \frac{|\Delta_i^n Y^m|}{\varepsilon} \wedge 1 \bigg)^p 
 \leq K  h\varepsilon^{-r_m},\quad m=0,1.
\end{equation}
Also, since $J^0$ is compound Poisson, writing $N^0(ds,dx)$ for the jump measure of $J^0$,
 $$
 \bE_{i-1} \Big( \Big| \frac{1}{\varepsilon}\int_{t_{i-1}}^{t_i}\chi_ s dJ^0_s \Big| \wedge 1 \Big)^p \leq  \bP\big( N^0([t_{i-1},t_i],\bR)>0\big)\leq K h.
 $$
Thus we obtain
$\bE_{i-1}\left(\left|\Delta_i^n V \right|^p {\bf 1}_{\{|\Delta_i^n V|\leq \varepsilon\}}\right)
 \leq K h \varepsilon^{p-r}.
$
 Therefore, since $|x+v|\leq \varepsilon$ and $|x|\leq \varepsilon$ imply $|v|\leq 2\varepsilon$, from \eqref{e:T_1_semimg_momentEst}, we have
\begin{align*}
 \bE_{i-1}T_{1}&\leq K \Big(\bE_{i-1}(\Delta_i^n V )^{2k} {\bf 1}_{\{|\Delta_i^n V|\leq2 \varepsilon\}}  + \varepsilon^{2k-1}\bE_{i-1}|\Delta_i^n V| {\bf 1}_{\{|\Delta_i^n V|\leq 2\varepsilon\}}  \Big) \\
 &\leq K h \varepsilon^{2k-r}.
\end{align*}
On $T_2$, write 
\begin{align}\nonumber
 T_2&=   (\Delta_i^n X)^{2k}  {\bf 1}_{\{|\Delta_i^n X|\leq  \varepsilon,~|\Delta_i^n X'|> \varepsilon + \varepsilon^\alpha\}}+(\Delta_i^n X)^{2k}  {\bf 1}_{\{|\Delta_i^n X|\leq  \varepsilon < |\Delta_i^n X'| \leq \varepsilon+ \varepsilon^\alpha \}}\\
 &=: T_2' + T_2'',
 \label{DfnT2T2pr}
\end{align}
where  $\alpha\geq1$ is to be later chosen. 
Using that $|x'+v|\leq \varepsilon$ and $|x'|>\varepsilon + \varepsilon^\alpha$ imply $|v|> \varepsilon^\alpha$, we get 
\begin{align*}
\bE_{i-1} T_2' & \leq  K \varepsilon^{2k}\bP_{i-1}(|\Delta_i^n V|>\varepsilon^\alpha,|\Delta_i^n X'|>\varepsilon + \varepsilon^\alpha\big)\\
&\leq  K \varepsilon^{2k}\Big( \bP_{i-1}(|\Delta_i^n (V-Y^1)|>\varepsilon^\alpha/2\big) \\
&\qquad\qquad+ \bP_{i-1}\big(|\Delta_i^n Y^1|>\varepsilon^\alpha/2\big)\bP\big(|\Delta_i^n X'|>\varepsilon + \varepsilon^\alpha\big) \Big)
\\ &\leq K \varepsilon^{2k}\big( \varepsilon^{-\alpha r_0}h  + h^2\varepsilon^{-Y-r_1\alpha}\big),
\end{align*}
where we used that $\bP(|\Delta X'|>\varepsilon)\leq K h\varepsilon^{-Y}$ as a consequence of Lemma \ref{lemma:|J|>eps}  and that $ \bP_{i-1}\big(|\Delta_i^n Y^1|>u\big)\leq \bE_{i-1}\big(\frac{|\Delta_i^n Y^1|}{u}\wedge 1\big)\leq h u^{-r_1}$ per \eqref{InBndVarWK}.  

On the other hand, for $T''_2$, since $|x'+v|\leq \varepsilon$ and $|x'|\leq \varepsilon+\varepsilon^\alpha$ give $|v|\leq 2\varepsilon + \varepsilon^\alpha,$ 
\begin{align*}
\bE_{i-1} T_2''& =  \bE_{i-1} (\Delta_i^n X)^{2k}  {\bf 1}_{\{|\Delta_i^n X|\leq  \varepsilon < |\Delta_i^n X'| \leq \varepsilon+ \varepsilon^\alpha \}} \\
&\leq K  \bE_{i-1} \big[\big((\Delta_i^n X')^{2k}  + (\Delta_i^n V)^{2k} \big){\bf 1}_{\{|\Delta_i^n X|\leq  \varepsilon < |\Delta_i^n X'| \leq \varepsilon+ \varepsilon^\alpha \}}\big]  \\
& \leq  \bE_{i-1} (\Delta_i^n X')^{2k}{\bf 1}_{\{\varepsilon < |\Delta_i^n X'| \leq \varepsilon+ \varepsilon^\alpha \}} + O_P(\varepsilon^{2k-  r}h) .
\end{align*}
Since $\alpha\geq1$, arguing as in \eqref{e:D_i1_derivative}, by applying  Lemma \ref{ErroDiscrHigherOrder}, we obtain 
\begin{align*}
 \bE_{i-1} (\Delta_i^n X')^{2k}{\bf 1}_{\{\varepsilon < |\Delta_i^n X'| \leq \varepsilon+ \varepsilon^\alpha \}} &=  O_P\big(h\big((\varepsilon + \varepsilon^\alpha)^{2k-Y}- \varepsilon ^{2k-Y}\big)\big)\\
& = O_P( h\varepsilon^{2k-1-Y+\alpha}).
\end{align*}
Now, turning to  $T_3,$ write
\begin{align*}
 T_3&=   (\Delta_i^n X')^{2k}  {\bf 1}_{\{|\Delta_i^n X|>  \varepsilon, ~|\Delta_i^n X'|\leq \varepsilon -\varepsilon^\alpha \}}+({ \Delta_i^n X'})^2  {\bf 1}_{\{{  \varepsilon -\varepsilon^\alpha< |\Delta_i^n X'|\leq\varepsilon \}}} \\
  & = T_3' + T_3''.
\end{align*}
Since $|x'+v|>\varepsilon$ and $|x'|\leq \varepsilon - \varepsilon^\alpha$ imply $|v|> \varepsilon^\alpha$, the same arguments for $T_2'$ and $T_2''$ apply to $T_3',T_3''$, giving 
\eqref{e:X'_to_X_k2_pre}.

To obtain \eqref{e:X'_to_X_k2} under $h_n^{\frac{4}{8+Y}}\ll \varepsilon_n\ll  h_n^{\frac{1}{4-Y}}$, consider first $k=1$.  In this case, we may simultaneously satisfy each of
$\varepsilon^{2-\alpha r_0}h\ll h \varepsilon^{2-\frac Y2}
$, $(\!\iff  \alpha<\frac{Y}{2r_0}$),  $h^2 \varepsilon^{2-Y-r_1\alpha} \ll h \varepsilon^{2-\frac Y2}$  ($\iff  h^{\frac2{Y+2\alpha r_1}}\ll \varepsilon\;\Longleftarrow\varepsilon \gg h^{\frac{4}{8+Y}}$ if $\alpha<\frac{16 - 2Y}{8r_1}$),  and  $h\varepsilon^{1-Y+\alpha}\ll h \varepsilon^{2-\frac Y2}$ (\!$\iff \alpha > 1 + Y/2)$  provided $\alpha$ is chosen such that
$$
1+ \frac{Y}{2} < \alpha<\Big(\frac{Y}{2r_0} \Big)\wedge \Big(\frac{16 - 2Y}{8r_1}\Big),
$$
which is always possible  under the restrictions on $r_0,r_1$ in Assumption \ref{assump:Coef0}. Since $r_1\leq Y/2$, we also have $\varepsilon^{2-r_1}\ll \varepsilon^{2-\frac Y2}$, concluding \eqref{e:X'_to_X_k2} for $k=1$.  When $k\geq 2$, $\varepsilon^{2k-\alpha r_0}h\ll h \varepsilon^{2-\frac Y2}
$, $ h^2 \varepsilon^{2k-Y-r_1\alpha} \ll h \varepsilon^{2-\frac Y2}$, and  $h\varepsilon^{2k-1-Y+\alpha}\ll h \varepsilon^{2-\frac Y2}$ all hold by taking $\alpha=1$.

Finally, since $\varepsilon \ll h^{\frac{1}{4-Y}}$ implies $h \varepsilon^{2-Y/2}\ll h^{3/2}$, expression \eqref{e:X'_to_X_k2} implies each of \eqref{CTNITP}, \eqref{e:trucated_difference_estimate} and \eqref{CTNITPbb} hold with $X$ in place of $X'$.
\end{proof}

\begin{lemma}\label{lemma:|J|>eps} Suppose $0<Y<2$, and $\varepsilon \to 0 $ with $\varepsilon \gg h^{1/Y}$.  For all $p>1\vee Y$,   the following estimates hold:
\begin{equation}\label{e:tailbound_J_>eps}
\mathbb{E}\left( \frac{|\Delta_i^n J|}{\varepsilon}\wedge 1\right)^p \leq  K h \varepsilon^{-Y},\quad \mathbb{E}_{ i-1}\left( \frac{|\Delta_i^n X|}{\varepsilon}\wedge 1\right)^{p} \leq   K\big(  h \varepsilon^{-Y} +  h^{p/2}\varepsilon^{-p}\big).
\end{equation}
In particular, with $V=J$ or $V=X$,
$
\mathbb{P}_{i-1}\left[|\Delta_i^n V|>\varepsilon\right]\leq h \varepsilon^{-Y}.
$
\end{lemma}
\begin{proof}
We only give the proof for $|\Delta_i^n X|$ since the proof for other term is similar.  It suffices to establish the bound for $\int_{t_{i-1}}^{t_i} b_t dt$, $\int_{t_{i-1}}^{t_i} \sigma_t dW_t$ and $\Delta_i^n X^j$ separately. It holds immediately for $ \int_{t_{i-1}}^{t_i} b_t dt$, since $\varepsilon^{-p}\big| \int_{t_{i-1}}^{t_i} b_t dt\big|^p\leq K\varepsilon^{-p}h^p\ll  {\varepsilon^{-p}h^{p/2}}$. for the second term, Burkholder-Davis-Gundy inequality gives $\varepsilon^{-p}\bE_{{i-1}} \big|\int_{t_{i-1}}^{t_i} \sigma_t dW_t|^{p}\leq K_{ p}(h^{1/2}\varepsilon^{-1})^{p}$. 
For $\Delta_i^n X^j$,  we first consider $\Delta_i^n X^{j,\infty}$ as in \eqref{eq:MnMdlX}. 
Applying Corollary 2.1.9(a)\footnote{{Strictly speaking, this corollary assumes $\varepsilon = h^q$ for some $q\in (0,1/Y)$, though a straightforward adaptation shows it holds provided  $h^{1/Y}\ll \varepsilon \ll 1$.}} in, we obtain
$$
\mathbb{E}_{{ i-1}}\left( \frac{|\Delta_i^n X^{j,\infty}|}{\varepsilon}\wedge 1\right)^{p} \leq  K h\varepsilon^{-Y}.
$$
Applying Corollary 2.1.9(c) in \cite{JacodProtter}, we have
$$
\mathbb{E}_{{ i-1}}\left( \frac{|\Delta_i^n X^{j,0}|}{\varepsilon}\wedge 1\right)^{p}\leq K h \varepsilon^{- r},
$$
 where $r= r_0\vee r_1$, which gives
\[
	\mathbb{E}_{{ i-1}}\left( \frac{|\Delta_i^n X^{j}|}{\varepsilon}\wedge 1\right)^{p} \leq K h\varepsilon^{-Y}.
\]
 The remaining statements follow since 
 \[
 	\mathbb{P}_{i-1}\left[|\Delta_i^n V|>\varepsilon\right] \leq \mathbb{E}_{{ i-1}}\left( \frac{|\Delta_i^n V|}{\varepsilon}\wedge 1\right)^{p},
\]
 for every $p>0$, and $h^{p/2}\varepsilon^{-p}\ll h\varepsilon^{-Y}$ for $p$ large enough. This completes the proof.
\end{proof}

\section{Expansions for truncated moments of L\'evy processes}\label{VeryTechProofs}
In this section, we obtain some new high-order asymptotic expansions for the truncated moments  $\mathbb{E}|\widehat{X}_{h_n}|^{2k}{\bf 1}_{\{|\widehat{X}_{h_n}|\leq \varepsilon_n\}}$ of the L\'evy process
\begin{align}\label{DfnHatX0}
	 \widehat{X}_t:=\sigma W_t+J_t,
\end{align}
for a constant volatility value $\sigma\in[0,\infty)$, where, unless otherwise stated, $h_n$ and $\varepsilon_n$ satisfies 
\begin{equation}\label{e:basic_asymptotics}
\varepsilon_n\to0\quad\text{ and }\quad h_n\to{}0\quad \text{such that}\quad \varepsilon_n \gg \sqrt h_n.
\end{equation}
For notational simplicity, we often omit the subscript $n$ in $h_n$ and $\varepsilon_n$. 

 Recall that here $J$ is a L\'evy process having a L\'evy density of the form \eqref{eq:levyden} with a function $q$ satisfying Assumption \ref{assump:Funtq}. 
 As a first step, we show our results hold  under the following additional conditions (iii)-(v), below. We later show these additional conditions can, in fact, be dropped (Lemma \ref{prop:Genq0}, below).
 
\begin{assumption}\label{assump:FuntqStrong}
\hfill
\begin{itemize}
\item [(iii)] $\displaystyle{\limsup_{|x|\rightarrow\infty}\frac{|\ln q(x)|}{|x|^p}<\infty}$, ~for some $p<Y\wedge 1$;
\item [(iv)] for any $\varepsilon>0$, $\displaystyle{\inf_{|x|<\varepsilon}q(x)>0}$;
\item [(v)] 
$\displaystyle{\int_{|x|>1}q(x)|x|^rdx<\infty}$, for any $r>0$.
\end{itemize}
\end{assumption}
The conditions (i)-(v) of Assumptions \ref{assump:Funtq} and  \ref{assump:FuntqStrong} imply that
\[
	\int_{\mathbb{R}_0}\left(q(x)^{-1/2}-1\right)^2\nu(dx)<\infty,
\]
which, according to Theorem 33.1 in \cite{Sato:1999}, is sufficient for the existence of a measure $\wt{\bP}$, locally equivalent to $\bP$, so that, under $\wt{\bP}$, 
$W$ is still a standard Brownian motion independent of $J$, but $J$ is now a stable L\'evy process with L\'evy triplet $(\tilde{b},0,\tilde{\nu})$, where $\wt{\nu}(dx)$ is given by 
\begin{align}\label{Dfntildenu0}
\wt{\nu}(dx):=\big(C_{+}{\bf 1}_{(0,\infty)}(x)+C_{-}{\bf 1}_{(-\infty,0)}(x)\big)|x|^{-Y-1}dx,
\end{align} 
and $\wt{b}:=b+\int_{0<|x|\leq 1} x(\tilde{\nu}-\nu)(dx)$. 
More specifically (see \cite{FigueroaLopezGongHoudre:2016}), for any $t\in\bR_{+}$, 
\begin{equation}\label{e:change_of_measure}
	\frac{d\wt{\bP}\big|_{\sF_{t}}}{d\bP\big|_{\sF_{t}}}=e^{U_t},
\end{equation}
 with $U$  given by 
\begin{equation}\label{eq:DenTranTildePP}
U_t:=
\wt{U}_{t}+\eta t:=\int_{0}^{t}\!\!\int_{\bR_{0}}\varphi(x)\wt{N}(ds,dx)+\eta t,
\end{equation}
where $\varphi(x)=-\ln q(x)$, $\eta:=\int_{\bR_{0}}\big(e^{-\varphi(x)}-1+\varphi(x)\big)\wt{\nu}(dx)$, $N(dt,dx)$ is the jump measure of the process $J$, and $\wt{N}(dt,dx):=N(dt,dx)-\wt{\nu}(dx)dt$ is its compensated version under $\wt{\bP}$. The representation \eqref{eq:DenTranTildePP} (including the $\eta$ appearing in \eqref{eq:DenTranTildePP}) is well-defined under the conditions (i)-(v) of Assumptions \ref{assump:Funtq} and \ref{assump:FuntqStrong}.  We refer the reader to \cite{gong2021} for additional details. 

Under $\wt{\bP}$, for an appropriate shift $\widetilde \gamma$, the process 
\begin{equation}\label{e:def_St}
S_{t}:=J_{t}-t\wt{\gamma} 
\end{equation}
is a strictly $Y$-stable process {with location parameter 0 and its scale and skewness parameters given by $[(C_{+}+C_{-})\Gamma(-Y)|\cos(\pi Y/2)|]^{1/Y}$, $(C_{+}-C_{-})/(C_{+}+C_{-})$, respectively}.  With slight abuse of notation, 
for convenience we assume $S$ has the same distribution under both probability measures $\bP$ and $\wt \bP$.

We first study the truncated moments of the L\'evy process
\begin{align}\label{DfnTildeX0}
	\widetilde{X}_t:=\sigma W_t+S_t,
\end{align}
and subsequently extend our results to the more general L\'evy process $\widehat{X}_t=\sigma W_t+J_t$, first with $J$ satisfying all the conditions (i)-(v) of Assumptions \ref{assump:Funtq} and \ref{assump:FuntqStrong}, and then with only conditions (i)-(ii).

\smallskip
The following result is a refinement of Theorem 3.1 in \cite{gong2021}.
\begin{proposition}\label{lemma:Estable}
Suppose $Y\in (0,1)\cup(1,2)$, and $\varepsilon \to 0$ with  $\sqrt h \ll \varepsilon$.  As $h=1/n\rightarrow 0$, 
\begin{align}\label{ExpForTildeX2}
\bE\left( \wt X_h^2 \,{\bf 1}_{\{|\wt X_h|\leq \varepsilon\}}\right)&= \sigma^2 h + \widetilde{A}(\varepsilon,h)h +  O\left(h^3 \varepsilon^{-Y-2}\right) + O\left(h^{2}\varepsilon^{2-2Y} \right)\\
\nonumber
& \quad + O\left(h^{\frac{1}{2}}\varepsilon e^{-\frac{\varepsilon^2}{2\sigma^2 h}}\right)+ 
{O\left(h^{\frac{1}{2}} \varepsilon^{3-Y}  e^{-\frac{\varepsilon^2}{8\sigma^2 h}} \right)},
\end{align} 
where  $\widetilde{A}(\varepsilon,h)$ is defined as:
\begin{align}\label{e:def_a(eps,h)}
     \widetilde{A}(\varepsilon,h) 
    := \frac{  C_++C_-}{2-Y}\varepsilon^{2-Y} - (   C_++C_-)\frac{(Y+1)(Y+2)}{2Y} \sigma^2 h\varepsilon^{-Y}.
\end{align}
\end{proposition}
\begin{proof}
Throughout the proof, we set $\bar{C}:=  C_++C_-$. 
 We may  write
\begin{align}\label{eq:Decompb1epsstable}
\bE\big( \wt X_h^2 \,{\bf 1}_{\{|\wt X_h|\leq \varepsilon\}}\big)   &=\sigma^{2}\,\wt \bE\Big(W_{h}^{2}\,{\bf 1}_{\{|\sigma W_{h}+S_{h}|\leq\varepsilon\}}\Big)\!\\
\nonumber
&\quad+2\sigma \wt \bE\Big(W_h S_{h}\,{\bf 1}_{\{|\sigma W_{h}+S_{h}|\leq\varepsilon\}}\Big)\!+\!\wt\bE\Big(S_{h}^{2}\,{\bf 1}_{\{|\sigma W_{h}+S_{h}|\leq\varepsilon\}}\Big). %
\end{align}
The sum of the first two terms of \eqref{eq:Decompb1epsstable} is denoted $\mathcal{T}_{1,n}$. 
By symmetry of $W_{1}$, we have:
\begin{align} 
\nonumber
\mathcal{T}_{1,n}&=
\sigma^2 h- \sigma^2 h{\wt\bE}\Big(W_{1}^{2}[{\bf 1}_{\{\sigma\sqrt{h}W_{1}+S_{h}>\varepsilon\}}+{\bf 1}_{\{\sigma\sqrt{h}W_{1}-S_{h}>\varepsilon\}}]\Big) \\
\nonumber
&  \qquad \qquad \qquad +2\sigma\sqrt{h}\wt\bE\Big(W_1 S_{h}\,{\bf 1}_{\{|\sigma\sqrt{h} W_{1}+S_{h}|\leq\varepsilon\}}\Big)\\
\label{eq:12decomp}
&=\sigma^{2}h-\sigma\sqrt{h}\varepsilon\wt\bE\left(\phi\bigg(\frac{\varepsilon+S_{h}}{\sigma\sqrt{h}}\bigg)+\phi\bigg(\frac{\varepsilon-S_{h}}{\sigma\sqrt{h}}\bigg)\right)\\
\nonumber
&\qquad\quad+\sigma\sqrt{h}\wt\bE\left(S_{h}\bigg(\phi\bigg(\frac{\varepsilon+S_{h}}{\sigma\sqrt{h}}\bigg)-\phi\bigg(\frac{\varepsilon-S_{h}}{\sigma\sqrt{h}}\bigg)\bigg)\right)\\
\nonumber
& \qquad\quad -\sigma^{2}h\wt\bE\left( \overline{\Phi}\bigg(\frac{\varepsilon+S_{h}}{\sigma\sqrt{h}}\bigg)+ \overline{\Phi}\bigg(\frac{\varepsilon-S_{h}}{\sigma\sqrt{h}}\bigg)\right),
\end{align}
where in the last equality we condition on $S_{h}$ and apply the identities ${\wt\bE}\left(W_{1}^{2}{\bf 1}_{\{W_{1}>x\}}\right)=x\phi(x)+\overline{\Phi}(x)$ and ${\wt\bE}\left(W_{1}{\bf 1}_{\{x_1<W_{1}<x_2\}}\right)=\phi(x_1)-\phi(x_2)$.
From \eqref{eq:EJk}, we have $\wt{\bE}\phi\left(\frac{\varepsilon\pm {S_h}}{\sigma\sqrt{h}}\right)=\frac{\sigma\sqrt{h}}{2\pi} \mathcal{I}(0)$ and
\begin{align}
\nonumber
    &\wt{\bE}\left(\left(\mp  S_h \right)\,\phi\bigg(\frac{\varepsilon\pm S_h}{\sigma\sqrt{h}}\bigg)\right)\\
    \nonumber
    &\quad=  \varepsilon \frac{\sigma\sqrt{h}}{2\pi}\left(\mathcal{I}(0)+(-1)^{1}\,2 \left(i\varepsilon\right)^{-1} \left(\frac{\sigma^2 h}{2}\right) \mathcal{I}(1)\right)\\
    \nonumber
    &\quad= \varepsilon\wt\bE\phi\left( \frac{\varepsilon \pm S_{h} }{\sigma\sqrt{h}}\right)
    +O \left(h^{5/2} \varepsilon^{-Y-2}\right)\\
    & \qquad \qquad \qquad + O \left(\varepsilon e^{-\frac{\varepsilon^2}{2\sigma^2 h}} \right)+ O\left(h^{\frac{4-Y}{2Y}}e^{-\delta h^{\frac{Y-2}{Y}}}\right), 
    \label{eq:ESphi}
\end{align}
where in the last equality we used \eqref{e:moments_of_FT_expression}. Therefore, plugging \eqref{eq:ESphi} into \eqref{eq:12decomp}, we obtain  that
\begin{align}\nonumber
\mathcal{T}_{1,n}
&= \sigma^2 h  - 2 \sigma\sqrt{h}\, \varepsilon\, {\wt \bE}\left(\phi\bigg(\frac{\varepsilon-S_{h}}{\sigma\sqrt{h}}\bigg) +\phi\bigg(\frac{\varepsilon+S_{h}}{\sigma\sqrt{h}}\bigg)\right)\\
\nonumber
&\qquad \quad - \sigma^2 h \, {\wt \bE}\left(\overline{\Phi}\bigg(\frac{\varepsilon -S_{h}}{\sigma\sqrt{h}}\bigg) + \overline{\Phi}\bigg(\frac{\varepsilon +S_{h}}{\sigma\sqrt{h}}\bigg)\right)\\
&\qquad \quad+  O\left(h^3 \varepsilon^{-Y-2}\right)+O\bigg(\varepsilon\sqrt{h}e^{-\frac{\varepsilon^{2}}{2\sigma^{2}h}}\bigg)
+ O\left(h^{\frac{2}{Y}}e^{-\delta h^{\frac{Y-2}{Y}}}\right).
 \label{eq:w2}
\end{align}
By Lemma \ref{lemmaC1},
\begin{align}%
\quad{\wt\bE}\left(\phi\bigg(\frac{\varepsilon\pm S_{h}}{\sigma\sqrt{h}}\bigg)\right)&={{C}_{\pm}\sigma h ^{\frac{3}{2}} \varepsilon^{-Y-1} + O\left(h^{\frac{5}{2}}\varepsilon^{-Y-3}\right)+O\bigg(e^{-\frac{\varepsilon^{2}}{2\sigma^{2}h}}\bigg).}  
\label{eq:Ephi}
\end{align}
Also, by Lemma \ref{lemma:lambda}, we have
\begin{align} \label{eq:Ephibar}
\qquad    {\wt \bE}\left(\overline{\Phi}\left( \frac{\varepsilon \pm S_{h} }{\sigma\sqrt{h}}\right)\right) &= \frac{C_\mp}{Y}h\varepsilon^{-Y}  + O\left(h^2\varepsilon^{-2-Y}\right)+{O\left(h\varepsilon^{2-2Y}\right)}\\
    &\quad+{O\left(h^{1/2}\varepsilon^{-1}e^{-\frac{\varepsilon^2}{2\sigma^2 h}}\right)+O\left(h^{-1/2} \varepsilon^{3-Y}  e^{-\frac{\varepsilon^2}{8\sigma^2 h}} \right)}.\nonumber
\end{align}
Finally, {we combine \eqref{eq:w2}, \eqref{eq:Ephi}, \eqref{eq:Ephibar}, Lemma \ref{lemma:S2k_new}, and \eqref{eq:Decompb1epsstable}} to accomplish the proof.
\end{proof}

The following result gives an expansion for the truncated second moment of a L\'evy process with a tempered stable jump component.
\begin{proposition}\label{lemma:2E} Suppose $Y\in (0,1)\cup(1,2)$, and $\varepsilon \to 0$ with  $\sqrt h \ll \varepsilon$.
%
Let $\widehat{X}$ and $\widetilde{X}$ be defined as in \eqref{DfnHatX0} and \eqref{DfnTildeX0}, and suppose that $J$  is a L\'evy process having a L\'evy density of the form \eqref{eq:levyden} with $q$ satisfying conditions (i)-(v) of Assumptions \ref{assump:Funtq} and  \ref{assump:FuntqStrong}. Then, as  $h\rightarrow 0$, for any fixed $\bar\delta>0$,
\begin{align}\nonumber
    &\bE\left(  \widehat{X}_h^2 \,{\bf 1}_{\{|\widehat{X}_h|\leq \varepsilon\}} \right)- \bE\left( \wt X_h^2 \,{\bf 1}_{\{|\wt X_h|\leq \varepsilon\}}\right)\\
    &=O\left(h^{3}\varepsilon^{-Y-2}\right)+O\big(h^{\frac{3}{2}}\varepsilon^{1-\frac{Y}{2}}\big) + O\left(h^2\varepsilon^{2-2Y}\right) + O(h \varepsilon^{2-\bar\delta}) \\
    &\quad+ O\left(\sqrt{h}\varepsilon e^{-\frac{\varepsilon^2}{2\sigma^2h}}\right) + {O\left(\sqrt{h}\varepsilon^{3-Y}  e^{-\frac{\varepsilon^2}{8\sigma^2 h}} \right)}.
    \nonumber
\end{align}
 In particular, in light of \eqref{ExpForTildeX2}, we have:
\begin{align}\nonumber
	\bE\left( \wh X_h^2 \,{\bf 1}_{\{|\wh X_h|\leq \varepsilon\}}\right)&= \sigma^2 h + \widetilde{A}(\varepsilon,h)h\\
	\label{EstNIK}
	&\quad + O\left(h^{3}\varepsilon^{-Y-2}\right)+O\big(h^{\frac{3}{2}}\varepsilon^{1-\frac{Y}{2}}\big) + O\left(h^2\varepsilon^{2-2Y}\right)\\
	\nonumber
& \quad +  O(h \varepsilon^{2-\bar\delta})+O\left(h^{\frac{1}{2}}\varepsilon e^{-\frac{\varepsilon^2}{2\sigma^2 h}}\right)+ 
{O\left(h^{\frac{1}{2}} \varepsilon^{3-Y}  e^{-\frac{\varepsilon^2}{8\sigma^2 h}} \right)} .
\end{align} 
\end{proposition}
\begin{proof}
 We set $\bar{C}:=C_++C_-$ and again consider the simple decomposition%
\begin{align}
\nonumber
\bE\left(\widehat X_h^2\, {\bf 1}_{\{|X_h^2|\leq\varepsilon\}}\right)
& =\bE\Big(A_n(\sigma W_{h},J_h)\Big)+
 2\bE\Big(B_n(\sigma W_{h},J_h)\Big)\\
&\quad +\bE\Big(C_n(\sigma W_{h},J_h)\Big),
 \label{eq:Decompb1epsdiff}
\end{align}
where $A_n(w,s)=w^2{\bf 1}_{\{|w+s|\leq{}\varepsilon_n\}}$, $B_n(w,s)=w s{\bf 1}_{\{|w+s|\leq{}\varepsilon_n\}}$, and
$C_n(w,s)=s^2{\bf 1}_{\{|w+s|\leq{}\varepsilon_n\}}$.
We prove the result in two steps.

\bigskip
\noindent
\textbf{Step 1.} We first analyze the difference between each of the first two terms in \eqref{eq:Decompb1epsdiff} and their counterparts under $\wt \bP$. By \eqref{eq:DenTranTildePP} and  \eqref{e:def_St},
\begin{align}
\nonumber
\mathcal{R}_{n,1}
&:=\bE(A_n(\sigma W_{h},J_h))-\wt\bE(A_n(\sigma W_{h},S_h))\\
&\qquad\qquad+
 2\big\{\bE(B_n(\sigma W_{h},J_h))-\wt\bE(B_n(\sigma W_{h},S_h))\big\}
\end{align}
can be decomposed as 
\begin{align}
\nonumber
\mathcal{R}_{n,1}&=e^{-\eta h} \left(-\sigma^2h \,\wt{\bE}\Big(W_{1}^{2}\,{\bf 1}_{\{|\sigma\sqrt{h}W_{1}+J_{h}|>\varepsilon\}}\Big) + 2\sigma \sqrt{h}\,\wt{\bE}\Big(W_{1}J_{h}{\bf 1}_{\{|\sigma\sqrt{h}W_{1}+J_{h}|\leq\varepsilon\}}\Big)\right)\\
\nonumber
&\qquad -\left(- \sigma^2 h{\wt\bE}\Big(W_{1}^{2}\,{\bf 1}_{\{|\sigma\sqrt{h}W_{1}+S_{h}|>\varepsilon\}}\Big) + 2\sigma\sqrt{h}\wt\bE\Big(W_1 S_{h}\,{\bf 1}_{\{|\sigma\sqrt{h} W_{1}+S_{h}|\leq\varepsilon\}}\Big)\right)\\
\nonumber
&\qquad -\sigma^2he^{-\eta h}\,\wt{\bE}\Big(\Big(e^{-\wt{U}_{h}}-1\Big)W_{1}^{2}\,{\bf 1}_{\{|\sigma\sqrt{h}W_{1}+J_{h}|>\varepsilon\}}\Big) \\
\nonumber
&\qquad + 2\sigma\sqrt{h}\,e^{-\eta h}\,\wt{\bE}\Big(\Big(e^{-\wt{U}_{h}}-1\Big)W_{1}J_{h}{\bf 1}_{\{|\sigma\sqrt{h}W_{1}+J_{h}|\leq\varepsilon\}}\Big)\\
&=:e^{-\eta h} \, I_{1}(h) - I_{2}(h) -\sigma^2he^{-\eta h}I_{3}(h) + 2\sigma\sqrt{h} e^{-\eta h}\,I_{4}(h).
\label{eq:Decompb1eps1J} 
\end{align}
Note that, written in terms of the term $\mathcal{T}_{1,n}$ of the proof of Proposition \ref{lemma:Estable}, $I_2=-\sigma^{2}h+\mathcal{T}_{1,n}$. Therefore, combining \eqref{eq:w2}-\eqref{eq:Ephibar}, we get
\begin{align*}%
I_2(h)  &= -\bar{C}\frac{2Y+1}{Y}\sigma^2 h ^{2} \varepsilon^{-Y} +  O\left(h^3 \varepsilon^{-Y-2}\right)+{O\left(h^2\varepsilon^{2-2Y}\right)}\\
&\quad+{O\bigg(\sqrt{h}\varepsilon e^{-\frac{\varepsilon^2}{2\sigma^2 h}}\bigg) 
+O\left(\sqrt{h}\varepsilon^{3-Y}  e^{-\frac{\varepsilon^2}{8\sigma^2 h}} \right)}.
\end{align*}
For $I_1(h)$, first note that, as in (\ref{eq:12decomp}), we have:
\begin{align} \nonumber
I_{1}(h) &= - \sigma\sqrt{h}\, \varepsilon\, \wt\bE\left(\phi\bigg(\frac{\varepsilon-J_{h}}{\sigma\sqrt{h}}\bigg) + \phi\bigg(\frac{\varepsilon+J_{h}}{\sigma\sqrt{h}}\bigg)\right)\\
\label{eq:s2J}
&\quad - \sigma\sqrt{h}\, \wt\bE\left(J_h\phi\bigg(\frac{\varepsilon-J_{h}}{\sigma\sqrt{h}}\bigg) - J_h\phi\bigg(\frac{\varepsilon+J_{h}}{\sigma\sqrt{h}}\bigg)\right) \\
\nonumber
&\quad - \sigma^2 h \wt\bE\left(\overline{\Phi}\bigg(\frac{\varepsilon -J_{h}}{\sigma\sqrt{h}}\bigg) + \overline{\Phi}\bigg(\frac{\varepsilon +J_{h}}{\sigma\sqrt{h}}\bigg)\right).
\end{align}
Recall that, under $\wt{\bP}$, $J_t=S_t+\tilde{\gamma}t$. Then, fixing $\tilde{\varepsilon}_{\pm}=\varepsilon\pm \tilde{\gamma}h$, we can apply \eqref{eq:Ephi} to deduce that:
\begin{align}%
\nonumber
\widetilde{\bE}\phi\bigg(\frac{\varepsilon\pm J_{h}}{\sigma\sqrt{h}}\bigg)&=\wt\bE\phi\bigg(\frac{\tilde{\varepsilon}_{\pm}\pm S_{h}}{\sigma\sqrt{h}}\bigg)\\
& = C_{\pm}\sigma h ^{\frac{3}{2}} \tilde{\varepsilon}_{\pm}^{-Y-1} + O\left(h^{\frac{5}{2}}\varepsilon^{-Y-3}\right)+O\bigg(e^{-\frac{\varepsilon^{2}}{2\sigma^{2}h}}\bigg).
\label{eq:Ephi0}
\end{align}
Similarly, by \eqref{eq:ESphi} and \eqref{eq:Ephibar},
\begin{align}%
\label{eq:EphiForJ}
\qquad \widetilde{\bE}\bigg(J_h\phi&\bigg(\frac{\varepsilon\pm J_{h}}{\sigma\sqrt{h}}\bigg)\bigg)\\
\nonumber
&={\wt\bE}\left(S_h\phi\bigg(\frac{\tilde{\varepsilon}_{\pm}\pm S_{h}}{\sigma\sqrt{h}}\bigg)\right)\, +\,\tilde{\gamma}h\,{\wt\bE}\left(\phi\bigg(\frac{\tilde{\varepsilon}_{\pm}\pm S_{h}}{\sigma\sqrt{h}}\bigg)\right)\\
\nonumber
&=\mp\varepsilon\,{\wt\bE}\left(\phi\bigg(\frac{\tilde{\varepsilon}_{\pm}\pm S_{h}}{\sigma\sqrt{h}}\bigg)\right)+O\left(h^{\frac{5}{2}} \varepsilon^{-Y-2}\right)+O\bigg(\varepsilon e^{-\frac{\varepsilon^{2}}{2\sigma^{2}h}}\bigg),
\end{align}
and
\begin{align}\nonumber
  ~ ~~ \widetilde{\bE}\overline{\Phi}\bigg(\frac{\varepsilon\pm J_{h}}{\sigma\sqrt{h}}\bigg) &= \widetilde{\bE}\overline{\Phi}\bigg(\frac{\tilde{\varepsilon}_{\pm}\pm S_{h}}{\sigma\sqrt{h}}\bigg)\\
\label{AsymBarPhi2}
    & =\frac{C_\mp}{Y}h\tilde\varepsilon_{\pm}^{-Y} +  O\left(h^2\varepsilon^{-2-Y}\right)+{O\left(h\varepsilon^{2-2Y}\right)}\\
    \nonumber
    &\quad +O\left(\sqrt{h}\varepsilon^{-1}e^{-\frac{\varepsilon^2}{2\sigma^2 h}}\right) +{O\left(h^{-1/2} \varepsilon^{3-Y}  e^{-\frac{\varepsilon^2}{8\sigma^2 h}} \right).}
\end{align}
Together,  \eqref{eq:Ephi0}, \eqref{eq:EphiForJ}, and \eqref{AsymBarPhi2} imply that $I_1(h)$ has the following asymptotic behavior:
\begin{align}\nonumber
I_1(h)  &= -\frac{ 2C_+ Y+C_-}{Y}\sigma^2 h ^{2} \tilde\varepsilon_+^{-Y}-\frac{2C_-Y+C_+}{Y}\sigma^2 h ^{2} \tilde\varepsilon_-^{-Y}+  O\left(h^3 \varepsilon^{-Y-2}\right)\\
\label{AsympI1J}
& \qquad +{O\left(h^2\varepsilon^{2-2Y}\right)}  +O\bigg(\sqrt{h}\varepsilon e^{-\frac{\varepsilon^2}{2\sigma^2 h}}\bigg) %
+{O\left(h^{1/2} \varepsilon^{3-Y}  e^{-\frac{\varepsilon^2}{8\sigma^2 h}} \right)}.
\end{align}
Therefore, we obtain that 
\begin{align}\nonumber
e^{-\eta h} &\, I_{1}(h) - I_{2}(h) \\
\nonumber
&= I_{1}(h) - I_{2}(h) + I_1(h)O\left(h\right)
\\
\nonumber
&= -\sigma^{2}h^2\frac{2 C_+Y+C_-}{Y}\left( \tilde\varepsilon_+^{-Y}-\varepsilon^{-Y}\right)\\
\nonumber
&\qquad -\sigma^{2}h^2\frac{2C_-Y+C_+}{Y}\left( \tilde\varepsilon_-^{-Y}-\varepsilon^{-Y}\right)+O\left(h^3 \varepsilon^{-Y-2}\right)\\
\nonumber
&\qquad+{O\left(h^2\varepsilon^{2-2Y}\right)}+O\Big(\sqrt{h}\varepsilon e^{-\frac{\varepsilon^2}{2\sigma^2 h}}\Big)
+{O\Big(\sqrt{h} \varepsilon^{3-Y}  e^{-\frac{\varepsilon^2}{8\sigma^2 h}} \Big)}\\
\nonumber
&= O\left(h^3 \varepsilon^{-Y-2}\right)+{O\left(h^2\varepsilon^{2-2Y}\right)}\\
& \qquad\qquad +O\Big(\sqrt{h}\varepsilon e^{-\frac{\varepsilon^2}{2\sigma^2 h}}\Big)
+{O\left(\sqrt{h} \varepsilon^{3-Y}  e^{-\frac{\varepsilon^2}{8\sigma^2 h}} \right)},
\label{eq:I12J}
\end{align}
where in the last equality we used that $h^{2}(\tilde\varepsilon_{\pm}^{-Y}-\varepsilon^{-Y})=O(h^3\varepsilon^{-Y-1})$.
Next, using the proof in the Step 1.2 and Step 3 of Theorem 3.1 of \cite{gong2021}, we know, as $h\rightarrow 0$, 
\begin{align}
    I_3(h) &= \wt{\bE}\Big(\Big(e^{-\wt{U}_{h}}-1\Big)W_{1}^{2}\,{\bf 1}_{\{|\sigma\sqrt{h}W_{1}+J_{h}|>\varepsilon\}}\Big) = o\big(h^{1/Y}\big),\label{eq:I3J}\\
    I_4(h) &= \wt{\bE}\Big(\Big(e^{-\wt{U}_{h}}-1\Big)W_{1}J_{h}{\bf 1}_{\{|\sigma\sqrt{h}W_{1}+J_{h}|\leq\varepsilon\}}\Big) = O\big(h\varepsilon^{1-Y/2}\big).\label{eq:I4J}
\end{align}
Thus, combining \eqref{eq:Decompb1eps1J}, \eqref{eq:I12J}, \eqref{eq:I3J}, and \eqref{eq:I4J}, we  have
\begin{align}
    \label{eq:first2diff} 
    \mathcal{R}_{n,1}&=
O\left(h^{3}\varepsilon^{-Y-2}\right)+{O\left(h^2\varepsilon^{2-2Y}\right)}+ O\big(h^{\frac{3}{2}}\varepsilon^{1-\frac{Y}{2}}\big) + o\big(h^{1+\frac{1}{Y}}\big)\\
\nonumber
&\quad 
+O\left(\sqrt{h}\varepsilon e^{-\frac{\varepsilon^2}{2\sigma^2h}}\right)+{O\left(\sqrt{h} \varepsilon^{3-Y}  e^{-\frac{\varepsilon^2}{8\sigma^2 h}} \right)}.
\end{align}

\smallskip
\noindent
\textbf{Step 2.} In this step, we will study the asymptotic behavior of the third term in \eqref{eq:Decompb1epsdiff}, as $h\rightarrow 0$. By \eqref{e:change_of_measure}, \eqref{eq:DenTranTildePP} 
and \eqref{e:def_St}, we  have
\begin{align}
\nonumber
\mathcal{R}_{n,2}&:=\bE\Big(J_{h}^{2}\,{\bf 1}_{\{|\sigma W_{h}+J_{h}|\leq\varepsilon\}}\Big) - \wt\bE\Big(S_{h}^{2}\,{\bf 1}_{\{|\sigma W_{h}+S_{h}|\leq\varepsilon\}}\Big)\\
\nonumber
&= \left(\wt{\bE}\Big(S_{h}^{2}\,{\bf 1}_{\{|\sigma W_{h}+S_{h}+\wt{\gamma}h|\leq\varepsilon\}}\Big) - \wt\bE\Big(S_{h}^{2}\,{\bf 1}_{\{|\sigma W_{h}+S_{h}|\leq\varepsilon\}}\Big)\right) \\
\nonumber
&\qquad \qquad \qquad \qquad + e^{-\eta h}\, \wt{\bE}\Big(\Big(e^{-\wt{U}_{h}}-1\Big)S_{h}^{2}\,{\bf 1}_{\{|\sigma W_{h}+S_{h}+\wt{\gamma}h|\leq\varepsilon\}}\Big)\\
\nonumber
& \qquad \qquad \qquad \qquad +2\wt{\gamma}he^{-\eta h}\,\wt{\bE}\Big(e^{-\wt{U}_{h}}S_{h}{\bf 1}_{\{|W_{h}+S_{h}+\wt{\gamma}h|\leq\varepsilon\}}\Big) \\
\nonumber
& \qquad \qquad \qquad \qquad +\wt{\gamma}^{2}h^{2}e^{-\eta h}\,\wt{\bE}\Big(e^{-\wt{U}_{h}}\,{\bf 1}_{\{|W_{h}+S_{h}+\wt{\gamma}h|\leq\varepsilon\}}\Big)\\
\nonumber
&\qquad \qquad \qquad \qquad +O(h)\cdot \wt{\bE}\Big(S_{h}^{2}\,{\bf 1}_{\{|\sigma W_{h}+S_{h}+\wt{\gamma}h|\leq\varepsilon\}}\Big)\\
 &=:I_{5,1}(h)  +e^{-\eta h}I_{5,2}(h)+2\wt{\gamma}he^{-\eta h}I_{6}(h)+O\big(h^{2}\big)+ O\Big(\varepsilon^{1-Y}h^{5/2}e^{-\frac{\varepsilon^2}{2\sigma^2 h}}\Big),%
\label{eq:Decompb1eps2J} 
\end{align}
where in the last equality we used  Lemma \ref{lemma:S2k_new}.
We begin with the analysis of $I_{5,1}(h)$. 
First note that, for large enough $n$,  
\begin{align*}
	    \wt\bE\Big(S_{h}^{2}[&{\bf 1}_{\{|\sigma W_{h}+S_{h}|\leq\varepsilon - |\tilde\gamma| h\}}-{\bf 1}_{\{|\sigma W_{h}+S_{h}|\leq\varepsilon\}}]\Big)\\
	    &\leq{}I_{5,1}(h)\leq\wt\bE\Big(S_{h}^{2}[{\bf 1}_{\{|\sigma W_{h}+S_{h}|\leq\varepsilon + |\tilde\gamma| h\}}-{\bf 1}_{\{|\sigma W_{h}+S_{h}|\leq\varepsilon\}}]\Big),
\end{align*}
and, since, by   Lemma \ref{lemma:S2k_new},
\begin{align*}
\wt\bE&\Big(S_{h}^{2}\,{\bf 1}_{\{|\sigma W_{h}+S_{h}|\leq\varepsilon \pm |\tilde\gamma| h\}}\Big) - \wt\bE\Big(S_{h}^{2}\,{\bf 1}_{\{|\sigma W_{h}+S_{h}|\leq\varepsilon\}}\Big)\\
    &= \frac{\bar{C}}{2-Y} h \left(\left(\varepsilon\pm|\wt\gamma| h\right)^{2-Y} -\varepsilon^{2-Y}\right) \\
    & \quad + \frac{\bar{C}(1-Y)}{2} \sigma^2 h^{2} \left(\left(\varepsilon\pm|\wt\gamma| h\right)^{-Y} - \varepsilon^{-Y}\right) \\
    &\quad + O\left(h^3\varepsilon^{-2-Y}\right) + O\left(h^2\varepsilon^{2-2Y}\right) +  O\left(\varepsilon^{1-Y}h^{\frac{3}{2}}e^{-\frac{\varepsilon^2}{2\sigma^2 h}}\right)\\%
    &=   O\left(h^3\varepsilon^{-2-Y}\right)+O\left(h^2\varepsilon^{2-2Y}\right) +  O\left(h^{\frac{3}{2}} \varepsilon^{1-Y} e^{-\frac{\varepsilon^2}{2\sigma^2h}}\right),
\end{align*}
we conclude that $I_{5,1}(h)=  O\left(h^3\varepsilon^{-2-Y}\right)+O\left(h^2\varepsilon^{2-2Y}\right) + O\big(h^{\frac{3}{2}} \varepsilon^{1-Y} e^{-\frac{\varepsilon^2}{2\sigma^2h}}\big)$. 
In \ref{MoreMoreAuxLm}, expression \eqref{e:I_{52}_improved}, we show
\begin{align}\label{ImprovEstI652}
	I_{5,2}(h) =  O(h \varepsilon^{2-\bar\delta}),
\end{align}
valid for all $\bar\delta>0$. Further, an application of Cauchy-Schwarz (c.f. Step 2.2 in  \cite{gong2021}) and Lemma \ref{lemma:S2k_new} together show 
$I_6(h)  = O\big(h^{1/2}\varepsilon^{1-Y/2}\big).$
Thus, 
we obtain 
\begin{align}
\nonumber
\mathcal{R}_{n,2}&=  O\left(h^3\varepsilon^{-2-Y}\right)+ O\left(h^2\varepsilon^{2-2Y}\right) +  O\left(h^{\frac{3}{2}} \varepsilon^{1-Y} e^{-\frac{\varepsilon^2}{2\sigma^2h}}\right) \\
\nonumber
    &\quad + O(h \varepsilon^{3-Y} \big(h \varepsilon^{-Y}\big)^{-\delta})
    + O\big(h^{3/2}\varepsilon^{1-Y/2}\big)\\
    \nonumber
    &\quad+O\big(h^{2}\big)+ O\left(\varepsilon^{1-Y}h^{5/2}e^{-\frac{\varepsilon^2}{2\sigma^2 h}}\right)\\
    \nonumber
&=  O\left(h^3\varepsilon^{-2-Y}\right)+O\left(h^2\varepsilon^{2-2Y}\right)\\
& \quad + {O(h \varepsilon^{2-\bar\delta})} + O\left(h^{\frac{3}{2}} \varepsilon^{1-Y} e^{-\frac{\varepsilon^2}{2\sigma^2h}}\right).\label{eq:thirddiff}
\end{align}
Finally, combining \eqref{eq:first2diff} and \eqref{eq:thirddiff}, we conclude the result.
\end{proof}

Next, we give an expansion for the truncated higher-order  moments of $\widehat{X}$.
\begin{proposition}\label{prop:EX2} {Suppose $Y\in (0,1)\cup(1,2)$, and $\varepsilon \to 0$ with  $\sqrt h \ll \varepsilon$.}
Let $k$ be a positive integer. Then, as $h\to 0$,
\begin{align}\label{EstNIKbb}
    \bE\left(\widehat{X}_h^{2k}\,{\bf 1}_{\{|\widehat{X}_h|\leq\varepsilon\}}\right) &= (2k-1)!!\,\sigma^{2k} h^k + \frac{C_{+}+C_{-}}{2k-Y}\,h\varepsilon^{2k-Y}\\
    \nonumber
    & \quad+ O\left(h^2 \varepsilon^{2k-Y-2}\right)
    +O\big(h\varepsilon^{2k-Y/2}\big)+{O\left(h^{\frac{1}{2}}\varepsilon^{2k-1} e^{-\frac{\varepsilon^2}{2\sigma^2 h}}\right)}.
\end{align}
\end{proposition}
\begin{proof}
By Lemmas \ref{lemma:W2k}, \ref{lemma:J2k}, and \ref{lemma:WJ}, we can compute%
\begin{align*}
    &\bE\left(\widehat X_h^{2k}\,{\bf 1}_{\{|X_h|\leq\varepsilon\}}\right) = \bE\Big(\big(\sigma W_{h}+J_{h}\big)^{2k}\,{\bf 1}_{\{|\sigma W_{h}+J_{h}|\leq\varepsilon\}}\Big)\\
    &\quad =\bE\left(\sigma^{2k}W_h^{2k}{\bf 1}_{\{|\sigma W_{h}+J_{h}|\leq\varepsilon\}}\right) + \bE\left(J_h^{2k}{\bf 1}_{\{|\sigma W_{h}+J_{h}|\leq\varepsilon\}}\right) \\
    & \quad\qquad+ \bE\left(\sum_{j=1}^{k-1}\frac{(2k)!}{(2j)! \, (2k-2j)!} (\sigma W_h)^{2j} J_h^{2k-2j}{\bf 1}_{\{|\sigma W_{h}+J_{h}|\leq\varepsilon\}}\right)\\
    & \quad\qquad + \bE\left(\sum_{j=0}^{k-1}\frac{(2k)!}{(2j+1)! \, (2k-2j-1)!} (\sigma W_h)^{2j+1} J_h^{2k-2j-1} {\bf 1}_{\{|\sigma W_{h}+J_{h}|\leq\varepsilon\}}\right)\\
    &\quad= (2k-1)!!\,\sigma^{2k} h^k + 
    O\big(h^2 \varepsilon^{2k-Y-2}\big)+ o\left(h^{k+\frac{1}{Y}}\right)+ {O\left(h^{\frac{1}{2}}\varepsilon^{2k-1} e^{-\frac{\varepsilon^2}{2\sigma^2 h}}\right)}\\%
    &\quad\qquad+ \frac{C_{+}\!+\!C_{-}}{2k-Y}h\varepsilon^{2k-Y} +  O\big(h^2\varepsilon^{2k-Y-2}\big) + O\big(h\varepsilon^{2k-Y/2}\big)%
   \\
    & \quad\qquad + O\left(h^2 \varepsilon^{2k-Y-2}\right) + O\left(h^2 \varepsilon^{2k-Y-2}\right) + O\left(h^{3/2} \varepsilon^{2k-Y/2-1}\right).
\end{align*}
The result then follows. 
\end{proof}
We are now ready to show that the asymptotic results above hold without the conditions (iii)-(vi) of Assumption \ref{assump:FuntqStrong}.
\begin{lemma}\label{prop:Genq0}
	The asymptotic expansions in Eqs.~\eqref{EstNIK} and \eqref{EstNIKbb} hold true for a L\'evy process $J$ having a L\'evy density of the form \eqref{eq:levyden} with $q$ satisfying only conditions (i)-(ii) of Assumption \ref{assump:Funtq}.
\end{lemma}
\begin{proof}
Let us recall the processes $J^\infty$ and $J^0_t$ introduced in \eqref{e:def_J^infty} and \eqref{e:def_J^0}.
Write 
\begin{align}
	\widehat X '_t&=\sigma W_t+J^\infty_t.
\end{align}
Note that $J^\infty_t$ satisfies all the conditions of Assumptions \ref{assump:Funtq} and  \ref{assump:FuntqStrong} and, therefore, $\widehat X '_t$ enjoys the expansions \eqref{EstNIK} and \eqref{EstNIKbb}. We also have that $J^0_t$ is of finite jump activity.  Next, writing $V=J^0=\widehat{X}-\widehat X '$, we  have
\begin{align}\nonumber
\big|\widehat{X}_h^2 {\bf 1}_{\{|\widehat{X}_h|\leq \varepsilon\}} -(\widehat X'_h)^2 {\bf 1}_{\{|\widehat X'_h|\leq \varepsilon\}} \big| &= \big|V_h^2 + 2  V_h \widehat X '_h \big| {\bf 1}_{\{|\widehat{X}_h|\leq \varepsilon, ~|\widehat X'_h|\leq \varepsilon\}}\\
\nonumber
 & \qquad + \widehat{X}_h^2  {\bf 1}_{\{|\widehat{X}_h|\leq \varepsilon, ~|\widehat X'_h|> \varepsilon\}}\\
 \nonumber
  & \qquad + (\widehat X'_h)^2  {\bf 1}_{\{|\widehat{X}_h|>\varepsilon, |\widehat X'_h|\leq \varepsilon\}}\\
  \label{FncyDcmIbb}
    & =: T_1 + T_2 + T_3. 
\end{align}
{Since $V$ has finite jump activity, for any $p>0$ we have}
$$
\bE_{i-1}\left(\left|\frac{V_h}{\varepsilon} \right|^p {\bf 1}_{\{|V_h|\leq \varepsilon\}}\right)\leq \bP(|V_h|>0) \leq Kh.
$$
For $T_1$, since ${\bf 1}_{\{|x+v|\leq \varepsilon,|x|\leq \varepsilon\}}\leq{}{\bf 1}_{\{|v|\leq 2\varepsilon\}}$, we have
\begin{align*}
 \bE T_{1}&\leq K \Big(\bE V_h^2 {\bf 1}_{\{|V_h|\leq { 2}\varepsilon\}}  + \varepsilon\bE|V_h| {\bf 1}_{\{|V_h|\leq { 2}\varepsilon\}}  \Big) \leq K h \varepsilon^2.
\end{align*}
 For $T_{2}$, note ${\bf 1}_{\{|x+v|\leq \varepsilon,|x|> \varepsilon\}}\leq{\bf 1}_{\{|v|>0\}}$.  So, writing $N^0$ for the jump measure of $V$, we have $\bP(|V_h|>0)\leq \bP\big(N^0([0,h],\bR)>0\big) \leq Kh$, giving 
 \begin{align*}
 \bE T_{2}&\leq K \varepsilon^2 \bP(|V_h|>0) \leq K \varepsilon^{2}h.
  \end{align*}
  For $T_{3}$, since  ${\bf 1}_{\{|x+v|> \varepsilon,|x|\leq  \varepsilon\}}\leq{\bf 1}_{\{|v|>0\}}$, we again have
   \begin{align*}
\bE T_{3}&\leq K \varepsilon^2 \bP(|V_h|>0) \leq K \varepsilon^{2}h.
  \end{align*}
 Therefore, we obtain
 \begin{equation}\label{e:X-X'_estimatebb}
 \big|\bE\widehat{X}_h^2 {\bf 1}_{\{|\widehat{X}_h|\leq \varepsilon\}} -\bE(\widehat X'_h)^2 {\bf 1}_{\{|\widehat X '_h|\leq \varepsilon\}} \big| =O(h \varepsilon^{2}),
 \end{equation}
 which together with the fact that $\widehat X '$ satisfies the asymptotic \eqref{EstNIK}, implies that 
 \begin{align}
 \nonumber
	&\bE\left( \wh X_h^2 \,{\bf 1}_{\{|\wh X_h|\leq \varepsilon\}}\right)\\
\nonumber	
	\quad&= \sigma^2 h + \widetilde{A}(\varepsilon,h)h\\
\label{EstNIKbc}	
	&\quad + O\left(h^{3}\varepsilon^{-Y-2}\right)+O\big(h^{\frac{3}{2}}\varepsilon^{1-\frac{Y}{2}}\big) + O\left(h^2\varepsilon^{2-2Y}\right)\\
	\nonumber
& \quad +  {O(h \varepsilon^{2-\bar\delta})}+O\left(h^{\frac{1}{2}}\varepsilon e^{-\frac{\varepsilon^2}{2\sigma^2 h}}\right)+ 
{O\left(h^{\frac{1}{2}} \varepsilon^{3-Y}  e^{-\frac{\varepsilon^2}{8\sigma^2 h}} \right)}.
\end{align} 
  We can similarly prove \eqref{EstNIKbb}. Concretely, applying Corollary 2.1.9 in \cite{JacodProtter} and Proposition \ref{prop:EX2}, we obtain
 \begin{equation}\label{e:X-X'_estimatecc}
 \big|\bE\widehat{X}_h^{2k} {\bf 1}_{\{|\widehat{X}_h|\leq \varepsilon\}} -\bE(\widehat X'_h)^{2k} {\bf 1}_{\{|\widehat X '_h|\leq \varepsilon\}} \big| =O(h \varepsilon^{2k}),
 \end{equation}
 which together with the fact that $\widehat X '$ satisfies \eqref{EstNIKbb}, implies the same for $ \bE\big(\widehat{X}_h^{2k}\,{\bf 1}_{\{| \widehat{X}_h|\leq\varepsilon\}}\big) $.
\end{proof}

%
%
%



%

\section{Bounds for the discretization error of integrals}\label{OtherSuperTechProofs}

One of the key tools behind our results is to discretize the integrals involved in the increment $\Delta_i^nX$. Concretely, we approximate $\Delta_i^nX$ with
\begin{align}\label{NtnEXs00}
	x_i:=b_{t_{i-1}}h+\sigma_{t_{i-1}}\Delta_i^n W
	+\chi_{t_{i-1}}\Delta_i^n J=:x_{i,1}+x_{i,2}+x_{i,3}.
\end{align}
In this section, we aim to provide bounds for the corresponding discretization errors:
\begin{align}
	\label{NtnEXs}
	\quad&\mathcal{E}_{i}:=\Delta_i^n X-x_i,\quad
	\mathcal{E}_{i,1}=\int_{t_{i-1}}^{t_i}b_s ds-{b}_{t_{i-1}}h,\\
	\nonumber
	&\mathcal{E}_{i,2}=\int_{t_{i-1}}^{t_i}\sigma_s dW_s-\sigma_{t_{i-1}}\Delta_i W,
	\qquad \mathcal{E}_{i,3}=\int_{t_{i-1}}^{t_i}\chi_s dJ_s-\chi_{t_{i-1}}\Delta_i J.
\end{align}

\begin{lemma} \label{l:ldiscretiz_error_moment_bounds} Suppose $Y\in(0,1)\cup(1,2)$ and let $\mathcal E_{i,\ell}$, $\ell=1,2,3$ be defined as in \eqref{NtnEXs}.   Then, 
\begin{align}\label{MDEE00}
\mathbb{E}_{i-1}\left[\left|\mathcal{E}_{i,\ell}\right|^p\right]\leq C \begin{cases}  h^{p}& \ell=1,~p>0,\\
h^{\frac{(2 \wedge p)+p}{2}}&\ell=2,~p>0, \\
 h^{1+\frac{p}{2}}   & \ell=3,~ p \in [1,\infty)\cap (Y,\infty).
  \end{cases}
\end{align}
\end{lemma}

\begin{proof}
Recall, by virtue of a localization argument, we can assume that the processes $b,\chi$ are bounded (by a nonrandom constant).
 For $\mathcal E_{i,1}$, $\bE_{i-1} |\mathcal E_{i,1}|^r \leq K(\int_{t_i}^{t_{i-1}} ds)^r\leq Kh^r$.
 For $\mathcal E_{i,2}$, using the Burkholder-Davis-Gundy inequality and Jensen's inequality we have
\begin{align*}
\bE_{i-1}|\mathcal E_{i,2}|^{r}&\leq K  \begin{cases}
\Big( \bE_{i-1}\int_{t_{i-1}}^{t_i} |\sigma_s-\sigma_{t_{i-1}}|^2 ds\Big)^{r/2} & 0<r<2\\
h^{\frac{r}{2}-1} \int_{t_{i-1}}^{t_i} \bE_{i-1}  |\sigma_s-\sigma_{t_{i-1}}|^r ds & r\geq 2.
\end{cases}\\
&\leq K \begin{cases}
h^r & 0<r<2\\
h^{1+\frac{r}{2}} & r\geq 2.
\end{cases}
\end{align*}
Turning now to $\mathcal E_{i,3}$, let $\delta(s,z) = (\chi_{s^-}- \chi_{t_{i-1}})z$, so that we may write $\mathcal E_{i,3}= \int_{t_{i-1}}^{t_i}\int_\bR \delta(s,z)(N(ds,dz)-\nu(dz)ds)$.  First observe, for any $r>Y,$
\begin{align}
\nonumber
\bE_{i-1}\int_{t_{i-1}}^{t_i}\int_\bR |\delta(s,z)|^r  \nu(dz)ds&\leq K   \bE_{i-1}\int_{t_{i-1}}^{t_i} |\chi_{s}-\chi_{t_{i-1}}|^r ds \int_{\bR}|z|^r \nu(dz)\\
& \leq K  h^{1+\frac{r}{2}}.\label{e:naive_M_bound}
\end{align}
Thus, for the case $p\in (Y,2]\cap [1,2]$, by Lemma 2.1.5 in \cite{JacodProtter} 
\begin{align*}
\bE_{i-1}(|\mathcal E_{i,3}|^p) &\leq K\bE_{i-1}\Big(\int_{t_{i-1}}^{t_i} \int_\bR |\delta(s,z)|^p  \nu(dz)ds\Big)\\
& \leq K h^{1+\frac{p}{2}}.
\end{align*}
For the case $p>2$, again by Lemma 2.1.5 in \cite{JacodProtter}, we obtain
\begin{align*}
\bE_{i-1}(|\mathcal E_{i,3}|^p) &\leq K\bigg( \bE_{i-1}\int_{t_{i-1}}^{t_i} \int_\bR  |\delta(s,z)|^p  \nu(dz)ds\\
& \qquad\qquad+ h^{p/2}\bE_{i-1}\bigg(\frac{1}{h}\int_{t_{i-1}}^{t_i} \int_\bR|\delta(s,z)|^2  \nu(dz)ds\bigg)^{p/2}\bigg)\\
& \leq  K \bigg(h^{1+\frac{p}{2}}  + h^{p/2-1} \int_{t_{i-1}}^{t_i} \bE_{i-1}\bigg(\int_{\bR}|\delta(s,z)|^2  \nu(dz)\bigg)^{p/2}ds\bigg)\\
& \leq Kh^{1+\frac{p}{2}}.
\end{align*}
This completes the proof.
\end{proof}


For the following lemma, we need to laying out some additional notation related to the process $J$.
Recall that $N$ denotes the Poisson jump measure of $J$. Let $\bar{N}$ be its compensated measure under $\mathbb{P}$. Observe that due to condition (ii) in Assumption \ref{assump:Funtq}, there exists $\delta_0\in(0,1)$ such that $q(x)>0$ for all $|x|\leq{}\delta_0$. Next, {for $p<Y \wedge 1$}, let $\breve J$ be a pure-jump L\'evy process independent  of $J$ with triplet $(0,0,\breve \nu)$, where $\breve{\nu}(dx)={e^{-|x|^p}}(C_{+}{\bf 1}_{(0,\infty)}(x)+C_{-}{\bf 1}_{(-\infty,0)}(x)){\bf 1}_{|x|>\delta_0}\,|x|^{-1-Y}dx$ and define the L\'evy process
\begin{equation}\label{e:def_J^infty}
	J^\infty_t=\Big(b+\int_{\delta_0<|x|\leq 1}x\nu(dx)\Big)t+\int_{0}^t \int_{|x|\leq{}\delta_0}x\bar{N}(ds,dx)+\breve{J}_t.
\end{equation}
In other words, $J^\infty$ has L\'evy measure $\nu(dx){\bf 1}_{\{|x|\leq \delta_0\}} + \breve{\nu}(dx){\bf 1}_{\{|x|>\delta_0\}}$, and, in particular,  $J^\infty_t$ satisfies all the conditions of Assumptions \ref{assump:FuntqStrong}.  Next, we write
\begin{equation}\label{e:def_J^0}
	J^0_t:=J_t-J^\infty_t=\int_{0}^t \int_{|x|>{}\delta_0}x{N}(ds,dx)-\breve{J}_t,
\end{equation}
and observe $J^0$ has finite jump activity.

The next lemma gives, for an appropriate $\delta_n\to 0$, an upper bound for the quantity $\bE[(\Delta_i^nJ^{\infty})^{p} {\mathbf 1}_{\{|\Delta_i^nJ^{\infty}|\leq \varepsilon\}}{\bf 1}_{\{|\mathcal E_i|>\delta_n\}}]$.
\begin{lemma} \label{l:V_i3} Suppose $Y\in (0,1)\cup(1,2)$, and  $\varepsilon=\varepsilon_n\to 0$ with $h^{\frac{1}{2}-s}\ll \varepsilon $ for some $s>0$. Consider any sequence $\delta_n\to 0$ satisfying $\delta_n \geq h^{1/2}\varepsilon^{Y/2}$.  Then, for any $p\geq 0$, with $\mathcal E_i$ as in \eqref{NtnEXs},
\begin{equation}\label{e:less_than_h_threeHalves}
\bE[(\Delta_i^nJ^\infty)^{p} {\mathbf 1}_{\{|\Delta_i^nJ^\infty|\leq \varepsilon\}}{\bf 1}_{\{|\mathcal E_i|>\delta_n\}}] =O(h^2 \varepsilon ^{{p}-2Y}). 
\end{equation}
\end{lemma}
\begin{proof}
Recall that $J^\infty$ satisfies both Assumptions \ref{assump:Funtq} and  \ref{assump:FuntqStrong}. For simplicity in this proof we write $J$ in place of $J^\infty$, and with slight abuse of notation denote its jump and L\'evy measure by $N(ds,dz)$ and $\nu(dz)$, respectively.
Using the relation ${\bf 1}_{\{|\mathcal E_i|>\delta_n\}} \leq  {\bf 1}_{\{|\mathcal E_{i,1}|>\delta_n/3\}} + {\bf 1}_{\{|\mathcal E_{i,2}|>\delta_n/3\}} + {\bf 1}_{\{ |\mathcal E_{i,3}|>\delta_n/3\}}$, we may consider each $\mathcal E_{i,\ell}$ $\ell=1,2,3$ separately.  For $\mathcal{E}_{i,1}$, note $\{|\mathcal E_{i,1}|>\delta_n/3\} =\emptyset$ for large enough $n$, since $\delta_n\gg h^{1/2}\varepsilon \gg h$, and $|\mathcal E_{i,1}|\leq K h$ almost surely.  For $\mathcal E_{2,i}$,  clearly $\bE|\mathcal E_{i,2}|^2 \leq K h^2$ which using Lemma \ref{lemma:J2k} implies 
\begin{align*}
\bE[(\Delta_i^nJ)^{p} {\mathbf 1}_{\{|\Delta_i^nJ|\leq \varepsilon\}}{\bf 1}_{\{|\mathcal E_{i,2}|>\delta_n\}}] & \leq  \bE[(\Delta_i^nJ)^{p} {\mathbf 1}_{\{|\Delta_i^nJ|\leq \varepsilon\}}] \delta_n^{-2}\bE|\mathcal E_{i,2}|^2\\
& \leq K h^3\varepsilon^{p-Y} \delta_n^{-2}\\
& \leq K h^2 \varepsilon^{p-2Y}. 
\end{align*}
So, it remains to consider $\mathcal E_{i,3}$. Choose an $s_0>0$ small enough so that
\begin{equation}\label{e:rate_of_u_n}
 h^{1/2}\varepsilon \ll  h^{s_0}\delta_n =: u_n.
 \end{equation}
For simplicity below we write $u=u_n$. Letting $\delta(s,z)= (\chi_{s^-}- \chi_{t_{i-1}})z$, and $\bar{N}(ds,dz) = N(ds,dz)-\nu(dz)ds$, define
$$
N(u) = \int_{t_{i-1}}^{t_i}\int_{{\{|\delta(s,z) |>u\}}} N(ds,dz)
$$
$$
 M(u) =\int_{t_{i-1}}^{t_i}\int_{{\{|\delta(s,z) |\leq u\}}}\delta(s,z) {\bar{N}(ds,dz)}
 $$
 $$
 B(u) =-\int_{t_{i-1}}^{t_i}\int_{{\{|\delta(s,z) |>u\}}} \delta(s,z)  \nu(dz)ds
$$
Since $N(u) = 0$ implies $\mathcal E_{i,3}= M(u) + B(u)$, we have the following:
\begin{align}
\nonumber
(\Delta_i^nJ)^{p} &{\bf 1}_{\{|\Delta_i^nJ|\leq \varepsilon\}}{\bf 1}_{\{|\mathcal E_{i,3}|>\delta_n\}}\\
& = (\Delta_i^nJ)^{p} {\bf 1}_{\{|\Delta_i^nJ|\leq \varepsilon\}}{\bf 1}_{\{|\mathcal E_{i,3}|>\delta_n\}}\big( {\bf 1}_{\{N(u)>0\}}+  {\bf 1}_{\{N(u)=0\}}\big)\notag \\
&  \leq  \varepsilon^{p} ({\bf 1}_{\{N(u)>0,|\Delta_i^n J|\leq \varepsilon\}} +  {\bf 1}_{\{|M(u)|>\delta_n/2\}} + {\bf 1}_{\{|B(u)|>\delta_n/2\}})\notag \\
& =: T_1 + T_2 + T_3.\label{e:3term}
\end{align}
 Now, for the first term in \eqref{e:3term}, we split up $\Delta_i^nJ$ based on jumps of size $2\varepsilon$.  For simplicity write $w=2\varepsilon$ and consider the decomposition
 $$
 \Delta_i^n J = L(w) + M'(w) + B'(w),
 $$
 where 
 $$
L(w) = \int_{t_{i-1}}^{t_i}\int_{{\{|z |>w\}}} z N(ds,dz),
$$
$$
 M'(w) =\int_{t_{i-1}}^{t_i}\int_{{\{|z|\leq w\}}}z {\bar N(ds,dz)} \quad B'(w) =-\int_{t_{i-1}}^{t_i}\int_{{\{|z |>w\}}}z \nu(dz)ds,
$$
Further, let
$$
L_+(w) =  \int_{t_{i-1}}^{t_i}\int_{{\{z >w\}}} z N(ds,dz) \quad L_-(w) =  \int_{t_{i-1}}^{t_i}\int_{{\{z <-w\}}} z N(ds,dz),
$$
so that $L(w) = L_+(w) + L_-(w)$.  Finally, analogous to $N(u)$, let $$N'(w)= \int_{t_{i-1}}^{t_i}\int_{{\{|z |>w\}}}  N(ds,dz) = \mu\big(\{z:|z |>w\} ,(t_{i-1},t_i]\big),
$$
and similarly let
$$
N'_+(w) = \mu\big(\{z:z >w\} ,(t_{i-1},t_i]\big) \quad N'_-(w)=\mu\big(\{z:z <-w\} ,(t_{i-1},t_i]\big),
$$
i.e., $N'(w)=N'_+(w) + N'_+(w)$. 
Returning to the  first term in \eqref{e:3term}, first observe $$|\delta(s,z)|>u\iff |z|> |\chi_{s^-}- \chi_{t_{i-1}}|^{-1}u.$$ 
We now split the set $\{|z|>|\chi_{s^-}- \chi_{t_{i-1}}|^{-1}u\}$ and consider the cases  and $w<|\chi_{s^-}- \chi_{t_{i-1}}|^{-1}u$ and $w\geq |\chi_{s^-}- \chi_{t_{i-1}}|^{-1}u$ separately. For the case when $w<|\chi_{s^-}- \chi_{t_{i-1}}|^{-1}u$, note
$$
w<|\chi_{s^-}- \chi_{t_{i-1}}|^{-1}u  \implies N(u) \leq N'(w), 
$$
since $\{|z|>w\} \supseteq \{|z| \geq |\chi_{s^-}- \chi_{t_{i-1}}|^{-1}u \}$. Using the decomposition $\{N'(w) >0\} = \{N_+'(w)>0,N_-'(w)=0\}\cup \{N_+'(w)=0,N_-'(w)>0\} \cup\{N_+'(w)N_-'(w)>0\}$, we therefore have
\begin{align}
\nonumber
&{\bf 1}_{\{w<|\chi_{s^-}- \chi_{t_{i-1}}|^{-1}u\}}{\bf 1}_{\{N(u)>0,|\Delta_i^n J|\leq \varepsilon\}}\\
\nonumber
& \leq {\bf 1}_{\{N'(w) >0,|\Delta_i^n J|\leq \varepsilon\}}\\
\nonumber
&={\bf 1}_{\{N_+'(w)>0,N_-'(w)=0,|\Delta_i^n J|\leq \varepsilon\}} + {\bf 1}_{\{N_+'(w)=0,N_-'(w)>0,|\Delta_i^n J|\leq \varepsilon\}}\\
& \qquad\quad+ {\bf 1}_{\{N_+'(w)N_-'(w)>0,|\Delta_i^n J|\leq \varepsilon\}}  \nonumber
\\
\label{e:case1_T1+T2+T3}
&=: T_{1,1} + T_{1,2} + T_{1,3}.
\end{align}
For $T_{1,1},$ note the event $\{N_+'(w)>0,N_-'(w)=0\}=\{N_+'(w)\geq 1,N_-'(w)=0\}$ implies  
$\{ L_+(w)\geq w, |L(w)|=L_+(w)\}$.  Also, using that $|\Delta_i^n J|\geq  |L(w)| - |M'(w)| - |B'(w)|$, we have 
$\{|\Delta_i^n J|\leq \varepsilon\}\subseteq \{ |L(w)| - |M'(w)| - |B'(w)| \leq \varepsilon \} =  \{ |M'(w)| \geq |L(w)| - \varepsilon - |B'(w)|  \} $. Therefore, we obtain
\begin{align*}
\bE T_{1,1} &\leq \bP\Big( L_+(w)\geq w, |L(w)|=L_+(w),~|\Delta_i^n J|\leq \varepsilon\Big)\\
& \leq \bP\Big( L_+(w)\geq w, |L(w)|=L_+(w),~|M'(w)| \geq |L(w)| - \varepsilon - |B'(w)|\Big)\\
&  \leq \bP\Big(L_+(w)\geq w, |M'(w)| \geq w - \varepsilon - |B'(w)|\Big)\\
&  \leq \bP\Big(L_+(w)\geq w\Big)\bP\Big( |M'(w)| \geq w - \varepsilon - |B'(w)|\Big)\\
& \leq K (h w^{-Y})\cdot hw^{2-Y}  (w - \varepsilon - |B'(w)|)^{-2}\\
& \leq K (h \varepsilon ^{-Y})^2,
\end{align*}
where we used independence of $L_+(w)$ and $M'(w)$, the inequalities $\bE L_+(w) \leq Khw^{1-Y}$, $|B'(w)|\leq K h \varepsilon^{1-Y}\ll \varepsilon $, and also that $w=2\varepsilon$.  Analogous arguments show that $\bE T_{1,2} \leq K h^2 \varepsilon ^{-2Y}$.  For $T_{1,3}$, simply note
$$
\bE T_{1,3} \leq \bP(N_+'(w)\geq 1, N_-'(w)\geq 1) 
\leq K (h \varepsilon^{-Y})^2.
$$
Thus, turning back to the first term in \eqref{e:3term},  we obtain from  \eqref{e:case1_T1+T2+T3},
\begin{equation}\label{e:T1_bound_partA}
\bE  \big(T_1{\bf 1}_{\{w<|\chi_{s^-}- \chi_{t_{i-1}}|^{-1}u\}}\big) \leq \varepsilon^p \bE(T_{1,1} + T_{1,2} + T_{1,3}) \leq  h^2 \varepsilon ^{p-2Y}.
\end{equation}
Now we consider the case $\{w\geq |\chi_{s^-}- \chi_{t_{i-1}}|^{-1}u\} = \{|\chi_{s^-}- \chi_{t_{i-1}}|> u/w\} = \{|\chi_{s^-}- \chi_{t_{i-1}}|> u\varepsilon^{-1}/2\}$.  Based on \eqref{e:rate_of_u_n}, we may write $u\varepsilon^{-1} = h^{1/2} \phi_n$ where $\phi_n\gg h^{-s_1}$ for some small $s_1>0$. Thus, using Markov's inequality, for every $r>0$,
\begin{align}
 \bE  \big(T_1 {\bf 1}_{\{w\geq |\chi_{s^-}- \chi_{t_{i-1}}|^{-1}u\}}\big) & \leq \varepsilon^{p} \bE  \big({\bf 1}_{  \{|\chi_{s^-}- \chi_{t_{i-1}}|> h^{1/2}\phi_n/2\}}\big)\notag\\
& \leq K\varepsilon^{p} (h^{-1/2}\phi_n^{-1})^{r} \bE |\chi_{s^-}- \chi_{t_{i-1}}|^{r}\notag\\
& \leq K \varepsilon ^{p} \phi_n^{-r}=o(h^2 \varepsilon ^{-2Y}),  \label {e:T1_bound_partB}
\end{align}
by taking $r$ large enough. 
\noindent Now, turning to $T_2$ in \eqref{e:3term}, By Lemma 2.1.5 in \cite{JacodProtter}, with $r\geq 2$,
\begin{align*}
&\bE_{i-1}(|M(u)|^{r}) \\
&\leq  K\Bigg( \bE_{i-1}\Big(\int_{t_{i-1}}^{t_i} \int_{\{|\delta(s,z)|<u\}}  |\delta(s,z)|^{r}  \nu(dz)ds\Big)\\
&\qquad\qquad\qquad\qquad\quad+ h^{{r}/2} \bE_{i-1}\Big(\frac{1}{h}\int_{t_{i-1}}^{t_i} \int_{\{|\delta(s,z)|<u\}}  |\delta(s,z)|^2 \nu(dz)ds\Big)^{{r}/2} \Bigg)\\
& \leq K \bigg(u^{{r}-Y} h^{1+Y/2}+ {h^{{r}/2}\bE_{i-1}\Big[\frac{1}{h }\int_{t_{i-1}}^{t_i}\Big( \int_{\{|\delta(s,z)|<u\}}  |\delta(s,z)|^2 \nu(dz)\Big)^{{r}/2}ds\Big]}\bigg)\\
& \leq K \big(u^{{r}-Y} h^{1+Y/2} + h^{r/2}\big)\\
& \leq  K u^{{r}-Y} h^{1+Y/2},
\end{align*}
 where the last is a consequence of the lower bound in assumption \eqref{e:rate_of_u_n} and that $\varepsilon \gg h^{1/2}$.  Thus, taking $r$ large enough, we obtain
\begin{align}
\nonumber
\bE T_2 \leq& \varepsilon^{p} \bP\big(|M(u)|>\delta_n/2\big) \leq K \varepsilon^{p} \delta_n^{-{r}} \bE|M(u)|^{r} \\
&\leq  \varepsilon^{p} \delta_n^{-{r}} u^{{r}-Y} h^{1+Y/2} =o(h^2 \varepsilon ^{-2Y}).\label{e:T2_bound}
\end{align}
Turning to $T_3$, since $\chi$ is bounded we have $|B(u)| \leq K h u^{1-Y} \ll K h^{3/2}\varepsilon^{1-Y}  \ll \delta_n$, implying $T_3= 0$ for all large $n$.
Thus, based on \eqref{e:T1_bound_partA}, \eqref{e:T1_bound_partB}, \eqref{e:T2_bound}, we obtain
\begin{align}
\bE[(\Delta_i^nJ)^{p} {\bf 1}_{\{|\Delta_i^nJ|\leq \varepsilon\}}{\bf 1}_{\{|\mathcal E_{i,3}|>\delta_n\}}] & = \bE T_1 + \bE T_2 \notag\\
& \leq K \big( h^2 \varepsilon ^{p-2Y} + o(h^2 \varepsilon ^{p-2Y} )\big).
\end{align}
This establishes  \eqref{e:less_than_h_threeHalves}.
\end{proof}


 Our last result of  this section gives an estimate for the discretization error of approximating Riemann integrals by Riemann sums in the case where the integrand is an It\^o semimartingale. 
\begin{lemma}\label{l:Riemann}
As $h\to0$, 
	$$\sum_{i=1}^{n}\sigma^2_{t_{i-1}}h-\int_0^1\sigma_s^2ds=o_{P}(n^{-1/2}),$$
	and, for all $0<Y<2$, with $\sqrt h \ll \varepsilon \ll 1$,
	$$\varepsilon^{2-Y}\sum_{i=1}^{n}|\chi_{t_{i-1}}|^{Y}h-\varepsilon^{2-Y}\int_0^1|\chi_s|^Yds=O_P(\varepsilon^{2-Y}h^{\frac{Y\wedge 1}{2}})$$
	$$\text{and} \quad h\varepsilon^{-Y}\sum_{i=1}^{n}|\chi_{t_{i-1}}|^{Y}\sigma_{t_{i-1}}^2h-h\varepsilon^{-Y}\int_0^1|\chi_s|^Y\sigma_s^2 ds=o_P(\varepsilon^{2-Y}h^{\frac{Y\wedge 1}{2}}).$$
\end{lemma}

\begin{proof}
Consider
$$
A_n:=\sum_{i=1}^{n}f_n(\sigma_{t_{i-1}},\chi_{t_i-1})h-\int_0^1f_n(\sigma_s,\chi_s)ds
$$
for the cases $f_n(v,x) = v^2$, $f_n(v,x)=\varepsilon^{2-Y}|x|^Y$, and $f_n(v,x)= h \varepsilon^{-Y}v^2|x|^{Y} $. The case for $f_n(v,x)=v^2$, establishing $A_n=o_P(n^{-1/2})$ is standard (e.g. one can argue as in (5.3.24) in \cite{JacodProtter}) and hence we omit its proof.  For the case $f_n(v,x)= h \varepsilon^{-Y}v^2|x|^{Y}$, 
$$|f_n(\sigma_{t},\chi_{t})-f_n(\sigma_s,\chi_s)|\leq K \varepsilon^{2-Y} |\chi_t-\chi_s|^{Y\wedge 1}$$
which implies $\bE| A_n| \leq \varepsilon^{2-Y}h^{\frac{Y \wedge 1}2}.$ For the remaining case $f_n(v,x)= h \varepsilon^{-Y}v^2|x|^{Y}$, since $h\ll \varepsilon^2$, we have the same bound: $|f_n(\sigma_{t},\chi_{t})-f_n(\sigma_s,\chi_s)|\leq K  h \varepsilon^{-Y} |\chi_t-\chi_s|^{Y\wedge 1}\ll \varepsilon^{2-Y}|\chi_t-\chi_s|^{Y\wedge 1}$, which again gives $\bE| A_n| =o(\varepsilon^{2-Y}h^{\frac{Y \wedge 1}{2}})$.
\end{proof}

\section{Proofs {for the  stable-in-law} CLT results in \cite{BonieceFLHan:2022}}\label{CLTACnJ} 
In this section, we verify the technical condition (2.2.40) in \cite{JacodProtter} to conclude the stable convergence of our estimators.
\begin{lemma}\label{ErroDiscrHigherOrder22}
Suppose that $h^{\beta}\ll \varepsilon\ll h^{\frac{1}{4-Y}}$, for some $\frac{1}{4-Y}<\beta<\frac{1}{2}$. 
Define 
\[
	\xi_i^n := \sqrt{n}\left(\big(\Delta_{i}^{n}X\big)^{2}{\bf 1}_{\{|\Delta_{i}^{n}X|\leq\varepsilon\}}-\bE_{i-1}\big(\big(\Delta_{i}^{n}X\big)^{2}\,{\bf 1}_{\{|\Delta_{i}^{n}X|\leq\varepsilon\}}\big)\right).
\] 
Then, as $n\to\infty$,
\begin{align}\label{NdTchCnd}
	\sum_{i=1}^{n}\mathbb{E}_{i-1}\left[\xi_{i}^n({\Delta_i^nM})\right]\,\stackrel{\mathbb{P}}{\longrightarrow}\,0,
\end{align}
when $M=W$ or $M$ is a square-integrable martingale orthogonal to $W$.
\end{lemma} 

\begin{proof}
Recall the notation \eqref{NtnEXs00}-\eqref{NtnEXs}. 
We need to show that 
\begin{align*}
	\sqrt{n}\sum_{i=1}^{n}\mathbb{E}_{i-1}\left[\big(\Delta_{i}^{n}X\big)^{2}{\bf 1}_{\{|\Delta_{i}^{n}X|\leq\varepsilon\}}({\Delta_i^nM})\right]\,\stackrel{\mathbb{P}}{\longrightarrow}\,0,
\end{align*}
We treat each term in the expansion $(x_i+\mathcal{E}_i)^2=x_i^2+2x_i\mathcal{E}_i+\mathcal{E}_i^2$ separately. 
For the second and third terms, we use Cauchy's inequality. Indeed, for $\mathcal E_i^2$,
\begin{align*}
	n\left|\sum_{i=1}^{n}\mathbb{E}_{i-1}\left[\mathcal{E}_i^2{\bf 1}_{\{|\Delta_{i}^{n}X|\leq\varepsilon\}}({\Delta_i^nM})\right]\right|^2&\leq
	n\left|\sum_{i=1}^{n}\mathbb{E}_{i-1}[\mathcal{E}_i^4]^{\frac{1}{2}}\mathbb{E}_{i-1}[({\Delta_i^nM})^2]^{\frac{1}{2}}\right|^2\\
	&\leq
	n\sum_{i=1}^{n}\mathbb{E}_{i-1}[\mathcal{E}_i^4]\sum_{i=1}^{n}\mathbb{E}_{i-1}[({\Delta_i^nM})^2]\\
	&\leq{}n\sum_{i=1}^{n}O_P(h^{3})O_P(1)\,\stackrel{\mathbb{P}}{\longrightarrow}\,0,
\end{align*}
since $\mathbb{E}[\sum_{i=1}^{n}\mathbb{E}_{i-1}[({\Delta_i^nM})^2]]=\sum_{i=1}^{n}\mathbb{E}[({\Delta_i^nM})^2]=
\mathbb{E}[(M_{1}-M_{0})^2]<\infty$. For the term $2x_i\mathcal{E}_i$, when $|\mathcal{E}_i|>{}\varepsilon$, since $ h^{\frac{1}{2}}/\varepsilon\ll h^{\frac{1}{2}-\beta}$, we can use Markov's inequality to conclude  $\bP(|\mathcal E_i|>\varepsilon)\leq K h^{1+(\frac{1}{2}-\beta)r}$ for all $r>Y$ (c.f. \eqref{MDEE00}) and can apply H\"older's inequality to show that 
\[
	\sqrt{n}\sum_{i=1}^{n}\mathbb{E}_{i-1}\left[|x_i\mathcal E_i|{\bf 1}_{\{|\Delta_{i}^{n}X|\leq\varepsilon\}}{\bf 1}_{\{|\mathcal E_i|>\varepsilon\}}({\Delta_i^nM})\right]=o_P(1).
\] 
 When $|\mathcal{E}_i|\leq{}\varepsilon$, note  $|\Delta_i^n X|=|x_i+\mathcal{E}_{i}|\leq\varepsilon$ implies that $|x_i|\leq{}2\varepsilon$. Next,
\begin{align*}
	n\Big|\sum_{i=1}^{n}&\mathbb{E}_{i-1}\left[|\mathcal{E}_i x_i|{\bf 1}_{\{|x_i|\leq2\varepsilon\}}{\bf 1}_{\{|\mathcal E_i|\leq\varepsilon\}}|{\Delta_i^nM}|\right]\Big|^2\\
	&\leq
	n\left|\sum_{i=1}^{n}\mathbb{E}_{i-1}[\mathcal{E}_i^2x_i^2{\bf 1}_{\{|x_i|\leq2\varepsilon\}}]^{\frac{1}{2}}\mathbb{E}_{i-1}[({\Delta_i^nM})^2]^{\frac{1}{2}}\right|^2\\
	&\leq
	n\sum_{i=1}^{n}\mathbb{E}_{i-1}[\mathcal{E}_i^2x_i^2{\bf 1}_{\{|x_i|\leq2\varepsilon\}}]\sum_{i=1}^{n}\mathbb{E}_{i-1}[({\Delta_i^nM})^2]\\
	&\leq{}n\sum_{i=1}^{n}O_P(h^{\frac{3}{2}})O_P(h)\sum_{i=1}^{n}\mathbb{E}_{i-1}[({\Delta_i^nM})^2]\,\stackrel{\mathbb{P}}{\longrightarrow}\,0,
\end{align*}
where in the last inequality we apply Proposition 
\ref{prop:EX2}.
All that remains is to show
\begin{align}\label{SLTF0}
	\sqrt{n}\sum_{i=1}^{n}\mathbb{E}_{i-1}\left[x_i^2{\bf 1}_{\{|x_i+\mathcal{E}_{i}|\leq\varepsilon\}}({\Delta_i^nM})\right]\,\stackrel{\mathbb{P}}{\longrightarrow}\,0.
\end{align}
Now, when $x_i$ is replaced with $b_{t_{i-1}}h$ in \eqref{SLTF0}, we proceed as follows:
\begin{align}\label{IKICD}
\sqrt{n}\left|\sum_{i=1}^{n}b^2_{t_{i-1}}h^2\mathbb{E}_{i-1}\left[{\bf 1}_{\{|x_i+\mathcal{E}_{i}|\leq\varepsilon\}}({\Delta_i^nM})\right]\right|\leq C h^2
\sqrt{n}\sum_{i=1}^{n}\mathbb{E}_{i-1}\left[|{\Delta_i^nM}|^2\right]^{\frac{1}{2}}\\
\leq C h^2
n\Big(\sum_{i=1}^{n}\mathbb{E}_{i-1}\left[|{\Delta_i^nM}|^2\right]\Big)^{\frac{1}{2}}\to{}0,\nonumber
\end{align}
where we used 
$\sum_{i=1}^{n}\mathbb{E}_{i-1}\left[|{\Delta_i^nM}|^2\right]^{\frac{1}{2}}\leq{}n^{1/2}\big(\sum_{i=1}^{n}\mathbb{E}_{i-1}\left[|{\Delta_i^nM}|^2\right]
\big)^{1/2}$.
We also have 
\begin{align*}
\sqrt{n}\Big|&\sum_{i=1}^{n}b_{t_{i-1}}h\mathbb{E}_{i-1}\left[\hat{x}_{i}{\bf 1}_{\{|x_i+\mathcal{E}_{i}|\leq\varepsilon\}}({\Delta_i^nM})\right]\Big|\\
&\leq C h\sqrt{n}\sum_{i=1}^{n}\mathbb{E}_{i-1}\left[\hat{x}_i^2\right]^{\frac{1}{2}}\mathbb{E}_{i-1}\left[|{\Delta_i^nM}|^2\right]^{\frac{1}{2}}\\
&\leq C h\sqrt{n}\Big(\sum_{i=1}^{n}\mathbb{E}_{i-1}\left[\hat{x}_i^2\right]\Big)^{\frac{1}{2}}\Big(\sum_{i=1}^{n}\mathbb{E}_{i-1}\left[|{\Delta_i^nM}|^2\right]\Big)^{\frac{1}{2}}\stackrel \bP \to{}0.
\end{align*}
So, it suffices to show 
\begin{align}\label{SLTFbb}
	\sqrt{n}\sum_{i=1}^{n}\mathbb{E}_{i-1}\left[\hat{x}_i^2{\bf 1}_{\{|x_i+\mathcal{E}_{i}|\leq\varepsilon\}}({\Delta_i^nM})\right]\,\stackrel{\mathbb{P}}{\longrightarrow}\,0,
\end{align}
where recall that $\hat{x}_{i}=\sigma_{t_{i-1}}\Delta_i^n W
	+\chi_{t_{i-1}}\Delta_i^n J$. 
We first prove it when $M=W$. 
When we replace $\hat{x}_i$  with $\sigma_{t_{i-1}}\Delta_{i}^n W$ in \eqref{SLTFbb}, we have
\begin{align*}
	\sqrt{n}\Big|&\sum_{i=1}^{n}\sigma_{t_{i-1}}^2\mathbb{E}_{i-1}\left[\Delta_{i}^nW^2{\bf 1}_{\{|x_i+\mathcal{E}_{i}|\leq\varepsilon\}}\Delta_i^nM\right]\Big|\\
	&=	\sqrt{n}\Big|\sum_{i=1}^{n}\sigma_{t_{i-1}}^2\mathbb{E}_{i-1}\left[\Delta_{i}^nW^3{\bf 1}_{\{|x_i+\mathcal{E}_{i}|>\varepsilon\}}\right]\Big|\\
		&\leq	\sqrt{n}\sum_{i=1}^{n}\sigma_{t_{i-1}}^2\frac{\mathbb{E}_{i-1}\left[|\Delta_{i}^nW|^3|x_i+\mathcal{E}_{i}|\right]}{\varepsilon}\\
	&\leq C	\sqrt{n}\sum_{i=1}^{n}\mathbb{E}_{i-1}\left[\Delta_{i}^nW^6\right]^{\frac{1}{2}}\frac{\mathbb{E}_{i-1}\left[|x_i+\mathcal{E}_{i}|^2\right]^{\frac{1}{2}}}{\varepsilon}\\
	&\leq Cn^{\frac{3}{2}}h^{\frac{3}{2}}\frac{h^{\frac{1}{2}}}{\varepsilon}\to{}0.
\end{align*}
When $\hat{x}_i$  is replaced with $\chi_{t_{i-1}}\Delta_{i}^n J$ in \eqref{SLTFbb}, from boundedness of $\chi$, we need only to show that
\begin{align}\label{SLTFcc}
	\sqrt{n}\sum_{i=1}^{n}\mathbb{E}_{i-1}\left[\Delta_{i}^nJ^2{\bf 1}_{\{|x_i+\mathcal{E}_{i}|\leq\varepsilon\}}|\Delta_i^nW|\right]\,\stackrel{\mathbb{P}}{\longrightarrow}\,0.
\end{align}
Note that when $|\Delta_i^nW|>{}\varepsilon$ or $|\mathcal{E}_i|>{}\varepsilon$, we can use H\"older's inequality, and that  $\bP(|\mathcal E_i|>\varepsilon)$ and  $\bP(|\Delta^n_iW|>\varepsilon)$ decay faster than any power of $h$ (a consequence of $ h^{\frac{1}{2}}/\varepsilon\ll h^{\frac{1}{2}-\beta}$) to obtain $$\sqrt n\sum_{i=1}^{n}\mathbb{E}_{i-1}[\Delta_{i}^nJ^2{\bf 1}_{\{|x_i+\mathcal{E}_{i}|\leq\varepsilon\}}|\Delta_i^nW|({\bf 1}_{\{|\Delta_i^nW|>{}\varepsilon\}} + {\bf 1}_{\{|\mathcal{E}_i|>{}\varepsilon\}})=o_P(1).$$
  So, when $|\Delta_i^nW|\leq{}\varepsilon$ and $|\mathcal{E}_i|\leq{}\varepsilon$, note $|x_i+\mathcal{E}_{i}|\leq\varepsilon$ implies that $|\Delta_i^{n}J|\leq{}C\varepsilon$, for some positive constant $C$ that we assume for simplicity is $1$. Then, for \eqref{SLTFcc}, it suffices to establish that $\sqrt{n}\sum_{i=1}^{n}\mathbb{E}_{i-1}\left[\Delta_{i}^nJ^2{\bf 1}_{\{|\Delta_i^nJ|\leq\varepsilon\}}|\Delta_i^nW|\right]\,\stackrel{\mathbb{P}}{\longrightarrow}\,0$. By H\"older's inequality and Lemma \ref{lemma:J2k}, for any $p,q>1$ such that $1/p+1/q=1$,
\begin{align}\label{SLTFdd}
	\sqrt{n}\sum_{i=1}^{n}&\mathbb{E}_{i-1}\left[\Delta_{i}^nJ^2{\bf 1}_{\{|\Delta_i^nJ|\leq\varepsilon\}}|\Delta_i^nW|\right]\\
	\nonumber
	&\leq{}\sqrt{n}\sum_{i=1}^{n}\mathbb{E}_{i-1}\left[\Delta_{i}^nJ^{2p}{\bf 1}_{\{|\Delta_i^nJ|\leq\varepsilon\}}\right]^{\frac{1}{p}}\mathbb{E}_{i-1}\left[|\Delta_i^nW|^q\right]^{\frac{1}{q}}\\
	\nonumber
	&\leq C\sqrt{n}\sum_{i=1}^{n}\left(h\varepsilon^{2p-Y}\right)^{\frac{1}{p}}h^{\frac{1}{2}}\\
	\nonumber
	&\leq{}C n^{\frac{3}{2}}h^{\frac{1}{p}+\frac{1}{2}}\varepsilon^{2-\frac{Y}{p}};
\end{align}
so, making $p$ close to $1$ (and $q$ large), we get the convergence to $0$, establishing \eqref{SLTFbb} when $M=W$.


Now, we show \eqref{SLTFbb} for any bounded square integrable $M$ that is orthogonal to $W$. We first consider the case when $\hat{x}_i$ is replaced with $\chi_{t_{i-1}}\Delta_i^nJ$. As before, when $|\Delta_i^nW|>{}\varepsilon$ or $|\mathcal{E}_i|>{}\varepsilon$ we can use H\"older's inequality and the fast decay of $\bP(|\Delta_i^n W|>\varepsilon)$ and $\bP(|\mathcal E_i|>\varepsilon)$ to conclude $$\sqrt n\sum_{i=1}^{n}\mathbb{E}_{i-1} [|\Delta_{i}^nJ|^2{\bf 1}_{\{|\Delta_i^nJ|\leq3\varepsilon\}}|\Delta_i^nM|({\bf 1}_{\{|\Delta_i^nW|>{}\varepsilon\}} + {\bf 1}_{\{|\mathcal{E}_i|>{}\varepsilon\}})] =o_P(1).$$
 When $|\Delta_i^nW|\leq{}\varepsilon$ and $|\mathcal{E}_i|\leq{}\varepsilon$, then  $|x_i+\mathcal{E}_{i}|\leq\varepsilon$ implies $|\Delta_i^nJ|\leq{}3\varepsilon$ and we proceed as follows:
\begin{align*}
	\sqrt{n}\sum_{i=1}^{n}&\mathbb{E}_{i-1}\left[|\Delta_{i}^nJ|^2{\bf 1}_{\{|\Delta_i^nJ|\leq3\varepsilon\}}|\Delta_i^nM|\right]\\
	&\leq{}\sqrt{n}\sum_{i=1}^{n}\mathbb{E}_{i-1}\left[|\Delta_{i}^nJ|^{4}{\bf 1}_{\{|\Delta_i^nJ|\leq3\varepsilon\}}\right]^{\frac{1}{2}}\mathbb{E}_{i-1}\left[|\Delta_i^nM|^2\right]^{\frac{1}{2}}\\
	&\leq \sqrt{n}\Big(\sum_{i=1}^{n}\mathbb{E}_{i-1}\left[|\Delta_{i}^nJ|^{4}{\bf 1}_{\{|\Delta_i^nJ|\leq3\varepsilon\}}\right]\Big)^{\frac{1}{2}}\Big(\sum_{i=1}^{n}\mathbb{E}_{i-1}\left[|{\Delta_i^nM}|^2\right]\Big)^{\frac{1}{2}}\\
	&\leq C \sqrt{n}(n h\varepsilon^{4-Y})^{\frac{1}{2}}O_P(1)=C(h^{-1}\varepsilon^{4-Y})^{\frac{1}{2}}O_P(1)=o_P(1),
\end{align*}
since, by assumption, $\varepsilon\ll h^{1/(4-Y)}$. The cross-product when expanding $\hat{x}_{i}^2$ can be analyzed as follows (again, assuming that $|\Delta_iJ|\leq{}3\varepsilon$):
\begin{align*}
	&\sqrt{n}\sum_{i=1}^{n}\mathbb{E}_{i-1}\left[|\Delta_{i}^nJ||\Delta_i^nW|{\bf 1}_{\{|\Delta_i^nJ|\leq3\varepsilon\}}|\Delta_i^nM|\right]\\
	&\quad\leq{}\sqrt{n}\sum_{i=1}^{n}\mathbb{E}_{i-1}\left[|\Delta_{i}^nJ|^{2}{\bf 1}_{\{|\Delta_i^nJ|\leq3\varepsilon\}}\right]^{\frac{1}{2}}\mathbb{E}_{i-1}\left[|\Delta_{i}^nW|^{2}\right]^{\frac{1}{2}}\mathbb{E}_{i-1}\left[|\Delta_i^nM|^2\right]^{\frac{1}{2}}\\
	&\quad\leq \Big(\sum_{i=1}^{n}\mathbb{E}_{i-1}\left[|\Delta_{i}^nJ|^{2}{\bf 1}_{\{|\Delta_i^nJ|\leq3\varepsilon\}}\right]\Big)^{\frac{1}{2}}\Big(\sum_{i=1}^{n}\mathbb{E}_{i-1}\left[|{\Delta_i^nM}|^2\right]\Big)^{\frac{1}{2}}\\
	&\quad \leq C (nh\varepsilon^{2-Y})^{\frac{1}{2}}O_P(1)=o_P(1).
\end{align*}
Finally, we consider the case when $\hat{x}_i$ is replaced with $\sigma_{t_{i-1}}\Delta_i^nW$. Note that $\Delta_i^nW^2=2\int_{t_{i-1}}(W_s-W_{t_{i-1}}) dW_s+h$ and, since $M$ is orthogonal to $W$, $\mathbb{E}_{i-1}\left[(\Delta_{i}^nW^2-h)\Delta_i^nM\right]=0$, which in turn implies that $\mathbb{E}_{i-1}\left[\Delta_{i}^nW^2\Delta_i^nM\right]=0$. Thus, for any $p,q>0$ such that $\frac{1}{p}+\frac{1}{q}=\frac{1}{2}$, we have
\begin{align*}
	&\sqrt{n}\Big|\sum_{i=1}^{n}\sigma_{t_{i-1}}^2\mathbb{E}_{i-1}\left[\Delta_{i}^nW^2{\bf 1}_{\{|x_i+\mathcal{E}_{i}|\leq\varepsilon\}}\Delta_i^nM\right]\Big|\\
	&\quad=	\sqrt{n}\Big|\sum_{i=1}^{n}\sigma_{t_{i-1}}^2\mathbb{E}_{i-1}\left[\Delta_{i}^nW^2{\bf 1}_{\{|x_i+\mathcal{E}_{i}|>\varepsilon\}}\Delta_i^n M\right]\Big|\\
		&\quad\leq	\frac{C\sqrt{n}}{\varepsilon} \sum_{i=1}^{n}\mathbb{E}_{i-1}\left[|\Delta_{i}^nW|^2|x_i+\mathcal{E}_{i}||\Delta_i^nM|\right]\\
		&\quad\leq	\frac{C\sqrt{n}}{\varepsilon} \sum_{i=1}^{n}\mathbb{E}_{i-1}\left[|\Delta_{i}^nW)|^{2p}\right]^{\frac{1}{p}}\mathbb{E}_{i-1}\left[|x_i+\mathcal{E}_{i}|^q\right]^{\frac{1}{q}}\mathbb{E}_{i-1}\left[|\Delta_i^nM|^{2}\right]^{\frac{1}{2}}\\
		&\quad\leq	\frac{C\sqrt{n}}{\varepsilon} \left(\sum_{i=1}^{n}\mathbb{E}_{i-1}\left[|\Delta_{i}^nW|^{2p}\right]^{\frac{2}{p}}\mathbb{E}_{i-1}\left[|x_i+\mathcal{E}_{i}|^q\right]^{\frac{2}{q}}\right)^{\frac{1}{2}}\\
		& \qquad \qquad \qquad \qquad \qquad \qquad \qquad \times \left(\sum_{i=1}^{n}\mathbb{E}_{i-1}\left[|\Delta_i^nM|^{2}\right]\right)^{\frac{1}{2}}\\
		&\quad\leq	\frac{C\sqrt{n}}{\varepsilon} \left(nh^2h^{\frac{2}{q}}\right)^{\frac{1}{2}}O_P(1)
				\leq C	\frac{h^{\frac{1}{q}}}{\varepsilon}O_P(1)\ll 
				Ch^{\frac{1}{2}-\beta-s}\varepsilon^{-1}=o_P(1),
\end{align*}
for all small $s>0$ by taking $q$ close enough to $2$ and $p$ large enough. This concludes the proof. 
\end{proof}


\begin{lemma}\label{ErroDiscrHigherOrderb}

Suppose that 
${h^{\beta}}\ll \varepsilon\ll h^{\frac{1}{4-Y}}$ {with $\frac{1}{4-Y}<\beta< \frac{1}{2}\wedge \frac{2}{3 Y}$,} and let
\[
	\tilde\xi_i^n := \varepsilon^{\frac{Y-4}{2}}\left(\big(\Delta_{i}^{n}X\big)^{2}{\bf 1}_{\{\varepsilon<|\Delta_{i}^{n}X|\leq\zeta\varepsilon\}}-\bE_{i-1}\big(\big(\Delta_{i}^{n}X\big)^{2}\,{\bf 1}_{\{\varepsilon<|\Delta_{i}^{n}X|\leq\zeta\varepsilon\}}\big)\right).
\] 
Then, as $n\to\infty$,
\begin{align}\label{NdTchCndbb}
	\sum_{i=1}^{n}\mathbb{E}_{i-1}\left[\tilde\xi_{i}^n({\Delta_i^nM})\right]\,\stackrel{\mathbb{P}}{\longrightarrow}\,0,
\end{align}
when $M=W$ or $M$ is a square-integrable martingale orthogonal to $W$.
\end{lemma} 
\begin{proof}
%
Recall the notation \eqref{NtnEXs00}-\eqref{NtnEXs}. 
Fix $f(x)={\bf 1}_{\{\varepsilon<|x|\leq{}\zeta\varepsilon\}}$. We need to prove that 
\begin{align}\label{NdTchCndb}
	\varepsilon^{\frac{Y-4}{2}}\sum_{i=1}^{n}\mathbb{E}_{i-1}\left[(\Delta_{i}^{n}X)^2f\big(\Delta_{i}^{n}X\big)({\Delta_i^nM})\right]=o_P(1).
\end{align}
Similar to the proof of Lemma 7 in \cite{SahaliaJacod2009}, there exists a $C^2$ smooth approximation $f_n$ of $f$ such that, for any $\eta>0$,  
\begin{align}\label{DfnfnandD}
	\left\{\begin{array}{l}
	{\bf 1}_{\{\varepsilon(1+\frac{2}{3}\varepsilon^\eta)<|x|<\zeta\varepsilon(1-\frac{2}{3}\varepsilon^\eta)\}}\leq{}f_n(x)\leq{\bf 1}_{\{\varepsilon(1+\frac{1}{3}\varepsilon^\eta)<|x|<\zeta\varepsilon(1-\frac{1}{3}\varepsilon^\eta)\}},\\
	|f_n'(x)|\leq\frac{C}{\varepsilon^{1+\eta}},\quad |f_n''(x)|\leq{}\frac{C}{\varepsilon^{2+2\eta}}.\end{array}\right.
\end{align}
To show \eqref{NdTchCndb}, first observe it suffices to prove the statement replacing $f$ with $f_n$. To see this, we need to show that $$\varepsilon^{\frac{Y-4}{2}}\sum_{i=1}^{n}\mathbb{E}_{i-1}\left[\Delta_{i}^{n}X^2(f-f_n)(\Delta_{i}^{n}X)({\Delta_i^nM})\right]\to0,$$
 which in turn follows from 
\begin{align*}
	 A_{n,1}&:=\varepsilon^{\frac{Y-4}{2}}\sum_{i=1}^{n}\mathbb{E}_{i-1}\left[\Delta_{i}^nX^2{\bf 1}_{\{\varepsilon\leq{}|\Delta_{i}^{n}X|\leq{}\varepsilon(1+\varepsilon^\eta)\}}|{\Delta_i^nM}|\right] =o_P(1),\\
	 A_{n,2}&:=\varepsilon^{\frac{Y-4}{2}}\sum_{i=1}^{n}\mathbb{E}_{i-1}\left[\Delta_{i}^nX^2{\bf 1}_{\{\zeta\varepsilon(1-\varepsilon^\eta)\leq |\Delta_{i}^{n}X|\leq \zeta \varepsilon\}}|{\Delta_i^nM}|\right]=o_P(1).
\end{align*}
Applying Cauchy's inequality and \eqref{CTNITPbb},
\begin{align*}
	|A_{n,1}|&\leq\varepsilon^{\frac{Y-4}{2}}\sum_{i=1}^{n}\mathbb{E}_{i-1}\left[\Delta_{i}^nX^4{\bf 1}_{\{\varepsilon\leq{}|\Delta_{i}^{n}X|\leq{}\varepsilon(1+\varepsilon^\eta)\}}\right]^{\frac{1}{2}}
	\mathbb{E}_{i-1}\left[|{\Delta_i^nM}|^2\right]^{\frac{1}{2}}\\
	&\leq{}C\varepsilon^{\frac{Y-4}{2}}\left(h\varepsilon^{4-Y}\varepsilon^{\eta}\right)^{\frac{1}{2}}\sum_{i=1}^{n}
	\mathbb{E}_{i-1}\left[|{\Delta_i^nM}|^2\right]^{\frac{1}{2}}=O_P(\varepsilon^{\frac{\eta}{2}}),
\end{align*}
since, by Cauchy's inequality,  
\begin{align}\label{CAIFMINK}
	\sum_{i=1}^{n}
	\mathbb{E}_{i-1}\left[|{\Delta_i^nM}|^2\right]^{\frac{1}{2}}\leq{}n^{\frac{1}{2}}\left(\sum_{i=1}^{n}
	\mathbb{E}_{i-1}\left[|{\Delta_i^nM}|^2\right]\right)^{\frac{1}{2}}=O_P( h_n^{-\frac{1}{2}}),
\end{align}
because $\mathbb{E}[\sum_{i=1}^{n}\mathbb{E}_{i-1}[({\Delta_i^nM})^2]]=\sum_{i=1}^{n}\mathbb{E}[({\Delta_i^nM})^2]=
\mathbb{E}[(M_{1}-M_{0})^2]<\infty$. The same argument shows $A_{n,2}=o_P(1)$. Thus it suffices to show 
\begin{align}\label{NdTchCndff}
	\varepsilon^{\frac{Y-4}{2}}\sum_{i=1}^{n}\mathbb{E}_{i-1}\left[(\Delta_{i}^{n}X)^2f_n\big(\Delta_{i}^{n}X\big)({\Delta_i^nM})\right]=o_P(1).
\end{align}
We recall the notation \eqref{NtnEXs00}-\eqref{NtnEXs}. We consider separately the three terms in the expansion $(\Delta_{i}^{n}X)^2=(x_i+\mathcal{E}_i)^2=x_i^2+2x_i\mathcal{E}_i+\mathcal{E}_i^2$. 
For the second and third terms, we use Cauchy's inequality. Indeed, by Lemma \ref{l:ldiscretiz_error_moment_bounds},
$$
	\varepsilon^{\frac{Y-4}{2}}\sum_{i=1}^{n}\mathbb{E}_{i-1}\left[\mathcal{E}_i^2 f_n\big(\Delta_{i}^{n}X\big)|{\Delta_i^nM}|\right]\leq
	\varepsilon^{\frac{Y-4}{2}}\sum_{i=1}^{n}\mathbb{E}_{i-1}[\mathcal{E}_i^4]^{\frac{1}{2}}\mathbb{E}_{i-1}[({\Delta_i^nM})^2]^{\frac{1}{2}}
	$$
	$$
	\leq{}C\varepsilon^{\frac{Y-4}{2}}h^{\frac{3}{2}}\sum_{i=1}^{n}\mathbb{E}_{i-1}[({\Delta_i^nM})^2]^{\frac{1}{2}}=o_P(1),
$$
since recall that $\sum_{i=1}^{n}\mathbb{E}_{i-1}[({\Delta_i^nM})^2]^{1/2}=O_P(h^{-1/2})$ and 
\[
	\varepsilon^{\frac{Y-4}{2}}h\ll 1\;\Longleftrightarrow\; h^{\frac{2}{4-Y}}\ll \varepsilon
	\;\Longleftarrow\; h^{1/2}\ll \varepsilon\;\text{ because }\; \frac{2}{4-Y}>\frac12.
\]
For the term $x_i\mathcal{E}_i$, note that when $|\mathcal{E}_i|>\varepsilon$, we can use the bound 
\[
	\bP_{i-1}(|\mathcal E_i|>\varepsilon)\leq \frac{\mathbb{E}_{i-1}[|\mathcal E_i|^{r}]}{\varepsilon^{r}}\leq K h^{1+r(\frac12-\beta)}.
\]
(which is valid for all large $r$ due to \eqref{MDEE00}  and  $ \varepsilon\gg h^\beta$) 
and H\"older's inequality to conclude 
\[
	\varepsilon^{\frac{Y-4}{2}}\sum_{i=1}^{n}\mathbb{E}_{i-1}\left[|\mathcal{E}_i x_i|{\bf 1}_{\{|x_i|\leq C\varepsilon\}}|{\Delta_i^nM}|{\bf 1}_{\{|\mathcal{E}_i|>\varepsilon\}}\right]=o_P(1).
\]
  When  $|\mathcal{E}_{i}|\leq\varepsilon$, observe $|\Delta_i^nX|\leq{}\zeta\varepsilon$ implies that $|x_{i}|=|\Delta_i^n X-\mathcal{E}_{i}|\leq(1+\zeta)\varepsilon=:C\varepsilon$. Next, by \eqref{EstNIKbb}, \eqref{MDEE00}, and \eqref{CAIFMINK},
\begin{align*}
	\varepsilon^{\frac{Y-4}{2}}\sum_{i=1}^{n}\mathbb{E}_{i-1}&\left[|\mathcal{E}_i x_i|{\bf 1}_{\{|x_i|\leq C\varepsilon\}}|{\Delta_i^nM}|\right]\\
		&\leq
	\varepsilon^{\frac{Y-4}{2}}\sum_{i=1}^{n}\mathbb{E}_{i-1}[\mathcal{E}_i^4]^{\frac{1}{4}}\mathbb{E}_{i-1}[x_i^4{\bf 1}_{\{|x_i|\leq C\varepsilon\}}]^{\frac{1}{4}}\mathbb{E}_{i-1}[({\Delta_i^nM})^2]^{\frac{1}{2}}\\
	&\leq{}C\varepsilon^{\frac{Y-4}{2}}h^{\frac{3}{4}}h^{\frac{1}{2}}\sum_{i=1}^{n}\mathbb{E}_{i-1}[({\Delta_i^nM})^2]^{\frac{1}{2}}=o_P(1),
\end{align*}
since $h^\beta \ll \varepsilon$
implies that $\varepsilon^{(Y-4)/2}h^{3/4}\ll h^{3/4-\beta (4-Y)/2}\ll1$.

Therefore, we only need to show that 
\begin{align}\label{SLTF}
	\varepsilon^{\frac{Y-4}{2}}\sum_{i=1}^{n}\mathbb{E}_{i-1}\left[x_i^2f_n(\Delta_i^nX)({\Delta_i^nM})\right]=o_P(1).
\end{align}
When $x_i$ is replaced with $b_{t_{i-1}}h$ in \eqref{SLTF}, we have:
\begin{align}\label{IKICD}
\varepsilon^{\frac{Y-4}{2}}&\Big|\sum_{i=1}^{n}b^2_{t_{i-1}}h^2\mathbb{E}_{i-1}\left[f_n(\Delta_i^nX)({\Delta_i^nM})\right]\Big|\\
& \leq C h^2
\varepsilon^{\frac{Y-4}{2}}\sum_{i=1}^{n}\mathbb{E}_{i-1}\left[|{\Delta_i^nM}|^2\right]^{\frac{1}{2}}=O_P(h^{\frac{3}{2}}\varepsilon^{\frac{Y-4}{2}})=o_P(1),\nonumber
\end{align}
which follows from the assumption $h^\beta \ll \varepsilon$, since $\frac{3}{4-Y}>\frac{1}2>\beta$. 
We also have
\begin{align*}
	\varepsilon^{\frac{Y-4}{2}}\Big|\sum_{i=1}^{n}b_{t_{i-1}}h&\mathbb{E}_{i-1}\left[\hat{x}_{i}f_n(\Delta_i^nX)({\Delta_i^nM})\right]\Big|\\
	&\leq C h
\varepsilon^{\frac{Y-4}{2}}\sum_{i=1}^{n}\mathbb{E}_{i-1}\left[\hat{x}_i^2\right]^{\frac{1}{2}}\mathbb{E}_{i-1}\left[|{\Delta_i^nM}|^2\right]^{\frac{1}{2}}\\
&\leq C h\varepsilon^{\frac{Y-4}{2}}h^{\frac{1}{2}}\sum_{i=1}^{n}\mathbb{E}_{i-1}\left[|{\Delta_i^nM}|^2\right]^{\frac{1}{2}}=O_P(h\varepsilon^{\frac{Y-4}{2}}),
\end{align*}
whose convergence to $0$ follows from $h^{\beta}\ll \varepsilon$, since  $\frac{2}{4-Y}>\frac{1}2>\beta$. So, it suffices to show 
\begin{equation}\label{SLTFbb2}
	\varepsilon^{\frac{Y-4}{2}}\sum_{i=1}^{n}\mathbb{E}_{i-1}\left[\hat{x}_i^2f_n(\Delta_iX)({\Delta_i^nM})\right]=o_P(1),
\end{equation}
 where recall that $\hat{x}_{i}=\sigma_{t_{i-1}}\Delta_i^n W+\chi_{t_{i-1}}\Delta_i^n J$.
We first prove it when $M=W$. 
When we replace $\hat{x}_i$  with $\sigma_{t_{i-1}}\Delta_{i}^n W$ in \eqref{SLTFbb}, for any $q,p>0$ such that $1/q+1/p=1$, we have
\begin{align*}
	\varepsilon^{\frac{Y-4}{2}}\Big|\sum_{i=1}^{n}&\sigma_{t_{i-1}}^2\mathbb{E}_{i-1}\left[\Delta_{i}^nW^2 f_n(\Delta_i^nX)\Delta_i^nM\right]\Big|\\
	&\leq C\varepsilon^{\frac{Y-4}{2}}\sum_{i=1}^{n}\mathbb{E}_{i-1}\left[|\Delta_{i}^nW|^3{\bf 1}_{\{|\Delta_i^nX|>\varepsilon\}}\right]\Big|\\
	&\leq C	\varepsilon^{\frac{Y-4}{2}}\sum_{i=1}^{n}\mathbb{E}_{i-1}\left[\Delta_{i}^nW^{3p}\right]^{\frac{1}{p}}\mathbb{P}_{i-1}\left[|\Delta_i^n X|>\varepsilon\right]^{\frac{1}{q}}\\
	&\leq C\varepsilon^{\frac{Y-4}{2}}n h^{\frac{3}{2}}(h\varepsilon^{-Y})^{\frac{1}{q}},
\end{align*}
where we applied Lemma \ref{lemma:|J|>eps}  in the third line.
We can make the above bound to converge to $0$ by taking $q$ close to $1$. Indeed, 
\begin{align*}
	&\varepsilon^{\frac{Y-4}{2}}h^{\frac{1}{2}}(h\varepsilon^{-Y})^{\frac{1}{q}}\ll 1\;\Longleftarrow\; 
	h^{\frac{3}{2}}\varepsilon^{\frac{-4-Y}{2}}\ll 1
	\;\Longleftarrow\; 
	h^{\frac{3}{4+Y}}\ll \varepsilon
	\;\Longleftarrow\; h^{1/2}\ll \varepsilon
\end{align*}
since $\frac{3}{4+Y}>\frac12$. When $\hat{x}_i$  is replaced with $\chi_{t_{i-1}}\Delta_{i}^n J$ in \eqref{SLTFbb}, using boundedness of $\chi$ we need only to show that
\begin{equation}\label{AnySrs}
	\varepsilon^{\frac{Y-4}{2}}\sum_{i=1}^{n}\chi_{t_{i-1}}^2\mathbb{E}_{i-1}\left[\Delta_{i}^nJ^2f_n(\Delta_i^nX)|\Delta_i^nW|\right]=o_P(1),
\end{equation}
When $|\Delta_i^nW|>{}\varepsilon$ or $|\mathcal{E}_i|>{}\varepsilon$, we can use H\"older's inequality, 
 and that  $\bP(|\mathcal E_i|>\varepsilon)$ and $\bP(|\Delta^n_iW|>\varepsilon)$ both decay faster than any power of $h$ (a consequence of $ h^{\frac{1}{2}}/\varepsilon\ll h^{\frac{1}{2}-\beta}$) to conclude that 
 $$\varepsilon^{\frac{Y-4}{2}}\sum_{i=1}^{n}\mathbb{E}_{i-1}[\Delta_{i}^nJ^2f_n(\Delta_i^nX)|\Delta_i^nW|({\bf 1}_{\{|\Delta_i^nW|>{}\varepsilon\}} + {\bf 1}_{\{|\mathcal{E}_i|>{}\varepsilon\}})] =o_P(1).$$
  So,  together $|\Delta_i^nW|\leq{}\varepsilon$, $|\mathcal{E}_i|\leq{}\varepsilon$, and $|\Delta_i X|=|x_i+\mathcal{E}_{i}|\leq\zeta \varepsilon$ imply that $|\Delta_i^{n}J|\leq{}C\varepsilon$, for some positive constant $C$ that, for simplicity, we assume  is $2$. Then, for  \eqref{AnySrs} to hold, it suffices that 
  $$\varepsilon^{\frac{Y-4}{2}}\sum_{i=1}^{n}\mathbb{E}_{i-1}\left[\Delta_{i}^nJ^2{\bf 1}_{\{|\Delta_i^nJ|\leq2\varepsilon\}}|\Delta_i^nW|\right]\,\stackrel{\mathbb{P}}{\longrightarrow}\,0.$$
To show the latter, by H\"older's inequality and Lemma \ref{lemma:J2k}, for any $p,q>1$ such that $1/p+1/q=1$,
\begin{align}\label{SLTFdd}
	\varepsilon^{\frac{Y-4}{2}}\sum_{i=1}^{n}\mathbb{E}_{i-1}&\left[\Delta_{i}^nJ^2{\bf 1}_{\{|\Delta_i^nJ|\leq2\varepsilon\}}|\Delta_i^nW|\right]\\
	\nonumber
	&\leq{}\varepsilon^{\frac{Y-4}{2}}\sum_{i=1}^{n}\mathbb{E}_{i-1}\left[\Delta_{i}^nJ^{2p}{\bf 1}_{\{|\Delta_i^nJ|\leq 2\varepsilon\}}\right]^{\frac{1}{p}}\mathbb{E}_{i-1}\left[|\Delta_i^nW|^q\right]^{\frac{1}{q}}\\
	\nonumber
	&\leq C\varepsilon^{\frac{Y-4}{2}}\sum_{i=1}^{n}\left(h\varepsilon^{2p-Y}\right)^{\frac{1}{p}}h^{\frac{1}{2}}\leq{}C \varepsilon^{\frac{Y-4}{2}}h^{\frac{1}{p}-\frac{1}{2}}\varepsilon^{2-\frac{Y}{p}};
\end{align}
thus, making $p$ very close to $1$ (and $q\to\infty$), we get the convergence to $0$ since
$$
	\varepsilon^{\frac{-Y}{2}}h^{\frac{1}{2}}\ll 1\;\Longleftrightarrow\; h^{\frac{1}{Y}}\ll\varepsilon
	\;\Longleftarrow\; h^{\frac{1}{2}}\ll \varepsilon.
$$


We now show \eqref{SLTFbb2} for any bounded square integrable $M$ that is orthogonal to $W$.  The cross-product when expanding $\hat{x}_{i}^2$ can be analyzed as follows  (again, assuming that $|\Delta_iJ|\leq{}C\varepsilon$):
\begin{align*}
	&\varepsilon^{\frac{Y-4}{2}}\sum_{i=1}^{n}\mathbb{E}_{i-1}\left[|\Delta_{i}^nJ||\Delta_i^nW|{\bf 1}_{\{|\Delta_i^nJ|\leq C\varepsilon\}}|\Delta_i^nM|\right]\\
	&\quad\leq{}\varepsilon^{\frac{Y-4}{2}}\sum_{i=1}^{n}\mathbb{E}_{i-1}\left[|\Delta_{i}^nJ|^{2}{\bf 1}_{\{|\Delta_i^nJ|\leq C\varepsilon\}}\right]^{\frac{1}{2}}\mathbb{E}_{i-1}\left[|\Delta_{i}^nW|^{2}\right]^{\frac{1}{2}}\mathbb{E}_{i-1}\left[|\Delta_i^nM|^2\right]^{\frac{1}{2}}\\
	&\quad \leq C \varepsilon^{\frac{Y-4}{2}} (h\varepsilon^{2-Y})^{\frac{1}{2}}h^{\frac{1}{2}}\sum_{i=1}^{n}\mathbb{E}_{i-1}\left[|\Delta_i^nM|^2\right]^{\frac{1}{2}}\\
	&\quad=C h^{1/2}\varepsilon^{-1}
	O_P(1)=o_P(1),
\end{align*}
because $h^{1/2}\varepsilon^{-1}\ll 1$. 
Next, we consider the case when $\hat{x}_i$ is replaced with $\sigma_{t_{i-1}}\Delta_i^nW$. Since $f_n(x)\leq{}{\bf 1}_{\{|x|\geq{}\varepsilon\}}$, it suffices to show that 
\begin{align*}
	D_n:=\varepsilon^{\frac{Y-4}{2}}\sum_{i=1}^{n}\sigma_{t_{i-1}}^2\mathbb{E}_{i-1}\left[\Delta_{i}^nW^2 {\bf 1}_{\{|\Delta_i^nX|\geq{}\varepsilon\}}|\Delta_i^nM|\right]=o_P(1).
\end{align*}
As before, $|\Delta_i X|=|x_i+\mathcal{E}_{i}|\geq{}\varepsilon$ implies that $|\Delta_iW|>\varepsilon/3$, $|\mathcal{E}_i|>\varepsilon/3$, or $|\Delta_iJ|>\varepsilon/3$. The first two cases are straightforward. For the last one, we have, for any $p,q>0$ such that $\frac{1}{p}+\frac{1}{q}=\frac{1}{2}$,
\begin{align*}
	D_n
		&\leq C\varepsilon^{\frac{Y-4}{2}}\sum_{i=1}^{n}\mathbb{E}_{i-1}\left[(\Delta_{i}^nW)|^{2p}\right]^{\frac{1}{p}}\mathbb{P}_{i-1}\left[|\Delta_i^n J|\geq \varepsilon/3\right]^{\frac{1}{q}}\mathbb{E}_{i-1}\left[|\Delta_i^nM|^{2}\right]^{\frac{1}{2}}\\
		&\quad\leq
		C\varepsilon^{\frac{Y-4}{2}}h(h\varepsilon^{-Y})^{\frac{1}{q}}
		\sum_{i=1}^{n}\mathbb{E}_{i-1}\left[|\Delta_i^nM|^{2}\right]^{\frac{1}{2}}=C\varepsilon^{\frac{Y-4}{2}}h^{\frac{1}{2}}(h\varepsilon^{-Y})^{\frac{1}{q}}O_P(1).
\end{align*}
Using that $\varepsilon \gg h^\beta$, 
the last expression above is such that 
\[
	\varepsilon^{\frac{Y-4}{2}}h^{\frac{1}{2}}(h\varepsilon^{-Y})^{\frac{1}{q}}\ll h^{\frac{1}{2}+\frac{1}{q}-\beta\left( 2-\frac{Y}{2}+\frac{Y}{q}\right)}
\]	
As $q\downarrow 2$, the exponent ${\frac{1}{2}+\frac{1}{q}-\beta\left( 2-\frac{Y}{2}+\frac{Y}{q}\right)}\to 1-2\beta>0$, hence for $q>2$ close enough to 2 we have 
$D_n=o_P(1)$. 

Finally, we consider the case when $\hat{x}_i$ is replaced with $\chi_{t_{i-1}}\Delta_i^nJ$ in \eqref{SLTFbb2}. That is, we need to show that:
\begin{align}\label{SLTFdd}
	\varepsilon^{\frac{Y-4}{2}}\sum_{i=1}^{n}\mathbb{E}_{i-1}\left[\Delta_{i}^nJ^2f_n(\Delta_i^nX)(\Delta_i^nM)\right]=o_P(1).
\end{align}
As shown in \ref{MoreMoreAuxLm},
\begin{align}\label{SINHN1}
	\varepsilon^{\frac{Y-4}{2}}\sum_{i=1}^{n}\mathbb{E}_{i-1}\left[\Delta_{i}^nJ^2|f_n(\Delta_i^nX)-f_n(\chi_{t_{i-1}}\Delta_i^nJ)||\Delta_i^nM|\right]=o_P(1).
\end{align}
Therefore, it suffices to show 
\begin{align}\label{SLTFjj}
	\varepsilon^{\frac{Y-4}{2}}\sum_{i=1}^{n}\mathbb{E}_{i-1}\left[\Delta_{i}^nJ^2f_n(\chi_{t_{i-1}}\Delta_i^nJ)(\Delta_i^nM)\right]=o_P(1).
\end{align}
For simplicity in the remaining part of the argument we assume $\chi\equiv 1$; the general case can be obtained by a near-identical argument. 
Let us first introduce some notation. By redefining the process $b$ if necessary, we may assume that the L\'evy-It\^o decomposition of $J$ takes the form $$J_t=\int_0^t\int x{\bf 1}_{\{|x|\leq{}1\}}\bar{N}(ds,dx)+\int_0^t\int x{\bf 1}_{\{|x|>{}1\}}N(ds,dx),$$
 where $\bar{N}(ds,dx)=N(ds,dx)-\nu(dx)ds$ is the compensated jump measure of $J$.  We also recall $M$ admits the decomposition $$M_t=M'_t+\int_0^t\int\delta(s,x)\bar{N}(ds,dx),$$ where $M'$ is a martingale orthogonal to both $W$ and $N$, and $\delta$ is bounded such that $\mathbb{E}\left[\left.\int_{t}^u\int\delta^2(s,x)\nu(dx)ds\right|\mathcal{F}_{t}\right]\leq{}\mathbb{E}\left[\left.(M_{u}-M_{t})^2\right|\mathcal{F}_{t}\right]<\infty$ (see expression~(80) in \cite{SahaliaJacod2009} for details).  For future reference, set $N^{i}(ds,dx):=N(t_{i}+ds,dx)$, $\delta^{i}(s,x):=\delta(t_{i}+s,x)$, $p_n(x)=x^2f_n(x)$, and 
\[
	v_n(y,x)=p_n(x+y)-p_n(y),\quad 
	w_n(y,x)=p_n(x+y)-p_n(y)-xp_n'(y){\bf 1}_{\{|x|\leq{}1\}}.
\]
By It\^o's formula applied to $Y^{i-1}_u:=J_{t_{i-1}+u}-J_{t_{i-1}}$, we have
\begin{align*}
	p_n(Y^{i-1}_{u})&=\int_0^{u}p_n'(Y^{i-1}_{s^-})dY^{i-1}_s\\
	& \qquad+\int_0^u\int (p_n(Y^{i-1}_{s^-}+x)-p_n(Y^{i-1}_{s^-})-p_n'(Y^{i-1}_{s^-})x)N^{i-1}(ds,dx)\\
	&=\int_0^{u}\int_{|x|\leq{}1} p_n'(Y^{i-1}_{s^-})x\bar{N}^{i-1}(ds,dx)+\int_0^u\int w_n(Y^{i-1}_{s^-},x)N^{i-1}(ds,dx).
\end{align*}
With the notation $Z^{i-1}_u:=M_{t_{i-1}+u}-M_{t_{i-1}}$, It\^o's formula gives:
\begin{align}
\nonumber
	p_n(Y^{i-1}_{u})Z^{i-1}_u&=\int_0^{u}p_n(Y^{i-1}_{s^-})dZ^{i-1}_s\\
	\nonumber
	&\qquad +\int_0^{u}Z^{i-1}_{s^-}d(p_n(Y^{i-1}_{s}))+\sum_{s\leq{}u}\Delta (p_n(Y^{i-1}_s))\Delta Z^{i-1}_s\\
	\nonumber
	&={\rm Mrtg}+\int_0^{u}\int \delta^{i-1}(s,x)v_n(Y^{i-1}_{s},x)\nu(dx)ds\\
	& \quad+\int_0^u\int w_n(Y^{i-1}_{s},x)Z^{i-1}_{s}\nu(dx)ds.\label{DcmProdYZ}
\end{align}
Fixing $u=h_n$ in  \eqref{DcmProdYZ} and taking expectations, we get
\begin{align}
\nonumber
	\mathbb{E}_{i-1}[p_n(\Delta_i J)\Delta_i^n M] 
	&=\mathbb{E}_{i-1}\bigg[\int_0^{h_{n}}\!\!\!\!\int \delta^{i-1}(s,x)v_n(Y_{s},x)\nu(dx)ds\\
	&\qquad\qquad+\int_0^{h_n}\!\!\!\!\int w_n(Y_{s},x)Z_{s}\nu(dx)ds\bigg],\label{DcmProdYZ2}
\end{align}
where, for simplicity, we omitted  the superscript $i-1$ in $Y$ and $Z$. 
We will analyze each of the two terms in \eqref{DcmProdYZ2} separately. We start with the second term. We claim the following bound follows from \eqref{DfnfnandD}:
\begin{align}
	|w_n(y,x)|&\leq{}K{\bf 1}_{|y|\leq \varepsilon\text{ or }|y|\geq{}\zeta\varepsilon}{\bf 1}_{|x|>\frac{1}{3}\varepsilon^{1+\eta}}\varepsilon^2\label{e:wnbound}\\
	&\quad+K {\bf 1}_{\varepsilon(1+\varepsilon^\eta)\leq|y|\leq \zeta\varepsilon(1-\varepsilon^\eta)}{\bf 1}_{|x|\leq\frac{1}{3}\varepsilon^{1+\eta}}x^2\notag\\
	&\quad+K {\bf 1}_{\varepsilon(1+\varepsilon^\eta)\leq|y|\leq \zeta\varepsilon(1-\varepsilon^\eta)}{\bf 1}_{|x|\geq\frac{1}{3}\varepsilon^{1+\eta}}(2\varepsilon|x|{\bf 1}_{|x|\leq{}1}+\varepsilon^2\notag\\
	&\quad+K {\bf 1}_{\varepsilon\leq{}|y|\leq\varepsilon(1+\varepsilon^\eta)}(\varepsilon^{-2\eta}{\bf 1}_{|x|\leq\varepsilon^{1+\eta}}x^2+\varepsilon^{1-\eta}{\bf 1}_{|x|\geq\varepsilon^{1+\eta}}(|x|\wedge 1))\notag\\
	&\quad+K {\bf 1}_{\zeta\varepsilon(1-\varepsilon^\eta)\leq{}|y|\leq \zeta\varepsilon}(\varepsilon^{-2\eta}{\bf 1}_{|x|\leq\varepsilon^{1+\eta}}x^2+\varepsilon^{1-\eta}{\bf 1}_{|x|\geq\varepsilon^{1+\eta}}(|x|\wedge 1)).\notag
\end{align}
To see this, first note $|p_n(z)|\leq K(\varepsilon\wedge |z|)^2$ for all $z$.  For the first term, observe $|y|\leq \varepsilon$ or $|y|>\zeta\varepsilon\implies  p_n(y)=0$, giving $|w_n(y,x)|=p_n(x+y)=p_n(x+y){\mathbf 1}_{\{|x|>\frac{1}{3}\varepsilon^{1+\eta}\}} \leq \varepsilon^2{\mathbf 1}_{\{|x|>\frac{1}{3}\varepsilon^{1+\eta}\}}$; for the second term, when both $\varepsilon(1+\varepsilon^\eta)\leq|y|\leq \zeta\varepsilon(1-\varepsilon^\eta)$ and $|x|\leq\frac{1}{3}\varepsilon^{1+\eta}$, then $\varepsilon(1+\frac{2}{3}\varepsilon^\eta)\leq|x+y|\leq\zeta \varepsilon(1-\frac{2}{3}\varepsilon^\eta)$, implying $f_n(x+y)=1=f_n(y)$, so $w_n(y,x)= x^2$; for the third term, under the same restriction on $y$ again we have $f_n(y)=1$ and $p'_n(y)=2y$, and using $p_n(x+y)\leq K\varepsilon^2$ gives $|w_n(y,x)|\leq K(\varepsilon^2 +\varepsilon|x|{\bf 1}_{|x|\leq{}1})$; for the fourth term, when $|x|\leq \varepsilon^{1+\eta}$, a second order Taylor expansion and the restriction $|y|\leq \varepsilon(1+\varepsilon^\eta)$ gives $|w_n(y,x)|\leq K x^2\varepsilon^{-2\eta}$ and when $|x|>\varepsilon^{1+\eta}$, the mean value theorem gives $|w_n(y,x)|\leq 2 |x  p'_n(y)| {\bf 1}_{\{|x|\leq 1\}}  + p_n(y){\bf 1}_{\{|x|>1\}} \leq K \big(|x|\varepsilon^{1-\eta}{\bf 1}_{\{|x|\leq 1\}} + \varepsilon^2 {\bf 1}_{\{|x|>1\}}\big) \leq K\varepsilon^{1-\eta}(1\wedge |x|);$ the fifth term is similar to the fourth term.

Therefore, using that $\mathbb{E}_{i-1}(|Z_{s}|)\leq{}\mathbb{E}_{i-1}(|Z_{h_n}|)=\mathbb{E}_{i-1}(|\Delta_i^n M|)$ and that $M$ is bounded, %
\begin{align}
	\mathbb{E}_{i-1}\Big(&\int_0^{h_n}\!\!\!\!\int |w_n(Y_{s},x)||Z_{s}|\nu(dx)ds\Big)\nonumber\\
	&\;\leq K(\varepsilon^{2-Y(1+\eta)}\vee \varepsilon)\int_0^{h_n}\mathbb{E}_{i-1}(|Z_{s}|)ds \label{InqFFT0_0}\\
&\quad+ K\varepsilon^{2-Y(1+\eta)}h_n\int_0^{1}\mathbb{P}_{i-1}({\varepsilon\leq{}|Y_{uh_n}|\leq\varepsilon(1+\varepsilon^\eta)})du
\nonumber
\\
\nonumber
&\quad
+ K\varepsilon^{2-Y(1+\eta)}h_n\int_0^{1}\mathbb{P}_{i-1}({\zeta\varepsilon(1-\varepsilon^\eta)\leq{}|Y_{uh_n}|\leq \zeta\varepsilon})du\\
&\leq{}K(\varepsilon^{2-Y(1+\eta)}\vee \varepsilon)h_n\mathbb{E}_{i-1}(|\Delta_i^n M|)+
K\varepsilon^{2-Y(1+\eta)}h_n^2{\varepsilon^{-Y}}.
\label{InqFFT0}
\end{align}
(In \eqref{InqFFT0_0} the first term following the inequality bounds the first three terms in the corresponding bound \eqref{e:wnbound}.) To obtain the second bound in \eqref{InqFFT0}, we used Lemma \ref{lemma:J2k} and the following argument: 
\begin{align}
\nonumber
	\mathbb{P}(\varepsilon\leq{}|J_{h_n u}|\leq\varepsilon(1+\varepsilon^\eta))&\leq 
	\frac{1}{\varepsilon^2}\mathbb{E}\left[J_{h_n u}^2{\bf 1}_{\{\varepsilon\leq{}|J_{h_n u}|\leq\varepsilon(1+\varepsilon^\eta)\}}\right]\\
	&\leq
 C\frac{u h_n\varepsilon^{2-Y}}{\varepsilon^2}\leq{}Ch_n\varepsilon^{-Y}.\label{CInPON}
\end{align}
From expression \eqref{DcmProdYZ2} and the bound \eqref{InqFFT0}, we need to show that
\begin{align*}
	\varepsilon^{\frac{Y-4}{2}}(\varepsilon^{2-Y(1+\eta)}\vee \varepsilon)h_n\sum_{i=1}^{n}\mathbb{E}_{i-1}(|\Delta_i^n M|)+\varepsilon^{\frac{Y-4}{2}}\varepsilon^{2-Y(1+\eta)}h_n^2\varepsilon^{\eta-Y}h_n^{-1}=o_P(1).
\end{align*}
The first term converges to $0$ when $Y>1$ since $h_n^{\frac{1}{2}}\sum_{i=1}^{n}\mathbb{E}_{i-1}(|\Delta_i^n M|)=O_P(1)$ and $\varepsilon^{\frac{Y-4}{2}}\varepsilon^{2-Y(1+\eta)}h_n^{\frac{1}{2}}=\varepsilon^{-\frac{Y}{2}-Y\eta}h_n^{\frac{1}{2}}\ll 1$, if $\eta$ is small enough; when $Y\leq 1$, the first term is of order $\varepsilon^{\frac{Y-2}{2}}h^{\frac12}\ll \varepsilon^{-1}h^{\frac12}\ll1$.  The order of the second term can be written as $\varepsilon^{\frac{Y-4}{2}}\varepsilon^{2-Y(1+\eta)}h_n^2\varepsilon^{\eta-Y}h_n^{-1}=\varepsilon^{-\frac{3Y}{2}+(1-Y)\eta}h_n$ which converges to 0 for all small $\eta$,  since by assumption $\beta <\frac2{3Y}$.

It remains to bound the first term of \eqref{DcmProdYZ2}. In other words, we need to show that 
\begin{align}
	F_n:=\varepsilon^{\frac{Y-4}{2}}\sum_{i=1}^{n}\mathbb{E}_{i-1}\left[\int_0^{h_{n}}\!\!\!\!\int \delta^{i-1}(s,x)v_n(Y_{s},x)\nu(dx)ds\right]=o_P(1).\label{DcmProdYZ2pp}
\end{align}
To show the above, we separately consider cases where $|Y_s|\leq \varepsilon$, $|Y_s|>\zeta \varepsilon$, $\varepsilon(1+\varepsilon^{\eta})<|Y_s|<\zeta\varepsilon(1-\varepsilon^{\eta})$, $\varepsilon\leq{}|Y_s|\leq{}\varepsilon(1+\varepsilon^{\eta})$ and $\zeta\varepsilon\leq{}|Y_s|\leq{}\zeta\varepsilon(1+\varepsilon^{\eta})$.

So, suppose first that $|Y_s|\leq{}\varepsilon$, so that $f_n(Y_s)=0$,  $v_n(Y_s,x)=(Y_s+x)^2 f_n(Y_s+x)$ and, in particular, $|x|\geq{}\frac{1}{3}\varepsilon^{1+\eta}$ (otherwise, $f_n(Y_s+x)=0$). Starting from the bound $|v_n(Y_s,x)|\leq Y_s^2 + 2|xY_s| +x^2$, we proceed to control each term corresponding to $Y_s^2$,  $2|xY_s|$, and  $x^2$, separately. 
For $Y_s^2$, since Lemma \ref{lemma:J2k} implies  $\bE_{i-1}[Y_s^{2k}{\bf 1}_{\{|Y_s|\leq{}\varepsilon\}}]\leq K h\varepsilon^{2k-Y}$ for all $0<s<h_n$, we have:
\begin{align*}
	&\varepsilon^{\frac{Y-4}{2}}\sum_{i=1}^{n}
	\mathbb{E}_{i-1}\left[\int_0^{h_{n}}\!\!\!\!\int |\delta^{i-1}(s,x)|Y_s^2{\bf 1}_{\{|Y_s|\leq{}\varepsilon\}}{\bf 1}_{\{|x|\geq{}\frac{1}{3}\varepsilon^{1+\eta}\}}\nu(dx)ds\right]\\
	&\quad\leq
	\varepsilon^{\frac{Y-4}{2}}\sum_{i=1}^{n}\mathbb{E}_{i-1}\left[\int_0^{h_{n}}\!\!\!\!\int Y_s^{4}{\bf 1}_{\{|Y_s|\leq{}\varepsilon\}}{\bf 1}_{\{|x|\geq{}\frac{1}{3}\varepsilon^{1+\eta}\}}\nu(dx)ds\right]^{\frac{1}{2}}\\
	&   \qquad\qquad\qquad\qquad\qquad\qquad\times\mathbb{E}_{i-1}\left[\int_0^{h_{n}}\!\!\!\!\int |\delta^{i-1}(s,x)|^2\nu(dx)ds\right]^{\frac{1}{2}}\\
&\quad\leq C
	\varepsilon^{\frac{Y-4}{2}}\left(\varepsilon^{-Y(1+\eta)}h^2\varepsilon^{4-Y}\right)^{\frac{1}{2}}\sum_{i=1}^{n}
	\mathbb{E}_{i-1}\left[\int_{t_{i-1}}^{t_{i}}\int |\delta(u,x)|^2\nu(dx)du\right]^{\frac{1}{2}}\\
	&\quad \leq C h\varepsilon^{-\frac{Y}{2}(1+\eta)}
	\sum_{i=1}^{n}
	\left[\mathbb{E}_{i-1}(\Delta_i^n M)^2\right]^{\frac{1}{2}}= h^{\frac{1}{2}}\varepsilon^{-\frac{Y}{2}(1+\eta)}O_P(1),
\end{align*}
which, due to our assumption $\varepsilon\gg h^{\beta}$, converges to $0$ if $\eta$ is small enough, because $h^{\frac{1}{2}}\varepsilon^{-\frac{Y}{2}}\ll 1 \Longleftarrow h^{\frac1Y}\ll  \varepsilon\Longleftarrow h^{\frac12}\ll \varepsilon$. 
For $2Y_s x$,
\begin{align*}
	&\varepsilon^{\frac{Y-4}{2}}\sum_{i=1}^{n}
	\mathbb{E}_{i-1}\left[\int_0^{h_{n}}\!\!\!\!\int |\delta^{i-1}(s,x)||Y_s x|{\bf 1}_{\{|Y_s|\leq{}\varepsilon\}}{\bf 1}_{\{|x|\geq{}\frac{1}{3}\varepsilon^{1+\eta}\}}\nu(dx)ds\right]\\
	&\quad\leq
	\varepsilon^{\frac{Y-4}{2}}\sum_{i=1}^{n}
	\mathbb{E}_{i-1}\left[\int_0^{h_{n}}\!\!\!\!\int |\delta^{i-1}(s,x)|^2\nu(dx)ds\right]^{\frac{1}{2}}\\
	&\qquad\qquad\qquad\qquad\qquad\times \mathbb{E}_{i-1}\left[\int_0^{h_{n}}\!\!\!\!\int Y_s^{2}x^2{\bf 1}_{\{|Y_s|\leq{}\varepsilon\}}{\bf 1}_{\{|x|\geq{}\frac{1}{3}\varepsilon^{1+\eta}\}}\nu(dx)ds\right]^{\frac{1}{2}}\\
&\quad\leq C
	\varepsilon^{\frac{Y-4}{2}}\left(h^2\varepsilon^{2-Y}\right)^{\frac{1}{2}}\sum_{i=1}^{n}
	\mathbb{E}_{i-1}\left[\int_{t_{i-1}}^{t_{i}}\int |\delta(u,x)|^2\nu(dx)du\right]^{\frac{1}{2}}\\
	&\quad \leq C h\varepsilon^{-1}
	\sum_{i=1}^{n}
	\left[\mathbb{E}_{i-1}(\Delta_i^n M)^2\right]^{\frac{1}{2}}= h^{\frac{1}{2}}\varepsilon^{-1}O_P(1),
\end{align*}
which again is $o_P(1)$ since, $\varepsilon \gg h^{\frac12}$. 
The proof for the case $x^2$ is more involved. First note that when $f_n(Y_s+x)>0$,  $|Y_s+x|\leq{}\zeta\varepsilon$,  and together with our assumption $|Y_s|\leq{}\varepsilon$ this gives $|x|\leq{}(1+\zeta)\varepsilon=:C\varepsilon$. Then, for any sequence $\alpha_n>0$ such that $\alpha_n\to{}0$, 
\begin{align*}
	&\varepsilon^{\frac{Y-4}{2}}\sum_{i=1}^{n}
	\mathbb{E}_{i-1}\left[\int_0^{h_{n}}\!\!\!\!\int_{|\delta^{i-1}(s,x)|\leq{}\alpha_n} |\delta^{i-1}(s,x)|x^2{\bf 1}_{\{|x|\leq{}C\varepsilon\}}\nu(dx)ds\right]\\
	&\quad\leq
	\varepsilon^{\frac{Y-4}{2}}\sum_{i=1}^{n}\mathbb{E}_{i-1}\left[\int_0^{h_{n}}\!\!\!\!\int_{|x|\leq{}C\varepsilon} x^4\nu(dx)ds\right]^{\frac{1}{2}}\\
	&\qquad\qquad\qquad\qquad\qquad \times \mathbb{E}_{i-1}\left[\int_0^{h_{n}}\!\!\!\!\int_{|\delta^{i-1}(s,x)|\leq{}\alpha_n} |\delta^{i-1}(s,x)|^2\nu(dx)ds\right]^{\frac{1}{2}}\\
	&\quad \leq C\varepsilon^{\frac{Y-4}{2}}\left(h\varepsilon^{4-Y}\right)^{\frac{1}{2}}h^{-\frac{1}{2}}\sum_{i=1}^{n}
	\mathbb{E}_{i-1}\left[\int_{t_{i-1}}^{t_{i}}\int_{|\delta(u,x)|\leq{}\alpha_n} |\delta(u,x)|^2\nu(dx)ds\right]\\
	&=o_P(1),
\end{align*}
since 
$$
\mathbb{E}\Big[\sum_{i=1}^{n}
	\mathbb{E}_{i-1}\Big[\int_{t_{i-1}}^{t_{i}}\int_{|\delta(u,x)|\leq{}\alpha_n} |\delta(u,x)|^2\nu(dx)ds\Big]\Big]\\
$$
$$
=\mathbb{E}\Big[\int_0^{T}\int_{|\delta(u,x)|\leq{}\alpha_n} |\delta(u,x)|^2\nu(dx)ds\Big]\to{}0.
$$
On the other hand, 
\begin{align*}
	&\varepsilon^{\frac{Y-4}{2}}\sum_{i=1}^{n}
	\mathbb{E}_{i-1}\left[\int_0^{h_{n}}\!\!\!\!\int_{|\delta^{i-1}(s,x)|>{}\alpha_n} |\delta^{i-1}(s,x)|x^2{\bf 1}_{\{|x|\leq{}C\varepsilon\}}\nu(dx)ds\right]\\
	&\quad\leq C\frac{1}{\alpha_n}
	\varepsilon^{\frac{Y-4}{2}}\varepsilon^2\sum_{i=1}^{n}
	\mathbb{E}_{i-1}\left[\int_0^{h_{n}}\!\!\!\!\int_{|\delta^{i-1}(s,x)|>{}\alpha_n} |\delta^{i-1}(s,x)|^2\nu(dx)ds\right]\\
	&\quad \leq C\frac{\varepsilon^{\frac{Y}{2}}}{\alpha_n}\sum_{i=1}^{n}
	\mathbb{E}_{i-1}\left[\int_{t_{i-1}}^{t_{i}}\int|\delta(u,x)|^2\nu(dx)ds\right]=\frac{\varepsilon^{\frac{Y}{2}}}{\alpha_n}O_P(1).
\end{align*}
Thus, it suffices to take $\alpha_n\to{}0$ such that $\varepsilon^{Y/2}/\alpha_n\to0$ (e.g., $\alpha_n=\varepsilon^{Y/4}$). This finishes the bound for $F_n$ for the case when $|Y_s|\leq \varepsilon$.

For the case of $|Y_{s}|\geq{}\zeta\varepsilon$, we have $f_n(Y_s)=0$ and $v_n(Y_s,x)=(Y_s+x)^2 f_n(Y_s+x)\leq{}(\zeta \varepsilon)^2$. Moreover, whenever $f_n(Y_s+x)> 0$, we have $|Y_s+x|\leq \zeta\varepsilon(1-\frac{1}{3}\varepsilon^\eta)$, implying $|x|>\zeta \varepsilon^{1+\eta}/3$.   Thus, 
\begin{align}\label{WTUxL}
	&\qquad\varepsilon^{\frac{Y-4}{2}}\sum_{i=1}^{n}
	\mathbb{E}_{i-1}\bigg[\int_0^{h_{n}}\!\!\!\!\int |\delta^{i-1}(s,x)|(Y_s+x)^2\\
	\nonumber
	& \qquad \qquad \qquad \qquad  \times f_{n}(Y_s+x){\bf 1}_{\{|Y_s|\geq{}\zeta\varepsilon\}}{\bf 1}_{\{|x|\geq{}C\varepsilon^{1+\eta}\}}\nu(dx)ds\bigg]\\
	\nonumber
	&\quad\leq C
	\varepsilon^{\frac{Y-4}{2}}\varepsilon^2\sum_{i=1}^{n}
	\mathbb{E}_{i-1}\bigg[\int_0^{h_{n}}\!\!\!\!\int |\delta^{i-1}(s,x)|^2\nu(dx)ds\bigg]^{\frac{1}{2}}\\
	\nonumber
	& \qquad\qquad\qquad\qquad\qquad\times \mathbb{E}_{i-1}\left[\int_0^{h_{n}}\!\!\!\!\int {\bf 1}_{\{|Y_s|\geq{}\zeta\varepsilon\}}{\bf 1}_{\{|x|\geq{}C\varepsilon^{1+\eta}\}}\nu(dx)ds\right]^{\frac{1}{2}}\\
	\nonumber
&\quad\leq C
	\varepsilon^{\frac{Y-4}{2}}\varepsilon^2\left(h^2\varepsilon^{-Y-{(1+\eta) Y}}\right)^{\frac{1}{2}}\sum_{i=1}^{n}
	\mathbb{E}_{i-1}\left[\int_{t_{i-1}}^{t_{i}}\int |\delta(u,x)|^2\nu(dx)du\right]^{\frac{1}{2}}\\
	\nonumber
	&\quad \leq C h\varepsilon^{-{\frac{(1+\eta) Y}{2}}}
	\sum_{i=1}^{n}
	\left[\mathbb{E}_{i-1}(\Delta_i^n M)^2\right]^{\frac{1}{2}}= h^{\frac{1}{2}}\varepsilon^{-{\frac{(1+\eta) Y}{2}}}O_P(1)=o_P(1),
\end{align}
for all small $\eta>0$ due to our assumption $\varepsilon\gg h^{\beta}$, which completes the bound for $F_n$ in the case $|Y_{s}|\geq{}\zeta\varepsilon$.

We now consider the case when $\varepsilon(1+\varepsilon^{\eta})<|Y_s|<\zeta\varepsilon(1-\varepsilon^{\eta})$.  When $|x|\geq{}\frac{1}{3}\varepsilon^{1+\eta}$, using the bound $|v_n(Y_s,x)|\leq C\varepsilon^2$ we can follow a similar proof as in \eqref{WTUxL}. If $|x|\leq{}\frac{1}{3}\varepsilon^{1+\eta}$, then $f_n(x+Y_s)=1=f(Y_s)$, implying $v_n(Y_s,x)=x^2+2xY_s$; we bound each term separately.  For the case of $x^2$,
\begin{align*}
	&\varepsilon^{\frac{Y-4}{2}}\sum_{i=1}^{n}
	\mathbb{E}_{i-1}\bigg[\int_0^{h_{n}}\!\!\!\!\int|\delta^{i-1}(s,x)|x^2{\bf 1}_{\{\varepsilon(1+\varepsilon^{\eta})<|Y_s|<\zeta\varepsilon(1-\varepsilon^{\eta})\}}{\bf 1}_{\{|x|\leq{}C\varepsilon^{1+\eta}\}}\nu(dx)ds\bigg]\\
	&\quad\leq
	\varepsilon^{\frac{Y-4}{2}}\sum_{i=1}^{n}
	\mathbb{E}_{i-1}\left[\int_0^{h_{n}}\!\!\!\!\int |\delta^{i-1}(s,x)|^2\nu(dx)ds\right]^{\frac{1}{2}}\\
	& \qquad\qquad\qquad\qquad\quad \times \mathbb{E}_{i-1}\left[\int_0^{h_{n}}\!\!\!\!\int_{|x|\leq{}C\varepsilon^{1+\eta}} x^4{\bf 1}_{\{|Y_s|\geq{}\varepsilon\}}\nu(dx)ds\right]^{\frac{1}{2}}\\
	&\quad \leq C\varepsilon^{\frac{Y-4}{2}}\left(\varepsilon^{(4-Y)(1+\eta)}h^2\varepsilon^{-Y}\right)^{\frac{1}{2}}\sum_{i=1}^{n}
	\left[{\mathbb{E}_{i-1}(\Delta_i^n M)^2}\right]^{\frac{1}{2}}\\
	&\quad=\varepsilon^{\frac{\eta(4-Y)-Y}{2}}h^{\frac{1}{2}}O_P(1)=o_P(1),
\end{align*}
for all $\eta>0$. For the case of $2|xY_s|$, since $|Y_s|<\zeta\varepsilon(1-\varepsilon^{\eta})$, for some $C$,
\begin{align*}
	&\varepsilon^{\frac{Y-4}{2}}\sum_{i=1}^{n}
	\mathbb{E}_{i-1}\bigg[\int_0^{h_{n}}\!\!\!\!\int|\delta^{i-1}(s,x)||xY_s|{\bf 1}_{\{\varepsilon(1+\varepsilon^{\eta})<|Y_s|<\zeta\varepsilon(1-\varepsilon^{\eta})\}} {\bf 1}_{\{|x|\leq{}C\varepsilon^{1+\eta}\}}\nu(dx)ds\bigg]\\
	&\quad\leq
	\varepsilon^{\frac{Y-4}{2}}\sum_{i=1}^{n}
	\mathbb{E}_{i-1}\bigg[\int_0^{h_{n}}\!\!\!\!\int |\delta^{i-1}(s,x)|^2\nu(dx)ds\bigg]^{\frac{1}{2}}\\
	& \qquad\qquad \qquad\qquad\qquad \quad\times \mathbb{E}_{i-1}\bigg[\int_0^{h_{n}}\!\!\!\!\int_{|x|\leq{}C\varepsilon^{1+\eta}} x^2Y_s^2{\bf 1}_{\{|Y_s|\leq{}\zeta \varepsilon\}}\nu(dx)ds\bigg]^{\frac{1}{2}}\\
	&\quad \leq C\varepsilon^{\frac{Y-4}{2}}\left(\varepsilon^{(2-Y)(1+\eta)}h^2\varepsilon^{2-Y}\right)^{\frac{1}{2}}\sum_{i=1}^{n}
	\left[{\mathbb{E}_{i-1}(\Delta_i^n M)^2}\right]^{\frac{1}{2}}\\
	&\quad=\varepsilon^{-\frac{Y}{2}+\frac{\eta(2-Y)}{2}}h^{\frac{1}{2}}O_P(1)=o_P(1),
\end{align*}
{again for all $\eta>0$. This completes the bound for $F_n$ for the case $\varepsilon(1+\varepsilon^{\eta})<|Y_s|<\zeta\varepsilon(1-\varepsilon^{\eta})$.}

The only cases left are when $\varepsilon\leq{}Y_s\leq{}\varepsilon(1+\varepsilon^{\eta})$ and $\zeta\varepsilon\leq{}Y_s\leq{}\zeta\varepsilon(1+\varepsilon^{\eta})$. We analyze only the first case (the other case follows similarly). If $|x|\geq{}\varepsilon^{1+\eta}$, we can repeat the same argument as in \eqref{WTUxL}. Suppose that $|x|\leq{}\varepsilon^{1+\eta}$. Note that $|p_n'(y)|\leq{}2|yf_n(y)|+|y^2 f_n'(y)|\leq C\varepsilon^{1-\eta}$, {when $|y|\leq C\varepsilon$} and, in particular, $|v_n({Y_s},x)|\leq {|p'_n(Y_s,x)||x|}\leq C \varepsilon^{1-\eta}|x|$. Thus, $F_n$ is bounded by
\begin{align*}
	&\varepsilon^{\frac{Y-4}{2}}\varepsilon^{1-\eta}\sum_{i=1}^{n}
	\mathbb{E}_{i-1}\left[\int_0^{h_{n}}\!\!\!\!\int|\delta^{i-1}(s,x)||x|{\bf 1}_{\{\varepsilon\leq{}|Y_s|\leq{}\varepsilon(1+\varepsilon^\eta)\}}{\bf 1}_{\{|x|\leq{}\varepsilon^{1+\eta}\}}\nu(dx)ds\right]\\
	&\quad\leq
	\varepsilon^{\frac{Y-4}{2}}\varepsilon^{1-\eta}\sum_{i=1}^{n}
	\mathbb{E}_{i-1}\left[\int_0^{h_{n}}\!\!\!\!\int |\delta^{i-1}(s,x)|^2\nu(dx)ds\right]^{\frac{1}{2}}\\
	&\qquad\qquad \qquad\qquad \qquad\times \mathbb{E}_{i-1}\left[\int_0^{h_{n}}\!\!\!\!\int_{|x|\leq{}C\varepsilon^{1+\eta}} x^2{\bf 1}_{\{\varepsilon\leq{}|Y_s|\leq{}\varepsilon(1+\varepsilon^\eta)\}}\nu(dx)ds\right]^{\frac{1}{2}}\\
	&\quad \leq C\varepsilon^{\frac{Y-4}{2}}\varepsilon^{1-\eta}\left(\varepsilon^{(2-Y)(1+\eta)}h^2\varepsilon^{\eta-Y}\right)^{\frac{1}{2}}\sum_{i=1}^{n}
	\left[{\mathbb{E}_{i-1}(\Delta_i^n M)^2}\right]^{\frac{1}{2}}\\
	&\quad=\varepsilon^{-\frac{Y}{2}+\frac{\eta(1-Y)}{2}}h^{\frac{1}{2}}O_P(1),
\end{align*}
where we used \eqref{CInPON}. The expression above is $o_P(1)$ if $\eta$ is small enough.
This concludes the proof.
 \end{proof}

\section{Other Technical Proofs}\label{MoreMoreAuxLm}

We follow the same notation and assumptions introduced in Sections \ref{Sec:Model} and \ref{VeryTechProofs}. In particular, unless otherwise stated, throughout this appendix we continue to assume \eqref{e:basic_asymptotics}  in all asymptotic expressions  below and assume that the L\'evy process $J$ satisfies all the  conditions (i)-(vi) of Assumptions \ref{assump:Funtq} and  \ref{assump:FuntqStrong}.  Let us recall the Fourier transform and its inverse:
\begin{align}\label{FourierTrsf}
 (\cF g)(z):=\frac{1}{\sqrt{2\pi}}\int_{\bR}g(x)e^{-izx}\,dx,\quad\big(\cF^{-1}g\big)(x):=\frac{1}{\sqrt{2\pi}}\int_{\bR}g(z)e^{izx}\,dz.
\end{align}

We recall the following bounds for the tail probabilities and density of $S_1$ under $\wt\bP$, which hereafter we denote $p_S$:
\begin{align}
\label{eq:1stOrderEstDenZ}
p_{S}(\pm z)&\leq{K}\,z^{-Y-1},\\
\label{eq:2ndOrderEstDenZ} \Big|p_{S}(\pm z)-C_{\pm} z^{-Y-1}\Big|&\leq{K}\big(z^{-Y-1}\wedge z^{-2Y-1}\big),\\
\label{eq:1stOrderEstTailZ} \wt{\bP}\big(\!\pm\!S_{1}>z\big)&\leq{K}z^{-Y},\\
\label{eq:2ndOrderEstTailZ} \bigg|\wt{\bP}\big(\!\pm\!S_{1}>z\big)-\frac{C_{\pm}}{Y}z^{-Y}\bigg|&\leq{K}z^{-2Y},
\end{align}
for a constant ${K}<\infty$. The inequalities above can be derived from the asymptotics in Section 14 of \cite{Sato:1999} (see \cite{FigueroaLopezGongHoudre:2016} for more details).

\smallskip
The following lemma is used in the proof of Proposition \ref{lemma:Estable}.
\begin{lemma} \label{lemmaC1}
Let $Y\in(0,1)\cup(1,2)$, fix  $\gamma_0\in \bR$, and let $J'_h =  S_h +\gamma_0 h$.  Then, %
as $h\to 0$, with $ \sqrt h\ll \varepsilon\ll 1$,
\begin{align}%
{\wt\bE}\left(\phi\bigg(\frac{\varepsilon\pm J'_{h}}{\sigma\sqrt{h}}\bigg)\right) &= C_{\pm}\sigma h ^{\frac{3}{2}} (\varepsilon\pm \gamma_0h)^{-Y-1} + O\left(h^{\frac{5}{2}}\varepsilon^{-Y-3}\right)+O\bigg(e^{-\frac{\varepsilon^{2}}{2\sigma^{2}h}}\bigg).
\label{eq:Ephi_lemma}
\end{align}
\end{lemma}
\begin{proof}
Let $p_{J,h}^{\pm}$ be the density of $\pm J'_{h}$ under $\wt{\bP}$, and observe
\begin{align*}
\big(\cF p_{J,h}^{\pm}\big)(u)%
&=\frac{1}{\sqrt{2\pi}}\exp\Big(c_{1}|u|^{Y}h\,{ \pm}\, ic_{2}|u|^{Y}h\,\text{sgn}(u)\mp iu\gamma_0h\Big),
\end{align*} 
where
\begin{align}\label{e:def_c1,c2}
&c_{1}:=(C_{+}\!+C_{-})\,\cos\bigg(\frac{\pi Y}{2}\bigg)\Gamma(-Y)<0,\notag\\
&c_{2}:=(C_{-}-C_{+})\sin\bigg(\frac{\pi Y}{2}\bigg)\Gamma(-Y).
\end{align}
 Also let
\begin{align*}
    \psi(x)&:=\bigg(\cF^{-1}\phi\bigg(\frac{\cdot}{\sigma\sqrt{h}}-\frac{\varepsilon}{\sigma\sqrt{h}}\bigg)\bigg)(x)
    =\frac{\sigma\sqrt{h}}{\sqrt{2\pi}}\exp\bigg(i\varepsilon x-\frac{1}{2}\sigma^{2}x^{2}h\bigg).
\end{align*}
Then, using Plancherel,
\begin{align}
\nonumber
    \wt{\bE}\phi\bigg(\frac{\varepsilon\pm J'_h}{\sigma\sqrt{h}}\bigg)
    &= \int_{\bR}(\cF\psi)(z)p_{J,h}^{\mp}(z)\,dz = \int_{\bR}\psi(u)\cF\big(p_{J,h}^{\mp}(z)\big)(u)\,du\\
\nonumber    
    &= \frac{\sigma\sqrt{h}}{\pi} \int_{0}^{\infty}  e^{c_1u^Y h - \frac{1}{2}\sigma^2u^2 h} \, \cos \left( \mp c_2u^Y h + u(\varepsilon \pm {\gamma}_0 h)\right )du  \\
\nonumber    
        &=\int_{0}^{\infty}g_1(\omega^{Y}h^{1-\frac{Y}{2}})e^{-\frac{\omega^{2}}{2}}\cos\bigg(\frac{\omega\varepsilon^{\pm}_0}{\sigma\sqrt{h}}\bigg)d\omega\\
        \nonumber
        & \qquad\qquad\qquad\qquad
       -\int_{0}^{\infty}g^\pm_2(\omega^{Y}h^{1-\frac{Y}{2}})e^{-\frac{\omega^{2}}{2}}\sin\bigg(\frac{\omega\varepsilon^{\pm}_0}{\sigma\sqrt{h}}\bigg)d\omega\\
       \label{TrmI2Pa}
       &=:I_1-I_2,
\end{align}
where $\varepsilon^{\pm}_0=\varepsilon\pm{\gamma_0}h,$ and
\[
	g_1(u)=\frac{1}{\pi}e^{c_{1}u/\sigma^Y}\cos(c_{2}u/\sigma^Y),\quad
	g^\pm_2(u)= \mp \frac{1}{\pi}e^{c_{1}u/\sigma^Y}\sin(c_{2}u/\sigma^Y).
\]
Fix $\delta>0$ and set $D_{\delta,h}:=\{w\in(0,\infty): w^{Y}h^{1-Y/2}<(2\delta)^{Y/2}\}$.
\[	
	I_1=\left(\int_{D_{\delta,h}}+\int_{D_{\delta,h}^c}\right)g_1(\omega^{Y}h^{1-\frac{Y}{2}})e^{-\frac{\omega^{2}}{2}}\cos\bigg(\omega\frac{\varepsilon^{\pm}_0}{\sigma\sqrt{h}}\bigg)d\omega=:I_{11}+I_{12}.
\]
Clearly, since $c_1<0$, 
\[
	|I_{12}|\leq{}\frac{ 1}{\pi}\int_{(2\delta)^{1/2}/h^{1/Y-1/2}}^{\infty}e^{-w^{2}/2}dw=O\left(h^{\frac{2-Y}{2Y}}e^{-\delta h^{\frac{Y-2}{Y}}}\right).
\]
We decompose the integral on $D_{\delta,h}$ as follows:
\begin{align}\label{Term1I11}
	\qquad I_{11}&=\sum_{j=0}^{m-1}\frac{1}{j!}g_1^{(j)}(0)h^{j(1-\frac{Y}{2})}\int_{0}^{\infty}w^{jY}e^{-\frac{\omega^{2}}{2}}\cos\bigg(\omega\frac{\varepsilon_0^{\pm}}{\sigma\sqrt{h}}\bigg)d\omega\\
	\label{Term2I11}
	&\quad-\sum_{j=0}^{m-1} \frac{1}{j!}g_1^{(j)}(0)h^{j(1-\frac{Y}{2})}\int_{(2\delta)^{1/2}/h^{1/Y-1/2}}^{\infty}w^{jY}e^{-\frac{\omega^{2}}{2}}\cos\bigg(\omega\frac{\varepsilon_0^{\pm}}{\sigma\sqrt{h}}\bigg)d\omega\\
	\label{Term3I11}
	&\quad +\frac{h^{m(1-\frac{Y}{2})}}{m!}\int_{0}^{(2\delta)^{1/2}/h^{1/Y-1/2}}g_{1}^{(m)}(\theta_{w})w^{mY}e^{-\frac{\omega^{2}}{2}}\cos\bigg(\omega\frac{\varepsilon_0^{\pm}}{\sigma\sqrt{h}}\bigg)d\omega,
\end{align}
where $\theta_{\omega}\in(0,(2\delta)^{Y/2})$. The $j^{th}$ term in \eqref{Term2I11} is %
\[
	O\left(h^{j(1-\frac{Y}{2})}\int_{(2\delta)^{1/2}/h^{1/Y-1/2}}^{\infty}w^{jY}e^{-\frac{\omega^{2}}{2}}d\omega\right)
	=O\left(h^{\frac{2-Y}{2Y}}e^{-\delta h^{\frac{Y-2}{Y}}}\right),
\]
where above we expressed the integral in terms of the incomplete gamma function and then apply the standard asymptotics behavior for such a function (e.g., \cite[Section 8.35]{gradshteyn:ryzhik:2007}).  For \eqref{Term1I11}, the term corresponding to $j=0$ is clearly $O(e^{-\varepsilon^{2}/(2\sigma^{2}h)})$. For $j\geq{}1$, we can apply the same argument as in (A.11) of \cite{gong2021} (see also \cite[expression 13.7.2]{Olveretal}) to show that 
\begin{align}\label{eq:KummerPowerCos0}
&h^{j(1-\frac{Y}{2})}\int_{0}^{\infty}\omega^{jY}\cos\bigg(\omega\frac{\varepsilon_0^{\pm}}{\sigma\sqrt{h}}\bigg)e^{-\omega^{2}/2}\,d\omega\\
\nonumber
&= h^{j(1-\frac{Y}{2})}2^{(jY-1)/2}\,\Gamma\bigg(\frac{jY+1}{2}\bigg) M\bigg(\frac{jY+1}{2},\frac{1}{2},-\frac{(\varepsilon_0^{\pm})^{2}}{2\sigma^{2}h}\bigg)\\
\nonumber
&=h^{j(1-\frac{Y}{2})}2^{\frac{jY-1}{2}}\,\Gamma\bigg(\frac{jY+1}{2}\bigg)\left(\frac{\Gamma(1/2)}{\Gamma\big(-jY/2\big)}\bigg(\frac{(\varepsilon_0^{\pm})^{2}}{2\sigma^{2}h}\bigg)^{-\frac{jY+1}{2}}\left(1+O\left(\frac{h}{\varepsilon^{2}}\right)\right)\right.\\
\nonumber
&\qquad\qquad\qquad\qquad\qquad\,\,\left.+\frac{\Gamma(1/2)\,e^{-\varepsilon^{2}/(2\sigma^{2}h)}}{\Gamma\big((jY+1)/2\big)}\bigg(\frac{(\varepsilon^{\pm}_0)^{2}}{2\sigma^{2}h}\bigg)^{jY/2}\right)=O\bigg(\frac{h^{j+1/2}}{\varepsilon^{jY+1}}\bigg).
\end{align}
Therefore, the $j=1$ term in \eqref{Term1I11} is the leading order term,  and the second-order term therein is $O(h^{5/2}/\varepsilon^{2Y+1})$. Furthermore, by picking $m>5/(2-Y)$, we see the term in \eqref{Term3I11} -- which is   $O(h^{m(1-Y/2)})$ -- decays faster than the second-order term  in \eqref{Term1I11}. %
Putting together the above estimates and noting that $g_1'(0)=c_1/(\pi\sigma^Y)$, we conclude that 
\[
	I_{1}=\frac{c_1}{\pi}\sigma2^{Y}\frac{\Gamma(Y/2+1/2)\Gamma(1/2)}{\Gamma(-Y/2)}\frac{h^{3/2}}{(\varepsilon_0^{\pm})^{Y+1}}+O\left(e^{-\frac{\varepsilon^{2}}{2\sigma^{2}h}}\right)+O\bigg(\frac{h^{5/2}}{\varepsilon^{Y+3}}\bigg),%
\]
{where above we used that $h^{\frac{2-Y}{2Y}}e^{-\delta h^{\frac{Y-2}{Y}}}\ll h^{5/2}\varepsilon^{-(Y+3)}$.} %
Applying the same analysis to the term $I_2$ in \eqref{TrmI2Pa}%
\[
	I_{2}=\pm \frac{c_2}{\pi}\sigma2^{Y+1}\frac{\Gamma(Y/2+1)\Gamma(3/2)}{\Gamma((1-Y)/2)}\frac{h^{3/2}}{(\varepsilon_0^{\pm})^{Y+1}}+O\left(e^{-\frac{\varepsilon^{2}}{2\sigma^{2}h}}\right)+{O\bigg(\frac{h^{5/2}}{\varepsilon^{Y+3}}\bigg)}.%
\]
We then get that
\begin{align*}
	 &\wt{\bE}\phi\bigg(\frac{\varepsilon\pm J'_h}{\sigma\sqrt{h}}\bigg)\\
	& =\frac{\sigma 2^{Y}}{\pi}\left(c_1\frac{\Gamma(Y/2+1/2)\Gamma(1/2)}{\Gamma(-Y/2)}\, \mp\,c_2\frac{2\Gamma(Y/2+1)\Gamma(3/2)}{\Gamma((1-Y)/2)}\right)\frac{h^{3/2}}{(\varepsilon_0^{\pm})^{Y+1}}\\
	&\qquad+O\left(e^{-\frac{\varepsilon^{2}}{2\sigma^{2}h}}\right)+{O\bigg(\frac{h^{5/2}}{\varepsilon^{Y+3}}\bigg)}.%
\end{align*}
By applying the properties 
$\Gamma(x)\, \Gamma\left(x+1/2 \right) = \sqrt{\pi}\, 2^{1-2x}\,\Gamma\left(2x\right)$
and $\Gamma(1-x)\,\Gamma(x) = \frac{\pi}{\sin (\pi x)}$, we can simplify the coefficients as
\begin{align*}
	&\frac{2^{Y}}{\pi}c_1\frac{\Gamma(Y/2+1/2)\Gamma(1/2)}{\Gamma(-Y/2)}=\frac{C_++C_-}{2}\\
	 &\frac{2^{Y+1}}{\pi}c_2\frac{\Gamma(Y/2+1)\Gamma(3/2)}{\Gamma((1-Y)/2)}=\frac{C_--C_+}{2}.
\end{align*}
We finally {obtain} \eqref{eq:Ephi_lemma}.
\end{proof}

The following lemma is used in the proofs of Lemmas \ref{lemma:W2k} and \ref{lemma:WJ}.
\begin{lemma} \label{lemma:4results}
Suppose $Y\in(0,1)\cup(1,2)$, fix $\gamma_0\in \bR$, any integer $k\geq 0$, and let $J'_h =  S_h +\gamma_0 h$. Then, with $\sqrt h \ll \varepsilon \ll 1$,%
as $h\to 0$,
\begin{align}\label{e:J^kphi_asymptotics_1}
    \wt{\bE}\left(\left(\mp  J'_h \right)^k\,\phi\bigg(\frac{\varepsilon\pm J'_h}{\sigma\sqrt{h}}\bigg)\right) &= O \left(h^{\frac{3}{2}} \varepsilon^{k-Y-1}\right) +{O \left( \varepsilon^{k}e^{-\frac{\varepsilon^2}{2\sigma^2 h} } \right)},\\%
    \nonumber
    \wt{\bE}\left(\left(\frac{\varepsilon \pm  J'_h}{\sigma\sqrt{h}}\right)^k\,\phi\bigg(\frac{\varepsilon\pm J'_h}{\sigma\sqrt{h}}\bigg)\right) &= 
O \left(h^{\frac{3-k}{2}} \varepsilon^{k-Y-1}\right) +O \left(\varepsilon^{k}h^{\frac{-k}{2}}\,e^{-\frac{\varepsilon^2}{2\sigma^2 h}} \right).%
\end{align}
\end{lemma}
\begin{proof}
We use the same notation as in the proof of Lemma \ref{lemmaC1}.
We first establish the asymptotic bound  for any fixed $\ell=0,1,\dots$, 
\begin{align}\nonumber
\mathcal{I}(\ell)&:=\int_{\bR}  u^\ell \exp{\left (iu (\varepsilon \pm  \gamma_0 h) - \frac{1}{2}\sigma^2u^2 h +|u|^Y h \left(c_1 \,  \mp\, ic_2\sgn (u)\right) \right)} du\\
&
= O \left(h \varepsilon^{-Y-\ell-1}\right)+ O \left(h^{-\ell-\frac{1}{2}} \, \varepsilon^{\ell} e^{-\frac{\varepsilon^2}{2\sigma^2 h}} \right).%
\label{e:moments_of_FT_expression}
\end{align}
 We suppose $\ell\geq 0$ is even (the case for $\ell$ odd can be established similarly).  As in \eqref{TrmI2Pa}, we can write
\begin{align}
\nonumber
    \mathcal{I}(\ell)
    &= 2 \, (\sigma \sqrt{h})^{-\ell-1} \int_0^\infty w^\ell g_1(h^{1-Y/2} \omega^Y)e^{-\frac{\omega^2}{2}}\cos \left(\frac{\varepsilon_0^{\pm}}{\sigma \sqrt{h}}\omega\right) d\omega \\
    & \quad - 2 \, (\sigma \sqrt{h})^{-\ell-1} \int_0^\infty w^\ell g_2^{ \pm}(h^{1-Y/2} \omega^Y)e^{-\frac{\omega^2}{2}}\sin\left(\frac{\varepsilon_0^{\pm}}{\sigma \sqrt{h}}\omega\right) d\omega.
    \label{eq:inteven}
\end{align}
Applying the same arguments as in the proof of Lemma \ref{lemmaC1}, we can show that, for $\delta>0$ and positive integer $m$,
\begin{align}\label{T1NH1}
   \qquad\; \mathcal{I}(\ell)&=\sum_{j=0}^{m-1}O\left(h^{-\frac{\ell+1}{2}}h^{j(1-\frac{Y}{2})}\int_{(2\delta)^{1/2}/h^{1/Y-1/2}}^{\infty}w^{jY+\ell}e^{-\frac{\omega^{2}}{2}}d\omega\right)\\
    \label{T2NH1}
    &\quad+\sum_{j=1}^{m-1}O\left(h^{-\frac{\ell+1}{2}}h^{j(1-\frac{Y}{2})}\int_{0}^{\infty}\omega^{jY+\ell}\cos\bigg(\omega\frac{\varepsilon_0^{\pm}}{\sigma\sqrt{h}}\bigg)e^{-\omega^{2}/2}\,d\omega\right)\\
        \label{T2NH1b}
    &\quad+\sum_{j=1}^{m-1}O\left(h^{-\frac{\ell+1}{2}}h^{j(1-\frac{Y}{2})}\int_{0}^{\infty}\omega^{jY+\ell}\sin\bigg(\omega\frac{\varepsilon_0^{\pm}}{\sigma\sqrt{h}}\bigg)e^{-\omega^{2}/2}\,d\omega\right)\\
    \label{T3NH1}
    &\quad+O\left(h^{-\frac{\ell+1}{2}}h^{m(1-Y/2)}\right)+O \left(h^{-\ell-\frac{1}{2}} \, \varepsilon^{\ell} e^{-\frac{\varepsilon^2}{2\sigma^2 h}} \right).
 \end{align}
Each of the $m$ terms in \eqref{T1NH1} is $O(h^{( 1-Y-\ell)/Y}e^{-\delta h^{\frac{Y-2}{Y}}})$. For the terms in \eqref{T2NH1}-\eqref{T2NH1b}, we recall the estimates (cf (43) and (44) in \cite{gong2021}): 
\begin{align*} %
h^{j(1-Y/2)}\int_{0}^{\infty}\omega^{jY+r}\cos\bigg(\omega\cdot\frac{\varepsilon}{\sigma\sqrt{h}}\bigg)e^{-\omega^{2}/2}\,d\omega = O\bigg(\frac{h^{j+(r+1)/2}}{\varepsilon^{jY+r+1}}\bigg),\\
h^{j(1-Y/2)}\int_{0}^{\infty}\omega^{jY+r}\sin\bigg(\omega\cdot\frac{\varepsilon}{\sigma\sqrt{h}}\bigg)e^{-\omega^{2}/2}\,d\omega= O\bigg(\frac{h^{j+(r+1)/2}}{\varepsilon^{jY+r+1}}\bigg),
\end{align*}
valid when  $h\to{}0$ and $\varepsilon/\sqrt{h}\to{}\infty$. Then, we can conclude that the leading term in \eqref{T2NH1} and  \eqref{T2NH1b} corresponds to $j=1$ and is of order $O \left(h \varepsilon^{-Y-\ell-1}\right)$, which is asymptotically larger than every term in \eqref{T1NH1}. The first term in \eqref{T3NH1} can be made of higher order by choosing $m$ large enough. We then conclude \eqref{e:moments_of_FT_expression}.

Turning to \eqref{e:J^kphi_asymptotics_1}, 
we apply Plancherel to get:
\begin{align}\nonumber
    &\wt{\bE}\left(\left(\mp  J'_h \right)^k\,\phi\bigg(\frac{\varepsilon\pm J'_h}{\sigma\sqrt{h}}\bigg)\right)\\
        \nonumber
    &\quad= \int_{\bR}(\cF\psi)(z)z^k p_{J,h}^{\mp}(z)\,dz= (-i)^k \int_{\bR}\psi^{(k)}(u)\cF\big(p_{J,h}^{\mp}(z)\big)(u)\,du\\
    \nonumber
    &\quad=  \varepsilon^k \, \frac{\sigma\sqrt{h}}{2\pi} \, \sum_{l=0}^k \sum_{m=0}^{\left \lfloor l/2 \right \rfloor} \frac{(-1)^{l+m}\,2^{l-2m}  \, k!}{(k-l)! \, m! \, (l-2m)!} \left(i\varepsilon\right)^{-l} \left(\frac{\sigma^2 h}{2}\right)^{l-m} \mathcal{I}({l-2m})\\
    \quad&\quad=    O \left(h^{\frac{3}{2}} \varepsilon^{k-Y-1}\right)+O \left(\varepsilon^{k}\,e^{-\frac{\varepsilon^2}{2\sigma^2 h}} \right)+O\left(\varepsilon^{k}h^{\frac{2-Y}{2Y}}\,e^{-\delta h^{\frac{Y-2}{Y}}}\right),  \label{eq:EJk}
\end{align}
where on the last line we made use of the asymptotic expression \eqref{e:moments_of_FT_expression} for $\mathcal I(\ell)$.  This establishes \eqref{e:J^kphi_asymptotics_1}.  
For the second asymptotic bound, based on \eqref{eq:EJk}, we can compute
\begin{align*}
    & \wt{\bE}\left(\left(\frac{\varepsilon \pm  J'_h}{\sigma\sqrt{h}}\right)^k\,\phi\bigg(\frac{\varepsilon\pm J'_h}{\sigma\sqrt{h}}\bigg)\right) \\
    &\quad = (\sigma\sqrt{h})^{-k} \, \wt\bE\left(\sum_{m=0}^k\binom{k}{m}\varepsilon^{k-m}(\pm  J'_h)^m \phi\bigg(\frac{\varepsilon\pm J'_h}{\sigma\sqrt{h}}\bigg)\right)\\
    &\quad= 
    (\sigma\sqrt{h})^{-k} \sum_{m=0}^k \binom{k}{m}\varepsilon^{k-m} \bigg[O \left(h^{\frac{3}{2}} \varepsilon^{m-Y-1}\right) +O \left(\varepsilon^{m}\,e^{-\frac{\varepsilon^2}{2\sigma^2 h}} \right)\\
    & \qquad \qquad \qquad \qquad \qquad\qquad \qquad \qquad+O\left(\varepsilon^{m}h^{\frac{2-Y}{2Y}}\,e^{-\delta h^{\frac{Y-2}{Y}}}\right)\bigg].
\end{align*}
This clearly implies the result.
\end{proof}


The following Lemma is used in the proof of Proposition \ref{prop:EX2}.
\begin{lemma} \label{lemma:W2k}
Suppose $Y\in(0,1)\cup(1,2)$. Fix a positive integer $k$. Then,  %
as $h\to 0$, $\sqrt h \ll \varepsilon \ll 1$,
\begin{align*}
  &\bE\left(W_h^{2k}\,{\bf 1}_{\{|\sigma W_{h}+J_{h}|\leq\varepsilon\}}\right)\\
  &~=    {(2k-1)!!\,h^k + %
   O\left(h^2 \varepsilon^{2k-Y-2}\right)+ o\left(h^{k+1/Y}\right)+ O\left(h^{\frac{1}{2}}\varepsilon^{2k-1} e^{-\frac{\varepsilon^2}{2\sigma^2 h}}\right)}.%
\end{align*}
\end{lemma}
\begin{proof}
The case of $k=1$ has been proved in the Step 1 of Lemma 3.1 of \cite{gong2021}. Following a similar procedure, we generalize to the case of {$k \geq1$}.  Note
\begin{align}
\nonumber
\bE\Big(W_{h}^{2k}\,{\bf 1}_{\{|\sigma W_{h}+J_{h}|\leq\varepsilon\}}\Big) &= (2k-1)!! \, h^k - \bE\Big(W_{h}^{2k}\,{\bf 1}_{\{|\sigma W_{h}+J_{h}|>\varepsilon\}}\Big)\\
\nonumber
&= (2k-1)!! \, h^k - \wt{\bE}\Big(e^{-\wt{U}_{h}-\eta h}\,W_{h}^{2k}\,{\bf 1}_{\{|\sigma W_{h}+J_{h}|>\varepsilon\}}\Big)\\
\nonumber
&= (2k-1)!! \, h^k - h^k e^{-\eta h}\,\wt{\bE}\Big(W_{1}^{2k}\,{\bf 1}_{\{|\sigma\sqrt{h}W_{1}+J_{h}|>\varepsilon\}}\Big)\\
\nonumber
&\quad -h^k e^{-\eta h}\,\wt{\bE}\Big(\Big(e^{-\wt{U}_{h}}-1\Big)W_{1}^{2k}\,{\bf 1}_{\{|\sigma\sqrt{h}W_{1}+J_{h}|>\varepsilon\}}\Big)\\
\label{eq:kDecompb1eps1} &=:(2k-1)!! \, h^k - h^k e^{-\eta h}I_{1}(h) - h^ke^{-\eta h}I_{2}(h).
\end{align}
We first analyze the asymptotic behavior of $I_1(h)$. Write
\begin{align}
I_{1}(h)&=\wt{\bE}\Big(W_{1}^{2k}\,{\bf 1}_{\{\sigma\sqrt{h}W_{1}+J_{h}>\varepsilon\}}\Big)+\wt{\bE}\Big(W_{1}^{2k}\,{\bf 1}_{\{\sigma\sqrt{h}W_{1}-J_{h}>\varepsilon\}}\Big)\notag\\
&=:I_{1}^{+}(h)+I_{1}^{-}(h).\label{eq:kDecompb1eps12}
\end{align}
 Recall the notation $\phi(x):=\frac{1}{\sqrt{2\pi}}e^{-x^{2}/2}$, and $\overline{\Phi}(x):=\int_{x}^{\infty}\phi(x)\,dx$.
By repeating integration by parts, we have for all $x \in \bR$,
\begin{align*}
    \wt\bE\left(W_1^{2k}\,{\bf 1}_{\{W_1 >x\}}\right) = (2k-1)!!\, \overline{\Phi}(x) + \sum_{j=0}^{k-1}\frac{(2k-1)!!}{(2j+1)!!} \, x^{2j+1} \phi(x).
\end{align*}
Also, note that expression (A.21) of \cite{gong2021} shows
\begin{align}\label{eq:LimitExpPhi}
\wt{\bE}\left(\overline{\Phi}\bigg(\frac{\varepsilon}{\sigma\sqrt{h}}\pm\frac{J_{h}}{\sigma\sqrt{h}}\bigg)\right)=O\Big(e^{-\varepsilon^{2}/(2\sigma^{2}h)}\Big)+O\big(h\varepsilon^{-Y}\big),%
\quad\text{as }\,h\rightarrow 0.
\end{align}
Thus, by conditioning on $J_{h}$ and then applying Lemma \ref{lemma:4results},
\begin{align}\nonumber
I_{1}^{\pm}(h) &= \wt{\bE}\Bigg((2k-1)!!\,  \overline{\Phi}\bigg(\frac{\varepsilon}{\sigma\sqrt{h}}\mp\frac{J_{h}}{\sigma\sqrt{h}}\bigg) \\
& \qquad\qquad\quad + \sum_{j=0}^{k-1}\frac{(2k-1)!!}{(2j+1)!!} \, \bigg(\frac{\varepsilon}{\sigma\sqrt{h}}\mp\frac{J_{h}}{\sigma\sqrt{h}}\bigg)^{2j+1} \phi\bigg(\frac{\varepsilon}{\sigma\sqrt{h}}\mp\frac{J_{h}}{\sigma\sqrt{h}}\bigg) \Bigg)\notag\\
&= O\Big(e^{-\frac{\varepsilon^{2}}{2\sigma^{2}h}}\Big)+O\big(h\varepsilon^{-Y}\big) +  O \left(h^{2-k} \varepsilon^{2k-2-Y}\right)\notag\\
&\qquad \qquad \quad+ O\left(h^{\frac{1}{2}-k}\varepsilon^{2k-1} e^{-\frac{\varepsilon^2}{2\sigma^2 h}}\right)+O\left(h^{\frac{1}{Y}-k}\varepsilon^{2k-1} e^{-\delta h^{\frac{Y-2}{Y}}}\right) \notag
\end{align}
\begin{equation}
= O \left(h^{2-k} \varepsilon^{2k-Y-2}\right) + O\left(h^{\frac{1}{2}-k}\varepsilon^{2k-1} e^{-\frac{\varepsilon^2}{2\sigma^2 h}}\right)+O\left(h^{\frac{1}{Y}-k}\varepsilon^{2k-1} e^{-\delta h^{\frac{Y-2}{Y}}}\right). 
\label{eq:Decompb1eps12pm}
\end{equation}
Using the same procedure as shown in the Step 1.2 of Theorem 3.1 of \cite{gong2021}, we have that  $I_2(h) = o(h^{1/Y})$. Note that the last term \eqref{eq:Decompb1eps12pm} is of smaller order than the first term in \eqref{eq:Decompb1eps12pm}. The result then follows.
\end{proof}

The lemma below gives an expansion for a type of truncated higher-order moment of a pure-jump stable L\'evy process. 
This is used in the proofs of Propositions \ref{lemma:Estable} and \ref{lemma:2E}.
\begin{lemma} \label{lemma:S2k_new} 
Suppose $Y\in(0,1)\cup(1,2)$. Let $\gamma_0\in \bR$,  $\sigma\in[0,\infty)$, and  consider $S_h$ as in \eqref{e:def_St}.   Then, for any fixed $k\geq 1$, as $h\to 0$ with $\sqrt{h}\ll \varepsilon\ll 1$,
\begin{align}
&\wt\bE\left(S_h^{2k}\,{\bf 1}_{\{|\sigma W_{h}+S_{h} + h\gamma_0|\leq\varepsilon\}}\right)\notag \\
&= \frac{C_{+}\!+\!C_{-}}{2k-Y}h\varepsilon^{2k-Y} +  \frac{\sigma^2(C_++C_-)(2k-1-Y)}{2} h^{2}\,\varepsilon^{2k-Y-2}\notag\\
& \qquad +O\left(h^3\varepsilon^{2k-4-Y}\right) + O\big(h^{2}\,\varepsilon^{2k-2Y}\big) +  O\left(\varepsilon^{2k-1-Y}h^{\frac{3}{2}}e^{-\frac{\varepsilon^2}{2\sigma^2 h}}\right). \label{e:E-tilde_Sh_expansion}
\end{align}
\end{lemma}
\begin{proof}
It will be evident from the proof below that the statement remains valid for the case $\sigma=0$. However, for brevity, we only give the details for the case $\sigma>0$, which is more involved. First suppose $\gamma_0=0$. For simplicity, let $I(x)=({\frac{-\varepsilon  - x \sigma \sqrt{h}}{h^{1/Y}}},{\frac{\varepsilon -   x \sigma \sqrt{h}}{h^{1/Y}}})$ and also let  $r:=r(h) := \frac{\varepsilon}{\sigma \sqrt h}\to \infty$, as $h\to 0$.  Then, in terms of the density $p_S$ of $S_1$ under $\wt \bP$,
\begin{align}
\nonumber
\wt\bE\Big(S_{h}^{2k}\,{\bf 1}_{\{|{\sigma W_{h}}+{S_{h}}|\leq\varepsilon\}}\Big)&={h^{2k/Y}\wt\bE\Big(S_{1}^{2k}\,{\bf 1}_{\{|\sigma \sqrt h W_{1}+h^{1/Y}S_{1}|\leq\varepsilon\}}\Big)}\\
\nonumber
&=h^{2k/Y} \Big\{\int_{\{|x|<r\}} + \int_{\{|x|>r\}}\Big\}\int_{I(x)}u^{2k} p_S(u) du \phi(x) dx\\
&=:h^{2k/Y}(T_1 + T_2).\label{e:i1,i2}
\end{align}
We first focus on $T_1$. Write $\mathcal{D}(u):=p_S(u)-(C_-\mathbf 1_{\{u<0\}} + C_+\mathbf 1_{\{u>0\}}) |u|^{-1-Y}$, recall $\bar{C}=C_++C_-$, and observe
\begin{align}
T_1
 &= \frac{ \bar{C}}{2k-Y}(\varepsilon h^{-1/Y})^{2k-Y}\int_{-r}^r\Big( 1 - \frac{x}{r} \Big)^{2k-Y}  \phi(x)dx\notag\\
 &\qquad + \int_{-r}^r\int_{I(x)} u^{2k} \mathcal{D}(u)du \phi(x) dx.
  \label{e:C_term1}
 \end{align}
Letting $f(u) = (1-u)^{2k-Y}$, $q_{ m}(u) = \sum_{\ell=0}^{ m} \frac{f^{(\ell)}(0)}{\ell!} u^\ell$, and $R_{ m}(u) = f(u)-q_{ m}(u)$,
the integral in first term of  \eqref{e:C_term1} may then be written as
\begin{align*}
&\int_{-r}^r\Big( 1 - \frac{x}{r} \Big)^{2k-Y}\phi(x)dx  \notag\\
& ~=\Big\{\int_{-r/2}^{r/2 } + \int_{(-r,r)\cap{ (-r/2,r/2)^c}}\Big\}f(x/r)\phi(x)dx \\
& ~=  \int_{-r/2}^{r/2} q_{m}(x/r) \phi(x) dx + \int_{-r/2}^{r/2} R_{ m}(x/r) \phi(x)dx +O(r^{-1}e^{-r^2/ 8}).
\end{align*}
For $|u|<\frac{1}{2}$,  there is a constant $K'=K'({m},Y)$ such that $|R_{m}(u)|\leq K'|u|^{{m}+1}$. Thus, 
\[ \int_{-r/2}^{r/2} |R_{ m}(x/r)| \phi(x)dx \leq K'r^{-({m}+1)} \int_{-r/2}^{r/2}x^{{m}+1}\phi(x)dx =O(r^{-({ m}+1)}).\]
Moreover, using that $\int_{r/2}^\infty x^{m}\phi(x) dx=O(r^{({ m}-1)} e^{-(r/2)^2/2})$,
$$
\int_{-r/2}^{r/2} q_{ m}(x/r) \phi(x) dx  = \int_{-\infty}^\infty q_{ m}(x/r) \phi(x) dx + O(r^{({m}-1)} e^{-r^2/8}).
$$
Hence, taking ${m}=3$, as $r\to\infty,$
\begin{align}\label{NIE0}
\qquad\int_{-r}^r\Big( 1 - \frac{x}{r} \Big)^{2k-Y}\phi(x)dx 
= 1 + \frac{(2k-Y)(2k-1-Y)}{2} \frac{1}{r^2} + O(r^{-4}).
\end{align}
%
%
%
%
To study the second term on the right-hand side of   \eqref{e:C_term1}, we consider the cases $k=1$ and $k>1$ separately.  First suppose $k=1$. We first show 
\begin{equation}\label{e:i12,t2_keq1} 
 \int_{-r}^r\int_{I(x)} u^2\mathcal{D}(u)du \phi(x) dx =  O((\varepsilon h^{-1/Y})^{2-2Y}).
\end{equation}
When $Y<1$, \eqref{e:i12,t2_keq1}  follows from a straightforward computation, so we provide details only for $Y>1$. In this case, by applying Lemma \ref{K0Is0}  and using the symmetry of  $\phi$,
\begin{align}\label{INHYl1}
 \int_{-r}^r\int_{I(x)} u^2\mathcal{D}(u)du \phi(x) dx=
- \int_{-r}^r \int_{\bR\setminus(-b(x),b(x))} u^2 \mathcal{D}(u) du\, \phi(x) dx,
 \end{align}
 where $b(x) = \frac{\varepsilon + x \sigma \sqrt{h}}{h^{1/Y}}=\frac{\varepsilon(1+x/r)}{h^{1/Y}} \in (0,2\varepsilon h^{-1/Y})$.
Observe 
$0<b(x)<1$ if and only if $x\in \big(-\frac{\varepsilon}{\sigma \sqrt h},\frac{h^{1/Y}-\varepsilon}{\sigma \sqrt h}\big) =\big(-r, r(\frac{h^{1/Y} }{\varepsilon }-1)\big)$. So, we consider 
\begin{align*}
    &\int_{-r}^r\int_{b(x)}^\infty u^{2} |\mathcal{D}(u)|du\, \phi(x) dx\\
    & = \left(\int_{-r}^{r(\frac{h^{1/Y} }{\varepsilon }-1)} + \int_{ r(\frac{h^{1/Y} }{\varepsilon }-1)}^{r}\right) \int_{b(x)}^\infty u^{2} |\mathcal{D}(u)|du\, \phi(x) dx.
\end{align*}
Applying (\ref{eq:1stOrderEstTailZ}), the first integral above can be estimated as follows:
\begin{align}
\nonumber
&\int_{-r}^{(\frac{h^{1/Y} }{\varepsilon }-1)r} \int_{b(x)}^\infty u^2 |\mathcal{D}(u)|du \phi(x) dx\\
\nonumber
&\leq K\int_{-r}^{(\frac{h^{1/Y} }{\varepsilon }-1)r}\bigg( \int_{b(x)}^1 u^{1-Y} du + \int_1^\infty u^{1-2Y}du\bigg) \phi(x)dx\\
\nonumber
&\leq K\int_{-r}^{(\frac{h^{1/Y} }{\varepsilon }-1)r}\Big(\frac{1}{2-Y} + \frac{1}{|2-2Y|}\Big) \phi(x)dx\\
& \leq K\int_{-r}^{-r/2} \phi(x)dx \leq O(r^{-1} e^{-(r/2)^2/2}).
\label{e:i12,t0}
\end{align}
Now, for $x\in((\frac{h^{1/Y} }{\varepsilon }-1)r,r)$, $b(x)>1$ and, so, again applying (\ref{eq:1stOrderEstTailZ}),
\begin{align}
\nonumber
&\int_{(\frac{h^{1/Y} }{\varepsilon }-1)r}^r \int_{b(x)}^\infty u^{2}|\mathcal{D}(u)|du \phi(x) dx \\
\nonumber
&\quad\leq \int_{-r(1-\varepsilon^{-1}h^{1/Y})}^\infty \int_{b(x)}^\infty u^{2} |\mathcal{D}(u)|du \phi(x) dx\\
\nonumber
&\quad \leq   
K (\varepsilon h^{-1/Y})^{2-2Y}\int_{(\frac{h^{1/Y} }{\varepsilon }-1)r}^\infty\Big(1 + \frac{x}{r}\Big)^{2-2Y} \phi(x) dx\\
&=  O((\varepsilon h^{-1/Y})^{2-2Y})\label{e:i12,t1}
\end{align}
Expressions \eqref{e:i12,t0}-\eqref{e:i12,t1} show that%
\[\int_{-r}^r\int_{b(x)}^\infty u^{2}|\mathcal{D}(u)|du\, \phi(x) dx = 
  O((\varepsilon h^{-1/Y})^{2-2Y}).\]
Similarly, 
$\int_{-r}^r \int_{-\infty}^{-b(x)} u^{2} |\mathcal{D}(u)|du\, \phi(x) dx =O((\varepsilon h^{-1/Y})^{2-2Y})$ and thus \eqref{e:i12,t2_keq1} holds.
For the case $k>1$, using \eqref{eq:1stOrderEstTailZ},
\begin{align}
\nonumber
& \int_{-r}^r\int_{I(x)} u^{2k}| \mathcal{D}(u)|du \phi(x) dx \\
 &\leq  K\int_{-r}^r \int_{I(x)} u^{2k-2Y-1} du \phi(x)dx = O((\varepsilon h^{-1/Y})^{2k-2Y})\label{e:i12,t2_kgeq1}.
\end{align}
Plugging \eqref{NIE0}, \eqref{e:i12,t2_keq1}, and \eqref{e:i12,t2_kgeq1} into 
\eqref{e:C_term1}, we conclude that for  $k\geq 1$,
\begin{align}
\nonumber
T_1&= \frac{\bar{C}}{2k-Y}\varepsilon^{2k-Y} h^{1-\frac{2k}{Y}} \Big( 1 + \frac{(2k-Y)(2k-1-Y)}{2}\frac{1}{r^2} +  O(r^{-4})\Big) \\
 & \quad + O((\varepsilon h^{-1/Y})^{2k-2Y}).\label{e:I3_expansion}
\end{align}
Now we examine $T_2$ in \eqref{e:i1,i2}.
Again applying \eqref{eq:1stOrderEstTailZ},%
\begin{align}
\nonumber
 \int_{r}^\infty\int_{\frac{-\varepsilon - x \sigma \sqrt{h}}{h^{1/Y}}}^{\frac{\varepsilon - x \sigma \sqrt{h}}{h^{1/Y}}} &|{u^{2k}} p_S(u) \phi(x) | du\, dx \\
 \nonumber
 &\leq 
 K (\varepsilon h^{-1/Y})^{2k-Y}\int_r^\infty (1+x/r)^{2k-Y} \phi(x) dx\\
 \nonumber
  & \leq K  (\varepsilon h^{-1/Y})^{2k-Y}\int_r^\infty (2x/r)^{2k-Y} \phi(x) dx\\
  & = \label{e:I4_asymptotics}O\left((\varepsilon h^{-1/Y})^{2k-Y} r^{-1} e^{-r^2/2}\right).
  \end{align}
Analogously, 
\[\int_{-\infty}^{-r}\int_{\frac{-\varepsilon - x \sigma \sqrt{h}}{h^{1/Y}}}^{\frac{\varepsilon - x \sigma \sqrt{h}}{h^{1/Y}}} |{u^{2k}} p_S(u) \phi(x) | du\, dx= O\left((\varepsilon h^{-1/Y})^{2k-Y} r^{-1} e^{-r^2/2}\right).\]
 Thus, based on  \eqref{e:i1,i2}, \eqref{e:I3_expansion}, and  \eqref{e:I4_asymptotics},
\begin{align}
\label{eq:s2}
\bE\Big(S_{h}^{ 2k}\,&{\bf 1}_{\{|\sigma W_{h}+S_{h}|\leq\varepsilon\}}\Big) \\
 \nonumber
= &\frac{\bar{C}}{2k-Y} h \varepsilon^{2k-Y} + \frac{\bar{C}(2k-1-Y)}{2} \sigma^2 {h^{2}\,\varepsilon^{2k-Y-2}}+ O\left(h^3\varepsilon^{2k-4-Y}\right)\\
 \nonumber
&\quad  + O\big(h^{2}\,\varepsilon^{2k-2Y}\big) +  O\left(\varepsilon^{2k-1-Y}h^{3/2}e^{-\varepsilon^2/(2\sigma^2 h)}\right).
\end{align}
This completes the case $\gamma_0=0$.  For $\gamma_0 \neq 0$, take $h>0$ small enough so that $|\gamma_0|h<\varepsilon$, and observe %
\begin{align*}
&\left|{\bf 1}_{\{|\sigma W_{h}+S_h + h\gamma_0|\leq\varepsilon\}}-{\bf 1}_{\{|\sigma W_{h}+S_h |\leq\varepsilon\}}\right| \\
& \qquad\leq {\bf 1}_{\{\varepsilon- h|\gamma_0|<|\sigma W_{h}+S_h | \leq  \varepsilon\}}+{\bf 1}_{ \{\varepsilon<|\sigma W_{h}+S_h | \leq  \varepsilon + h|\gamma_0|\}}\\
& \qquad\leq 2 {\bf 1}_{\{\varepsilon- h|\gamma_0|<|\sigma W_{h}+S_h | \leq   \varepsilon+h|\gamma_0|\}}.
\end{align*}
 Thus,
\begin{align*}
 \wt\bE&\left(S_h^{2k}\,\left|{\bf 1}_{\{|\sigma W_{h}+S_h + h\gamma_0|\leq\varepsilon\}}-{\bf 1}_{\{|\sigma W_{h}+S_h |\leq\varepsilon\}}\right|\right) \\
 \ \ \ & \leq  2 h^{2k/Y} \wt\bE\left(S_1^{2k}\, {\bf 1}_{\{\varepsilon- h|\gamma_0|<|\sigma W_{h}+S_h | \leq   \varepsilon+h|\gamma_0|\}}\right) \\
 \ \ \  & \leq  K h^{2k/Y}\int_\bR\bigg\{  \int_{b(x)-\gamma_0h^{1-1/Y}}^{b(x)+\gamma_0h^{1-1/Y}} +  \int_{-b(x)-\gamma_0h^{1-1/Y}}^{-b(x)+\gamma_0h^{1-1/Y}}\bigg\}u^{2k} p_S(u) du \phi(x) dx \\
 \ \ \   & \leq  K h^{2k/Y} \int_\bR \int_{b(x)-\gamma_0h^{1-1/Y}}^{b(x)+\gamma_0h^{1-1/Y}}|u|^{2k-1-Y}  du \phi(x) dx \\
 \ \ \     & \leq  K h^{2} \varepsilon^{2k-Y-1} \int_\bR \Big(1  + \frac{x}{r} +  o(1)\Big)^{2k-Y-1} \phi(x) dx\\
  \ \ \    &= O(h^{2} \varepsilon^{2k-Y-1}) = o(h^{2}\,\varepsilon^{2k-2Y}).
 \end{align*}
 Hence,
\begin{align*}
 \wt\bE(&S_h^{2k} {\bf 1}_{\{|\sigma W_{h}+S_h + h\gamma_0|\leq\varepsilon\}}) \\
 &= \wt\bE(S_h^{2k} {\bf 1}_{\{| \sigma W_{h}+S_h |\leq\varepsilon\}})+ \wt\bE\left(S_h^{2k}\,\left({\bf 1}_{\{|\sigma W_{h}+S_h + h\gamma_0|\leq\varepsilon\}}-{\bf 1}_{\{|\sigma W_{h}+S_h |\leq\varepsilon\}}\right)\right)\\
 & =  \wt\bE(S_h^{2k} {\bf 1}_{\{| \sigma W_{h}+S_h|\leq\varepsilon\}}) + o(h^{2}\,\varepsilon^{2k-2Y}).
 \end{align*}
In view of \eqref{eq:s2}, this implies the result for arbitrary $\gamma_0\in\bR$.
\end{proof}

 The lemma below gives an expansion for the truncated higher-order moments of the pure-jump tempered stable L\'evy process $J$.
This lemma is similar to a result used in Step 2 of the proof of Theorem 3.1 of \cite{gong2021}; here we provide a sharper result and generalize it to an arbitrary order $k\geq 1$. 
\begin{lemma} \label{lemma:J2k}
Suppose $Y\in (0,1)\cup (1,2)$. Let $k\geq 1$ be an integer.  Then,  for any $\sigma\in[0,\infty)$, as $h\to 0$ with $\sqrt{h}\ll \varepsilon\ll 1$,
\begin{align}\label{e:J2k}
 \bE\left(J_h^{2k}\,{\bf 1}_{\{|\sigma W_{h}+J_{h}|\leq\varepsilon\}}\right)
  &= \frac{C_{+}\!+\!C_{-}}{2k-Y}h\varepsilon^{2k-Y} + O\big(h^2\varepsilon^{2k-Y-2}\big) + O\big(h\varepsilon^{2k-Y/2}\big) ,
\end{align}
and when $\sigma=0$, for every $p\geq 1$ (not necessarily integer),
\begin{align}\label{e:Jp}
 \bE\left(|J_h|^{p}\,{\bf 1}_{\{|J_{h}|\leq\varepsilon\}}\right) \leq  K h\varepsilon^{p-Y}. 
\end{align}
\end{lemma}
\begin{proof}
We begin with \eqref{e:J2k}. It will be evident from the proof below that the statement remains valid for the case $\sigma=0$. However, for brevity, we only give the details for the case $\sigma>0$, which is more involved. Recall $J_h = S_h + \wt\gamma h$, where $S_h$ is $Y$-stable under $\widetilde \bP$. %
Note that
\begin{align}
\nonumber
\bE\Big(J_{h}^{2k}\,&{\bf 1}_{\{|\sigma W_{h}+J_{h}|\leq\varepsilon\}}\Big)=\wt{\bE}\Big(e^{-\wt{U}_{h}-\eta h}\,J_{h}^{2k}\,{\bf 1}_{\{|\sigma W_{h}+J_{h}|\leq\varepsilon\}}\Big)\\
&=e^{-\eta h}\,\sum_{j=0}^{{2k}}\frac{(2k)!}{j!\,(2k-j)!}(\wt\gamma h)^{2k-j}\, \wt{\bE}\Big(e^{-\wt{U}_{h}}S_h^{j}{\bf 1}_{\{|W_{h}+S_h+\wt{\gamma}h|\leq\varepsilon\}}\Big).
\nonumber
\end{align}
Now, denoting $\wt{C}_{\ell}=\int_{\bR_{0}}\big(e^{-\ell\varphi(x)}-1+\ell\varphi(x)\big)\tilde{\nu}(dx)$, $\ell=1,2$, by {Assumption \ref{assump:FuntqStrong}-(v)} we have
\begin{equation} \label{eq:EwtU}
\wt{\bE}\left(\big(e^{-\wt{U}_{h}}-1\big)^{2}\right)=e^{\wt{C}_{2}h}-2e^{\wt{C}_{1}h}+1\sim\big(\wt{C}_{2}-2\wt{C}_{1}\big)h,\quad\text{as }\,h\rightarrow 0.
\end{equation}
Therefore,  %
\begin{align}
\nonumber
\wt{\bE}\Big(e^{-\wt{U}_{h}}&S_{h}^{j}{\bf 1}_{\{|W_{h}+S_h+\wt{\gamma}h|\leq\varepsilon\}}\Big)\\
\nonumber
&=\wt{\bE}\Big(S_h^{j}\,{\bf 1}_{\{|\sigma W_{h}+S_h+\wt{\gamma}h|\leq\varepsilon\}}\Big)+\wt{\bE}\left(\big(e^{-\wt{U}_{h}}-1\big)S_h^{j}\,{\bf 1}_{\{|\sigma W_{h}+S_h+\wt{\gamma}h|\leq\varepsilon\}}\right)\\
\nonumber
&\leq \bigg(\wt{\bE}\Big(S_h^{2j}\,{\bf 1}_{\{|\sigma W_{h}+S_h+\wt{\gamma}h|\leq\varepsilon\}}\Big)\bigg)^{1/2}\left(1+\Big(\wt{\bE}\big(e^{-\wt{U}_{h}}-1\big)^{2}\Big)^{1/2}\right)\\
&=\bigg(\wt{\bE}\Big(S_h^{2j}\,{\bf 1}_{\{|\sigma W_{h}+S_h+\wt{\gamma}h|\leq\varepsilon\}}\Big)\bigg)^{1/2} \left( 1+ O\left(h^{1/2}\right)\right),\quad h\to 0. \label{e:E-tilde_exp(-U)S_h}%
\end{align}
In particular, applying Lemma \ref{lemma:S2k_new} with $\gamma_0=\wt\gamma$, based on expression \eqref{e:E-tilde_exp(-U)S_h}, for any $0\leq j \leq 2k-1$,  we see that
\begin{align*}
 h^{2k-j}\wt{\bE}&\Big(e^{-\wt{U}_{h}}S_{h}^{j}{\bf 1}_{\{|W_{h}+S_h+\wt{\gamma}h|\leq\varepsilon\}}\Big)\\
 & = h^{2k-j} O\left(\left(\wt{\bE}S_h^{2j}\,{\bf 1}_{\{|\sigma W_{h}+S_h+\wt{\gamma}h|\leq\varepsilon\}}\right)^{1/2}\right)
  =O(h^{3/2}\varepsilon^{2k-1-Y/2} ) .
\end{align*}

\noindent Thus, by Lemma \ref{lemma:S2k_new},%
\begin{align*}
    &\bE\Big(J_{h}^{2k}\,{\bf 1}_{\{|\sigma W_{h}+J_{h}|\leq\varepsilon\}}\Big)\\
    &=e^{-\eta h}\,\sum_{j=0}^{{2k}}\frac{(2k)!}{j!\,(2k-j)!}(\wt\gamma h)^{2k-j}\, \wt{\bE}\Big(e^{-\wt{U}_{h}}S_h^{j}{\bf 1}_{\{|W_{h}+S_h+\wt{\gamma}h|\leq\varepsilon\}}\Big)\\
    &= e^{-\eta h}\,\wt{\bE}\Big(e^{-\wt{U}_{h}}S_h^{2k}{\bf 1}_{\{|W_{h}+S_h+\wt{\gamma}h|\leq\varepsilon\}}\Big) + O(h^{3/2}\varepsilon^{2k-1-Y/2} )\\%
    &=e^{-\eta h}\,\wt{\bE}\Big(S_h^{2k}\,{\bf 1}_{\{|\sigma W_{h}+S_h+\wt{\gamma}h|\leq\varepsilon\}}\Big) \\
    & \quad\qquad+ e^{-\eta h}\,\wt{\bE}\Big(\Big(e^{-\wt{U}_{h}}-1\Big)S_h^{2k}\,{\bf 1}_{\{|\sigma W_{h}+S_h+\wt{\gamma}h|\leq\varepsilon\}}\Big) + O\left(h^{3/2}\varepsilon^{2k-1-Y/2}\right)\\
        &= e^{-\eta h}\wt{\bE}\Big(S_h^{2k}\,{\bf 1}_{\{|\sigma W_{h}+S_h+\wt{\gamma}h|\leq\varepsilon\}}\Big) \\
        & \quad \qquad + O\left(h^{1/2}\varepsilon^{2k-Y/2}\right) \cdot O\left(h^{1/2}\right) + O\left(h^{3/2}\varepsilon^{2k-1-Y/2}\right)\\
    &= \big(1+O(h)\big)\wt{\bE}\Big(S_h^{2k}\,{\bf 1}_{\{|\sigma W_{h}+S_h+\wt{\gamma}h|\leq\varepsilon\}}\Big) + O\left(h\varepsilon^{2k-Y/2}\right),
\end{align*}
where in the last line we used that $h^{3/2}\varepsilon^{2k-1-Y/2}\ll h \varepsilon^{2k-Y/2}$ since $\sqrt h \ll \varepsilon$, and on the second-to-last line we applied Lemma \ref{lemma:S2k_new}.  The statement \eqref{e:J2k} follows, since $h\wt{\bE}\Big(S_h^{2k}\,{\bf 1}_{\{|\sigma W_{h}+S_h+\wt{\gamma}h|\leq\varepsilon\}}\Big) = O(h^2 \varepsilon^{2k-Y}) = o\left(h\varepsilon^{2k-Y/2}\right)$ by Lemma \ref{lemma:S2k_new}.  

The statement \eqref{e:Jp} for $p\geq 2$ follows from the first statement; indeed, when $p=2k+s$ for $0<s<2$, \eqref{e:Jp} follows from the bound $|J_h|^{p}\,{\bf 1}_{\{|J_{h}|\leq\varepsilon\}}\leq \varepsilon^s|J_h|^{2k}\,{\bf 1}_{\{|J_{h}|\leq\varepsilon\}}$. The general case $p \geq 1$ can be established in a manner similar to the proof of Lemma \ref{e:tailbound_J_>eps} and is omitted.
\end{proof}

The next lemma gives a higher-order expansion for the right-tail of a L\'evy process with a stable L\'evy component of unbounded variation. This result is used in the proof of Proposition \ref{lemma:Estable}.%
\begin{lemma}\label{lemma:lambda}
Suppose $Y\in(0,1)\cup(1,2)$, and let
\begin{align*}
    P_n:=\wt \bP\left(\sigma W_h\mp S_h\geq{}\varepsilon\right)=\wt \bE\left(\overline{\Phi}\bigg( \frac{\varepsilon}{\sigma\sqrt{h}}\pm\frac{S_{h}}{\sigma\sqrt{h}}\bigg)\right).
\end{align*}
Then, as $h\to 0$ with $\sqrt{h}\ll \varepsilon\ll 1$,
\begin{align*}
    P_n&= \frac{C_\mp}{Y}h\varepsilon^{-Y} + C_\mp(1+Y)\sigma^2 h^2\varepsilon^{-2-Y} + \\
    &\quad   +\frac{C_\mp (1+Y)(2+Y)(3+Y)}{8} h^3\varepsilon^{-4-Y} + O\left(h^{4} \varepsilon^{-Y-6}\right) \\
    &\quad + O\left(h^{3/2}\varepsilon^{-1-Y}e^{-\frac{\varepsilon^2}{2\sigma^2 h}}\right)\cdot\left(1+h\varepsilon^{-2}\right) +O(h\varepsilon^{2-2Y}) \cdot (1 + h^{\frac{1}{2}} \varepsilon^{-1} + h \varepsilon^{-2})\\%
    &\quad  + O\left(h \varepsilon^{3-2Y-Y/2}\right)+ O\left( h^{1/2}\varepsilon^{-1}e^{-\frac{\varepsilon^2}{2\sigma^2 h}}\right)+ O\left(h^{-1/2} \varepsilon^{3-Y}  e^{-\frac{\varepsilon^2}{8\sigma^2 h}} \right).
\end{align*}
\end{lemma}
\begin{proof}
Fix $0<\bar\delta<1/32$.  Let $ S^{\bar\delta}$ be a L\'evy process with no diffusion part, L\'evy measure $\wt \nu(dz){\bf 1}_{|z|\leq \bar\delta\varepsilon}$, and third component of the characteristic triplet  $\gamma_{\bar\delta}$ ($\gamma_{\bar\delta}$ to be specified momentarily). %
Let $N^{\bar\delta}$ a standard Poisson process with intensity $\lambda_{\bar\delta}:= \wt\nu(\{z:|z|>{\bar\delta}\varepsilon\})$, $(\xi_i^{\bar\delta})_{i\geq 1}$ be a sequence of i.i.d. random variables with probability distribution $\wt\nu(dz){\bf 1}_{|z|> \bar\delta\varepsilon}/\lambda_{\bar\delta}$ (again under both $\wt\bP$ and $\bP$).
We also assume that $S^{\bar\delta}$, $N^{\bar\delta}$, and $(\xi_i^{\bar\delta})_{i\geq 1}$ are independent. Then, by construction, under $\wt\bP$,  we may take $\gamma_{\bar\delta}$ so that $S_h \eqDist S_h^{\bar\delta} + \sum_{i=1}^{N_h^{\bar\delta}} \xi_i^{\bar\delta}$.
Therefore, we can express $\wt \bE(\overline{\Phi}((\varepsilon\pm S_{h})/{\sigma\sqrt{h}})) = \wt \bP\left(\sigma W_h \mp S_h \geq \varepsilon\right)=\bP(\sigma W_h \mp S_h^{\bar\delta} \mp \sum_{i=1}^{N_h^{\bar\delta}} \xi_i^{\bar\delta} \geq \varepsilon)$ and conditioning on $N_{h}^{\bar\delta}$,
\begin{align}
\nonumber
     \wt \bE\left(\overline{\Phi}\bigg(\frac{\varepsilon}{\sigma\sqrt{h}}\pm\frac{S_{h}}{\sigma\sqrt{h}}\bigg)\right)
    &= e^{-\lambda_{\bar\delta} h} \, \bP\left(\sigma W_h \mp S_h^{\bar\delta} \geq \varepsilon\right)\\
    \label{eq:decomp}
    & \quad +\lambda_{\bar\delta} h\, e^{-\lambda_{\bar\delta} h} \, \bP\left(\sigma W_h \mp S_h^{\bar\delta} \mp \xi_1^{\bar\delta}\geq \varepsilon\right)\\
    \nonumber
    & \quad + e^{-\lambda_{\bar\delta} h}\sum_{k\geq 2}\frac{(\lambda_{\bar\delta} h\,)^k}{k!} \, \bP\left(\sigma W_h \mp S_h^{\bar\delta} \mp \sum_{i=1}^k\xi_i^{\bar\delta}\geq \varepsilon\right).
\end{align}
To estimate the first term in (\ref{eq:decomp}), by a Taylor approximation,
\begin{align}
\nonumber
    \bP\big(\sigma W_h \mp S_h^{\bar\delta} \geq \varepsilon\big) &= \bE\left(\overline{\Phi}\left( \frac{\varepsilon}{\sigma\sqrt{h}}\pm\frac{S_{h}^{\bar\delta}}{\sigma\sqrt{h}}\right)\right)\\
    \nonumber
    &= \overline{\Phi}\left(\frac{\varepsilon}{\sigma\sqrt{h}}\right) \pm \overline{\Phi}'\left(\frac{\varepsilon}{\sigma\sqrt{h}}\right)\,\bE \left(\frac{S_h^{\bar\delta}}{\sigma\sqrt{h}}\right) \\
    \nonumber
    &\qquad\qquad + \bE\left(\frac{(S_h^{\bar\delta})^2}{\sigma^2 h}\int_0^1 (1-\theta)  \overline{\Phi} '' \left(\frac{\varepsilon\pm\theta S_h^{\bar\delta}}{\sigma\sqrt{h}}\right) d\theta \right) \\
    \label{e:firstterm_decomp}
    &= \bP\left(\sigma W_h \geq \varepsilon\right) \, \mp\,\phi\left(\frac{\varepsilon}{\sigma\sqrt{h}}\right)\frac{\gamma_{\bar\delta} h}{\sigma \sqrt{h}} \\
    \nonumber
    & \quad + \frac{ h^{-3/2}}{\sigma^{3} (2\pi)^{1/2}} \, \bE\left((S_h^{\bar\delta})^2 \int_0^1 (1-\theta)\, e^{-\frac{(\varepsilon \pm \theta S_h^{\bar\delta})^2}{2\sigma^2 h}}  (\varepsilon \pm\theta S_h^{\bar\delta}) d\theta\right).
\end{align}
To estimate the last term in \eqref{e:firstterm_decomp}, first note that for any integer $k\geq 2$,
\begin{align*}
    \bE\left((S_h^{\bar\delta})^k\right)\sim h\int_{|z|\leq {\bar\delta} \varepsilon} z^k\wt\nu(dz) = O(h\varepsilon^{k-Y}).
\end{align*}

Since $\bar \delta<1/32$, by similar arguments to those shown in expression (46) of \cite{mijatovic2016new}, it holds that
    $\bP\left(|S_h^{\bar\delta}| > \frac{\varepsilon}{2}\right) \leq C_0(h\varepsilon^{-Y})^4$,
    for some constant $C_0$.
Thus, we can compute
\begin{align}
\nonumber
     \bE\Big((S_h^{\bar\delta})^2 \int_0^1 &(1-\theta)\, e^{-\frac{(\varepsilon\pm\theta S_h^{\bar\delta})^2}{2\sigma^2 h}} d\theta\Big) \\
     \nonumber
    &= \bE\left((S_h^{\bar\delta})^2\, {\bf 1}_{\{|S_h^{\bar\delta}| \leq \frac{\varepsilon}{2}\}} \int_0^1 (1-\theta)\, e^{-\frac{(\varepsilon\pm\theta S_h^{\bar\delta})^2}{2\sigma^2 h}} d\theta\right) \\
    \nonumber
    & \qquad\qquad\quad+ \bE\left((S_h^{\bar\delta})^2\,{\bf 1}_{\{
    |S_h^{\bar\delta}| > \frac{\varepsilon}{2}\}} \int_0^1 (1-\theta)\, e^{-\frac{(\varepsilon\pm\theta S_h^{\bar\delta})^2}{2\sigma^2 h}} d\theta\right)\\
    \nonumber
    & \leq  \bE\left((S_h^{\bar\delta})^2\right) \, e^{-\frac{\varepsilon^2}{8\sigma^2 h}} + \bE\left((S_h^{\bar\delta})^2\,{\bf 1}_{\{|S_h^{\bar\delta}| > \frac{\varepsilon}{2}\}} \right)\\
    \nonumber
    &\leq  O\left(h \varepsilon^{2-Y}  e^{-\frac{\varepsilon^2}{8\sigma^2 h}} \right) + \bE\left((S_h^{\bar\delta})^4\right)^{1/2}  \bP\left(|S_h^{\bar\delta}| > \frac{\varepsilon}{2}\right)^{1/2}\\
    &=  O\left(h \varepsilon^{2-Y}  e^{-\frac{\varepsilon^2}{8\sigma^2 h}} \right) + O\left(h^{5/2} \varepsilon^{2-2Y-Y/2}\right). \label{e:firstterm_remainder_term1}
\end{align}
Similarly, 
\begin{align}
\nonumber
     \bE\Big((S_h^{\bar\delta})^3 \int_0^1 &(1-\theta)\theta e^{-\frac{(\varepsilon\pm\theta S_h^{\bar\delta})^2}{2\sigma^2 h}} d\theta\Big)\\
     \nonumber
    & \leq  \bE\left((S_h^{\bar\delta})^3\right) \, e^{-\frac{\varepsilon^2}{8\sigma^2 h}} + \bE\left((S_h^{\bar\delta})^3\,{\bf 1}_{\{|S_h^{\bar\delta}| > \frac{\varepsilon}{2}\}} \right)\\
    &=  O\left(h \varepsilon^{3-Y}  e^{-\frac{\varepsilon^2}{8\sigma^2 h}} \right) + O\left(h^{5/2} \varepsilon^{3-2Y-Y/2}\right). \label{e:firstterm_remainder_term2}
\end{align}
Then, based on \eqref{e:firstterm_decomp}, using $\bP\left(\sigma W_h \geq \varepsilon\right) \pm \phi\left(\frac{\varepsilon}{\sigma\sqrt{h}}\right)\frac{\gamma_{\bar\delta} h}{\sigma \sqrt{h}} = O(h^{1/2}\varepsilon^{-1}e^{-\frac{\varepsilon^2}{2\sigma^2 h}})$ together with  \eqref{e:firstterm_remainder_term1} and  \eqref{e:firstterm_remainder_term2}, we obtain
\begin{align}
    \label{eq:lambda1}
    &\bP\big(\sigma W_h \mp S_h^{\bar\delta} \geq \varepsilon\big)\\
    \nonumber
    & = O\left( h^{1/2}\varepsilon^{-1}e^{-\frac{\varepsilon^2}{2\sigma^2 h}}\right) + O\left(h^{-1/2} \varepsilon^{3-Y}  e^{-\frac{\varepsilon^2}{8\sigma^2 h}} \right) + O\left(h \varepsilon^{3-2Y-Y/2}\right).
\end{align}
This establishes the asymptotic behavior of the first term in \eqref{eq:decomp}. For the second term in (\ref{eq:decomp}), again by a Taylor approximation,
\begin{align}
\label{e:secondterm_decomp1a}
    \bP&\big(\sigma W_h \mp S_h^{\bar\delta} \mp\xi_1^{\bar\delta} \geq \varepsilon\big) \\
    \nonumber
    &= \bP\left(\sigma W_h  \mp \xi_1^{\bar\delta} \geq \varepsilon\right) \mp \bE\left(\phi\left(\frac{\varepsilon \pm \xi_1^{\bar\delta}}{\sigma\sqrt{h}}\right)\right) \frac{\gamma_{\bar\delta} h}{\sigma \sqrt{h}} \\
    \nonumber
    & \quad +  \sigma^{-3} h^{-3/2} (2\pi)^{-1/2} \, \bE\left((S_h^{\bar\delta})^2 \int_0^1 (1-\theta)\, e^{-\frac{(\varepsilon \pm \xi_1^{\bar\delta} \pm\theta S_h^{\bar\delta})^2}{2\sigma^2 h}}  (\varepsilon \pm \xi_1^{\bar\delta} \pm\theta S_h^{\bar\delta}) d\theta\right)\\
    &= \bP\left(\sigma W_h  \mp \xi_1^{\bar\delta} \geq \varepsilon\right) + O\left(h^{1/2} \varepsilon^{1-Y}\right) + O\left( \varepsilon^{2-Y} \right), \label{e:secondterm_decomp1}
\end{align}
where the last equality is because
\begin{align*}
    \bE\Big((S_h^{\bar\delta})^2 &\int_0^1 (1-\theta)\, e^{-\frac{(\varepsilon \mp \xi_1^{\bar\delta} \mp\theta S_h^{\bar\delta})^2}{2\sigma^2 h}}  {|\varepsilon \mp \xi_1^{\bar\delta} \mp\theta S_h^{\bar\delta}|} d\theta\Big)\\
    & \leq\sigma \sqrt 2\sup_{x} |x| e^{ -x^2}\cdot \sqrt h \bE \left((S_h^{\bar\delta})^2\right)
    = O(h^{3/2}\varepsilon^{2-Y}).%
\end{align*}
For the first term in \eqref{e:secondterm_decomp1}, note
\begin{align}
\nonumber
    \lambda_{\bar\delta} \bP\big(&\sigma W_h  + \xi_1^{\bar\delta} \geq \varepsilon\big) \\
    \nonumber
    &= \int_{-\infty}^{-{\bar\delta} \varepsilon} \bP(\sigma W_h \geq \varepsilon -z) \wt \nu (dz) + \int_{{\bar\delta}\varepsilon}^\infty \bP(\sigma W_h \geq \varepsilon -z) \wt \nu (dz)\\
    &\qquad - \int_\varepsilon^\infty {\bf 1}_{\{|z|> {\bar\delta} \varepsilon\}}\wt \nu (dz) + \lambda_{\bar\delta} \bP\left( \xi_1^{\bar\delta} \geq \varepsilon\right).\label{eq:probdecomp}
\end{align}
Now, for the first term in (\ref{eq:probdecomp}), we employ the tail bound $\bP(|W_1|>x)\leq \frac{1}{x\sqrt{2\pi}}\exp\{-x^2/2\}$:
\begin{align}
\nonumber
    \int_{-\infty}^{-{\bar\delta} \varepsilon} \bP(\sigma W_h \geq \varepsilon -z) \wt \nu (dz)   
    &\leq \frac{\sigma \sqrt{h}}{\varepsilon\sqrt{2 \pi}} e^{-\frac{\varepsilon^2}{2\sigma^2 h}} \cdot \frac{C_-}{Y}({\bar\delta} \varepsilon)^{-Y}\\
    &= O\left(h^{1/2}\varepsilon^{-1-Y}e^{-\frac{\varepsilon^2}{2\sigma^2 h}}\right). \label{eq:probdecomp_0}
\end{align}
Then we estimate the second and the third term in (\ref{eq:probdecomp}),
\begin{align}
\nonumber
    \int_{{\bar\delta} \varepsilon}^\infty& \bP(\sigma W_h \geq \varepsilon -z) \wt\nu(dz)- \int_\varepsilon^\infty \wt\nu(dz)\\
    \nonumber
    &= \int_{{\bar\delta} \varepsilon}^\varepsilon \bP(\sigma W_h \geq \varepsilon -z) \wt\nu(dz) - \int_\varepsilon^\infty \bP(\sigma W_h \leq \varepsilon -z) \wt\nu(dz) \\
    \nonumber
    &= C_+\int_{0}^{\varepsilon (1-{\bar\delta})} \bP(\sigma W_h \geq u) \, (\varepsilon-u)^{-1-Y} du \\
    \nonumber
    & \qquad \quad- C_+\int_0^\infty \bP(\sigma W_h \leq -u) \, (\varepsilon+u)^{-1-Y} du\\
    \nonumber
    &= C_+\int_{0}^{\varepsilon (1-{\bar\delta})} \bP(\sigma W_h \geq u) \, \big((\varepsilon-u)^{-1-Y} - (\varepsilon+u)^{-1-Y}\big)du  \\
    \nonumber
    & \qquad\quad - C_+\int_{\varepsilon (1-{\bar\delta})}^\infty \bP(\sigma W_h \leq -u) \, (\varepsilon+u)^{-1-Y} du\\
    &=: T_1-T_2.\label{e:secondterm_decomp}
\end{align}
For the first term $T_1$ in \eqref{e:secondterm_decomp}, for simplicity let $r= r(h)= \frac{\varepsilon}{\sigma \sqrt h}\to \infty$ as $h\to 0$, and also let $f(u) =(1- u)^{-1-Y} - (1+u)^{-1-Y}$.  Also, recall $\bE\left(|W_1|^{k+1}\right) =2k \int_0^\infty  x^k \bP\left(W_1>x\right)\,dx$.  By a 4th order Taylor approximation of $f(u)$ at $u=0$,
\begin{align}
\nonumber
    T_1&= C_+\sigma h^{1/2} \varepsilon^{-1-Y}\int_0^{r(1-{\bar\delta})} \bP( W_1 \geq x) f(x/r) dx\\
    \nonumber
    & = C_+\sigma h^{1/2} \varepsilon^{-1-Y}\int_0^{r(1-{\bar\delta})} \bP( W_1 \geq x) \Big(\sum_{k=0}^4 \frac{f^{(k)}(0) x^k}{r^k\cdot k!} 
    + \frac{f^{(5)}(\vartheta(x)) x^5}{r^5\cdot5!} \Big) dx\\
                \nonumber
        & = C_+\sigma h^{1/2} \varepsilon^{-1-Y} \Bigg(\frac{f'(0) \bE W_1^2 }{r}   + \frac{f^{(3)}(0)  \bE W_1^4 }{r^3\cdot3!} + O(r^{-5})\\
        \nonumber
        &\qquad\qquad\qquad\qquad\qquad\qquad\qquad\qquad - O\Big(\frac{1}{r}\int_{r(1-{\bar\delta})}^\infty \bP( W_1 \geq x) dx \Big) \Bigg)\\
                    \nonumber
            \nonumber  
    &= C_+(1+Y)\sigma^2 h \varepsilon^{-2-Y} + \frac{C_+(1+Y)(2+Y)(3+Y)}{8}  \sigma^4 h^2\varepsilon^{-4-Y}\\
    &\qquad \qquad \qquad + O\left(h^{3} \varepsilon^{-Y-6}\right)- O\left(h^{3/2} \varepsilon^{-3-Y} e^{-\frac{\varepsilon^2(1-\delta)^2}{2\sigma^2 h}}\right). \label{e:secondterm_decomp_term1}
\end{align}
The second term in \eqref{e:secondterm_decomp} can be estimated as
\begin{align}
\nonumber
|T_2|
    & \leq C_+ \, \bP(\sigma W_h \leq \varepsilon ({\bar\delta} -1)) \int_{\varepsilon (1-{\bar\delta})}^\infty  \, (\varepsilon+u)^{-1-Y} du\\
    &= O\left(h^{1/2} \varepsilon^{-1-Y} e^{-\frac{\varepsilon^2(1-\delta)^2}{2\sigma^2 h}}\right). \label{e:secondterm_decomp_term2}
\end{align}
Finally, for the fourth term in (\ref{eq:probdecomp}),
\begin{align}
    \lambda_{\bar\delta} \bP\left( \xi_1^{\bar\delta} \geq \varepsilon\right) = \int_{ \varepsilon}^\infty \wt\nu(dz) = \frac{C_+}{Y}\varepsilon^{-Y} \label{e:fourthterm}
\end{align}
Then, combining \eqref{eq:probdecomp_0}, \eqref{e:secondterm_decomp_term1},  \eqref{e:secondterm_decomp_term2}, and \eqref{e:fourthterm}, and by repeating analogous arguments for $ \lambda_{\bar\delta} \bP\left(\sigma W_h  - \xi_1^{\bar\delta} \geq \varepsilon\right) $, based on \eqref{eq:probdecomp} we obtain%
\begin{align*}
    &\lambda_{\bar\delta} \bP\big(\sigma W_h  \mp \xi_1^{\bar\delta} \geq \varepsilon\big)\\
     &= \frac{C_\mp}{Y}\varepsilon^{-Y} + C_\mp(1+Y)\sigma^2 h \varepsilon^{-2-Y} + \frac{C_\mp(1+Y)(2+Y)(3+Y)}{8} \sigma^4 h^2\varepsilon^{-4-Y}\\
    &\qquad\quad + O\left(h^{3} \varepsilon^{-Y-6}\right) +  O\left(h^{1/2}\varepsilon^{-1-Y}e^{-\frac{\varepsilon^2}{2\sigma^2 h}}\right)\cdot\left(1+h\varepsilon^{-2}\right).
\end{align*}
Therefore, by \eqref{e:secondterm_decomp1}, and using that $\lambda_{\bar\delta} =O(\varepsilon^{-Y})$,
\begin{align}\label{eq:lambda2}
     \lambda_{\bar\delta} &h\, e^{-\lambda_{\bar\delta} h} \, \bP\left(\sigma W_h \mp S_h^{\bar\delta} \mp \xi_1^{\bar\delta}\geq \varepsilon\right)\\ %
     \nonumber
    &= \frac{C_\mp}{Y}h\varepsilon^{-Y} + C_\mp(1+Y)\sigma^2 h^2\varepsilon^{-2-Y} \\
    \nonumber
    & \quad+ \frac{C_\mp (1+Y)(2+Y)(3+Y)}{8} \sigma^4 h^3\varepsilon^{-4-Y} + O\left(h^{4} \varepsilon^{-Y-6}\right)\\
    \nonumber
    & \quad + O\left(h^{3/2}\varepsilon^{-1-Y}e^{-\frac{\varepsilon^2}{2\sigma^2 h}}\right)\cdot\left(1+h\varepsilon^{-2}\right) + {O(h\varepsilon^{2-2Y}) \cdot (1 + h^{1/2} \varepsilon^{-1} )}.
\end{align}
This establishes the asymptotic behavior of the second term in \eqref{eq:decomp}.  For the third term in \eqref{eq:decomp}, simply note
\begin{equation}\label{eq:lambda3}
e^{-\lambda_{\bar\delta} h}\sum_{k\geq 2}\frac{(\lambda_{\bar\delta} h\,)^k}{k!} \, \bP\left(\sigma W_h \mp S_h^{\bar\delta} \mp \sum_{i=1}^k\xi_i^{\bar\delta}\geq \varepsilon\right)\leq \sum_{k\geq 2}\frac{(\lambda_{\bar\delta} h\,)^k}{k!} = O(h^2\varepsilon^{-2Y}).
\end{equation}
Finally, combining (\ref{eq:lambda1}), (\ref{eq:lambda2}) and (\ref{eq:lambda3}), based on \eqref{eq:decomp} we obtain the result.
\end{proof}

The following lemma is used in the proof of Proposition \ref{prop:EX2}.
\begin{lemma} \label{lemma:WJ} Suppose $Y\in(0,1)\cup(1,2)$.
Assume as $h\to 0$, $\varepsilon \to 0$ with $\varepsilon \gg  h^{1/2-s}$ for some $s>0$.
Let $a,b\geq 1$ be integers. Then,
\begin{align*}
     \bE\big(W_h^{a}\,J_h^{b}&\,{\bf 1}_{\{|\sigma W_{h}+J_{h}|\leq\varepsilon\}}\big)\\
    &\qquad=  \begin{cases}
    O\left(h^{2} \varepsilon^{a+b-Y-2}\right) + O\big(h^{1+a/2} \varepsilon^{b-Y/2}\big), & a \text{ odd},\\
    O\left(h^{2} \varepsilon^{a+b-Y-2}\right)+O(h^{1+a/2}\varepsilon^{b-Y}), & a,b \text{ even},\\
        O\left(h^{2} \varepsilon^{a+b-Y-2}\right)  +  O(h^{\frac{1+a}{2}}\varepsilon^{b-Y/2}),& a \text{ even}, b \text{ odd}.
    \end{cases}
    \end{align*}
\end{lemma}
\begin{proof}
First, we decompose $\bE\left(W_h^{a}\,J_h^{b}\,{\bf 1}_{\{|\sigma W_{h}+J_{h}|\leq\varepsilon\}}\right)$ as
\begin{align}\nonumber
\bE\big(&W_h^{a}\,J_h^{b}\,{\bf 1}_{\{|\sigma W_{h}+J_{h}|\leq\varepsilon\}}\big) \\
\nonumber
&= h^{a/2} e^{-\eta h}\Big[\wt{\bE}\Big(W_{1}^{a}\,J_h^{b}\,{\bf 1}_{\{|\sigma W_{h}+J_{h}|\leq\varepsilon\}}+\big(e^{-\wt{U}_{h}}-1\big) \, W_{1}^{a}\,J_h^{b}\,{\bf 1}_{\{|\sigma W_{h}+J_{h}|\leq\varepsilon\}}\Big)\Big]\\
&\qquad=: h^{a/2} e^{-\eta h}\big(T_{1}(h) + T_{2}(h)\big).\label{eq:DecompWJ} 
\end{align}
We first analyze the asymptotic behavior of $T_1(h)$.
Note that by repeating integration by parts, we have for all $x_1, x_2 \in \bR$, and $x_1<x_2$,
\begin{align}
 \label{e:W_moments_formula}
    \wt\bE\big(W_1^{a} &\,{\bf 1}_{\{x_1<W_1<x_2\}}\big) \\
    &= \begin{cases}
      \sum_{j=0}^{d} \frac{(a-1)!!}{(2j)!!} \left(x_1^{2j} \phi(x_1) - x_2^{2j} \phi(x_2)\right),&  a = 2d+1, \\
     \begin{array}{l}
     \sum_{j=0}^{d-1} \frac{(a-1)!!}{(2j+1)!!} \left(x_1^{2j+1} \phi(x_1) - x_2^{2j+1} \phi(x_2)\right)\\
     \quad+ (a-1)!!\,\left( \overline{\Phi}(x_1) - \overline{\Phi}(x_2)\right) ,
     \end{array}&  a = 2d,
      \end{cases}
\nonumber
\end{align}
where $d\geq 0$ is an integer.
However, by applying Lemma \ref{lemma:4results} and noting that under the assumption $\varepsilon \gg  h^{1/2-s}$,  $\sqrt{h}\varepsilon^j e^{-\frac{\varepsilon ^2}{2\sigma^2 h}} \ll h^{\frac{3}{2}} \varepsilon^{j-Y-1} $ and $\sqrt{h}\varepsilon^j e^{-s h^{\frac{Y-2}{Y}}}\ll h^{\frac{3}{2}} \varepsilon^{j-Y-1} $, for any integer $j$, %
\begin{align}
\nonumber
\wt\bE\Big((\mp J_h)^b &\left(\frac{\varepsilon \mp J_h}{\sigma\sqrt{h}}\right)^{k} \phi \left(\frac{\varepsilon \pm J_h}{\sigma\sqrt{h}}\right) \Big)\\
\nonumber
& = (\sigma\sqrt{h})^{-k}\wt\bE\left( \sum_{m=0}^{k} \binom{k}{m} \varepsilon^{k-m}(\mp J_h)^{m+b} \, \phi \left(\frac{\varepsilon \pm J_h}{\sigma\sqrt{h}}\right) \right)\\
\nonumber
    &=  \left(\frac{\varepsilon}{\sigma\sqrt{h}}\right)^k  \sum_{m=0}^{k} \binom{k}{m} \varepsilon^{-m}\, O\left(h^{3/2} \varepsilon^{m+b-Y-1}\right)\\
    &= O\left(  h^{\frac{3-k}{2}} \varepsilon^{k+b-Y-1}\right). \label{e:moments_J_eq1}
\end{align}
Moreover, using Lemma \ref{lemma:J2k}, we see that
\begin{align}
\nonumber
&\wt\bE\left(J_h^b \, \overline{\Phi}\left(\frac{-\varepsilon - J_h}{\sigma\sqrt{h}}\right)\right) - \wt\bE\left(J_h^b\, \overline{\Phi}\left(\frac{\varepsilon - J_h}{\sigma\sqrt{h}}\right)\right)\\
    &\quad= \wt\bE\left(J_h^b \, {\bf 1}_{\{|\sigma W_{h}+J_{h}|\leq\varepsilon\}}\right) = \begin{cases}
    O(h\varepsilon^{b-Y}), & b \text{ even},\\
     O(h^{1/2}\varepsilon^{b-Y/2}),& b \text{ odd}.
    \end{cases}\label{e:moments_J_eq2}
\end{align}
 where we used Cauchy-Schwarz for the case when $b$ is odd. Thus, by conditioning on $J_h$, we may use the bounds \eqref{e:moments_J_eq1} and \eqref{e:moments_J_eq2} in expresssion \eqref{e:W_moments_formula} to obtain %
\begin{equation}
    T_1(h) =  \begin{cases}
    O\left(h^{2-a/2} \varepsilon^{a+b-Y-2}\right),  & a \text{ odd},\\
    O\left(h^{2-a/2} \varepsilon^{a+b-Y-2}\right)+O(h\varepsilon^{b-Y}), & a,b \text{ even},\\
        O\left(h^{2-a/2} \varepsilon^{a+b-Y-2}\right)  +  O(h^{1/2}\varepsilon^{b-Y/2}),& a \text{ even}, b \text{ odd}.
    \end{cases} \label{e:eq_T1}
\end{equation}
Now we turn to $T_2(h)$. Recall that by \eqref{eq:EwtU}, $\wt{\bE}(e^{-\wt{U}_{h}}-1\big)^{2}=O(h)$ as $h\rightarrow 0$. 
Then by Cauchy-Schwarz, %
\begin{align} 
\nonumber
    T_2(h) &= \wt{\bE}\Big(\big(e^{-\wt{U}_{h}}-1\big) \, W_{1}^{a}\,J_h^{b}\,{\bf 1}_{\{|\sigma\sqrt{h}W_{1}+J_{h}|\leq\varepsilon\}}\Big)\\
    \nonumber
    &\leq \left( \wt{\bE}\left(W^{2a}_1\big(e^{-\wt{U}_{h}}-1\big)^2\,\right)  \right)^{1/2} \cdot \left( \wt\bE \left(J_h^{2b}\,{\bf 1}_{\{|\sigma\sqrt{h}W_{1}+J_{h}|\leq\varepsilon\}}\right) \right)^{1/2}\\
    \nonumber
    &\leq K \left( \wt{\bE} \big(e^{-\wt{U}_{h}}-1\big)^2 \right)^{1/2} \left(O\big(h\varepsilon^{2b-Y}\big)\right)^{1/2}\\
    &=  O\big(h \varepsilon^{b-Y/2}\big), \label{e:eq_T2}
\end{align}
where we used the independence of $W_1$ and $U_t$ in the third line above.
Therefore, since $h \varepsilon^{b-Y/2} \ll h\varepsilon^{b-Y} $ and $h \varepsilon^{b-Y/2} \ll h^{1/2}\varepsilon^{b-Y/2}$, by \eqref{eq:DecompWJ}, \eqref{e:eq_T1} and \eqref{e:eq_T2}, the result holds.
\end{proof}

\begin{proof}[Proof of \eqref{ImprovEstI652}]
Let us start with the  decomposition
\begin{align*}
	I_{5,2}&=\wt{\bE}\Big(\Big(e^{-\wt{U}_{h}}-1+\wt{U}_{h}\Big)S_{h}^{2}\,{\bf 1}_{\{|\sigma W_{h}+S_{h}+\wt{\gamma}h|\leq\varepsilon\}}\Big)\\
	& \qquad \qquad -\wt{\bE}\Big(\wt{U}_{h}S_{h}^{2}\,{\bf 1}_{\{|\sigma W_{h}+S_{h}+\wt{\gamma}h|\leq\varepsilon\}}\Big)\\
	&=:T_1-T_2.
\end{align*}
Since $|\sigma W_{h}+S_{h}+\wt{\gamma}h|\leq\varepsilon$ implies that $|S_h|<\varepsilon+\sigma|W_h|+|\wt{\gamma}|h$ and $e^{-x}-1+x\geq{}0$ for all $x$, we can estimate the first term as follows:
\begin{align}\label{DBNH0}
0\leq{}T_1&\leq{}\wt{\bE}\Big(\left(e^{-\wt{U}_{h}}-1+\wt{U}_{h}\right)\left(\varepsilon+\sigma|W_h|+|\widetilde{\gamma}| h\right)^{2}\Big)\\
\nonumber
	&=O(\varepsilon^{2})\wt{\bE}\Big(e^{-\wt{U}_{h}}-1+\wt{U}_{h}\Big)=O(\varepsilon^{2}h).
\end{align}
For $T_2$, we consider the following decomposition
\begin{align*}
\wt{U}_{h}&=\int_{0}^{h}\int_{\bR_{0}}\big(\varphi(x)+\alpha_{\text{sgn}(x)}x{\bf 1}_{\{|x|\leq 1\}}\big)\wt{N}(ds,dx)\\
&\qquad \qquad -\int_{0}^{h}\int_{0<|x|<1}\alpha_{\text{sgn}(x)}x\wt{N}(ds,dx)\\
&=:\wt{U}^{\text{BV}}_{h}-S^{(0)}_t,
\end{align*}
where the first integral is well-defined in light of Assumption \ref{assump:Funtq}-(i) \& (ii). Then, we have 
\begin{align}
	\label{IANTH}
	T_2&=\wt{\bE}\Big(\wt{U}^{BV}_{h}S_{h}^{2}\,{\bf 1}_{\{|\sigma W_{h}+S_{h}+\wt{\gamma}h|\leq\varepsilon\}}\Big)- \wt{\bE}\Big(S_h^{(0)}S_{h}^{2}\,{\bf 1}_{\{|\sigma W_{h}+S_{h}+\wt{\gamma}h|\leq\varepsilon\}} \Big).
\end{align}
For the first term we proceed as in \eqref{DBNH0}:
\begin{align}
	\wt{\bE}\Big(&|\wt{U}^{BV}_{h}|S_{h}^{2}\,{\bf 1}_{\{|\sigma W_{h}+S_{h}+\wt{\gamma}h|\leq\varepsilon\}}\Big)	\nonumber\\
	&\leq{}
	\wt{\bE}\Big(|\wt{U}^{BV}_{h}| \left(\varepsilon+\sigma|W_h|+|\widetilde{\gamma}| h\right)^{2}\Big)
	=O(\varepsilon^2)\wt{\bE}\Big(|\wt{U}^{BV}_{h}| \Big)
	=O(\varepsilon^2 h) \label{e:T2_term1}
\end{align}
since $\wt{\bE}\Big(\big|\wt{U}^{\text{BV}}_{h}\big|\Big)\leq 2h\int_{\bR_{0}}\big|\varphi(x)+\alpha_{\text{sgn}(x)}x\big|\,\wt{\nu}(dx)<\infty$ due to Assumption \ref{assump:Funtq}.
For the expectations in \eqref{IANTH}, we apply H\"older's inequality to obtain
\begin{align}\label{HInq}
	\wt{\bE}\Big(|S_h^{(0)}|S_{h}^{2}\,&{\bf 1}_{\{|\sigma W_{h}+S_{h}+\wt{\gamma}h|\leq\varepsilon\}}\Big)\\
	& \leq \bigg(\wt{\bE}\left|S_h^{(0)}\right|^p\bigg)^{1/p} \bigg({\wt \bE } \Big(S_{h}^{2}\,{\bf 1}_{\{|\sigma W_{h}+S_{h}+\wt{\gamma}h|\leq\varepsilon\}}\Big)^q\bigg)^{1/q},
	\nonumber
\end{align}
where $p,q>1$ are such that $p^{-1} + q^{-1}=1$. %
However, by Lemma \ref{lemma:S2k_new},
\begin{align*}
\bigg({\wt \bE }\Big(S_{h}^{2}\,{\bf 1}_{\{|\sigma W_{h}+S_{h}+\wt{\gamma}h|\leq\varepsilon\}}\Big)^q\bigg)^{1/q} &=  \Big(O \big (h\varepsilon^{2q-Y} \big)\Big)^{1/q}= O\big(h^{1/q} \varepsilon^{2-Y/q} \big). %
\end{align*}
Therefore, plugging the above estimate into expression \eqref{HInq}, and since  $\bE \big|S_h^{(0)}\big|^p\leq C h$, we get
\begin{align*}
\wt{\bE}\Big(|S_h^{(0)}|S_{h}^{2}\,{\bf 1}_{\{|\sigma W_{h}+S_{h}+\wt{\gamma}h|\leq\varepsilon\}}\Big) &\leq
O(h^{1/p}) \times O\big(h^{1/q} \varepsilon^{2-Y/q} )\\
\quad  &= O(h\varepsilon^{2-Y/q}).
\end{align*}
Therefore,  for all small $\bar\delta>0$, by taking $q$ large enough,
\begin{align} 
\wt{\bE}\Big(|S_h^{(0)}|S_{h}^{2}\,{\bf 1}_{\{|\sigma W_{h}+S_{h}+\wt{\gamma}h|\leq\varepsilon\}}\Big)&=O(h\varepsilon^{2-\bar\delta}).
\label{e:E(SpmS^2)}
\end{align}
{Finally, combining \eqref{DBNH0}, \eqref{e:T2_term1}, and \eqref{e:E(SpmS^2)},} we get that 
\begin{align} \label{e:I_{52}_improved}
	I_{5,2}=O(\varepsilon^{2}h)+O(h\varepsilon^{2-\bar\delta})=O(h\varepsilon^{2-\bar\delta}).
\end{align}
\end{proof}

The following lemma is used in the proof of Lemma \ref{lemma:S2k_new}.
\begin{lemma} \label{K0Is0}
	Suppose $Y\in(1,2)$, and let $\mathcal{D}(u)$ be as in \eqref{e:C_term1}.
	Then, 
\begin{align}\label{K0Is0Prf}
	\int_{-\infty}^\infty u^2 \mathcal{D}(u) du =0.
\end{align}
\end{lemma}	
\begin{proof}
{Let $ \varphi_S(\omega)$ denote the characteristic function of $S_1$ under $\widetilde \bP$.  Then, with $\varsigma$ and $\rho$ denoting the scale and skewness parameters of $S_1$, respectively, we have}
\begin{align*}
\varphi_S(\omega)   &= \exp\Big(- |\varsigma \omega|^Y\big[1-i\rho\hspace{0.5mm}\sgn(\omega)\tan(\pi Y/2)\big]\Big)\\
 &=\exp\!\Bigg(\!-(C_{+}\!+C_{-})\,\bigg|\!\cos\bigg(\!\frac{\pi Y}{2}\!\bigg)\!\bigg|\Gamma(-Y)|\omega|^{Y} \\
 & \qquad \qquad \qquad\qquad\bigg[\!1 - i\frac{C_{+}\!-C_{-}}{C_{+}\!+C_{-}}\tan\!\bigg(\!\frac{\pi Y}{2}\!\bigg)\text{sgn}(\omega)\!\bigg]\!\Bigg).%
\end{align*}
Without loss of generality, suppose $\varsigma=1.$ 
For simplicity, let $\mathcal E(u):= u^2\mathcal D(u)$, and observe $\mathcal E \in L^1(\bR)$ since $ Y\in(1,2)$, and hence $\hat {\mathcal E}(\omega) := \int_\bR e^{i\omega u}\mathcal E(u)du$ exists.
 Now, {recall that, for any $\alpha\in(0,1)$,}
\[ \int_0^\infty e^{ix} x^{\alpha-1} dx = e^{i\pi \alpha /2} \Gamma(\alpha), \]
{(see \cite{gradshteyn:ryzhik:2007},  eq.~3.381.7)}. 
 Thus,
\begin{equation}\label{e:ft_x^(1-Y)}
\int_0^\infty e^{iu\omega}  u^{1-Y}du = - |\omega|^{Y-2} e^{-\text{sgn}(\omega)i\pi Y/2}\Gamma(2-Y), \quad Y\in(1,2).
\end{equation}
This implies
\begin{align*}
 C_+ &\int_{0}^\infty e^{ix\omega}|x| ^{1-Y}dx + C_-\int_{-\infty}^0 e^{ix\omega }|x| ^{1-Y}dx \\
    &= -|\omega|^{Y-2}\, \Gamma(2-Y) \left(C_+ e^{-i \,\text{sgn}(\omega)\pi Y/2}+ C_- e^{i\,\text{sgn}(\omega)\pi Y/2}\right)\\
    &= |\omega|^{Y-2} (C_{+}\!+C_{-})\,\Big|\cos\bigg(\!\frac{\pi Y}{2}\!\bigg) \Big| \Gamma(2-Y)\\
    & \qquad\qquad\qquad \times \Big[\!1 - i\frac{C_{+}\!-C_{-}}{C_{+}\!+C_{-}}\tan\!\bigg(\!\frac{\pi Y}{2}\!\bigg)\text{sgn}(\omega)\!\Big]\\
    & = |\omega|^{Y-2} Y(Y-1) \Big(1-i\rho\hspace{0.5mm}\sgn(\omega)\tan(\pi Y/2)\Big).%
\end{align*}
However, it can be shown that%
$$
 \int_\bR e^{iu\omega} u^2p_S(u)du = - \varphi_S''(\omega), \quad  \omega \neq 0,
$$
where the integral is interpreted as the limit {$\lim_{R\to\infty}\int_{-R}^R e^{iu\omega}u^2p_S(u)du.$}  Thus, by differentiating {$-\varphi_S(\omega)$} twice, we obtain
\begin{align*}
 \int_\bR e^{iu\omega} &u^2p_S(u)du \\
 &= \Big(Y(Y-1)|\omega|^{Y-2} - Y^2 |\omega|^{2Y-2}\Big)\Big(1-i\rho\hspace{0.5mm}\sgn(\omega)\tan(\pi Y/2)\Big)\varphi_S(\omega).%
 \end{align*}
Therefore, as $\omega \to 0,$
$$
    \hat {\mathcal E}(\omega) =  -\hat p''(\omega) + |\omega|^{Y-2} Y(Y-1) \left(1-i\rho\hspace{0.5mm}\sgn(\omega)\tan(\pi Y/2)\right)= O\left(|\omega|^{2Y-2}\right).
$$
In particular, $\int_\bR\mathcal E(x) dx =\hat{\mathcal{E}}(0)=0,\quad Y\in (1,2).$
\end{proof}

\begin{proof}[Proof of \eqref{ImprovEstI652}]
Let us start with the  decomposition
\begin{align*}
	I_{5,2}&=\wt{\bE}\Big(\Big(e^{-\wt{U}_{h}}-1+\wt{U}_{h}\Big)S_{h}^{2}\,{\bf 1}_{\{|\sigma W_{h}+S_{h}+\wt{\gamma}h|\leq\varepsilon\}}\Big)\\
	& \qquad \qquad -\wt{\bE}\Big(\wt{U}_{h}S_{h}^{2}\,{\bf 1}_{\{|\sigma W_{h}+S_{h}+\wt{\gamma}h|\leq\varepsilon\}}\Big)\\
	&=:T_1-T_2.
\end{align*}
Since $|\sigma W_{h}+S_{h}+\wt{\gamma}h|\leq\varepsilon$ implies that $|S_h|<\varepsilon+\sigma|W_h|+|\wt{\gamma}|h$ and $e^{-x}-1+x\geq{}0$ for all $x$, we can estimate the first term as follows:
\begin{align}\label{DBNH0}
0\leq{}T_1&\leq{}\wt{\bE}\Big(\left(e^{-\wt{U}_{h}}-1+\wt{U}_{h}\right)\left(\varepsilon+\sigma|W_h|+|\widetilde{\gamma}| h\right)^{2}\Big)\\
\nonumber
	&=O(\varepsilon^{2})\wt{\bE}\Big(e^{-\wt{U}_{h}}-1+\wt{U}_{h}\Big)=O(\varepsilon^{2}h).
\end{align}
For $T_2$, we consider the following decomposition
\begin{align*}
\wt{U}_{h}&=\int_{0}^{h}\int_{\bR_{0}}\big(\varphi(x)+\alpha_{\text{sgn}(x)}x{\bf 1}_{\{|x|\leq 1\}}\big)\wt{N}(ds,dx)\\
&\qquad \qquad -\int_{0}^{h}\int_{0<|x|<1}\alpha_{\text{sgn}(x)}x\wt{N}(ds,dx)\\
&=:\wt{U}^{\text{BV}}_{h}-S^{(0)}_t,
\end{align*}
where the first integral is well-defined in light of Assumption \ref{assump:Funtq}-(i) \& (ii). Then, we have 
\begin{align}
	\label{IANTH}
	T_2&=\wt{\bE}\Big(\wt{U}^{BV}_{h}S_{h}^{2}\,{\bf 1}_{\{|\sigma W_{h}+S_{h}+\wt{\gamma}h|\leq\varepsilon\}}\Big)- \wt{\bE}\Big(S_h^{(0)}S_{h}^{2}\,{\bf 1}_{\{|\sigma W_{h}+S_{h}+\wt{\gamma}h|\leq\varepsilon\}} \Big).
\end{align}
For the first term we proceed as in \eqref{DBNH0}:
\begin{align}
	\wt{\bE}\Big(&|\wt{U}^{BV}_{h}|S_{h}^{2}\,{\bf 1}_{\{|\sigma W_{h}+S_{h}+\wt{\gamma}h|\leq\varepsilon\}}\Big)	\nonumber\\
	&\leq{}
	\wt{\bE}\Big(|\wt{U}^{BV}_{h}| \left(\varepsilon+\sigma|W_h|+|\widetilde{\gamma}| h\right)^{2}\Big)
	=O(\varepsilon^2)\wt{\bE}\Big(|\wt{U}^{BV}_{h}| \Big)
	=O(\varepsilon^2 h) \label{e:T2_term1}
\end{align}
since $\wt{\bE}\Big(\big|\wt{U}^{\text{BV}}_{h}\big|\Big)\leq 2h\int_{\bR_{0}}\big|\varphi(x)+\alpha_{\text{sgn}(x)}x\big|\,\wt{\nu}(dx)<\infty$ due to Assumption \ref{assump:Funtq}.
For the expectations in \eqref{IANTH}, we apply H\"older's inequality to obtain
\begin{align}\label{HInq}
	\wt{\bE}\Big(|S_h^{(0)}|S_{h}^{2}\,&{\bf 1}_{\{|\sigma W_{h}+S_{h}+\wt{\gamma}h|\leq\varepsilon\}}\Big)\\
	& \leq \bigg(\wt{\bE}\left|S_h^{(0)}\right|^p\bigg)^{1/p} \bigg({\wt \bE } \Big(S_{h}^{2}\,{\bf 1}_{\{|\sigma W_{h}+S_{h}+\wt{\gamma}h|\leq\varepsilon\}}\Big)^q\bigg)^{1/q},
	\nonumber
\end{align}
where $p,q>1$ are such that $p^{-1} + q^{-1}=1$. %
However, by Lemma \ref{lemma:S2k_new},
\begin{align*}
\bigg({\wt \bE }\Big(S_{h}^{2}\,{\bf 1}_{\{|\sigma W_{h}+S_{h}+\wt{\gamma}h|\leq\varepsilon\}}\Big)^q\bigg)^{1/q} &=  \Big(O \big (h\varepsilon^{2q-Y} \big)\Big)^{1/q}= O\big(h^{1/q} \varepsilon^{2-Y/q} \big). %
\end{align*}
Therefore, plugging the above estimate into expression \eqref{HInq}, and since  $\bE \big|S_h^{(0)}\big|^p\leq C h$, we get
\begin{align*}
\wt{\bE}\Big(|S_h^{(0)}|S_{h}^{2}\,{\bf 1}_{\{|\sigma W_{h}+S_{h}+\wt{\gamma}h|\leq\varepsilon\}}\Big) &\leq
O(h^{1/p}) \times O\big(h^{1/q} \varepsilon^{2-Y/q} )\\
\quad  &= O(h\varepsilon^{2-Y/q}).
\end{align*}
Therefore,  for all small $\bar\delta>0$, by taking $q$ large enough,
\begin{align} 
\wt{\bE}\Big(|S_h^{(0)}|S_{h}^{2}\,{\bf 1}_{\{|\sigma W_{h}+S_{h}+\wt{\gamma}h|\leq\varepsilon\}}\Big)&=O(h\varepsilon^{2-\bar\delta}).
\label{e:E(SpmS^2)}
\end{align}
{Finally, combining \eqref{DBNH0}, \eqref{e:T2_term1}, and \eqref{e:E(SpmS^2)},} we get that 
\begin{align} \label{e:I_{52}_improved}
	I_{5,2}=O(\varepsilon^{2}h)+O(h\varepsilon^{2-\bar\delta})=O(h\varepsilon^{2-\bar\delta}).
\end{align}
\end{proof}

\begin{proof}[Proof of (\ref{SINHN1})]
First, we note that 
\begin{align*}
	\varepsilon^{\frac{Y-4}{2}}\sum_{i=1}^{n}\mathbb{E}_{i-1}\Big[\Delta_{i}^nJ^2&{\bf 1}_{\{|\chi_{t_{i-1}}\Delta_{i}^{n}J|>4\zeta\varepsilon\}}\\
	& |f_n(\Delta_i^nX)-f_n(\chi_{t_{i-1}}\Delta_i^nJ)||\Delta_i^nM|\Big]=o_P(1),
\end{align*}
since, recalling notation \eqref{NtnEXs00}-\eqref{NtnEXs}, 
\begin{align*}
	{\bf 1}_{\{|\chi_{t_{i-1}}\Delta_{i}^{n}J|>4\zeta\varepsilon\}}|f_n(\Delta_i^nX)-f_n(\chi_{t_{i-1}}\Delta_i^nJ)|&={\bf 1}_{\{|\chi_{t_{i-1}}\Delta_{i}^{n}J|>4\zeta\varepsilon\}}f_n(\Delta_i^nX)\\
	&\leq{}
{\bf 1}_{\{|b_{t_{i-1}}h+\sigma_{t_{i-1}}\Delta_{i}^{n}W+\mathcal{E}_{i}|>3\zeta\varepsilon\}},
\end{align*}
and Markov's inequality together with $\varepsilon\gg h^{\beta}$ and \eqref{MDEE00} for $r$ large enough will easily show that
\begin{align*}
	\varepsilon^{\frac{Y-4}{2}}\sum_{i=1}^{n}\mathbb{E}_{i-1}\left[\Delta_{i}^nJ^2({\bf 1}_{\{|b_{t_{i-1}}h>\zeta\varepsilon\}}
+
{\bf 1}_{\{|\sigma_{t_{i-1}}\Delta_{i}^{n}W|>\zeta\varepsilon\}}+
{\bf 1}_{\{|\mathcal{E}_{i}|>\zeta\varepsilon\}})\right]=o_P(1)
\end{align*}
So, we only need to prove that 
\begin{align*}
	\varepsilon^{\frac{Y-4}{2}}\sum_{i=1}^{n}\mathbb{E}_{i-1}\Big[\Delta_{i}^nJ^2&{\bf 1}_{\{|\chi_{t_{i-1}}\Delta_{i}^{n}J|\leq{}4\zeta\varepsilon\}}\\
	& |f_n(\Delta_i^nX)-f_n(\chi_{t_{i-1}}\Delta_i^nJ)||\Delta_i^nM|\Big]=o_P(1).
\end{align*}
We will make use of the following bound 
\begin{equation}\label{BndDifOff}
	|f_n(x+y)-f_n(y)|\leq{} {\bf 1}_{\{|x|\geq{}\frac{1}{3}\varepsilon^{1+\eta}\}}
	+{\bf 1}_{\{\varepsilon\leq{}|y|\leq{}\varepsilon(1+\varepsilon^{\eta})\;{\text or }\;\zeta\varepsilon(1-\varepsilon^{\eta})\leq{}|y|\leq{}\zeta\varepsilon\}}\frac{|x|}{\varepsilon^{1+\eta}},
\end{equation}
with $y=\chi_{t_{i-1}}\Delta_i^n J$ and $x=\Delta_i^nX-\chi_{t_{i-1}}\Delta_i^n J=b_{t_{i-1}}h+\sigma_{t_{i-1}}\Delta_i^nW+\mathcal{E}_{i}=:z_{i,1}+z_{i,2}+z_{i,3}$. It is easy to check that 
$$
	A_{\ell}:=\varepsilon^{\frac{Y-4}{2}}\sum_{i=1}^{n}\mathbb{E}_{i-1}\left[\Delta_{i}^nJ^2{\bf 1}_{\{|\chi_{t_{i-1}}\Delta_{i}^{n}J|\leq{}4\zeta\varepsilon\}}{\bf 1}_{\{|z_{i,\ell}|\geq{}\frac{1}{9}\varepsilon^{1+\eta}\}}|\Delta_i^nM|\right]=o_P(1),
$$
for $\ell=1,2,3;$ for instance, consider $\ell=2$. Assuming for simplicity $\chi_{t_{i-1}}=\sigma_{t_{i-1}}=1$, 
\begin{align*}
	A_{2}\leq\varepsilon^{\frac{Y-4}{2}}\sum_{i=1}^{n}\mathbb{E}_{i-1}&\left[\Delta_{i}^nJ^{8}{\bf 1}_{\{|\Delta_{i}^{n}J|\leq{}4\zeta\varepsilon\}}\right]^{\frac{1}{4}}\\
	&\times\mathbb{P}_{i-1}\left[|\Delta_i^n W|\geq{}\frac{1}{9}\varepsilon^{1+\eta}\right]^{\frac{1}{4}}\mathbb{E}_{i-1}\left[|\Delta_i^nM|^{2}\right]^{\frac{1}{2}}.
\end{align*}
For any $\alpha>0$, we have $\mathbb{P}_{i-1}\left[|\Delta_i^n W|\geq{}\frac{1}{9}\varepsilon^{1+\eta}\right]\leq{}C\mathbb{E}_{i-1}\left[|\Delta_i^n W|^{\alpha}\right]/\varepsilon^{\alpha(1+\eta)}\leq{}Ch^{\alpha/2}/\varepsilon^{\alpha(1+\eta)}$. Therefore, 
\begin{align*}
	A_{2}&\leq C\varepsilon^{\frac{Y-4}{2}}h^{\frac{1}{4}}\varepsilon^{\frac{8-Y}{4}}h^{\frac{\alpha}{8}}\varepsilon^{-\frac{\alpha(1+\eta)}{4}}\sum_{i=1}^{n}\mathbb{E}_{i-1}\left[|\Delta_i^nM|^{2}\right]^{\frac{1}{2}}\\
	&\leq C\varepsilon^{\frac{Y}{4}}h^{-\frac{1}{4}}\big(h^{\frac12}\varepsilon^{-(1+\eta)})^{\frac\alpha4}O_P(1).
	\end{align*}
Recalling that  $\varepsilon\gg h^\beta$, for small enough $\eta>0$, $\frac{1}{2}-\beta(1+\eta)>0$, hence taking $\alpha$ large enough gives  $A_2=o_P(1).$
 {The argument for $A_1$ is similar}. For $A_3$,
\begin{align*}
	A_{3}&\leq\varepsilon^{\frac{Y-4}{2}}\sum_{i=1}^{n}\mathbb{E}_{i-1}\left[\Delta_{i}^nJ^{8}{\bf 1}_{\{|\Delta_{i}^{n}J|\leq{}4\zeta\varepsilon\}}\right]^{\frac{1}{4}}\frac{\mathbb{E}_{i-1}\left[|\mathcal{E}_i|^{\alpha}\right]^{\frac{1}{4}}}{\varepsilon^{\frac{\alpha(1+\eta)}{4}}}\mathbb{E}_{i-1}\left[|\Delta_i^nM|^{2}\right]^{\frac{1}{2}}\\
	&\leq{}C\varepsilon^{\frac{Y-4}{2}}h^{\frac{1}{4}}\varepsilon^{\frac{8-Y}{4}}h^{\frac{1+\alpha/2}{4}}\varepsilon^{-\frac{\alpha(1+\eta)}{4}}\sum_{i=1}^{n}\mathbb{E}_{i-1}\left[|\Delta_i^nM|^{2}\right]^{\frac{1}{2}}\\
	&\leq\varepsilon^{\frac{Y}{4}}h^{\frac{\alpha}{8}}\varepsilon^{-\frac{\alpha(1+\eta)}{4}}O_P(1).
\end{align*}
proceeding as in the case of $A_2$, the above tends to 0 provided $\alpha$ is large enough and $\eta$ is sufficiently small.
 Next, we analyze the second term in \eqref{BndDifOff}. We consider only $\{\varepsilon\leq{}|y|\leq{}\varepsilon(1+\varepsilon^{\eta})\}$ (the other case is similar). In other words, we need to show that 
$$
	B_{\ell}:=\varepsilon^{\frac{Y-4}{2}-1-\eta}\sum_{i=1}^{n}\mathbb{E}_{i-1}\left[\Delta_{i}^nJ^2{\bf 1}_{\{\varepsilon\leq{}|\chi_{t_{i-1}}\Delta_{i}^{n}J|\leq{}\varepsilon(1+\varepsilon^\eta)\}}|z_{i,\ell}||\Delta_i^nM|\right]=o_P(1),
$$
for $\ell=1,2,3.$  Assume for simplicity of notation that $\chi_{t_{i-1}}=\sigma_{t_{i-1}}=1$. 
By Cauchy's inequality and Lemma \ref{lemma:J2k}, for any $p,q>0$ such that $\frac{1}{p}+\frac{1}{q}=\frac{1}{2}$,
\begin{align*}
	B_{2}\leq C \varepsilon^{\frac{Y-4}{2}-1-\eta}\sum_{i=1}^{n}\mathbb{E}_{i-1}\Big[\Delta_{i}^nJ^{2p}&{\bf 1}_{\{\varepsilon\leq{}|\Delta_{i}^{n}J|\leq{}\varepsilon(1+\varepsilon^\eta)\}}\Big]^{\frac{1}{p}}\\
	& \mathbb{E}_{i-1}\left[|\Delta_{i}W|^{q}\right]^{\frac{1}{q}}
\mathbb{E}_{i-1}\left[|\Delta_i^nM|^2\right]^{\frac{1}{2}}
\end{align*}
\begin{align*}
&\leq C
\varepsilon^{\frac{Y-4}{2}-1-\eta}\left(h\varepsilon^{2p-Y+\eta}\right)^{\frac{1}{p}}h^{\frac{1}{2}}\sum_{i=1}^{n}\mathbb{E}_{i-1}\left[|\Delta_i^nM|^2\right]^{\frac{1}{2}}\\
&=\varepsilon^{\frac{Y-4}{2}-1-\eta}\left(h\varepsilon^{2p-Y+\eta}\right)^{\frac{1}{p}}
O_P(1)
\end{align*}
(c.f.~\eqref{CInPON}.) Taking $p$ close to 2 (so $\frac{1}{p}=\frac{1}{2}-s$ for $s>0$ small), the last expression is of order {$\varepsilon^{\frac{Y-4}{2}-1-\eta}h^{\frac{1}{2}-s}\varepsilon^{2-\frac{Y}{2}+\frac{\eta}{2}+s(Y-\eta) }=h^{\frac{1}{2}-s}\varepsilon^{\frac{\eta}{2}-1+ s(Y-\eta)}\leq h^{\frac{1}{2}-s}\varepsilon^{\frac{\eta}{2}-1}$}. Using that $\varepsilon\gg h^\beta$, 
\[
	h^{\frac{1}{2}-s}\varepsilon^{\frac{\eta}{2}-1}
	\ll h^{\frac{1}{2}-s-(1-\frac{\eta}{2})\beta}\ll 1,
\]
if $\eta$ {and $s$ are} small enough.  {Similar arguments give $B_1=o_P(1)$ and $B_3=o_P(1)$, which establishes \eqref{SINHN1}.}
\end{proof}

\bibliographystyle{elsarticle-num-names}

\end{document}